\documentclass[12pt,dvipsnames]{article}

\usepackage[normalem]{ulem}

\usepackage{tools/arxiv}
\usepackage[colorinlistoftodos]{todonotes}
\usepackage{booktabs} 
\usepackage{nicefrac}  
\usepackage{microtype}
\usepackage[utf8]{inputenc}

\usepackage{enumerate}

\usepackage{xcolor}
\usepackage{multirow}

\usepackage{amssymb}
\usepackage{amsthm}
\usepackage{bbm}
\usepackage{array}
\usepackage[colorlinks=true,urlcolor=purple,
linkcolor=purple,citecolor=purple]{hyperref}
\usepackage{amsmath}
\usepackage{mathtools}
\usepackage{float}

\usepackage[nopostdot,style=super,sort=use]{glossaries}

\usepackage{mathrsfs}
\usepackage{graphicx}
\usepackage{color}

\usepackage{algorithm}
\usepackage[noend]{algpseudocode}
\usepackage{thmtools,thm-restate} %

\makeatletter
\def\paragraph{\@startsection{paragraph}{4}%
  \z@\z@{-\fontdimen2\font}%
  {\normalfont\bfseries}}
\makeatother

\usepackage{tikz}
\usetikzlibrary{decorations.markings, backgrounds, calc, positioning, shapes, shadows, arrows, fit, automata}
\usepackage{tikz-cd}
\usetikzlibrary{arrows.meta}
\tikzset{>={Latex}}

\makeglossaries

\newcommand{\N}{\mathbb{N}}
\newcommand{\Z}{\mathbb{Z}}
\newcommand{\Q}{\mathbb{Q}}
\newcommand{\R}{\mathbb{R}}
\newcommand{\C}{\mathbb{C}}

\newcommand{\mc}{\mathcal}
\newcommand{\mb}{\mathbb}
\newcommand{\mf}{\mathfrak}

\renewcommand{\a}{\alpha}
\renewcommand{\b}{\beta}
\newcommand{\g}{\gamma}
\renewcommand{\d}{\delta}
\newcommand{\e}{\varepsilon}
\newcommand{\w}{\omega}
\newcommand{\s}{\sigma}
\newcommand{\ph}{\varphi}
\newcommand{\p}{\partial}
\renewcommand{\t}{\tau}

\newcommand{\Om}{\Omega}
\newcommand{\Ga}{\Gamma}
\newcommand{\Si}{\Sigma}

\newcounter{desccount}

\newcommand{\descref}[1]{\hyperref[#1]{#1}}

\newcommand{\fst}{f^{*}}
\newcommand{\gst}{g^{*}}

\renewcommand{\P}{\mathbb{P}}

\newcommand{\lp}{\left(}
\newcommand{\rp}{\right)}

\newcommand\mattwo[4]{\left(\begin{smallmatrix}
			{#1} & {#2}\\
     			{#3} & {#4}
                     \end{smallmatrix}\right)}
\newcommand\mattres[9]{\left(\begin{smallmatrix}
			{#1} & {#2}&{#3}\\
     			{#4} & {#5}&{#6}\\
     				{#7} &{#8} &{#9}
                     \end{smallmatrix}\right)}

\newcommand{\borel}{\operatorname{Borel}}	

\newcommand{\set}[1]{\left\{#1\right\}}

\renewcommand{\r}{\rightarrow}
\newcommand{\xr}{\xrightarrow}

\newcommand{\norm}[1]{\|#1\|} 

\renewcommand{\emptyset}{\varnothing}

\newcommand{\supp}{\operatorname{supp}}

\newcommand{\diam}{\operatorname{diam}}

\newcommand{\ds}{\displaystyle}
\newcommand{\inv}{^{-1}}

\newcommand{\id}{\operatorname{id}}

\newcommand{\floor}[1]{\left \lfloor{#1}\right \rfloor }

\newcommand{\Ncal}{\mathcal{F}\mathcal{N}}		%

\newcommand{\Mgen}{\mathcal{M}}

\newcommand{\Ngen}{\mathcal{N}}
\newcommand{\Ncom}{\mathcal{CN}}
\newcommand{\Ndis}{\mathcal{FN}^{\operatorname{dis}}}
\newcommand{\Ncomdis}{\mathcal{CN}^{\operatorname{dis}}}
\newcommand{\Ngendis}{\mathcal{N}^{\operatorname{dis}}}

\newcommand{\Ngenult}{\mathcal{N}^{\operatorname{ult}}}
\newcommand{\Ncomult}{\mathcal{CN}^{\operatorname{ult}}}
\newcommand{\Nult}{\mathcal{FN}^{\operatorname{ult}}}

\newcommand{\Ncomsim}{\mathcal{CN}/{\cong^w}}
\newcommand{\Ncalsim}{\mathcal{FN}/{\cong^w}}

\newcommand{\dis}{\operatorname{dis}}		%
\newcommand{\cost}{\operatorname{cost}}

\newcommand{\dhdf}{d_{\operatorname{H}}}
\newcommand{\dgh}{d_{\operatorname{GH}}}		%
\newcommand{\gh}{\operatorname{GH}}
\newcommand{\dn}{d_{\mathcal{N}}}	%

\newcommand{\dbox}{d_{\Box}} %
\newcommand{\dnh}{\widehat{d}_{\mathcal{N}}}

\renewcommand{\H}{\mc{H}}		%
\newcommand{\M}{\operatorname{M}}		%
\newcommand{\us}{{\mathbb{S}}}
\newcommand{\congs}{\cong^s}
\newcommand{\congwo}{\cong^w_{\tiny \operatorname{I}}}
\newcommand{\congwt}{\cong^w_{\tiny \operatorname{II}}}

\newcommand{\Aut}{\operatorname{Aut}}
\newcommand{\cin}{\operatorname{in}}
\newcommand{\cout}{\operatorname{out}}

\newcommand{\mincout}{\operatorname{m^{\operatorname{out}}}}
\newcommand{\mincin}{\operatorname{m^{\operatorname{in}}}}

\newcommand{\spec}{\operatorname{spec}}
\newcommand{\skel}{\operatorname{sk}}
\newcommand{\sk}{\operatorname{sk}}

\newcommand{\im}{\operatorname{im}}

\newcommand{\Ropt}{R^{\operatorname{opt}}}
\newcommand{\opt}{\operatorname{opt}}
\newcommand{\matele}[1]{({\!}({#1})\!)}  
\newcommand{\Rsc}{\mathscr{R}}		%

\newcommand{\pow}{\operatorname{pow}}

\newcommand{\db}{d_{\operatorname{B}}}
\newcommand{\dgm}{\operatorname{Dgm}}
\newcommand{\vr}{\operatorname{\textbf{VR}}}
\newcommand{\cech}{\check{\operatorname{\textbf{C}}}}

\newcommand{\pvec}{{\operatorname{\mathbf{PVec}}}}

\newcommand{\di}{d_{\operatorname{I}}}

\newcommand{\so}{\operatorname{so}}
\newcommand{\si}{\operatorname{si}}

\newcommand{\hnr}{\mathcal{H}^{\scriptscriptstyle\operatorname{NR}}}
\newcommand{\unr}{u^{\scriptscriptstyle\operatorname{NR}}}
\newcommand{\hr}{\mathcal{H}^{\scriptscriptstyle\operatorname{R}}}
\newcommand{\ur}{u^{\scriptscriptstyle\operatorname{R}}}

\newcommand{\tr}{\operatorname{tr}}

\DeclareMathOperator*{\argmin}{arg\,min}

\newcommand{\wus}{\w_{\vec{\us}^1}}

\newcommand{\rus}{\w_{\mb{S}^1 \!\!,\rho}}
\newcommand{\sr}{\mb{S}^1 \!\!,\rho}
\newcommand{\rusn}{\w_{n,\rho}}

\newcommand{\dow}{\operatorname{\textbf{Dow}}}
\renewcommand{\path}{\operatorname{\textbf{Path}}}
\newcommand{\ot}{\operatorname{\textbf{OT}}}

\newtheorem{theorem}{Theorem}
\numberwithin{theorem}{subsection} %

\newtheorem{proposition}[theorem]{Proposition}
\newtheorem{lemma}[theorem]{Lemma}
\newtheorem{corollary}[theorem]{Corollary}

\theoremstyle{definition}
\newtheorem{example}[theorem]{Example}
\newtheorem{remark}[theorem]{Remark}

\newtheorem{definition}[theorem]{Definition}
\newtheorem{claim}{Claim}

\newenvironment{subproof}[1][\proofname]{%
  \begin{proof}[#1]%
}{%
  \end{proof}%
}

\makeatletter
\newcommand{\pushright}[1]{\ifmeasuring@#1\else\omit\hfill$\displaystyle#1$\fi\ignorespaces}
\newcommand{\pushleft}[1]{\ifmeasuring@#1\else\omit$\displaystyle#1$\hfill\fi\ignorespaces}
\makeatother

\newglossaryentry{Ncal}
{
    name=$\Ncal$,
    description={Finite networks}
}

\newglossaryentry{Ncom}
{
    name=$\Ncom$,
    description={Compact networks}
}

\newglossaryentry{Ngen}
{
    name=$\Ngen$,
    description={General networks}
}

\newglossaryentry{Ngendis}
{
    name=$\Ngendis$,
    description={General dissimilarity networks}
}

\newglossaryentry{dhdf}
{
    name=$\dhdf$,
    description={Hausdorff distance}
}

\newglossaryentry{dgh}
{
    name=$\dgh$,
    description={Gromov-Hausdorff distance}
}

\newglossaryentry{dn}
{
    name=$\dn$,
    description={Network distance}
}

\newglossaryentry{Rsc}
{
    name=$\Rsc$,
    description={Correspondences}
}

\newglossaryentry{dis}
{
    name=$\dis$,
    description={Distortion of a correspondence}
}

\newglossaryentry{Ropt}
{
    name=$\Rsc^{\opt}$,
    description={Optimal correspondences}
}

\newglossaryentry{congs}
{
    name=$\congs$,
    description={Strong isomorphism}
}

\newglossaryentry{congwo}
{
    name=$\congwo$,
    description={Type I weak isomorphism with strict tripods}
}

\newglossaryentry{congwt}
{
    name=$\congwt$,
    description={Type II weak isomorphism with $\e$-tripods}
}

\newglossaryentry{congw}
{
    name=$\cong^w$,
    description={Weak isomorphism in $\Ncom$}
}

\newglossaryentry{Aut}
{
    name=$\Aut$,
    description={Automorphisms of a network}
}

\newglossaryentry{posetx}
{
    name=$\mf{p}(X)$,
    description={Poset of weak isomorphism of $X$}
}

\newglossaryentry{blowupx}
{
    name=$X[\mathbf{k}]$,
    description={Blow-up of $X$}
}

\newglossaryentry{deltax}
{
    name=$\d_X$,
    description={Max-symmetrized Kuratowski metric}
}

\newglossaryentry{deltaxa}
{
    name=$\d_X^A$,
    description={Max-symmetrized Kuratowski metric relative to a subset $A$}
}

\newglossaryentry{skx}
{
    name=$\sk(X)$,
    description={Skeleton of a compact network}
}

\newglossaryentry{dnh}
{
    name=$\dnh$,
    description={The second network distance}
}

\newglossaryentry{rx}
{
    name=$\rho_X$,
    description={Finite reversibility parameter}
}

\newglossaryentry{us}
{
    name=$\us^1$,
    description={Unit circle in $\C$}
}

\newglossaryentry{S}
{
    name=$S^1$,
    description={Unit circle parametrized by the unit interval $[0,1)$}
}

\newglossaryentry{M}
{
    name=$\M_n$,
    description={Motifs of size $n$}
}

\newglossaryentry{wus}
{
    name=$\w_{\us^1}$,
    description={Counterclockwise arc length on $\us^1$}
}

\newglossaryentry{rus}
{
    name=$\rus$,
    description={Counterclockwise arc length with finite reversibility}
}

\newglossaryentry{pow}
{
    name=$\pow$,
    description={Nonempty elements of power set}
}

\newglossaryentry{db}
{
    name=$\db$,
    description={Bottleneck distance between persistence diagrams}
}

\newglossaryentry{di}
{
    name=$\di$,
    description={Interleaving distance between persistent vector spaces}
}

\newglossaryentry{pvec}
{
    name=$\pvec$,
    description={Persistent vector space}
}

\newglossaryentry{vr}
{
    name=$\vr$,
    description={Vietoris-Rips complex}
}

\newglossaryentry{cech}
{
    name=$\cech$,
    description={\v{C}ech complex}
}

\newglossaryentry{dow}
{
    name=$\dow$,
    description={Dowker complex}
}

\newglossaryentry{ot}
{
    name=$\ot$,
    description={Ordered tuple complex}
}

\newglossaryentry{path}
{
    name=$\path$,
    description={Persistent path homology}
}

\newglossaryentry{hnr}
{
    name=$\hnr$,
    description={Hierarchical nonreciprocal clustering}
}

\newglossaryentry{hr}
{
    name=$\hr$,
    description={Hierarchical reciprocal clustering}
}

\newglossaryentry{F}
{
    name=$\mb{F}$,
    description={Ground field $\mb{F}$ fixed throughout}
}

\newglossaryentry{rplus}
{
    name=$\R_+$,
    description={Nonnegative elements of $\R$}
}

\newglossaryentry{zplus}
{
    name=$\Z_+$,
    description={Nonnegative elements of $\Z$}
}

\newglossaryentry{dbox}
{
    name=$d_{\Box}$,
    description={Cut distance between graphs on different node sets}
}

\newglossaryentry{deltabox}
{
    name=$\d_{\Box}$,
    description={Cut distance between graphs on the same node set}
}

\newglossaryentry{supp}
{
    name=$\supp$,
    description={Support of a measure}
}

\newglossaryentry{mfm}
{
    name=$\mf{m}$,
    description={Minimal mass of a positive-measure set in an $\e$-system}
}

\newglossaryentry{mfM}
{
    name=$\mf{M}$,
    description={Minimal mass function optimized over all $\e$-systems}
}

\title{Distances and Isomorphism between Networks:  Stability and Convergence of Network Invariants}

\author{
Samir Chowdhury\\
  Department of Psychiatry and Behavioral Sciences\\
  Stanford University\\
  Stanford, CA 94305 \\
  \texttt{samirc@stanford.edu} \\
\And
Facundo M\'{e}moli\\
  Departments of Mathematics and Computer Science and Engineering\\
  The Ohio State University\\
  Columbus, OH 43210\\
  \texttt{memoli@math.osu.edu} \\
}

\begin{document}
\maketitle

\begin{abstract} 

We develop the theoretical foundations of a generalized Gromov-Hausdorff distance between functions on networks that has recently been applied to various subfields of topological data analysis and optimal transport. 
These functional representations of networks, or networks for short, specialize in the finite setting to (possibly asymmetric) adjacency matrices and derived representations such as distance or kernel matrices.
Existing literature utilizing these constructions cannot, however, benefit from continuous formulations because the continuum limits of finite networks under this distance are not well-understood. For example, while there are currently numerous persistent homology methods on finite networks, it is unclear if these methods produce well-defined persistence diagrams in the infinite setting. We resolve this situation by introducing the collection of compact networks that arises by taking continuum limits of finite networks and developing sampling results showing that this collection admits well-defined persistence diagrams.  
The difference between the network setting and metric setting arises as follows.
For metric spaces, the isomorphism class of the Gromov-Hausdorff distance consists of isometric spaces and is thus very simple. For networks, the isomorphism class is rather complex, and contains representatives having different cardinalities and different topologies. 
We provide an exact characterization of a suitable notion of isomorphism for compact networks as well as alternative, stronger characterizations under additional topological regularity assumptions. 
Toward data applications, we describe a unified framework for developing quantitatively stable network invariants, provide basic examples, and cast existing results on the stability of persistent homology methods in this extended framework. To illustrate our theoretical results, we introduce a model of directed circles with finite reversibility and characterize their Dowker persistence diagrams.

\end{abstract}

\tableofcontents
\setlength{\glsdescwidth}{0.69\textwidth}
\printglossary[title={List of Symbols}]

\section{Introduction}
\label{sec:intro}

\subsection{Background}

Networks which show the relationships within and between complex systems are key tools in a variety of current research areas. Network analysis techniques such as modularity, core-periphery structure, clustering coefficients, and centrality \cite{newman2010networks} have enjoyed enormous success in being used by researchers to derive insights from complex datasets, and this early success has ushered in a variety of approaches from different disciplines toward network data analysis. One such perspective that has appeared recently is to view a network as a \emph{generalized metric space (gms)}. In this viewpoint, a network is a set $X$ along with a real-valued \emph{edge weight function} $\w_X:X\times X \r \R$. When $X$ is finite, the pair $(X,\w_X)$ can be represented as a real-valued square matrix. 
Such a matrix may encode the adjacency matrix of a (possibly directed) combinatorial graph or any real-valued function on pairs of graph vertices derived from the adjacency matrix. Examples of the latter include shortest path length matrices or positive semidefinite matrices such as the heat kernel obtained by exponentiating the graph Laplacian. 

The gms perspective enables the import of numerous tools from the setting of metric spaces to the setting of networks. This program was initiated by the authors of \cite{carlsson2013axiomatic,clust-net}, who showed that single linkage hierarchical clustering could be defined for \emph{directed} networks, and that a certain reformulation of the well-established \emph{Gromov-Hausdorff (GH) distance ($\dgh$)} \cite{gromov1981structures, gromov-book}-- a \emph{network distance}, denoted $\dn$--could be used to show that such a clustering method would be stable with respect to perturbations of the input data \cite{clust-um}. This network distance was defined between pairs of finite networks $(X,\w_X), (Y,\w_Y)$ via a combinatorial optimization problem that would search over correspondences between the points of $X$ and $Y$ as follows (see Definition \ref{defn:dno} for a formal definition as used in this work):
\[\dn(X,Y):= \tfrac{1}{2}\min\bigg\{ \max_{(x,y),(x',y') \in R} |\w_X(x,x') - \w_Y(y,y')| : R \in \{0,1\}^{|X|\times |Y|} \text{with all row, column sums nonzero}\bigg\}.\] 

Beyond its origins in metric geometry \cite{burago}, the Gromov-Hausdorff distance between metric spaces 
has found applications in the context of shape and data analysis \cite{dgh-sgp,dghlp-focm,dgh-props,clust-um,dgh-pers}. 
Thus one expects that an appropriate modification of the Gromov-Hausdorff distance would be a valuable tool in network analysis. Following this thread, the gms viewpoint enabled the development of \emph{persistent homology} methods for directed networks \cite{dowker-asilo, dowker-jact, turner2019rips, pph, dey2020efficient} that came with stability guarantees obtained via $\dn$. However, the theoretical guarantees of these methods were limited by the theoretical study of $\dn$, which has to date been limited to the setting of finite networks. In this work, we remove this bottleneck and provide a comprehensive study of $\dn$ and the theoretical framework surrounding its applications to hierarchical clustering, persistent homology, and machine learning using (not necessarily metric) dissimilarities.   

\subsection{Setup and use cases}

Henceforth we will gradually stop referring to generalized metric spaces in favor of the ``network" terminology, but we stress that this gms perspective will always be implicit.
For the present discussion, the collection of networks is defined as (we will shortly add extra conditions; see Definition \ref{defn:nets} for the formal definition):
\[ \Ngen := \{(X,\w_X) : X \text{ a set, } \w_X : X\times X \r \R \text{ any function} \}. \]
A (possibly directed, weighted) graph $G=(V,E)$ admits a variety of principled embeddings into $\Ngen$, including:
\begin{itemize}
    \item \textbf{Adjacency:} $\w_V$ is defined to be the adjacency function of the graph.
    \item \textbf{Distance:} (for connected graphs/strongly connected digraphs) $\w_V$ is defined to be the geodesic distance function on the graph.
    \item \textbf{Kernel:} $\w_V$ is defined to be any of a number of graph kernels, many of which encode multiscale structure.
\end{itemize}
In this work, we study the metric structure of $(\Ngen,\dn)$ with the goal of producing a common framework for different embeddings into $\Ngen$. 
This naturally opens up a direction of future work on developing new embedding techniques--possibly parametrized and supervised, following \cite{litman2013learning}--that benefit from sharing a common ambient space.

Toward obtaining such characterizations, we first observe that the $\dn$ function defined above can be immediately extended to infinite networks $(X,\w_X),(Y,\w_Y)$ by replacing the $\min,\max$ with $\inf,\sup$. We also observe that $\dn$ only sees $\w_X$, $\w_Y$, and thus it is sensitive to structures on $X,Y$ only if these structures affect $\w_X,\w_Y$. One such structure, and the minimum needed to at least be able to talk about limits and continuity, is a topology. Because topologies on $X,Y$ seem to be a very basic requirement, we redefine $\Ngen$ as follows (cf. Definition \ref{defn:nets}):
\[ \Ngen:= \{(X,\w_X): X \text{ a first countable topological space, } \w_X \text{ continuous w.r.t. product topology}\}.\]

Here, first countability is a mild assumption needed to prove some of our technical results. An interesting point to note is that given $(X,\w_X),(Y,\w_Y)$, there may be numerous first countable topologies on $X,Y$ against which $\w_X,\w_Y$ remain continuous. Referring to these as admissible topologies, we note that $\dn(X,Y)$ is non sensitive to such admissible topological changes, as they do not perturb $\w_X,\w_Y$. Later in this work, we study the effects of constraining these admissible topologies. For now we remark on particular subcollections of $\Ngen$. The following definition of \emph{compact networks} is crucial for subsequent constructions:
\[ \Ncom:= \{(X,\w_X) \in \Ngen : X \text{ a compact, first countable topological space, } \w_X \text{ continuous w.r.t. product topology}\}.\] 
As an example of the difference between $\Ngen$ and $\Ncom$, we show that compact networks admit the following characterization via \emph{tripods}: if $X, Y \in \Ncom$, then $\dn(X,Y) = 0$ if and only if there exists a set $Z$ with surjective maps $\ph_X:Z \r X, \, \ph_Y:Z \r Y$ such that $\w_X(\ph_X(z), \ph_X(z')) = \w_Y(\ph_Y(z),\ph_Y(z'))$ for all $z,z' \in Z$
(see \cite{sturm2012space, memoli2017distance,blumberg2017universality,bauer-tripods-reeb} for related applications of this technique). In particular, this result is not true without some form of the compactness assumption.

From the practical perspective, a first step is to clarify the types of networks that motivate our constructions, and to which our theory can be applied. One primary motivation has come from the domain of neuroimaging \cite{sporns2011networks, sporns2012discovering}, where the technique of building a network out of thresholded correlations between neural systems (cf. Figure \ref{fig:atlas-FC}) has led to great successes in identifying the neural correlates of human behavior. 
The underlying correlation matrices are called \emph{functional connectivity} matrices, and are square, non-metric matrices with values in the interval $[-1,1]$ \cite{friston1994functional} (cf. Figure \ref{fig:applications}). 
It has been shown that embedding such matrices into spaces with non-Euclidean geometry is beneficial for downstream analysis \cite{venkatesh2020comparing}, suggesting that building new spaces with different geometries may yield further insights from the data. 
For search systems in computer vision, utilizing non-metric dissimilarity matrices has been shown to be effective \cite{jegou2008accurate,horster2007image}, and our work characterizes such matrices as members of a ``nice" space $\Ncom$. From these perspectives, our study of $\Ncom$ provides a new geometry in which to embed data and also provides posthoc theoretical justification for empirical methods involving non-metric dissimilarity measures.

\begin{figure}
    \centering
    \begin{tikzpicture}
\begin{scope}[xshift=-4cm,yshift=0.5cm]
\draw[fill] (0,0) circle(0.5ex);
\draw[->,black]  (0:0) -- (0:2cm) node[right]{$a$};
\draw[->,black]  (0:0) -- (90:2cm) node[right]{$c$};
\draw[->,black]  (0:0) -- (45:2cm) node[right]{$b$};
\draw[->,black]  (0:0) -- (225:2cm) node[right]{$d$};
\end{scope}
\begin{scope}
\draw[fill] (0,0) circle(0.5ex) node[below] {$b$};
\draw[fill] (0.75,1) circle(0.5ex) node[above] {$a$};
\draw[fill] (1.25,-0.25) circle(0.5ex) node[below] {$c$};
\draw[fill] (3,0.25) circle(0.5ex) node[right] {$d$};
\draw[-,black,thick] (0,0) -- (0.75,1) node[left, pos = 0.7] {$ \tfrac{1}{\sqrt{2}}$};
\draw[-,black,thick] (0,0) -- (1.25,-0.25) node[below, pos = 0.5] {$ \tfrac{1}{\sqrt{2}}$};
\draw[-,teal,thick] (0.75,1) -- (1.25,-0.25) node[right, pos = 0.5] {$0$};
\draw[-,purple,dashed] (0,0) -- (3,0.25) node[above, pos = 0.6] {$-1$};
\draw[-,purple,thick] (1.25,-0.25) -- (3,0.25) node[below, pos = 0.7] {$- \tfrac{1}{\sqrt{2}}$};
\draw[-,purple,thick] (0.75,1) -- (3,0.25) node[above, pos = 0.7] {$- \tfrac{1}{\sqrt{2}}$};
\end{scope}
\begin{scope}[xshift = 6cm,yshift = 0.5cm]
\matrix [inner sep=0pt, nodes={inner sep=.2em},matrix of math nodes,left delimiter=(,right delimiter=),row sep=0.01cm,column sep=0.01cm] (m) {
1 & \tfrac{1}{\sqrt{2}} & 0 & - \tfrac{1}{\sqrt{2}}\\
 \tfrac{1}{\sqrt{2}} & 1 &  \tfrac{1}{\sqrt{2}} & -1\\
0 &  \tfrac{1}{\sqrt{2}} & 1 & - \tfrac{1}{\sqrt{2}}\\
- \tfrac{1}{\sqrt{2}} & -1 & - \tfrac{1}{\sqrt{2}} & 1\\
};
\end{scope}
\end{tikzpicture}
    \caption{Four unit vectors and their inner products, i.e. their correlations, as well as their representations in the form of a matrix and a directed, weighted graph (self-loops are omitted from the figure). Networks with positive and negative values and zeros in arbitrary positions are the object of study in this work.}
    \label{fig:atlas-FC}
\end{figure}

\begin{figure}
    \centering
    \includegraphics[width=0.8\textwidth]{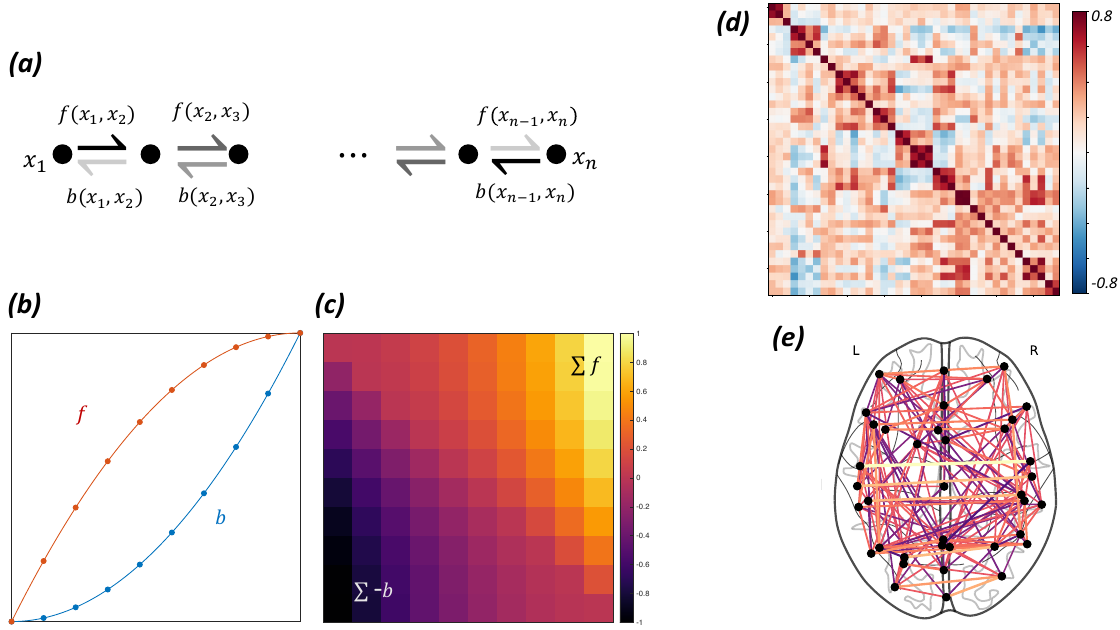}
    \caption{Examples of networks derived from natural systems that violate metric axioms. \textbf{(a):} A simple (unnormalized) Markov chain $X=\{x_1,\ldots, x_n\}$ with asymmetric forward and backward transition probabilities, as used in modeling biological/physical systems at broken detailed balance \cite{hill2005free, perez2020quantifying}. \textbf{(b):} Plot of $\sum_{i=1}^k f(x_i,x_{i+1})$ and $\sum_{i=1}^k b(x_i,x_{i+1})$ for $1\leq k \leq n$. \textbf{(c):} Weight function $\w_X$ defined for $i\leq j$ as $\w_X(x_i,x_j):= \sum_{k=i}^j f(x_k, x_{k+1})$ and $\w_X(x_j,x_i):= \sum_{k=i}^j - b(x_k,x_{k+1})$. 
    \textbf{(d),(e):} Functional connectivity (i.e. correlation) matrix and graph with edges corresponding to positive correlation values in the 80th percentile. Thresholding to obtain a graph is done in practice so that one may apply graph-connectivity tools for posthoc analysis, but the caveat is that there is no consensus on how to threshold weights, or how to deal with negative correlations. Our network framework provides posthoc analysis tools that accept arbitrary matrices without thresholding, thus avoiding potential information loss.}
    \label{fig:applications}
\end{figure}

Raw data in a variety of scientific applications may violate symmetry and/or triangle inequality, and this can be observed in natural phenomena such as human perception \cite{tversky1982similarity}, broken detailed balance \cite{hill2005free}, and discrepancies between time and distance measures \cite{matsumoto1989slope}. 
Exploratory data analysis techniques such as hierarchical clustering and persistent homology can be extended to accept nonmetric data as input. This was observed and successfully explored in \cite{carlsson2013axiomatic,clust-net,carlsson2018hierarchical, dowker-jact, pph, turner2019rips}. These works focused on the setting of finite datasets due to their emphasis on applications, but an underlying question---that was made explicit in \cite{turner2019rips}---was that of characterizing the convergence of these techniques as the sample size grows to infinity. Already in these works and surrounding literature \cite{clust-um, chazal2014persistence}, it was implicit that one way to prove such a result would be to use a stability property of an appropriately reformulated GH distance along with a notion of dense sampling analogous to taking $\e$-nets in a metric space. 
Looking toward the needs of a framework which can work with matrices that are possibly asymmetric and can have both positive and negative entries, we study 
the structure of $\Ncom$ and propose it as a convenient ambient space for the collection of finite networks. 
Studying $\Ncom$ necessitates the departure from spectral techniques and well-grounded structures such as the cone of symmetric positive definite matrices \cite{pennec2006riemannian}, but generates new possibilities via the lenses of applied topology and related machine learning approaches.

In this work we develop foundational tools for the theoretical integration of (possibly directed) networks into applied topology and broader machine learning frameworks. We study the notion of isomorphism under $\dn$ and produce multiple characterizations of isomorphism that hold under a hierarchy of topological regularity assumptions. We explore the limits obtained via $\dn$-convergent sequences, and use these insights to give conditions on networks admitting well-defined \emph{persistence diagrams}, i.e. networks on which persistent homology can be meaningfully applied. Toward grounding hierarchical clustering methods on directed networks, we develop the notion of sampling from networks equipped with probability measures and characterize the continuum limits of these methods. Finally, we derive network invariants that provide lower bounds on $\dn$, and provide algorithms and implementations for comparing networks via these lower bounds. 

\subsection{Main contributions and results}

The main contributions of this work are as follows:
\begin{enumerate}[I]

\item We introduce $(\Ncom,\dn)$ and characterize its metric structure, including multiple perspectives on its isomorphism structure. Among these, we highlight the results that require the minimal amount of setup:
\begin{enumerate}
    \item $(\Ncom,\dn)$ is a complete, geodesic, pseudometric space that includes all square matrices, notably matrices violating metric axioms (Theorems \ref{thm:dN1}, \ref{thm:complete},\ref{thm:geod-inf}).
    \item (Weak) Isomorphism in $(\Ncom,\dn)$ is exactly characterized via tripods (Theorem \ref{thm:weak-isom-cpt}).
\end{enumerate}

\item We extend known persistent homology methods on finite, possibly asymmetric networks to the infinite networks in $\Ncom$ (Corollary \ref{cor:stab}). We also study convergence of hierarchical clustering methods for finite, possibly asymmetric networks on certain classes of infinite\footnote{compact Polish spaces equipped with a bounded, measurable weight function $\w$, cf. Definition \ref{defn:measure-net}} networks equipped with Borel probability measures (Theorems \ref{thm:hnr-consistent}, \ref{thm:hr-consistent}).

\item We define a directed network model---the family of compact directed circles with finite reversibility---and characterize its Dowker persistent homology by linking to existing results \cite{adamaszek2016nerve,adamaszek2017vietoris} (cf. Figure \ref{fig:dowker-circle} and Theorem \ref{thm:dowker-circle}).

\end{enumerate}

\begin{figure}
    \centering
    \includegraphics[width=0.6\textwidth]{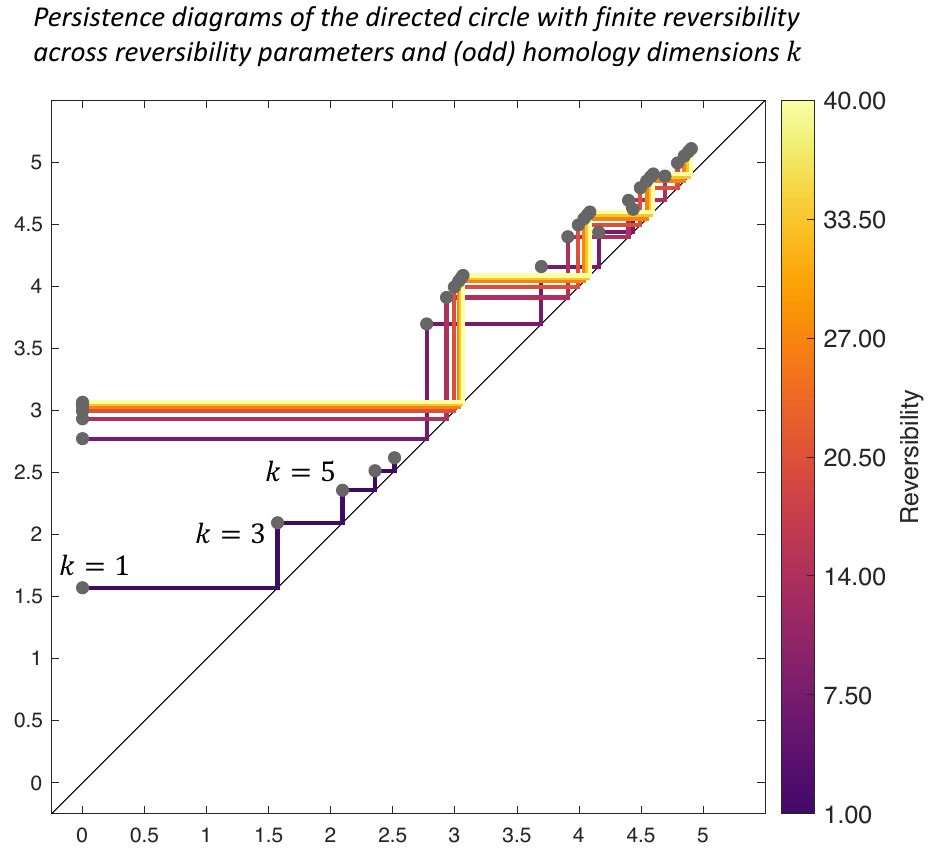}
    \caption{Dowker persistence diagrams of $(\us^1,\rus)$ across different reversibility parameters. Each broken line connects the persistence diagrams in homology dimensions $k=1,3,5,7,9$ for a single reversibility parameter. This illustrates the exact characterization given by Theorem \ref{thm:dowker-circle}.} 
    \label{fig:dowker-circle}
\end{figure}

We now mention some additional results contained in this work.
We show that the collection $(\Ncom,\dn)$ is complete and geodesic. Crucially, networks in $\Ncom$ also admit $\e$-approximation by finite networks in the $\dn$ sense. This result is used to show that compact networks have well-defined persistence diagrams, and to prove convergence results for both hierarchical clustering and persistent homology methods on networks.
One may take a network $(X,\w_X)$ and perturb the topology without incurring changes in $\dn(X,\cdot)$, as long as $\w_X$ remains unperturbed while still satisfying continuity. Combining this observation with the results stated above suggests that we would like descriptions of topologies on a pair $(X,\w_X)$ that are controlled by $\w_X$. We provide such descriptions, which we call the \emph{weak and strong coherent} topologies, by means of a Kuratowski-like embedding. Here we consider $\w_X(x,\cdot),\w_X(\cdot,x)$ in the $l^\infty$ function space, construct a metric in this space, and then study the topologies on $X$ that are compatible with the metric topology in $l^\infty$.
An interesting remark for practitioners is that an $\e$-ball in a coherent network may include points from disparate regions in the network. This is commonly seen in neuroimaging applications, where it is standard practice to study statistical correlations between disparate regions of the brain that coordinate to produce function.    

When restricted to compact networks equipped with coherent topologies, we recover a result showing that a network can be recovered from the knowledge of its $n\times n$ subnetworks, for $n \in \N$. This generalizes a result about reconstructibility of compact metric spaces from curvature sets \cite{gromov-book}. Notably in the metric setting, this result on curvature sets is the crucial ingredient in the work of \cite{filt-func} on producing families of polynomial-time lower bounds to the GH distance.

\subsection{Related literature}

A related thread which has received much attention from the machine learning community recently is that of the \emph{Gromov-Wasserstein (GW) distance} between metric measure spaces\footnote{metric spaces equipped with a Borel probability measure}, which was introduced in \cite{dgh-sm, dghlp-focm} to simplify the difficult combinatorial optimization problem posed above by taking a relaxation based on \emph{optimal transport (OT)}. This was later used in \cite{hendrikson} for network classification by taking graph geodesic distances to define the underlying metric.
This perspective was further developed in \cite{pcs16}, which discussed GW-averaging of distance and kernel matrix representations equipped with measures, and then in \cite{gwnets}, where the authors showed that all restrictions on the representations could be removed, and one would still obtain a pseudometric structure on such ``measure networks" with well-defined notions of isomorphism. This theoretical framework was leveraged in \cite{xu2019gromov,xu2019scalable} and \cite{chowdhury2020generalized}, which used adjacency and heat kernel representations on graphs, respectively, for a variety of learning tasks. Further applications of GW distances were demonstrated in \cite{gwa}, where the authors developed the notion of GW-averaging and principal component analysis from a Riemannian perspective. 
Yet another direction in this line of work appeared in the formulation of Gromov-Wasserstein Factorization \cite{xu2020gromov}, which unrolled the loopy computations of GW distances into a deep learning framework for graph dictionary learning. The use of measures in the network-GW setting has important theoretical and practical implications, and our work in this paper contributes to the understanding of the ``measure network" setting by studying the ``network" setting in isolation and explaining which properties of the former are inherited from the latter.

For the generalized GH distance $\dn$, the critical difference between the metric and network settings is in the zero set $\{ \dn(X,\cdot) = 0\}$, which can informally be thought of as the collection of symmetries. The collection of compact metric spaces equipped with $\dgh$ is a metric space up to isometry. So two compact metric spaces $(X,d_X)$, $(Y,d_Y)$ satisfying $\dgh(X,Y) =0$ only differ by a bijection. In the setting of networks $(X,\w_X)$, $(Y,\w_Y)$ with $\dn$, however, one may have $\dn(X,Y) = 0$ even when $X$ and $Y$ have different cardinalities. One far-reaching consequence of this phenomenon is the following seminal observation by Sturm \cite{sturm2012space} in the related setting of GW distance: a certain subfamily of networks equipped with a GW distance admits Riemannian-like structure such as tangent spaces and exponential maps, but the symmetries force this collection to have the structure of a Riemannian \emph{orbifold} rather than a manifold. The close relationship between GH and GW distances then makes it relevant and interesting to study the notion of \emph{isomorphism} that is compatible with $\dn$. This is one of the main directions of this work. Mildly relaxing the notion of isomorphism suggests considering limits and $\e$-approximation under $\dn$, and these constructions are also developed in this work.

Our approaches for studying $(\Ngen,\dn)$ are inspired by techniques used to study the structure $(\Mgen,\dgh)$ \cite{burago,petersen2006riemannian,ivanov2016gromov,dgh-era} as well as methods used by Sturm \cite{sturm2012space} to study the structure of metric measure \emph{(mm)} spaces equipped with the GW distance. For example, the completion of the space of mm spaces turns out to be the space of relaxed mm spaces that satisfy the triangle inequality almost everywhere. This in turn is embedded in an ambient space which is isometric to the quotient space of symmetric $L^2$ functions on the unit square modulo measure-preserving transformations of the unit interval. In our setting, the idea of progressively relaxing the ambient space proceeds as follows: the collection of finite networks $\Ncal$ is nested inside a collection of compact networks $\Ncom$ that is complete. This space admits a natural notion of equivalence class. Working in a space where the elements need not satisfy metric requirements allows one to perform linear operations such as ``addition" of networks as well as operations where the size of a network is changed while remaining in the equivalence class. These latter properties turn out to be fundamental in a parallel theory of Riemannian statistics on networks equipped with measures \cite{gwa}. 

For the related setting of quasimetric spaces, i.e. metric spaces where the symmetry condition is relaxed, there has been significant development of the topological consequences of quasimetry. Such spaces admit \emph{forward-open} and \emph{backward-open} $\e$-balls, and thus they may admit multiple topologies concurrently \cite{kelly1963bitopological}. These in turn lead to ``forward" and ``backward" notions of Cauchy sequences, completeness, and boundedness, which have been used in \cite{shen2010gromov} to extend results such as Gromov's precompactness theorem from Riemannian to Finsler manifolds. 
Such results require some form of topological regularity. In \cite{shen2010gromov} the regularity is provided by a condition called \emph{finite reversibility} that places bounds on the asymmetry of the space. We study this condition in the context of a model of \emph{directed circles with finite reversibility}, and in particular, we characterize the \emph{Dowker persistent homology} (cf. Section \ref{sec:ph-method-details}) of such circles across all reversibility parameters to explain the difference caused by asymmetry. Furthermore, we study a weaker topological regularity condition called \emph{coherence} that is implied by finite reversibility. This allows us to obtain a reconstruction result which, in the metric setting, states that compact metric spaces can be recovered from their curvature sets \cite[Section 3.27 $\frac{1}{2}$]{gromov-book}. 
This in turn leads to efficient computational implementations for estimating $\dgh$-type distances \cite{oles2019efficient}. While our focus is on characterizing settings where asymmetric networks can be studied via theoretical and computational tools developed for metric spaces, we remark that numerous works dating back to Busemann have studied specialized versions of length structures, geodesics, and curvature for quasimetric spaces \cite{busemann1950geometry,zaustinsky1959spaces, cobzas2012functional, mennucci2013asymmetric}. 
We also remark on an interesting connection to the work of \cite{perrault2011directed}, where directed graphs are modeled as samples from a manifold equipped with an asymmetric kernel: the directionality of the kernel is derived from a vector field on the manifold, which in turn provides the topological regularity suggested above. 

The ``generalized metric" approach to studying networks utilizes different techniques than combinatorial and spectral approaches \cite{newman2010networks}. A possible bridge between metric and combinatorial approaches is given by observing the structural similarities between GH-type distances and the \emph{cut metric} \cite{lovasz2012large, borgs2008convergent} that has been deeply influential in statistical physics, theoretical computer science, and extremal graph theory. This connection is not fully fleshed out at present, but we report on a potential approach and leave an explicit connection open for future work. 
A different approach is provided by the notion of \emph{structure space} \cite{jain2009structure}, wherein networks are modeled to have fixed size by appending null nodes \cite{calissano2020populations}. This latter approach also gives rise to well-defined notions of averaging and principal component analysis, and suggests that future work could further elucidate these connections.

\subsection{Organization}

The presentation is divided into three parts. Section \ref{sec:networks} introduces networks, notions of isomorphism, the network distance, topological considerations in defining families of networks, and a particular family of directed networks called the directed circles with finite reversibility. Section \ref{sec:metrics} develops metric properties of the space of networks equipped with the network distance and characterizes the isomorphism structure in this space. In particular, it sets up results on finite sampling and convergence of compact networks that we later use for proving convergence of network invariants. In Section \ref{sec:inv-conv} we set up the general framework of network invariants, which includes hierarchical clustering and persistent homology for infinite networks. We then establish the stability and convergence of these invariants. 

Certain subsections fall into natural groups that may be read reasonably independently of other parts. The directed circles with finite reversibility are introduced in Section \ref{sec:dir-s1}, and the full characterization of their (Dowker) persistent homology is presented in Section \ref{sec:dir-s1-pers}. The crucial structural aspects of compact networks, namely the sampling and isomorphism properties, are described in Sections \ref{sec:cpt-finite-samp} and \ref{sec:cpt-weak-isom}, and their application toward obtaining well-defined persistence diagrams is provided in Section \ref{sec:ph-method-details}. 

Computational experiments using these invariants are provided in the Appendix. To clarify presentation, longer proofs are relegated to the end of the section in which they appear.

\section{Networks, isomorphism, and network distances}
\label{sec:networks}

In this section, we formulate definitions for networks, provide examples, construct some \emph{model} networks, define distances between networks, and relate different notions of isomorphism. The notion of a unique limit for a $\dn$-convergent sequence of finite networks is somewhat complex, and the results in this section show why it is necessary to work in an appropriate subspace of $\Ngen$.

\subsection{Definitions}

For real-world applications, the object of interest is often the collection of all finite networks, which we denote by \gls{Ncal}. Formally, one writes:
\[\Ncal:=\set{(X,\w_X) : X \text{ a finite set, } \w_X:X\times X \r \R \text{ any map}}.\]

\begin{example}[Finite graphs and metric spaces] This definition of a finite network is an immediate relaxation of the definition of a finite metric space. The name ``network" is justified because graphs can be viewed as networks in the sense defined above. Given a connected, undirected graph $G=(V,E)$ having adjacency matrix $A$, degree matrix $D$, and geodesic distance matrix $d_G$, the pairs $(V, d_G)$ and $(V,A)$, are both examples of networks. Consider also the graph Laplacian $L:= D-A$, or its normalized form $\mc{L} := I - D^{-1}AD$. The pairs $(V,L), (V,\mc{L})$ are also networks. In the case of (strongly connected), directed graphs, one could consider the \emph{directed Laplacian} studied by Chung \cite{chung2005laplacians} that is defined by deriving a Markov transition matrix and considering its associated Perron vector. In summary, there are a variety of quantities with different properties that can be derived from the combinatorial structure of a graph, and the definition of a finite network is general enough to comprise all of these. Each choice of a weight function that one may derive from a graph is a ``lens" through which to study the data---such perspectives have been successfully harnessed in data analysis contexts \cite{mapper}.
\end{example}

In order to build a satisfactory theoretical foundation, one also needs to develop a formalism for infinite networks. Thus we proceed with the following definition. 
\begin{definition}[Networks \gls{Ngen}]\label{defn:nets}
Let $X$ be a first countable topological space, and let $\w_X$ be a continuous function from $X \times X$ (endowed with the product topology) to $\R$. By a \emph{network}, we will mean a pair $(X,\w_X)$. We will denote the collection of all networks by $\Ngen$. 
\end{definition}

Notice in particular that $\Ngen$ includes metric spaces (they are first countable, and the distance function is continuous) as well as spaces that are quasi-metric or directed (no symmetry), pseudometric (no nondegeneracy), semimetric (no triangle inequality), or all of the above. Recall that a space is first countable if each point in the space has a countable local basis (see \cite[p. 7]{counterexamples} for more details). First countability is a technical condition guaranteeing that when the underlying topological space of a network is compact, it is also sequentially compact.

Given a network $(X,\w_X)$, we will refer to the points of $X$ as \emph{nodes} and $\omega_X$ as the \emph{weight function} of $X$. Pairs of nodes will be referred to as \emph{edges}. Given a nonempty subset $A\subset X$, we will refer to $(A,\omega_X|_{A\times A})$ as the \emph{sub-network} of $X$ induced by $A$. For notational convenience, we will often write $X\in \Ngen$ to mean $(X,\w_X)\in\Ngen$.

Recall that any finite set $X$ can be equipped with the discrete topology, and any map $\w_X:X\times X \r \R$ is continuous with respect to the discrete topology. Thus the elements of $\Ncal$ trivially fit into the framework of $\Ngen$. Throughout the paper, we will always understand finite networks to be equipped with the discrete topology. 

While we are interested in $\Ncal$ for practical applications, a key ingredient of our theoretical framework is the collection of \emph{compact networks}. We define these to be the networks $(X,\w_X)$ satisfying the additional constraint that $X$ is compact. The collection of compact networks is denoted \gls{Ncom}. Specifically, we write:
\[\Ncom:=\set{(X,\w_X) : X \text{ compact, first countable topological space, }\w_X: X\times X \r \R \text{ continuous}}.\]

Compact networks are of special practical interest because they can be finitely approximated in a manner that we will make precise in Section \ref{sec:cpt-finite-samp}.  Real world networks that are amenable to computational tasks are necessarily finite and may be viewed as samples drawn from an underlying compact network, so whenever possible, we will state our results for compact networks. Occasionally we will provide examples of noncompact networks to illustrate interesting theoretical points.

The aforementioned finite reversibility property is defined for the special subclass of \emph{dissimilarity networks} that is described as follows.

\begin{definition}[Dissimilarity networks \gls{Ngendis}]
\label{def:nets-dissim}
A \emph{dissimilarity network} is a network $(X,A_X)$ where $A_X:X\times X \to \R_+$ and $A_X(x,x') = 0$ if and only if $x= x'$. Neither symmetry nor triangle inequality is assumed. The collection of all such networks is denoted $\Ngendis$. The finite and compact settings are denoted as $\Ndis$ and $\Ncomdis$, respectively.
\end{definition}

\begin{definition}[The reversibility parameter \gls{rx}]
\label{def:reversibility}
The \emph{reversibility} $\rho_X$ of a dissimilarity network $(X,A_X)$ is defined to be 
\[\rho_X:= \sup_{x\neq x' \in X}\frac{A_X(x,x')}{A_X(x',x)}.\]
$(X,A_X)$ is said to have \emph{finite reversibility} if $\rho_X < \infty$. Notice that $\rho_X\geq 1$ always, with equality iff $\w_X$ is symmetric. Finitely reversible networks will feature heavily in Sections \ref{sec:dir-s1} and \ref{sec:dir-s1-pers}. See also Figures \ref{fig:dowker-circle} and \ref{fig:dir-s1}.
\end{definition}

Dissimilarity networks satisfying the triangle inequality, but not symmetry, include the special class of objects called \emph{directed metric spaces}, which we define below.

\begin{definition} Let $(X,A_X)$ be a dissimilarity network. Given any $x\in X$ and $r \in \R_+$, the \emph{forward-open ball} of radius $r$ centered at $x$ is 
\[B^+(x,r):=\set{x' \in X: A_X(x,x') < r}.\] 
The \emph{forward-open topology induced by $A_X$} is the topology on $X$ generated by the collection $\set{B^+(x,r) : x\in X,\; r > 0}$. The idea of forward open balls is prevalent in the study of Finsler geometry; see \cite[p. 149]{bao2012introduction} for details. 
\end{definition}

\begin{definition}[Directed metric spaces] A \emph{directed metric space} or \emph{quasi-metric space} is a dissimilarity network $(X,\nu_X)$ such that $X$ is equipped with the forward-open topology induced by $\nu_X$ and $\nu_X: X\times X \r \R_+$ satisfies: 
\[\nu_X(x,x'') \leq \nu_X(x,x') + \nu_X(x',x'') \text{ for all } x,x',x'' \in X. \]
The function $\nu_X$ is called a directed metric or quasi-metric on $X$. 
Notice that compact directed metric spaces constitute a subfamily of $\Ncomdis$.
\end{definition}

\begin{definition}[Ultrametric/strong triangle inequality]
\label{def:ultrametric}
A network $(X,\w_X)$ is said to satisfy the \emph{ultrametric} or \emph{strong triangle inequality} if 
\[\w_X(x,x') \leq \max\set{\w_X(x,x''),\w_X(x'',x')} \text{ for all } x,x',x''\in X.\] 
Metric spaces satisfying this property are simply known as \emph{ultrametric spaces}. The collections of networks satisfying this inequality are denoted $\Ngenult, \Ncomult$, and $\Nult$ in the general, compact, and finite cases, respectively.
\end{definition}

\begin{example}
Finite metric spaces and finite ultrametric spaces constitute basic examples of dissimilarity networks. Also note that finite dissimilarity networks are basic examples of networks that satisfy finite reversibility (cf. Definition \ref{def:reversibility}).
\end{example}

Dissimilarity networks satisfying the symmetry condition, but not the triangle inequality, have a long history dating back to Fr\'{e}chet \cite{frechet1906quelques} and continuing with work by Pitcher and Chittenden \cite{pitcher1918foundations}, Niemytzki \cite{niemytzki1927third}, Galvin and Shore \cite{galvin1984completeness, galvin1991distance}, and many others, as summarized in \cite{gruenhage1984generalized}. One of the interesting directions in this line of work was the development of a ``local triangle inequality" and related metrization theorems \cite{niemytzki1927third}, which has been continued more recently in \cite{waszkiewicz2013local}.

Directed metric spaces with finite reversibility were studied in \cite{shen2010gromov}, and constitute important examples of networks that are strictly non-metric. More specifically, the authors of \cite{shen2010gromov} extended notions of Hausdorff distance and Gromov-Hausdorff distance to the setting of directed metric spaces with finite reversibility, and our network distance $\dn$ subsumes this theory while extending it to even more general settings. 

\begin{remark}[Finsler metrics] An interesting class of directed metric spaces arises from studying Finsler manifolds. A \emph{Finsler manifold} $(M,F)$ is a smooth, connected manifold $M$ equipped with an asymmetric norm $F$ (called a \emph{Finsler function}) defined on each tangent space of $M$ \cite{bao2012introduction}. A Finsler function induces a directed metric $d_F: M \times M \r \R_+$ as follows: for each $x,x'\in M$, 
\[d_F(x,x'):=\inf\set{\int_a^b \!F(\g(t),\dot{\g}(t)) \,dt : \g:[a,b] \r M \text{ a smooth curve joining $x$ and $x'$}}.\]  

Finsler metric spaces have received interest in the applied literature. In \cite{sabau2013metric}, the authors prove that Finsler metric spaces with \emph{reversible geodesics} (i.e. the reverse curve $\g'(t):=\g(1-t)$ of any geodesic $\g:[0,1]\r M$ is also a geodesic) is a \emph{weighted quasi-metric} \cite[p. 2]{sabau2013metric}. Such objects have been shown to be essential in biological sequence comparison \cite{stojmirovic2009geometric}. In Section \ref{sec:dir-s1} we will study directed circles with finite reversibility, which have the property that $\e$-balls in such circles are not spherically symmetric. This is motivated by the more general setting of Randers manifolds \cite{bao2004zermelo}, which are a special subfamily of Finsler manifolds. 
\end{remark}

Returning to the setting of general networks, a natural question in understanding the structure of $\Ngen$ would be: which elements of $\Ngen$ are equivalent? A suitable answer to this question requires us to develop notions of \emph{isomorphism} that show various degrees of restrictiveness. These notions of isomorphism form a recurrent theme throughout this paper. 

We first develop the notion of \emph{strong isomorphism} of networks. The definition follows below.

\begin{definition}[Weight preserving maps] Let $(X,\w_X), (Y,\w_Y) \in \Ngen$. A map $\ph: X \r Y$ is \emph{weight preserving} if:
\[\w_X(x,x') = \w_Y(\ph(x),\ph(x')) \text{ for all } x, x' \in X.\]
\end{definition}

\begin{definition}[Strong isomorphism] Let $(X,\w_X), (Y,\w_Y) \in \Ngen$. To say $(X,\w_X)$ and $(Y,\w_Y)$ are \emph{strongly isomorphic} means that there exists a weight preserving bijection $\ph: X \r Y$. We will denote a strong isomorphism between networks by $X \text{\gls{congs}}   Y$. Note that this notion is exactly the usual notion of isomorphism between weighted graphs.
\end{definition}

Strongly isomorphic networks formalize the idea that the information contained in a network should be preserved when we relabel the nodes in a compatible way.

\begin{example}\label{ex:simple-networks} Networks with one or two nodes (cf. Fig. \ref{fig:simplenets}) will be very instructive in providing examples and counterexamples, so we introduce them now with some special terminology.
\begin{itemize}
\item A network with one node $p$ can be specified by $\alpha \in \R$, and we denote this by $N_1(\alpha)$. We have $N_1(\a)\cong^s N_1(\a')$ if and only if $\a=\a'$.
\begin{figure}
\begin{center}
\begin{tikzpicture}[every node/.style={font=\footnotesize}]

\node[circle,draw,fill=purple!20](1) at (0,0){$p$};
\path[->] (1) edge [loop above, min distance=10mm,in=0,out=60] node[right]{$\a$}(1);

\node[circle,draw,above,fill=orange!20](2) at (3,0){$p$};
\node[circle,draw,fill=violet!30](3) at (5,0){$q$};
\path[->] (2) edge [loop left, min distance = 10mm] node[above]{$\a$}(2);
\path[->] (3) edge [loop right, min distance = 10mm] node[below]{$\beta$}(3);
\path[->] (2) edge [bend left] node[above]{$\d$} (3);
\path[->] (3) edge [bend left] node[below]{$\g$} (2);

\node at (0,-1){$N_1(\a)$};
\node at (4,-1){$N_2(\Om)$};

\end{tikzpicture}
\caption{Networks over one and two nodes with their weight functions.}
\label{fig:simplenets}
\end{center}
\end{figure}

\item A network with two nodes will be denoted by $N_2(\Omega)$, where $\Omega = \mattwo{\alpha}{\delta}{\gamma}{\beta}\in \R^{2\times 2}$. Given $\Omega,\Omega'\in \R^{2\times 2}$, $N_2(\Omega)\cong^s N_2(\Omega')$ if and only if there exists a permutation matrix $P$ of size $2\times 2$ such that $\Omega' = P\,\Omega\, P^T$.

\item Any $k$-by-$k$ matrix $\Sigma\in \R^{k\times k}$ induces a network on $k$ nodes, which we refer to as $N_k(\Sigma)$. Notice that $N_k(\Sigma)\cong^s N_\ell(\Sigma')$ if and only if $k=\ell$ and there exists a permutation matrix $P$ of size $k$ such that $\Sigma'=P\,\Sigma\,P^T.$
\end{itemize}
\end{example}

Having defined a notion of isomorphism between networks, the next goal is to present the network distance $\dn$ that is the central focus of this paper, and verify that $\dn$ is compatible with strong isomorphism. We remind the reader that restricted formulations of this network distance have appeared in earlier applications of \emph{hierarchical clustering} \cite{clust-net, carlsson2013axiomatic} and \emph{persistent homology} \cite{dowker-asilo,dowker-jact,pph} methods to network data, and our overarching goal in this paper is to provide a theoretical foundation for this useful notion of network distance. In our presentation, we use a formulation of $\dn$ that is more general than any other version available in the existing literature. As such, we proceed pedagogically and motivate the definition of $\dn$ by tracing its roots in the metric space literature. Before ending this section, however, we will provide a final definition that is related to the matrices in Example \ref{ex:simple-networks}. 

For a sequence $(x_i)_{i=1}^n$ of nodes in a network $X$, we will denote the associated weight matrix by $\matele{\w_X(x_i,x_j)}_{i,j=1}^n$. Entry $(i,j)$ of this matrix is simply $\w_X(x_i,x_j)$.

\begin{definition}[Motif sets \gls{M}]\label{defn:motif}

For each $n\in \mathbb{N}$ and each $X \in \Ncom$, define $\Psi^n_X: X^n \r \R^{n \times n}$ to be the map $(x_1, \cdots, x_n) \mapsto \matele{\w_X(x_i,x_j)}_{i,j=1}^n$. Note that $\Psi^n_X$ is simply a map that sends each sequence of length $n$ to its corresponding weight matrix. Let $\mc{C}(\R^{n\times n})$ denote the closed subsets of $\R^{n\times n}$. Then let $\M_n:\Ncom\r \mc{C}(\R^{n\times n})$ denote the map defined by 
\[(X,\w_X)\mapsto \set{\Psi^n_X(x_1,\ldots, x_n) : x_1,\ldots,x_n\in X}.\]
We refer to $\M_n(X)$ as the $n$-motif set of $X$. Notice that the image of $\M_n$ is closed in $\R^{n\times n}$ because each coordinate is the continuous image of the compact set $X\times X$ under $\w_X$, hence the image of $\M_n$ is compact in $\R^{n\times n}$ and hence closed. 
\end{definition}

Notice that for $X\in \Ncal$ and for a fixed $n\in \N$, the set $\M_n(X)$ is a \emph{finite} subset of $\R^{n\times n}$. The interpretation is that $\M_n(X)$ is a bag containing all the motifs of $X$ that one can form by looking at all subnetworks of size $n$ (with repetitions).

\begin{example}\label{ex:conf-sets}
For the networks from Example \ref{ex:simple-networks}, we have $\M_1(N_2(\Omega)) = \{\alpha,\beta\}$ and 
\[\M_2(N_2(\Omega))=\big\{\mattwo{\alpha}{\alpha}{\alpha}{\alpha},\mattwo{\beta}{\beta}{\beta}{\beta},\mattwo{\alpha}{\delta}{\gamma}{\beta},\mattwo{\beta}{\gamma}{\delta}{\alpha}\big\}, 
\qquad \M_2(N_1(\a)) = \set{\mattwo{\a}{\a}{\a}{\a}}.
\]
\end{example}

\begin{remark}
Our definition of motif sets is inspired by a definition made by Gromov, termed ``curvature classes,'' in the context of compact metric spaces \cite[\S 3.27]{gromov-book}.
\end{remark}

\subsection{The network distance}

One strategy for defining a notion of distance between networks would be to take a well-understood notion of distance between metric spaces and extend it to all networks. The network distance $\dn$ arises by following this strategy and extending the well-known Gromov-Hausdorff distance $\dgh$ between compact metric spaces \cite{gromov1981structures,burago, petersen2006riemannian}. The definition of $\dgh$ is rooted in the \emph{Hausdorff distance \gls{dhdf}}  between closed subsets of a metric space. Given a metric space $(Z,d_Z)$ and closed subsets $A,B \subseteq Z$, one defines:
\[\dhdf^Z(A,B):= \max( \sup_{a\in A} \inf_{b\in B} d_Z(a,b),
\sup_{b\in B} \inf_{a\in A} d_Z(a,b) ).\]

\begin{definition}\label{defn:dgh}
Given metric spaces $(X,d_X)$ and $(Y,d_Y)$, the \emph{Gromov-Hausdorff distance \gls{dgh}} between them is defined as:
\begin{align*}
\dgh((X,d_X),(Y,d_Y)) :=\inf \big\{\dhdf^Z(\ph(X),\psi(Y)) &: \text{$Z$ a metric space,}\\
&\ph:X \r Z,\; \psi:Y\r Z \text{ isometric embeddings} \big\}.
\end{align*}
\end{definition}

The Gromov-Hausdorff distance dates back to at least the early 1980s \cite{gromov1981structures}, and it satisfies numerous desirable properties. It is a valid metric on the collection of isometry classes of compact metric spaces, is complete, admits many precompact families, and has well-understood notions of convergence \cite[Chapter 7]{burago}. Moreover, it has found real-world applications in the \emph{shape matching} \cite{dgh-sgp,dgh-focm} and persistent homology literature \cite{dgh-pers}, and its computational aspects have been studied as well \cite{dgh-props}. As such, it is a strong candidate for use in defining a network distance.

Unfortunately, the formulation of $\dgh$ above is heavily dependent on a metric space structure, and the notion of Hausdorff distance may not make sense in the setting of networks. So $\dgh$ as defined above cannot be directly extended to a network distance. However, it turns out that there is a reformulation of $\dgh$ that utilizes the language of \emph{correspondences} \cite{kalton1999distances, burago}. We present this construction next, and note that the resulting network distance $\dn$ will agree with $\dgh$ when restricted to metric spaces.

\begin{definition}[Correspondence] Let $(X,\w_X), (Y,\w_Y) \in \Ngen$. A \emph{correspondence between $X$ and $Y$} is a relation $R \subseteq X \times Y$ such that $\pi_X(R) = X$ and $\pi_Y(R) = Y$, where $\pi_X$ and $\pi_Y$ are the canonical projections of $X\times Y$ onto $X$ and $Y$, respectively. The collection of all correspondences between $X$ and $Y$ will be denoted $\Rsc(X,Y)$, abbreviated to \gls{Rsc} when the context is clear.
\end{definition}

\begin{example}[1-point correspondence]\label{ex:corresp} Let $X$ be a set, and let $\set{p}$ be the set with one point. Then there is a unique correspondence $R = \set{(x,p) : x \in X}$ between $X$ and $\set{p}$.
\end{example} 

\begin{example}[Diagonal correspondence]\label{ex:diag-corresp}
Let $X = \set{x_1,\ldots, x_n}$ and $Y=\set{y_1,\ldots, y_n}$ be two enumerated sets with the same cardinality. A useful correspondence is the \emph{diagonal correspondence}, defined as $\Delta:=\set{(x_i,y_i) : 1\leq i \leq n}.$ When $X$ and $Y$ are infinite sets with the same cardinality, and $\ph:X \r Y$ is a given bijection, then we write the diagonal correspondence as $\Delta:=\set{(x,\ph(x)) : x \in X}.$
\end{example}

\begin{definition}[Distortion of a correspondence \gls{dis}]
Let $(X,\w_X),(Y,\w_Y) \in \Ngen$ and let $R\in\Rsc(X,Y)$. The \emph{distortion} of $R$ is given by:
\[\mathrm{dis}(R):=\sup_{(x,y),(x',y')\in R}|\w_X(x,x')-\w_Y(y,y')|.\] 
\end{definition}

\begin{remark}[Composition of correspondences]\label{rem:chained-corr} 
Let $(X,\w_X), \ (Y,\w_Y), \ (Z,\w_Z) \in \Ngen$, and let $R \in \Rsc(X,Y),\ S \in \Rsc(Y,Z)$. Then we define:
\[R\circ S := \{(x,z) \in X \times Z \mid \exists y, (x,y) \in R, (y,z) \in S\}.\]
In the proof of Theorem \ref{thm:dN1}, we verify that $R\circ S \in \Rsc(X,Z)$, and that $\dis(R\circ S) \leq \dis(R) + \dis(S)$. %
\end{remark}

\begin{definition}[The first network distance \gls{dn}]\label{defn:dno}
Let $(X,\w_X),(Y,\w_Y) \in \Ngen$. We define the \emph{network distance} between $X$ and $Y$ as follows:
\[\dn((X,\w_X),(Y,\w_Y)):=\tfrac{1}{2}\inf_{R\in\Rsc}\mathrm{dis}(R).\]

When the context is clear, we will often write $\dn(X,Y)$ to denote $\dn((X,\w_X),(Y,\w_Y))$. We define the collection of \emph{optimal correspondences} \gls{Ropt} between $X$ and $Y$ to be the collection $\set{R \in \Rsc(X,Y) : \dis(R) = 2\dn(X,Y)}.$ This set is always nonempty when $X,Y \in \Ncal$, but may be empty in general (see Example \ref{ex:non-opt-corr}).
\end{definition}

\begin{remark} The intuition behind the preceding definition of network distance may be better understood by examining the case of a finite network. Given a finite set $X$ and two edge weight functions $\w_X,\w'_X$ defined on it, we can use the $\ell^\infty$ distance as a measure of network similarity between $(X,\w_X)$ and $(X,\w'_X)$:
\[\norm{\w_X-\w'_X}_{\ell^\infty(X\times X)}:=\max_{x,x'\in X}|\w_X(x,x')-\w'_X(x,x')|.\]

A generalization of the $\ell^\infty$ distance is required when dealing with networks having different sizes: Given two sets $X$ and $Y$, we need to decide how to match up points of $X$ with points of $Y$. Any such matching will yield a subset $R\subseteq X\times Y$ such that $\pi_X(R)=X$ and $\pi_Y(R)=Y$, where $\pi_X$ and $\pi_Y$ are the projection maps from $X\times Y$ to $X$ and $Y$, respectively. This is precisely a correspondence, as defined above. A valid notion of network similarity may then be obtained as the distortion incurred by choosing an optimal correspondence---this is precisely the idea behind the definition of the network distance above.
\end{remark}

\begin{remark}\label{rem:dno-finite}
Some simple but important remarks are the following: 
\begin{enumerate}
\item When restricted to metric spaces, $\dn$ agrees with $\dgh$. This can be seen from the reformulation of $\dgh$ in terms of correspondences \cite[Theorem 7.3.25]{burago}, \cite{kalton1999distances}. Whereas $\dgh$ vanishes only on pairs of isometric compact metric spaces,
$\dn$ vanishes on a broader family of networks that we will describe more fully in Section \ref{sec:networks-wisom}.
\item Given $X,Y\in \Ncal$, the network distance reduces to the following:
\[\dn(X,Y)=\frac{1}{2}\min_{R\in \Rsc}\max_{(x,y),(x',y')\in R}|\w_X(x,x')-\w_Y(y,y')|.\]
Moreover, there is always at least one optimal correspondence $\Ropt$ for which $\dn(X,Y)$ is achieved; this is a consequence of considering finite networks.

\item For any $X,Y\in\Ncom$, we have $\Rsc(X,Y)\neq \emptyset$, and $\dn(X,Y)$ is always bounded. Indeed, $X\times Y$ is always a valid correspondence between $X$ and $Y$. So we have:
\[\dn(X,Y)\leq \frac{1}{2}\mathrm{dis}(X\times Y)
\leq\frac{1}{2}\left(\sup_{x,x'}\big|\omega_X(x,x')\big|+\sup_{y,y'\in Y}\big|\omega_Y(y,y')\big|\right)<\infty.\]
\end{enumerate}
\end{remark}

\begin{example}\label{ex:comparison} Now we give some examples to illustrate the preceding definitions.
\begin{itemize}

\item For $\alpha,\alpha'\in\R$ consider two networks with one node each: $N_1(\alpha) = (\{p\},\alpha)$ and $N_1(\alpha')=(\{p'\},\alpha')$. By Example \ref{ex:corresp} there is a unique correspondence $R=\{(p,p')\}$ between these two networks, so that $\dis(R)=|\alpha-\alpha'|$ and as a result $\dn(N_1(\alpha),N_1(\alpha'))=\frac{1}{2}|\alpha-\alpha'|.$

\item Let $(X,\omega_X)\in\Ncal$ be any network and for $\alpha\in \R$  let $N_1(\alpha)=(\{p\},\alpha)$. Then $R=\{(x,p),\,x\in X\}$ is the unique correspondence between $X$ and $\set{p}$, so that 
\[\dn(X,N_1(\alpha))=\frac{1}{2}\max_{x,x'\in X}\big|\omega_X(x,x')-\alpha\big|.\]

\end{itemize}
\end{example}

We now test whether $\dn$ is compatible with strong isomorphism. Given two strongly isomorphic networks, i.e. networks $(X,\w_X), (Y,\w_Y)$ and a weight preserving bijection $\ph:X \r Y$, it is easy to use the diagonal correspondence (Example \ref{ex:diag-corresp}) to verify that $\dn(X,Y) = 0$. However, it is easy to see that the reverse implication is not true in general. Using the one-point correspondence (Example \ref{ex:corresp}), one can see that $\dn(N_1(1), N_2(\mathbbm{1}_{2 \times 2})) =0$. Here $\mathbbm{1}_{n \times n}$ denotes the all-ones matrix of size $n \times n$ for any $n \in \N$. However, these two networks are not strongly isomorphic, because they do not even have the same cardinality. Thus we need to search for a different, perhaps weaker notion of isomorphism.

\subsection{Weak isomorphism}
\label{sec:networks-wisom}

To proceed in this direction, first notice that a strong isomorphism between two networks $(X,\w_X)$ and $(Y,\w_Y)$, given by a bijection $f:X\r Y$, is equivalent to the following ``tripod" condition (cf. Figure \ref{fig:tripod}): there exists a set $Z$ and bijective maps $\ph_X:Z \r X, \ph_Y:Z \r Y$ such that $\w_X(\ph_X(z),\ph_X(z'))=\w_Y(\ph_Y(z),\ph_Y(z'))$ for each $z,z'\in Z$. To see this, simply let $Z=\set{(x,f(x)) : x\in X}$ and let $\ph_X, \ph_Y$ be the projection maps on the first and second coordinates, respectively. Based on this observation, we make the next definition.

\begin{definition}\label{defn:typeI-wisom} Let $(X,\w_X)$ and $(Y,\w_Y) \in \Ngen$. We define $X$ and $Y$ to be \emph{Type I weakly isomorphic}, denoted $X \text{\gls{congwo}} Y$, if there exists a set $Z$ and surjective maps $\ph_X : Z \r X$ and $\ph_Y: Z \r Y$ such that 
\begin{equation}\w_X(\ph_X(z),\ph_X(z')) = \w_Y(\ph_Y(z),\ph_Y(z')) \text{ for each $z,z' \in Z$.}\end{equation} 
\end{definition}

\begin{figure}
\begin{center}
\begin{tikzpicture}
\node (1) at (0,0){$X$};
\node (2) at (1,1.5){$Z$};
\node (3) at (2,0){$Y$};
\node at (1,0){$\cong^s$};
\node at (-3,0.75)[text width = 4 cm]{Strong isomorphism: $\phi_X, \phi_Y$ injective and surjective};

\path[->] (2) edge node[left]{$\phi_X$}(1);
\path[->] (2) edge node[right]{$\phi_Y$}(3);

\begin{scope}[xshift=-1in]
\node (4) at (6,0){$X$};
\node (5) at (7,1.5){$Z$};
\node (6) at (8,0){$Y$};
\node at (7,0){$\congwo$};
\node at (11,0.75)[text width = 4 cm]{Type I weak isomorphism: $\phi_X, \phi_Y$ only surjective};

\path[->] (5) edge node[left]{$\phi_X$}(4);
\path[->] (5) edge node[right]{$\phi_Y$}(6);
\end{scope}
\end{tikzpicture}
\caption{Relaxing the requirements on the maps of this ``tripod structure" is a natural way to weaken the notion of strong isomorphism.}
\label{fig:tripod}
\end{center}
\end{figure}

Notice that Type I weak isomorphism is in fact a \emph{relaxation} of the notion of strong isomorphism. Indeed, if in addition to being surjective, we require the maps $\phi_X$ and $\phi_Y$ to be injective, then the strong notion of isomorphism is recovered. In this case, the map $\phi_Y\circ\phi_X^{-1}:X\rightarrow Y$ would be a weight preserving bijection between the networks $X$ and $Y$. The relaxation of strong isomorphism to a Type I weak isomorphism is illustrated in Figure \ref{fig:tripod}. Also observe that the relaxation is \emph{strict}. For example, the networks $X = N_1(1)$ and $Y = N_2(\mathbbm{1}_{2\times 2})$, are weakly but not strongly isomorphic via the map that sends both nodes of $Y$ to the single node of $X$. 

\begin{remark}[Surjective maps induce Type I isomorphism]\label{rem:surj} Let $(X,\omega_X),(Y,\omega_Y) \in \Ncom$ and suppose $\varphi:X\rightarrow Y$ is a surjective map such that $\omega_X(x,x')=\omega_Y(\varphi(x'),\varphi(x'))$ for all $x,x'\in X$. Then $X$ and $Y$ are Type I weakly isomorphic. This result follows from Definition \ref{defn:typeI-wisom} by: (1) choosing $Z = X$, (2) letting $\phi_X$ be the identity map, and (3) letting $\phi_Y=\varphi$. The converse implication, i.e. that Type I weak isomorphism implies the existence of a surjective map as above, is not true: an example is shown in Figure \ref{fig:dn0nets}.
\end{remark}

When dealing with infinite networks, it will turn out that an even weaker notion of isomorphism is required. We define this weakening next.

\begin{definition}\label{defn:typeII-wisom} Let $(X,\w_X)$ and $(Y,\w_Y)\in \Ngen$. We define $X$ and $Y$ to be \emph{Type II weakly isomorphic}, denoted $X\text{\gls{congwt}} Y$, if for each $\e >0$, there exists a set $Z_\e$ and surjective maps $\phi^\e_X:Z_\e\r X$ and $\phi^\e_Y:Z_\e\r Y$ such that 
\begin{equation}\label{eq:weak-isom-condition}
|\w_X(\phi_X^\e(z),\phi^\e_X(z')) - \w_Y(\phi^\e_Y(z),\phi^\e_Y(z'))| < \e \,\,\mbox{for all $z,z'\in Z_\e$}.
\end{equation}
\end{definition}

\begin{figure}[b]
\begin{center}
\begin{tikzpicture}[every node/.style={font=\footnotesize}]
\node (A) at (-1,1){$A$};
\node (B) at (4,1){$B$};
\node (C) at (9,1){$C$};
\node[draw,ellipse, fill=orange!20, minimum width= 7mm,minimum height= 11 mm] (4) at (-1,-0.25) {};
\node[inner sep = 0pt, minimum size = 3mm] (1) at (-1,0){$x$};
\node[inner sep = 0pt, minimum size = 3mm] (2) at (-1,-0.5){$y$};
\node[draw,circle,fill=purple!20] (3) at (0.5,-0.25){$z$};

\node[draw,circle,fill=orange!20] (5) at (4,-0.25){$u$};
\node[draw,fill=purple!20, ellipse, minimum width= 7mm,minimum height= 11 mm] (8) at (5.5,-0.25) {};
\node[inner sep = 0pt, minimum size = 3mm,] (6) at (5.5,0){$v$};
\node[inner sep = 0pt, minimum size = 3mm] (7) at (5.5,-0.5){$w$};

\node[draw,fill=orange!20,ellipse, minimum width= 7mm,minimum height= 11 mm] (9) at (9,-0.25){};
\node[draw,fill=purple!20,ellipse, minimum width= 7mm,minimum height= 11 mm] (10) at (10.5,-0.25){};
\node[inner sep = 0pt, minimum size = 3mm] (11) at (9,0){$p$};
\node[inner sep = 0pt, minimum size = 3mm] (12) at (9,-0.5){$q$};
\node[inner sep = 0pt, minimum size = 3mm] (13) at (10.5,0){$r$};
\node[inner sep = 0pt, minimum size = 3mm] (14) at (10.5,-0.5){$s$};

\path[->] (4) edge [loop left, min distance = 5mm] node[above left]{$2$}(4);
\path[->] (3) edge [loop right, min distance = 5mm] node[below right]{$3$}(3);
\path[->] (4) edge [bend left] node[above]{$1$} (3);
\path[->] (3) edge [bend left] node[below]{$1$} (4);

\path[->] (5) edge [loop left, min distance = 5mm] node[above left]{$2$}(5);
\path[->] (8) edge [loop right, min distance = 5mm] node[below right]{$3$}(8);
\path[->] (5) edge [bend left] node[above]{$1$} (8);
\path[->] (8) edge [bend left] node[below]{$1$} (5);

\path[->] (9) edge [loop left, min distance = 5mm] node[above left]{$2$}(9);
\path[->] (10) edge [loop right, min distance = 5mm] node[below right]{$3$}(10);
\path[->] (9) edge [bend left] node[above]{$1$} (10);
\path[->] (10) edge [bend left] node[below]{$1$} (9);

\node at (-1,-1.5) {$\Psi^3_A(x,y,z) = \mattres{2}{2}{1}{2}{2}{1}{1}{1}{3}$};

\node at (4,-1.5) {$\Psi^3_B(u,v,w) = \mattres{2}{1}{1}{1}{3}{3}{1}{3}{3}$};

\node at (9,-1.5) {$\Psi^4_C(p,q,r,s) = 
\left(
\begin{smallmatrix}
2 & 2 & 1 & 1\\
2 & 2 & 1 & 1\\
1 & 1 & 3 & 3\\
1 & 1 & 3 & 3\\
\end{smallmatrix}
\right)
$};

\end{tikzpicture}
\caption{Note that Remark \ref{rem:surj} \emph{does not} fully characterize weak isomorphism, even for finite networks: All three networks above, with the given weight matrices, are Type I weakly isomorphic since $C$ maps surjectively onto $A$ and $B$. But there are no surjective, weight preserving maps $A \r B$ or $B \r A$.}
\label{fig:dn0nets}
\end{center}
\vspace{-.2in}
\end{figure}

\begin{example}[Infinite networks without optimal correspondences]
\label{ex:non-opt-corr}
The following example illustrates the reason we had to develop multiple notions of weak isomorphism. The key idea is that the infimum in Definition \ref{defn:dno} is not necessarily obtained when $X$ and $Y$ are infinite networks. To see this, let $(X,\w_X)$ denote $[0,1]$ equipped with the Euclidean distance, and let $(Y,\w_Y)$ denote $\Q \cap [0,1]$ with the restriction of the  Euclidean distance. Since the closure of $Y$ in $[0,1]$ is just $X$, the Hausdorff distance between $X$ and $Y$ is zero (recall that given $A, B \subseteq \R$, we have $\dhdf^{\R}(A,B) = 0$ if and only if $\overline{A} = \overline{B}$ \cite[Proposition 7.3.3]{burago}). It follows from the definition of $\dgh$ (Definition \ref{defn:dgh}) and the equivalence of $\dn$ and $\dgh$ on metric spaces (Remark \ref{rem:dno-finite}) that $\dn(X,Y) = 0$.

However, one cannot define an optimal correspondence between $X$ and $Y$. To see this, assume towards a contradiction that $\Ropt$ is such an optimal correspondence, i.e. $\dis(\Ropt) = 0$. For each $x\in X$, there exists $y_x\in Y$ such that $(x,y_x)\in \Ropt$. By making a choice of $y_x\in Y$ for each $x\in X$, define a map $f: X \r Y$ given by $x\mapsto y_x$. Then $d_X(x,x') = d_Y(f(x),f(x'))$ for each $x,x'\in X$. Thus $f$ is an isometric embedding from $X$ into itself (note that $Y\subseteq X$). But $X=[0,1]$ is compact, and an isometric embedding from a compact metric space into itself must be surjective \cite[Theorem 1.6.14]{burago}. This is a contradiction, because $f(X) \subseteq Y \neq X$. 

We observe that $\dn(X,Y)=0$ and so $X$ and $Y$ are weakly isomorphic of Type II, but not of Type I. To see this, assume towards a contradiction that $X$ and $Y$ are Type I weakly isomorphic. Let $Z$ be a set with surjective maps $\ph_X: Z \r X$ and $\ph_Y:Z \r Y$ satisfying $\w_X\circ(\ph_X,\ph_X) = \w_Y\circ(\ph_Y,\ph_Y)$. Then $\set{(\ph_X(z),\ph_Y(z)) : z\in Z}$ is an optimal correspondence. This is a contradiction by the previous reasoning. 

\end{example}

Recall that our motivation for introducing notions of isomorphism on $\Ngen$ was to determine which networks deserve to be considered equivalent. It is easy to see that strong isomorphism induces an equivalence class on $\Ngen$. The same is true for both types of weak isomorphism, and we record this result in the following proposition.

\begin{restatable}{proposition}{weakisomequiv}
\label{prop:weak-isom-equiv}
Weak isomorphism of Types I and II both induce equivalence relations on $\Ngen$.
\end{restatable}

In the setting of $\Ncal$, it is not difficult to show that the two types of weak isomorphism coincide. This is the content of the next proposition. By virtue of this result, there is no ambiguity in dropping the ``Type I/II" modifier when saying that two finite networks are weakly isomorphic. 

\begin{restatable}{proposition}{weakisomfinite}
\label{prop:weak-isom-finite} Let $X,Y \in \Ncal$ be finite networks. Then $X$ and $Y$ are Type I weakly isomorphic if and only if they are Type II weakly isomorphic.
\end{restatable}

Type I weak isomorphisms will play a vital role in the content of this paper, but for now, we focus on Type II weak isomorphism. The next theorem justifies calling $\dn$ a network \emph{distance}, and shows that $\dn$ is compatible with Type II weak isomorphism. 

\begin{restatable}{theorem}{dN}
\label{thm:dN1}
$\dn$ is a metric on $\Ngen$ modulo Type II weak isomorphism.
\end{restatable}

Proofs of the preceding results are provided in Section \ref{sec:pf-dN1}. For finite networks, we immediately obtain:
\medskip
\begin{center}
\emph{The restriction of $\dn$ to $\Ncal$ yields a metric modulo Type I weak isomorphism.}
\end{center}
\medskip

The proof of Proposition \ref{prop:weak-isom-finite} will follow from the proof of Theorem \ref{thm:dN1}. In fact, an even stronger result is true: weak isomorphism of Types I and II coincide for compact networks as well. We present this statement in Section \ref{sec:cpt-weak-isom}. By virtue of this latter result, we will use the notation \gls{congw} without I/II qualifiers to mean Type I weak isomorphism between compact networks.

\begin{remark}[Making $\dn$ more sensitive to topology] 
\label{rem:dn-top-blind}
As will be clear from the proof of Theorem \ref{thm:dN1}, $\dn$ is actually a metric modulo Type II weak isomorphism on the collection $\{(X,\w_X) : X \text{ a set, } \w_X : X\times X \r \R \text{ any function}\}$. In other words, when applied in this general form, $\dn$ is insensitive to topology. However, restricting to subcollections of set-function pairs may improve this sensitivity. When one starts with a network in $\Ngen$ (comprising a continuous weight function over a first countable topological space) it is possible to refine the topology while maintaining first countability, even taking the discrete topology at the extreme, and remain in the same Type II weak isomorphism class.  Alternatively, one may start with a network in $\Ncom$ and refine the underlying topology by adding a finite number of open sets---the resulting network will still be in $\Ncom$, and will even be in the same Type I weak isomorphism class. In each of these cases, we progressively narrowed down the types of topological changes (possibly infinite refinements, and then finite refinements) against which $\dn$ would be guaranteed to be insensitive by studying nested subcollections of set-function pairs. Even the insensitivity to refinements need not be problematic, as it will turn out that each Type I weak isomorphism class has representatives with ``nice" topology. These ideas will be introduced in the next section.
\end{remark}

We end the current subsection with the following definition, which is reminiscent of the notion of $\e$-nets in metric spaces.
\begin{definition}[$\e$-approximations] 
\label{defn:e-approx}
Let $\e > 0$. A network $(X,\w_X)\in \Ngen$ is said to be \emph{$\e$-approximable} by $(Y,\w_Y)\in \Ngen$ if $\dn(X,Y) < \e$. In this case, $Y$ is said to be an \emph{$\e$-approximation} of $X$. Typically, we will be interested in the case where $X$ is infinite and $Y$ is finite, i.e. in $\e$-approximating infinite networks by finite networks.  
\end{definition}

\subsection{Interpolating between strong and weak isomorphism}

As we saw in the simple examples discussed above, strong isomorphism implies weak isomorphism, and weak isomorphism does not in general imply strong isomorphism.  
However, we will show that weakly isomorphic networks live over a ``base space" of strongly isomorphic networks.
The following definitions enable us to formulate the appropriate statement.

\begin{definition}[Automorphisms]\label{defn:automorph} Let $(X,\w_X) \in \Ncom$. We define the \emph{automorphisms (\gls{Aut}) of $X$ } to be the collection
\[\Aut(X):=\set{\ph:X \r X : \ph \text{ a weight preserving bijection}}.\]
\end{definition}

\begin{definition}[Poset of weak isomorphism]
\label{defn:poset} Let $(X,\w_X) \in \Ncom$. Define a set $\mf{p}(X)$ as follows:
\[\mf{p}(X):=\set{(Y,\w_Y)\in \Ncom : \text{ there exists a surjective, weight preserving map } \ph:X \r Y}.\]
Next we define a partial order $\preceq$ on $\mf{p}(X)$ as follows: for any $(Y,\w_Y), (Z,\w_Z) \in \mf{p}(X)$,
\[(Y,\w_Y) \preceq (Z,\w_Y) \iff \text{ there exists a surjective, weight preserving map } \ph:Z \r Y.\]
Then the set \gls{posetx} equipped with $\preceq$ is called the \emph{poset of weak isomorphism of $X$.}
\end{definition}

\begin{definition}[Terminal networks in $\Ncom$]\label{defn:terminal} 
Let $(X,\w_X) \in \Ncom$. A compact network $Z \in \mf{p}(X)$ is \emph{terminal} if:
\begin{enumerate}
\item For each $Y\in \mf{p}(X)$, there exists a weight preserving surjection $\ph: Y \r Z$.
\item Let $Y\in \mf{p}(X)$. If $f: Y \r Z$ and $g:Y \r Z$ are weight preserving surjections, then there exists $\ph \in \Aut(Z)$ such that $g=\ph \circ f$ (Fig. \ref{fig:terminal-diag}).  
\end{enumerate}
\end{definition}

\begin{figure}
\centering
\begin{tikzpicture}

\begin{scope}

\matrix[column sep={8em,between origins},
        row sep={6em}] at (0,0)
{ \node(X) {$X$};\\
  \node(Z) {$Z$};\\};
\path[->>] (X) edge [bend right] node [left] {$f$}(Z);
\path[->>] (X) edge [bend left] node [right] {$g$}(Z);
\path[->] (Z) edge [dashed, loop below, min distance = 10 mm] node [right] {$\ph$}(Z);
\end{scope}

\begin{scope}[xshift=2.5in]
[every node/.style={midway}]
\matrix[column sep={4em,between origins},
        row sep={6em}] at (0,0)
{ \node(X)   {$X$}  ;
 & \node(V) {$V$};
 & \node(Y) {$Y$};
 & \node(d) {$\cdots$};\\
 \node() {}; 
 & \node() {};
 & \node() {};
 & \node(Z) {$Z$};\\};
\path[->>] (X) edge [] node [] {}(V);
\path[->>] (V) edge [] node [] {}(Y);
\path[->>] (Y) edge [] node [] {}(d);
\path[->>] (d) edge [] node [] {}(Z);
\path[->>] (Y) edge [] node [] {}(Z);
\path[->>] (V) edge [] node [] {}(Z);
\path[->>] (X) edge [] node [] {}(Z);
\end{scope}
\end{tikzpicture}

\caption{\textbf{Left:} $Z$ represents a terminal object in $\mf{p}(X)$, and $f,g$ are weight preserving surjections $X \r Z$. Here $\ph \in \Aut(Z)$ is such that $g=\ph\circ f$. \textbf{Right:} Here we show more of the poset structure of $\mf{p}(X)$. In this case we have $X\succeq V \succeq Y \ldots \succeq Z$.}
\label{fig:terminal-diag}
\end{figure}

In Section \ref{sec:cpt-skel} we define a construction called the \emph{skeleton} of a network and show that it is terminal. One of our main results (Theorem \ref{thm:cpt-sisom}) shows that, under some mild topological regularity conditions, 
\begin{center}
    \emph{weakly isomorphic networks have strongly isomorphic skeleta.}
\end{center}

A terminal network captures the idea of a minimal substructure of a network. One may ask if anything interesting can be said about superstructures of a network. This motivates the following construction of a ``blow-up" network. We provide an illustration in Figure \ref{fig:blowup}.

\begin{definition}Let $(X,\omega_X)$ be any network. Let $\mathbf{k}=(k_x)_{x\in X}$ be a choice of an index set $k_x$ for each node $x\in X$. 
Consider the network $X[\mathbf{k}]$ with node set $\bigcup_{x\in X} \{(x,i) : i\in {k_x}\}$ and weights $\omega$ given as follows:  for $x,x'\in X$ and for $i \in {k_x}$, $i'\in {k_{x'}}$,
\[\omega\big((x,i),(x',i')\big):= \w_X(x,x').\]
The topology on $X[\mathbf{k}]$ is given as follows: the open sets are of the form $\bigcup_{x\in U}\set{(x,i) : i \in k_x}$, where $U$ is open in $X$. By construction, $X[\mathbf{k}]$ is first countable with respect to this topology. 
We will call any such \gls{blowupx} a \emph{blow-up network} of $X$.

\end{definition}

In a blow-up network of $X$, each node $x\in X$ is replaced by another network, indexed by ${k_x}$. All internal weights of this network are constant and all outgoing weights are preserved from the original network. If $X$ is compact, then so is $X[\mathbf{k}]$. 

We also observe that $X$ is weakly isomorphic to any of its blow-ups $Y = X[\mathbf{k}]$. To see this, let $Z = X[\mathbf{k}]$, let $\phi_Y:Z\rightarrow Y$ be the map sending each $(x,i)$ to $(x,i)$, and let $\phi_X:Z\rightarrow X$ be the map sending each $(x,i)$ to $x$. Then $\phi_X, \phi_Y$ are surjective, weight preserving maps from $Z$ onto $X$ and $Y$ respectively. %
By Remark \ref{rem:surj}, we obtain $X \congwo Y$. 

Later in Proposition \ref{prop:bijection}, we will show that Type I weakly isomorphic networks may be blown-up to strongly isomorphic networks.

\begin{figure}[t]
\begin{center}
\begin{tikzpicture}[every node/.style={font=\footnotesize}]

\node[circle,draw, fill=orange!20](2) at (-3,0){q};
\node[circle,draw,fill=violet!20](3) at (-1,0){r};

\path[->] (2) edge [loop left, min distance = 10mm] node{$1$}(2);
\path[->] (3) edge [loop right, min distance = 10mm] node{$4$}(3);
\path[->] (2) edge [bend left] node[above]{$2$} (3);
\path[->] (3) edge [bend left] node[below]{$3$} (2);

\node[above] (dummy1) at (1.5,2.5){};
\node[above] (dummy2) at (1.5,-2.5){};
\node[above] (dummy3) at (-1,1){};
\node[above] (dummy4) at (-1,-1){};

\path[->]
    (dummy3) edge[bend left,thick] node [above left] {blow-up} (dummy1)
    (dummy2) edge[bend left, thick] node [below left] {skeletonize} (dummy4);
    
\node[circle,draw,fill=orange!20](4) at (4,3){$(q,1)$};
\node[circle,draw,fill=orange!20](5) at (4,-3){$(q,2)$};

\node[circle,draw,fill = violet!20](6) at (10,3){$(r,1)$};
\node[circle,draw,fill = violet!20](7) at (10,-3){$(r,2)$};

\node (f1) at (2.5,4){};
\node (f2) at (5,-4){};

\begin{scope}[on background layer]
\node[draw,dashed,rounded corners, fill=orange!20, fill opacity = 0.5, fit = (4) (5) (f1) (f2) ]{};

\node (f3) at (9,4){};
\node (f4) at (11.5,-4){};

\node[draw, dashed,rounded corners, fill=purple!20, fill opacity = 0.5, fit = (6) (7) (f3) (f4)]{};
\end{scope}

\path[->] (4) edge [min distance = 10mm, in = 180,out=120] node[above]{$1$}(4);
\path[->] (5) edge [min distance = 10mm,in=240,out=180] node[below]{$1$}(5);
\path[->] (4) edge [bend left] node[left]{$1$} (5);
\path[->] (5) edge [bend left] node[left]{$1$} (4);

\path[->] (6) edge [min distance = 10mm, in =0, out=60] node[above]{$4$}(6);
\path[->] (7) edge [min distance = 10mm, out=0,in=300] node[below]{$4$}(7);
\path[->] (6) edge [bend left] node[right]{$4$} (7);
\path[->] (7) edge [bend left] node[right]{$4$} (6);

\path[->] (4) edge [bend left] node[above,pos=0.5]{$2$} (6);
\path[->] (5) edge [bend left] node[below,pos=0.5]{$2$} (7);
\path[->] (6) edge [bend left] node[above,pos=0.5]{$3$} (4);
\path[->] (7) edge [bend left] node[above,pos=0.5]{$3$} (5);

\path[->] (4) edge [bend left] node[pos=0.5,above]{$2$} (7);
\path[ ->] (5) edge [bend left] node[pos=0.5,above]{$2$} (6);
\path[->] (7) edge [bend left]  node[pos=0.5,above]{$3$} (4);
\path[->] (6) edge [bend left] node[pos=0.5,above]{$3$} (5);

\node at (-2,-6){$N_2(\mattwo{1}{2}{3}{4})$};
\node at (7,-6){$N_4
\left(\left(\begin{smallmatrix}
1&1&2&2\\
1&1&2&2\\
3&3&4&4\\
3&3&4&4
\end{smallmatrix}\right)\right)$};

\end{tikzpicture}
\caption{Interpolating between the skeleton and blow-up constructions.}
\label{fig:blowup}
\end{center} 
\end{figure}

\subsection{Coherence and Kuratowski-like embeddings}

Thus far we have maintained one-way control over the topology of a network, i.e. the topology of $X$ restricts the collection of weight functions $\w_X$ that are continuous and hence admissible. In certain settings it will be useful to impose control in the other direction, i.e. to restrict the collection of admissible topologies via $\w_X$. We produce this type of topological control via the embeddings $x\mapsto \w_X(x,\cdot), \, x\mapsto \w_X(\cdot,x)$. These are variants of the well-known \emph{Kuratowski embedding} \cite[Chapter 12]{heinonen2012lectures} that is defined for compact metric spaces $(X,d_X)$ as the map $x \mapsto d_X(x,\cdot)$. Such a map gives an isometric embedding of $(X,d_X)$ into the Banach space $l^\infty(X)$.

Let $(X,\w_X) \in \Ngen$. For any $x,x' \in X$, define:
\begin{align}
\delta_X^+(x,x') := & \| \w_X(x,\cdot) - \w_X(x',\cdot)\|_\infty, \qquad 
\delta_X^-(x,x') := \| \w_X(\cdot,x) - \w_X(\cdot,x')\|_\infty \nonumber \\ 
&\delta_X(x,x') := \max(\d_X^+(x,x'),\d_X^-(x,x')).  \label{eq:delta}
\end{align}
Both $\delta_X^+$ and $\delta_X^-$ are pseudometrics generated by the asymmetric Kuratowski-like embeddings, and \gls{deltax}, being the max of two pseudometrics, is also a pseudometric. A related generalization is later provided in Definition \ref{defn:delta-witness}, and a figure demonstrating its use is provided in Figure \ref{fig:circle-delta}.

\begin{definition}
\label{defn:coherence-weak}
Let $(X,\w_X) \in \Ngen$. We say that $X$ has a \emph{weak coherent topology} if for any $x \in X$ and any sequence $(x_n)_n$ in $X$, one has
\begin{equation}x_n \r x \quad \text{ iff } \quad \delta_X(x,x_n) \to 0.
\label{eq:coherence}
\end{equation} 
\end{definition}
\begin{definition}
\label{defn:coherence-strong}
Let $(X,\w_X) \in \Ngen$. We say that $X$ has a \emph{strong coherent topology} if for any $x \in X$, any sequence $(x_n)_n$ in $X$, and any subspace $(S,\w_X|_{S\times S})$ containing $x$ and $\{x_n\}_n$, the following are equivalent:
\begin{itemize}
\item $x_n \to x$
\item $\delta_X(x,x_n) \to 0$ 
\item $\delta_S(x,x_n) \to 0$ 
\end{itemize}
\end{definition}

Metric spaces satisfy strong coherence due to the triangle inequality---in particular, strong coherence encodes a generalized consequence of the linear constraints defining metric spaces. For clarification, let us verify that the statements in Definition \ref{defn:coherence-strong} hold for metric spaces.
In a metric space $(X,d_X)$, to say that a sequence $(x_n)_n$ converges to a point $x \in X$ means that $d_X(x_n,x) \to 0$, i.e. it suffices to check the value of $d_X$ on the points $S:= \{x\} \cup \{x_n\}_n$. It is well-known that the Kuratowski embedding is an isometric embedding for metric spaces, and hence $d_X(x,x_n) \to 0$ is equivalent to $\delta_X(x,x_n) \to 0$. Next, considering the space $(S,d_S)$ where $d_S:= d_X|_{S\times S}$, one applies the triangle inequality to obtain 
\[ \|d_S(x_n,\cdot)| - d_S(x,\cdot)| \|_\infty \to 0 \text{ and } \|d_S(\cdot,x_n)| - d_S(\cdot,x)| \|_\infty \to 0.\]
Hence $\delta_S(x,x_n) \to 0$.
Finally consider $(S,d_S)$ by itself and suppose $\delta_S(x,x_n) \to 0$ for some $x, \{x_n\}_n$ in $S$. Consider the isometric embedding given by the inclusion $(S,d_S) \hookrightarrow (X,d_X)$. As a function, $d_X:X\times X \r \R_+$ extends $d_S$, and the linear constraints coming from the triangle inequality force the following:
\[ \|d_X(x_n,\cdot) - d_X(x,\cdot) \|_\infty \to 0 \text{ and } \|d_X(\cdot,x_n)| - d_X(\cdot,x)| \|_\infty \to 0. \]
The notion of strong coherence captures this property of isometric embedding without specifying, a priori, the form of these linear constraints.

We refer to such topologies as the \emph{weak} or \emph{strong coherent topology generated by $\w_X$}, and refer to networks satisfying the property of Equation (\ref{eq:coherence}) as \emph{weak} or \emph{strong} \emph{coherent networks}. 

\begin{remark}
Directed metric spaces with finite reversibility comprise a collection of objects that are strictly more general than metric spaces and still satisfy strong coherence.
\end{remark}

\begin{definition} 
\label{def:kur-metric}
We refer to $\delta_X$ above as the \emph{canonical pseudometric} of the network $(X,\w_X)$.
\end{definition}

\begin{remark} Kelly \cite{kelly1963bitopological} has studied separation axioms in spaces with multiple topologies, as is the case here with the topologies generated by $\d_X^+, \d_X^-$. Also note that in a first countable space, the topology is determined precisely by the convergent sequences. Generalizing this property leads to \emph{sequential spaces}, which have been studied by Franklin \cite{franklin1965spaces}.
\end{remark}

\begin{remark}[Weak coherence is strictly weaker than strong coherence] 
The construction of $\delta_X$ above is \emph{extrinsic} in the sense that to assign $\delta_X(x,x')$, we used knowledge from \emph{all} of $X$. While $\delta_X$ formally gives us a metric, its construction suffers from consistency issues when embedding into a larger space. For example, suppose we have an injective, weight-preserving map $f:X \to Y$. A priori, we may have 
\[\delta_X(x,x') \neq \max \left( \| \w_Y(f(x),\cdot) - \w_Y(f(x'),\cdot)\|_\infty,
\| \w_Y(\cdot,f(x)) - \w_Y(\cdot,f(x'))\|_\infty \right) = \delta_Y(f(x),f(x'))\]
for $x,x' \in X$, which is an issue that one does not face with isometric embeddings of bona fide metric spaces. This is illustrated in Figure \ref{fig:net-embed}. Consequently, given $x_n \to x$, one may have $\d_X(x,x_n) \to 0$ and $\d_Y(f(x),f(x_n)) \not\to 0$. An alternative viewpoint is the following. When \emph{extending} a metric space (cf. \cite{melleray2007geometry}), say by a single point, one needs to choose distances from the new point to all other points that satisfy the linear constraints posed by the triangle inequality. When extending a network, however, there are no constraints in choosing weights. This flexibility comes at the cost of the inconsistent embedding described above. For strong coherence, however, one has that if $\d_X(x,x_n) = \d_{f(X)}(f(x),f(x_n)) \to 0$, then $\d_Y(f(x),f(x_n)) \to 0$ as well. This property is used below to generalize the familiar result that isometric embeddings of metric spaces are continuous.
\end{remark}

\begin{figure}
\centering
\begin{tikzpicture}
\begin{scope}
\draw[fill] (0,0) circle(0.5ex) node[below] {$b$};
\draw[fill] (0.75,1) circle(0.5ex) node[above] {$a$};
\draw[fill] (1.25,-0.25) circle(0.5ex) node[below] {$c$};
\node at (0.6,0.25) {$Z$};
\draw[-,black,thick] (0,0) -- (0.75,1) node[above left, pos = 0.5] {$1$};
\draw[-,black,thick] (0,0) -- (1.25,-0.25) node[below left, pos = 0.5] {$1$};
\draw[-,black,thick] (0.75,1) -- (1.25,-0.25) node[right, pos = 0.5] {$1$};
\end{scope}
\begin{scope}[xshift = 3 cm, yshift = 0.5cm]
\matrix [inner sep=0pt, nodes={inner sep=.2em},matrix of math nodes,left delimiter=(,right delimiter=),row sep=0.01cm,column sep=0.01cm] (m) {
0 & 1 & 1\\
1 & 0 & 1\\
1 & 1 & 0\\
};
\end{scope}
\begin{scope}[xshift = -3cm,yshift = -3cm]
\draw[fill] (0,0) circle(0.5ex) node[below] {$b$};
\draw[fill] (0.75,1) circle(0.5ex) node[above] {$a$};
\draw[fill] (1.25,-0.25) circle(0.5ex) node[below] {$c$};
\draw[fill] (3,0.25) circle(0.5ex) node[right] {$d$};
\node at (0.6,0.25) {$X$};
\draw[-,black,thick] (0,0) -- (0.75,1) node[above left, pos = 0.5] {$1$};
\draw[-,black,thick] (0,0) -- (1.25,-0.25) node[below left, pos = 0.5] {$1$};
\draw[-,black,thick] (0.75,1) -- (1.25,-0.25) node[above right, pos = 0.5] {$1$};
\draw[-,teal,dashed] (0,0) -- (3,0.25) node[above, pos = 0.6] {$2$};
\draw[-,teal,thick] (1.25,-0.25) -- (3,0.25) node[below, pos = 0.7] {$2$};
\draw[-,teal,thick] (0.75,1) -- (3,0.25) node[above, pos = 0.7] {$2$};
\end{scope}
\begin{scope}[xshift = -5 cm, yshift = -2cm]
\draw[fill=blue!20] (-0.7,0.8) rectangle (0.7,0);
\draw[fill=blue!20] (-0.7,0.8) rectangle (0,-0.9);
\matrix [inner sep=0pt, nodes={inner sep=.2em},matrix of math nodes,left delimiter=(,right delimiter=),row sep=0.01cm,column sep=0.01cm] (m) {
0 & 1 & 1 & 2\\
1 & 0 & 1 & 2\\
1 & 1 & 0 & 2\\
2 & 2 & 2 & 0\\
};
\end{scope}
\begin{scope}[xshift = 3cm,yshift = -3cm]
\draw[fill] (0,0) circle(0.5ex) node[below] {$b$};
\draw[fill] (0.75,1) circle(0.5ex) node[above] {$a$};
\draw[fill] (1.25,-0.25) circle(0.5ex) node[below] {$c$};
\draw[fill] (3,0.25) circle(0.5ex) node[right] {$d$};
\node at (0.6,0.25) {$Y$};
\draw[-,black,thick] (0,0) -- (0.75,1) node[above left, pos = 0.5] {$1$};
\draw[-,black,thick] (0,0) -- (1.25,-0.25) node[below left, pos = 0.5] {$1$};
\draw[-,black,thick] (0.75,1) -- (1.25,-0.25) node[above right, pos = 0.5] {$1$};
\draw[-,purple,dashed] (0,0) -- (3,0.25) node[above, pos = 0.6] {$-3$};
\draw[-,purple,thick] (1.25,-0.25) -- (3,0.25) node[below, pos = 0.7] {$-2$};
\draw[-,purple,thick] (0.75,1) -- (3,0.25) node[above, pos = 0.7] {$-1$};
\end{scope}
\begin{scope}[xshift = 8 cm, yshift = -2cm]
\draw[fill=red!20] (-1.1,0.8) rectangle (1.2,0);
\draw[fill=red!20] (-1.1,0.8) rectangle (0,-0.9);
\matrix [inner sep=0pt, nodes={inner sep=.2em},matrix of math nodes,left delimiter=(,right delimiter=),row sep=0.01cm,column sep=0.01cm] (m) {
0 & 1 & 1 & -1\\
1 & 0 & 1 & -3\\
1 & 1 & 0 & -2\\
-1 & -3 & -2 & 0\\
};
\end{scope}
\draw[->,thick,black] (-0.25,-0.75) -- (-1.5,-1.75);
\draw[->,thick,black] (1.5,-0.75) -- (2.5,-1.75);
\end{tikzpicture}
\caption{``Isometrically" embedding into a larger space is tricky for networks. Let $\delta$ denote the canonical pseudometric obtained via the Kuratowski embedding as in Definition \ref{def:kur-metric}. Consider the embedding of $Z$ into the metric space $X$. We have $\delta_X(a,b) = d_X(a,b) = 1 = \delta_Z(a,b)$. However, $\delta_Y(a,b) = 2 \neq 1 = \delta_Z(a,b)$. In Lemma \ref{lem:intrinsic} we show that this discrepancy can be resolved if the domain and codomain have the same motif sets.}
\label{fig:net-embed}
\end{figure}

\begin{proposition}\label{prop:strong-cohrt-cts} Let $(X,\w_X), (Y,\w_Y)$ be strongly coherent networks, and let $f:X \r Y$ be a weight-preserving map. Then $f$ is continuous. 
\end{proposition}

\begin{proof} Let $V\subseteq Y$ be open. We need to show that $U:= f\inv(V)$ is open. By first-countability, it suffices to show that any sequence converging to a point of $U$ is eventually inside $U$. Let $(x_n)_n$ be a sequence in $X$ converging to $x \in U$. Then we have
\begin{align*}
\norm{\w_Y(f(x_n),\cdot)|_{f(X)} - \w_Y(f(x),\cdot)|_{f(X)} }_\infty = 
\norm{\w_X(x_n,\cdot) - \w_X(x,\cdot)}_\infty \to 0.
\end{align*}
Here the first equality holds because $f$ is weight-preserving, and the limit holds by coherence in $X$. Similarly we also have
\[\norm{\w_Y(\cdot,f(x_n))|_{f(X)} - \w_Y(\cdot,f(x))|_{f(X)} }_\infty \to 0, \text{ and thus } \d_{f(X)}(f(x),f(x_n)) \to 0.\]
By strong coherence, this property transfers to all of $Y$, i.e. $\d_Y(f(x),f(x_n)) \to 0$.
Thus we have $f(x_n) \to f(x)$. But then there must exist $N \in \N$ such that $f(x_n) \in V$ for all $n\geq N$. Then $x_n \in U$ for all $n \geq N$. Since $(x_n)_n$ was arbitrary, it follows that $U$ is open. This concludes the proof. \end{proof}

\begin{corollary}[Uniqueness] 
\label{cor:str-cohrt-unique}
Let $(X,\w_X)$ be a strongly coherent network with a topology $\t_X$. Let $\t'$ be another strongly coherent topology on $X$. Then $\t' = \t_X$. 
\end{corollary}
\begin{proof}
The identity map is weight-preserving, and so we apply Proposition \ref{prop:strong-cohrt-cts} to show that each topology is finer than the other.
\end{proof}

Our ultimate application of coherence (Theorem \ref{thm:cpt-sisom}) will be to improve a statement about weak isomorphism to a statement about strong isomorphism. In such a setting, we have equality of motif sets. It turns out that equality of motif sets allows us to prove a form of the preceding result using \emph{only weak coherence}. We will develop this in Section \ref{sec:cpt-weak-isom}.

\subsection{Skeletons}
\label{sec:cpt-skel}

The preceding construction naturally suggests the equivalence relation $\sim$ defined on $X$ as follows:
\[x \sim x' \text{ iff } \w_X(x,z) = \w_X(x',z) \text{ and } \w_X(z,x) =\w_X(z,x') \text{ for all } z\in X.\]
We refer to the process of taking equivalence classes as \emph{passing to the skeleton}. Note that for a pseudometric $d_X$, one has $d_X(x,x') = 0$ iff $[x]=[x']$, i.e. $d_{X/\sim}$ is a bona fide metric. Next define $\s: X \r X/\sim$ to be the canonical map sending any $x\in X$ to its equivalence class $[x] \in X/\sim$. Also define $\w_{X/\sim}([x],[x']):=\w_X(x,x')$ for $[x],[x'] \in X/\sim$. To check that this map is well-defined, let $a,a' \in X$ be such that $a\sim x$ and $a'\sim x'$. Then,
\[\w_X(a,a') = \w_X(x,a') = \w_X(x,x'),\]
where the first equality holds because $a\sim x$, and the second equality holds because $a'\sim x'$. 
We equip $X/\sim$ with the quotient topology, i.e. a set is open in $X/\sim$ if and only if its preimage under $\s$ is open in $X$. Then $\s$ is a surjective, continuous map.

Observe that when $X$ is compact, $X/\sim$ is the continuous image of a compact space and so is compact. In general, first countability of a topological space is \emph{not} preserved under a surjective continuous map, but it is preserved when the surjective, continuous map is also open \cite[p. 27]{counterexamples}. The following proposition gives a sufficient condition on $X$ which will ensure that $X/\sim$ is first countable.

\begin{proposition}\label{prop:quotient-open} 
Suppose $(X,\w_X) \in \Ngen$ has a weakly coherent topology. Then the map $\s: X \r X/\sim$ is an open map, i.e. it maps open sets to open sets.
\end{proposition}

\begin{proof}[Proof of Proposition \ref{prop:quotient-open}] 
Let $U\subseteq X$ be open. Because $X/\sim$ has the quotient topology, we need to show $\s\inv(\s(U))$ is open. For convenience, define $V:= \s\inv(\s(U))$. By first countability, it suffices to show that any sequence converging to a point in $V$ is eventually inside $V$. Let $(v_n)_{n\in \N}$ be any sequence in $X$ converging to a point $v \in V$. Note that we have $\s(v)= [v] = [x]$ for some $x \in U$. Because $x \sim v$, we have by weak coherence that $v_n \r x$ as well. But because $x \in U$ and $U$ is open, there exists $N \in \N$ such that $v_n \in U \subseteq V$ for all $n \geq N$. Thus $V$ is open. This concludes the proof. \end{proof}

\begin{definition}[The skeleton (\gls{skx}) of a compact network ] Suppose $(X,\w_X) \in \Ncom$ has a weakly coherent topology. The \emph{skeleton of $X$} is defined to be $(\skel(X),\w_{\sk(X)}) \in \Ncom$, where $\skel(X):= X/\sim,$ and 
\begin{align*}
\w_{\sk(X)}([x],[x'])&:=\w_X(x,x') \text{ for all } [x],[x']\in \skel(X).
\end{align*}

Observe that $\skel(X)$ is compact because $X$ is compact, and first countable by Proposition \ref{prop:quotient-open} and the fact that the image of first countable space under an open, surjective, and continuous map is also first countable \cite[p. 27]{counterexamples}. Furthermore, $\w_{\sk(X)}$ is well defined by the definition of $\sim$. 
\end{definition}

\begin{restatable}[Skeletons are terminal]{theorem}{skelterminal}
\label{thm:skel-terminal} 
Let $(X,\w_X) \in \Ncom$ be a network with a weakly coherent topology. Then $(\skel(X),\w_{\sk(X)}) \in \Ncom$ is terminal in $\mf{p}(X)$.
\end{restatable} 

The proof of this and related results are provided in Section \ref{sec:pf-networks-skel-terminal}.

\subsection{The second network distance}
Even though the definition of $\dn$ is very general, in some restricted settings it may be convenient to consider a network distance that is easier to formulate. For example, in computational purposes it suffices to assume that we are computing distances between finite networks. Also, a potential reduction in computational cost is obtained if we restrict ourselves to computing distortions of bijections instead of general correspondences. The next definition (compare with Definition \ref{defn:dno}) arises from such considerations.

\begin{definition}[The second network distance \gls{dnh}]
\label{defn:dnh}
Let $(X,\w_X), (Y,\w_Y) \in \Ngen$ be such that $|X| = |Y|$. Then define: 
\[\dnh(X,Y) := \frac{1}{2}\inf_{\ph}\sup_{x,x'\in X} \big|\w_X(x,x')-\w_Y\left(\ph(x),\ph(x')\right)\big|,\] 
where $\ph:X\r Y$ ranges over all bijections from $X$ to $Y$ (at least one bijection exists because $|X|=|Y|$). 
\end{definition}
In analogy with the case for $\dn$, we note that $X\cong^s Y$ implies $\dnh(X,Y)=0$, and also that $\dnh$ satisfies symmetry and triangle inequality. In the setting of finite networks, and in contrast with $\dn$, we have that $\dnh(X,Y)=0$ also immediately implies $X\cong^s Y$. It turns out via Example \ref{prop:perm} that $\dn$ and $\dnh$ agree on networks over two nodes. However, the two notions do not agree in general. In particular, a minimal example where $\dn\neq \dnh$ occurs for three node networks, as we show in Remark \ref{remark:bijection}.

\begin{restatable}[Networks with two nodes]{example}{perm}
\label{prop:perm} Let $(X,\w_X), (Y,\w_Y) \in \Ncal$ where $X = \set{x_1,x_2}$ and $Y= \set{y_1,y_2}$. Then we claim $\dn(X,Y) = \dnh(X,Y).$
Furthermore, if $X = N_2\left(\mattwo{\a}{\d}{\b}{\g}\right)$ and $Y= N_2\left(\mattwo{\a'}{\d'}{\b'}{\g'}\right)$, then we have the explicit formula:
\begin{align*}\label{eq:dn1-2}
&\dn(X,Y) = \frac{1}{2}\min\left(\Ga_1, \Ga_2\right), \text{ where}\\
\Ga_1 &= \max\left(|\a-\a'|,|\b-\b'|,|\d-\d'|,|\g-\g'|\right),\\
\Ga_2 &= \max\left(|\a-\g'|,|\g-\a'|,|\d-\b'|,|\b-\d'|\right).
\end{align*}
\end{restatable}
Details for this calculation are in Section \ref{sec:pf-perm}.

\begin{remark}[A three-node example where $\dn\neq \dnh$]\label{remark:bijection}
Assume $(X,\omega_X)$ and $(Y,\omega_Y)$ are two networks with the same cardinality. Then 
\[\dn(X,Y) \leq \dnh(X,Y).\]

The inequality holds because each bijection induces a correspondence, and we are minimizing over all correspondences to obtain $\dn$. However, the inequality may be strict, as demonstrated by the following example.
Let $X = \set{x_1, \ldots, x_3}$ and let $Y=\set{y_1, \ldots, y_3}$. Define $\w_X(x_1,x_1) = \w_X(x_3,x_3) = \w_X(x_1,x_3) = 1, \w_X = 0$ elsewhere, and define $\w_Y(y_3,y_3) = 1, \w_Y= 0$ elsewhere. In terms of matrices, $X=N_3(\Sigma_X)$ and $Y=N_3(\Sigma_Y)$, where
\[\Sigma_X = \left(\begin{smallmatrix}
1&0&1\\
0&0&0\\
0&0&1
\end{smallmatrix}\right)\,\,\mbox{and}\,\,\Sigma_Y = \left(\begin{smallmatrix}
0&0&0\\
0&0&0\\
0&0&1
\end{smallmatrix}\right).\]

Define $\Ga(x,x',y,y') = |\w_X(x,x')-\w_Y(y,y')|$ for $x,x' \in X$, $y,y' \in Y$. Let $\ph$ be any bijection. Then we have: 
\begin{align*}
\max_{x,x' \in X}\Ga(x,x',\ph(x),\ph(x'))
&=\max\{\Ga(x_1,x_3,\ph(x_1),\ph(x_3)), \Ga(x_1,x_1,\ph(x_1),\ph(x_1)),\\
&\qquad \Ga(x_3,x_3,\ph(x_3),\ph(x_3)),
\Ga(\ph^{-1}(y_3),\ph^{-1}(y_3),y_3,y_3)\}\\
&= 1.
\end{align*}
So $\dnh(X,Y) = \frac{1}{2}.$ On the other hand, consider the correspondence \[R = \{(x_1,y_3),(x_2,y_2),(x_3,y_3),(x_2,y_1)\}.\]
Then $\max_{(x,y),(x'y') \in R}|\w_X(x,x')-\w_Y(y,y')| = 0$. Thus $\dn(X,Y) = 0 < \dnh(X,Y)$.
\end{remark}

\begin{example}[Networks with three nodes] Let $(X,\w_X), (Y,\w_Y) \in \Ncal$, where we write $X=\set{x_1,x_2,x_3}$ and $Y=\set{y_1,y_2,y_3}$. Because we do not necessarily have $\dn= \dnh$ on three node networks by Remark \ref{remark:bijection}, the computation of $\dn$ becomes more difficult than in the two node case presented in Example \ref{prop:perm}. A certain reduction is still possible, which we present next. Consider the following list $\mc{L}$ of matrices representing correspondences, where a $1$ in position $(i,j)$ means that $(x_i,y_j)$ belongs to the correspondence. 
\begin{center}
\begin{tabular}{ ccccc } 
 $\mattres{1}{}{}{}{1}{}{}{}{1}$ & 
 $\mattres{1}{}{}{}{}{1}{}{1}{}$ &
 $\mattres{1}{}{}{1}{}{}{}{1}{1}$ &
 $\mattres{1}{1}{}{}{}{1}{}{}{1}$ &
 $\mattres{1}{}{}{}{1}{1}{1}{}{}$\\[2ex] 
 
 \hline\\
 $\mattres{}{1}{}{1}{}{}{}{}{1}$ & 
 $\mattres{}{1}{}{}{}{1}{1}{}{}$ &
 $\mattres{}{1}{}{}{1}{}{1}{}{1}$ &
 $\mattres{}{1}{1}{1}{}{}{1}{}{}$ &
 $\mattres{}{1}{}{1}{}{1}{}{1}{}$\\[2ex] 
 
 \hline\\
 $\mattres{}{}{1}{}{1}{}{1}{}{}$ & 
 $\mattres{}{}{1}{1}{}{}{}{1}{}$ &
 $\mattres{}{}{1}{}{}{1}{1}{1}{}$ &
 $\mattres{1}{}{1}{}{1}{}{}{1}{}$ &
 $\mattres{}{}{1}{1}{1}{}{}{}{1}$\\
\end{tabular}
\end{center}

Now let $R\in \Rsc(X,Y)$ be any correspondence. Then $R$ contains a correspondence $S\in \Rsc(X,Y)$ such that the matrix form of $S$ is listed in $\mc{L}$. Thus $\dis(R) \geq \dis(S)$, since we are maximizing over a larger set. It follows that $\dn(X,Y)$ is obtained by taking $\argmin \tfrac{1}{2}\dis(S)$ over all correspondences $S\in \Rsc(X,Y)$ with matrix forms listed in $\mc{L}$. 

For an example of this calculation, let $S$ denote the correspondence $\set{(x_1,y_1),(x_2,y_2),(x_3,y_3)}$ represented by the matrix $\mattres{1}{}{}{}{1}{}{}{}{1}$. Then $\dis(S)$ is the maximum among the following:

\begin{center}
\begin{tabular}{ccc}
$|\w_X(x_1,x_1) - \w_Y(y_1,y_1)|$ & 
$|\w_X(x_1,x_2) - \w_Y(y_1,y_2)|$ & 
$|\w_X(x_1,x_3) - \w_Y(y_1,y_3)|$\\

$|\w_X(x_2,x_1) - \w_Y(y_2,y_1)|$ & 
$|\w_X(x_2,x_2) - \w_Y(y_2,y_2)|$ & 
$|\w_X(x_2,x_3) - \w_Y(y_2,y_3)|$\\

$|\w_X(x_3,x_1) - \w_Y(y_3,y_1)|$ & 
$|\w_X(x_3,x_2) - \w_Y(y_3,y_2)|$ & 
$|\w_X(x_3,x_3) - \w_Y(y_3,y_3)|$.
\end{tabular}
\end{center}
\end{example}

The following proposition provides an explicit connection between $\dn$ and $\dnh$. The proof is provided in Section \ref{sec:pf-networks-bijection}, and an illustration is provided in Figure \ref{fig:bijection}.

\begin{restatable}{proposition}{bijection}
\label{prop:bijection} Let $(X,\w_X),(Y,\w_Y) \in\Ngen$. Then, 
\[\dn(X,Y) = \inf\set{\dnh(X',Y') : X', Y' \in \Ngen, \; X' \congwt  X,\ Y' \congwt  Y, \text{ and } |X'| = |Y'|}.\]
\end{restatable}

The proof of this result is provided in Section \ref{sec:pf-networks-bijection}. In words, the result states that the $\dn$ distance between two networks $X,Y$ may be achieved by blowing up to networks $X',Y'$ of equal cardinality and then optimizing over bijections. Note that when restricted to $X,Y \in \Ncal$, the Type II weak isomorphism can be interchanged with Type I weak isomorphism. 

\begin{figure}[b]
\center
\begin{tikzpicture}
\begin{scope}[every node/.style={circle,thick,draw}]
\footnotesize
    \node (y1) at (5,2) {$y_1$};
    \node (y2) at (4,0) {$y_2$};
    \node (y3) at (6,0) {$y_3$};
\end{scope}

\begin{scope}
\node (Y) at (5,-1) {$Y$};
\end{scope}

\begin{scope}[every node/.style={circle,thick,draw},yshift=1cm]
\footnotesize
\node (x1) at (0,0) {$x_1$};
\node (x2) at (2,0) {$x_2$};
\end{scope}

\begin{scope}[yshift=1cm]
\node (X) at (1,-1) {$X$};
\end{scope}

\begin{scope}[every node/.style={circle,thick,draw},xshift=9cm]
\footnotesize
    \node (y11) at (5,2) {$y_1$};
    \node (y21) at (4,0) {$y_2$};
    \node (y31) at (6,0) {$y_3$};
\end{scope}

\begin{scope}[xshift=9cm]
\node (Yk) at (5,-1) {$Y$};
\end{scope}

\begin{scope}[every node/.style={circle,thick,draw},yshift=2cm,xshift=10cm,
rotate=270]
\footnotesize
\node (x11) at (0,0) {$x_1$};
\node (x21) at (2,-1) {$x_{21}$};
\node (x22) at (2,1) {$x_{22}$};
\end{scope}

\begin{scope}[xshift=10cm]
\node (Xk) at (0,-1) {$Z$};
\end{scope}

\begin{scope}[              
              every edge/.style={draw=black,very thick}]
    \path [-] (x1) edge node[above]{$3$} (x2);
    \path [-] (y1) edge node[left] {$5$} (y2);
    \path [-] (y1) edge node[right] {$5$} (y3);
    \path [-] (y2) edge node[above] {$1$} (y3);
    \path [-] (x11) edge node[left] {$3$} (x21);
    \path [-] (x11) edge node[right] {$3$} (x22);
    \path [-,dotted] (x21) edge node[above] {$0$} (x22);
    \path [-] (y11) edge node[left] {$5$} (y21);
    \path [-] (y11) edge node[right] {$5$} (y31); 
    \path [-] (y21) edge node[above] {$1$} (y31); 
\end{scope}

\begin{scope}[
              every node/.style={fill=white,circle},
              every edge/.style={draw=orange,very thick}]
    \path [->] (x11) edge[bend left] (y11);
    \path [->] (x21) edge[bend right] (y31);
    \path [->] (x22) edge[bend left] (y21);
    
\end{scope}

\end{tikzpicture}
\caption{The two networks on the left have different cardinalities, but computing correspondences shows that $\dn(X,Y)=1$. Similarly one computes $\dn(X,Z) = 0$, and thus $\dn(Y,Z) = 0$ by triangle inequality. On the other hand, the bijection given by the red arrows shows $\dnh(Y,Z)=1$. Applying Proposition \ref{prop:bijection} then recovers $\dn(X,Y) = 1$.}
\label{fig:bijection}
\end{figure}

\begin{remark}[Computational aspects of $\dn$ and $\dnh$]\label{rem:compute-hard}
The relation between $\dn$ and $\dnh$ in the setting of finite networks has been leveraged in \cite{gwa} to compute network Fr\'{e}chet means and geodesic principal components. From a computational complexity perspective, even though $\dnh$ has a simpler formulation than $\dn$, computing $\dnh$ still turns out to be a hard problem (cf. Section \ref{sec:np-hard}).

Instead of trying to compute $\dn$, we will focus on finding network invariants that can be computed easily. This is the content of Section \ref{sec:inv-conv}. For each of these invariants, we will prove a stability result to demonstrate its validity as a proxy for $\dn$.

\end{remark}

\subsection{Special families: dissimilarity networks and directed metric spaces} 
\label{sec:dissim-rev}

The second network distance $\dnh$ that we introduced in the previous section turned out to be compatible with strong isomorphism. Interestingly, by narrowing down the domain of $\dn$ to the setting of compact dissimilarity networks (cf. Definition \ref{def:nets-dissim}), we obtain a subfamily of $\Ngen$ where $\dn$ is compatible with strong isomorphism.

\begin{theorem}[\cite{carlsson2013axiomatic}] \label{thm:dis-dN1}
The restriction of $\dn$ to $\Ncal^\mathrm{dis}$ is a metric modulo strong isomorphism.
\end{theorem}

The expression for $\dn$ was used in the context of $\mathcal{FN}^{\mathrm{dis}}$ in \cite{carlsson2013axiomatic,clust-net} to study the stability properties of hierarchical clustering methods on metric spaces and directed dissimilarity networks. That setting is considerably simpler than the situation in this paper, because in general we allow $\dn(X,Y) = 0$ for $X,Y \in \Ngen$ even when $X$ and $Y$ are not strongly isomorphic. In Theorem \ref{thm:dis-cpt-sisom} below, we provide an extension of Theorem \ref{thm:dis-dN1} to a class of compact dissimilarity networks that contains all finite dissimilarity networks.

The following definition gives a continuous relaxation of the triangle inequality that is motivated by applications to computer vision \cite{fagin1998relaxing}.

\begin{definition}[$\Psi$-relaxed triangle inequality] 
\label{def:psi-relaxed}

Let $\Psi: \R_+ \times \R_+ \r \R_+$ be a continuous function such that $\Psi(0,0) = 0$. 
A dissimilarity network $(X,A_X)$ is said to satisfy a \emph{$\Psi$-relaxed} triangle inequality if we have
\[A_X(x,x') \leq \Psi \big(A_X(x,x''),A_X(x',x'') \big) \text{ for all } x,x',x'' \in X.\] 
This condition automatically encodes a notion of reversibility:
\begin{alignat*}{2}
A_X(x,x') &\leq \Psi\big( A_X(x,x),A_X(x',x) \big) &&= \Psi\big(0,A_X(x',x)\big),\\
A_X(x',x) &\leq \Psi\big( A_X(x',x'),A_X(x,x') \big) &&= \Psi\big(0,A_X(x,x')\big).
\end{alignat*}
In the sequel, whenever we write ``$(X,A_X) \in \Ngendis$ has a $\Psi$-relaxed triangle inequality" without explicit reference to a map $\Psi$, we mean that there exists a function $\Psi:\R_+ \times \R_+ \r \R_+$ satisfying the conditions above.

\end{definition}

\begin{remark} Any finite dissimilarity network is finitely reversible and has a $\Psi$-relaxed triangle inequality. For example, $\Psi$ can be taken to be a bump function that vanishes outside $\R_+^2 \setminus U$---where $U$ is some open set containing $\im(A_X)$ and excluding $(0,0)$---and constant at $\max_{x,x' \in X}A_X(x,x')$ on $U$. 
\end{remark}

We end this section with a strengthening of Theorem \ref{thm:dis-dN1} to the setting of compact networks. Recall that the interesting part of Theorem \ref{thm:dis-dN1} was to show that $\dn(X,Y) = 0 \implies X\cong^s Y$; we generalize this result to a certain class of \emph{compact} dissimilarity networks. The proof is provided in Section \ref{sec:pf-dis-cpt-sisom}.

\begin{restatable}{theorem}{discptsisom}
\label{thm:dis-cpt-sisom}
Let $(X,A_X), (Y,A_Y) \in \Ncomdis$ be equipped with the forward-open topologies induced by $A_X$ and $A_Y$, respectively. Suppose also that at least one of the two networks has a $\Psi$-relaxed triangle inequality. Then $\dn(X,Y)=0 \implies X \cong^s Y$.
\end{restatable}

\begin{remark}[Generalizations of Theorem \ref{thm:dis-dN1}] For any finite dissimilarity network $(X,A_X)$, the discrete topology is precisely the topology induced by $A_X$. We have already stated before that finite dissimilarity networks trivially satisfy finite reversibility and $\Psi$-relaxed triangle inequality. It folows that Theorem \ref{thm:dis-cpt-sisom} is a bona fide generalization of Theorem \ref{thm:dis-dN1}. 
 \end{remark}

\subsection{Network models: the directed circles}
\label{sec:dir-s1}

\begin{figure}
\centering
\begin{tikzpicture}[every node/.style={font=\footnotesize, minimum size=1cm}]
\tikzset{>={Latex[length=3mm,width=2mm,orange]}}
\begin{scope}
\draw[
        decoration={markings, 
        mark=at position 0.325 with {\arrow{>}},
        mark=at position 0.825 with {\arrow{>}} },
        postaction={decorate}
        ]
        (0,0) circle (1.8);
\node at (0,0){$(\us^1,\wus)$};     
\draw [|->,black,thick,domain=30:75] plot ({2*cos(\x)}, {2*sin(\x)});
\end{scope}  
\begin{scope}[xshift = 2in]
\draw[
        decoration={markings, 
        mark=at position 0.325 with {\arrow{>}},
        mark=at position 0.825 with {\arrow{>}} },
        postaction={decorate}
        ]
        (0,0) circle (1.8);
\node at (0,0){$(\us^1,\rus)$};     
\draw[fill] (30:2cm) circle (0.5ex);
\draw [->,black,thick,domain=30:75] plot ({2*cos(\x)}, {2*sin(\x)});
\draw [->,black,thick,domain=30:10] plot ({2*cos(\x)}, {2*sin(\x)});
\end{scope}
\end{tikzpicture}
\caption{The directed circle $(\us^1,\wus)$ and the directed circle $(\us^1,\rus)$ with reversibility $\rho$, for some $\rho \in [1,\infty)$ (cf. Definition \ref{def:reversibility}). The arrows show that traveling in a clockwise direction is possibly only in the directed circle with reversibility $\rho$. However, this incurs a penalty modulated by $\rho$, hence the shorter arrow in the clockwise direction.
}
\label{fig:dir-s1}
\end{figure}

The collections $\Ngen$, $\Ncom$, and $\Ncal$ contain the collections of all metric spaces, compact metric spaces, and finite metric spaces, respectively. It is interesting to identify networks in these families that are not just metric spaces. Here we provide models of \emph{directed circles}, which are infinite, asymmetric networks. 
See Figure \ref{fig:dir-s1} for an illustration, and also Section \ref{sec:dir-s1-pers} for applications of hierarchical clustering and persistent homology methods to these asymmetric network models. 

Define $\us^1:= \set{ e^{i\theta}\in \C : \theta\in [0,2\pi)},$ i.e. the standard unit circle in the complex plane with the standard topology. For any $x,y \in \us^1$, define $\w_{\us^1}(x,y)$ to be the counterclockwise geodesic distance (i.e. arc length) from $x$ to $y$. Then $(\text{\gls{us},\gls{wus}})$ becomes a model of a directed circle.

An issue that arises now is that $\w_{\us^1}$ is not continuous with respect to the standard topology on $\us^1$. Thus $(\us^1,\w_{\us^1})$ is a set-function pair that can be accepted as input into $\dn$-computations, but does not enjoy the nicer properties of $\Ngen$. One recourse is to define $(\us^1,\w_{\us^1})$ to be equipped with the discrete topology. This topology is first countable and makes $\w_{\us^1}$ continuous, and so $(\us^1,\w_{\us^1}) \in \Ngen$. 

However, the resulting network is not compact. Moreover, a coarser topology does not work to make $(\us^1,\w_{\us^1})$ fit the framework of $\Ncom$. To see why, let $x \in \us^1$. Suppose $\w_{\us^1}$ is continuous with respect to some topology on $\us^1$. Fix $0 < \e \ll 2\pi$, and define $V:= \w_{\us^1}\inv[(-\e,\e)]$. Then $V$ is open in the product topology, and in particular contains $(x,x)$. Since $V$ is a union of open rectangles, there exists an open set $U \subseteq \us^1$ such that $(x,x) \in U\times U \subseteq V$. Suppose towards a contradiction that $U \neq \{x\}$. Then there exists $y \in U$, for some $y \neq x$. Because $(x,y) \in U\times U$, one has $\w_{\us^1}(x,y) \in (0,\e)$. But then $\w_{\us^1}(y,x) \in (2\pi-\e,2\pi)$, which is a contradiction because $(y,x) \in U\times U \subseteq \w_{\us^1}\inv[(-\e,\e)]$. Thus $U=\{x\}$, and hence the topology on $(\us^1,\w_{\us^1})$ necessarily contains all singletons as open sets. Such a topology cannot be compact.

We would still like a model of a directed circle with sufficient topological regularity to belong to $\Ncom$. We therefore produce a family of \emph{directed circles with finite reversibility} as follows. Fix a reversibility parameter $\rho \geq 1$ (cf. Definition \ref{def:reversibility}). Then for each $x,y \in \us^1$, define \gls{rus} as:
\[\rus(x,y):= \min\lp \w_{\us^1}(x,y), \rho\w_{\us^1}(y,x)\rp .\]
In particular, $\rus$ has reversibility $\rho$. 
Finally, we equip $\us^1$ with the standard subspace topology generated by the open balls in $\C$. In this case, $\us^1$ is compact and first countable. It remains to check that $\rus$ is continuous. Before proceeding to the next result, we set the notation $d_{\us^1}$ to denote the standard (not necessarily counterclockwise) geodesic distance on $\us^1.$ Note that $d_{\us^1}(x,y) = \min ( \w_{\us^1}(x,y),\w_{\us^1}(y,x))$, and so we always have $\rus(x,y) \leq \rho d_{\us^1}(x,y).$

\begin{proposition}
\label{prop:dir-S1-cts} $\rus: \us^1 \times \us^1 \r \R$ is continuous.
\end{proposition}

\begin{proof}[Proof of Proposition \ref{prop:dir-S1-cts}] It suffices to show that the preimages of basic open sets under $\rus$ are open. Let $(a,b)$ be an open interval in $\R$, and let $(x,y)\in \rus\inv[(a,b)]$. Set 
$\e:= \frac{1}{2\rho}\min\{ \rus(x,y) - a, b - \rus(x,y) \}.$ Let $x' \in B(x,\e)$ and $y' \in B(y,\e)$, i.e. $d_{\us^1}(x,x')< \e$ and $d_{\us^1}(y,y')< \e$. It suffices to show $\rus(x',y') \in \rus\inv[(a,b)]$, for which, in turn, it suffices to show:
\[| \rus(x,y) - \rus(x',y') | < 2\rho\e =  \min\{ \rus(x,y) - a, b - \rus(x,y) \}.\]

\begin{claim} Let $x,y,z \in \us^1$. Then
$\rus(x,z) \leq \rus(x,y) + \rus(y,z).$
\end{claim}

Assuming the claim, we have $\rus(x,y) \leq \rus(x,x') + \rus(x',y') + \rus(y',y)$ and so 
\[\rus(x,y) - \rus(x',y') \leq \rus(x,x') + \rus(y',y) \leq \rho d_{\us^1}(x,x') + \rho d_{\us^1}(y',y) = \rho(d_{\us^1}(x,x') + d_{\us^1}(y,y')).\] 
Similarly we have $\rus(x',y') - \rus(x,y) \leq \rho d_{\us^1}(x',x) + \rho d_{\us^1}(y,y') = \rho(d_{\us^1}(x,x') + d_{\us^1}(y,y')).$ We thus have $|\rus(x,y) - \rus(x',y')| < 2\rho\e.$ Since $x'\in B(x,\e), y' \in B(y,\e)$ were arbitrary, it follows that $\rus\inv[(a,b)]$ is open. 

To conclude the proof, we need to verify the triangle inequality in the claim. First note that for $x,y,z \in \us^1$, we have the simpler triangle inequality $\w_{\us^1}(x,z) \leq \w_{\us^1}(x,y) + \w_{\us^1}(y,z)$. To see this, note that the counterclockwise arc from $x$ to $z$ is either equal to, or strictly included in, the union of the counterclockwise arcs from $x$ to $y$ and from $y$ to $z$. Next we treat four cases corresponding to the ways in which $\rus(x,y), \rus(y,z)$ may be realized. 

First suppose $\rus(x,y)=\w_{\us^1}(x,y)$ and $\rus(y,z)=\w_{\us^1}(y,z)$. Then $\rus(x,z) \leq \w_{\us^1}(x,z) \leq \w_{\us^1}(x,y)+ \w_{\us^1}(y,z) = \rus(x,y) + \rus(y,z)$. Here the first inequality follows by definition and the second follows by the triangle inequality for $\w_{\us^1}$. Similarly, for the case $\rus(x,y)=\rho\w_{\us^1}(y,x)$ and $\rus(y,z)=\rho \w_{\us^1}(z,y)$, we have $\rus(x,z)\leq \rho \w_{\us^1}(z,x) \leq \rho( \w_{\us^1}(z,y) + \w_{\us^1}(y,x) ) =  \rus(y,z) + \rus(x,y).$

Next suppose $\rus(x,y)=\w_{\us^1}(x,y)$ and $\rus(y,z)=\rho\w_{\us^1}(z,y).$ Note that either $z$ belongs to the counterclockwise arc from $x$ to $y$, or $x$ belongs to the counterclockwise arc from $z$ to $y$. In the first subcase, we have $\rus(x,z) = \w_{\us^1}(x,z) \leq \w_{\us^1}(x,y) \leq \rus(x,y) + \rus(y,z)$. Here we use the assumption $\rus(x,y)=\w_{\us^1}(x,y)$. In the second subcase, we have $\rus(x,z) \leq \rho \w_{\us^1}(z,x) \leq \rho \w_{\us^1}(z,y) \leq  \rus(y,z) + \rus (x,y).$ Here we use the assumption $\rus(y,z)=\rho\w_{\us^1}(z,y).$

Finally we have the case $\rus(x,y)=\rho \w_{\us^1}(y,x)$ and $\rus(y,z)=\w_{\us^1}(y,z).$ After splitting into subcases, this case proceeds analogously to the previous case. We thus prove the claim, and conclude the proof. \end{proof}

\begin{definition} 
Let $\rho \in [1,\infty)$. We define the \emph{directed unit circle with reversibility $\rho$} to be $(\us^1,\rus).$ This is a compact, asymmetric network in $\Ncom$, specifically in $\Ncomdis$. Figure \ref{fig:dir-s1} provides an illustration of such a network alongside $(\us^1,\w_{\us^1})$ for comparison.
\end{definition}

\begin{remark}[Directed circle with finite reversibility---forward-open topology version]
Instead of using the subspace topology generated by the standard topology on $\C$, we can also endow $(\us^1,\rus)$ with the forward-open topology generated by $\rus$. The open balls in this topology are precisely the open balls in the subspace topology induced by the standard topology, the only adjustment being the ``center" of each ball. The directed metric space $(\us^1,\rus)$ equipped with the forward-open topology is another example of a compact, asymmetric network in $\Ncomdis$.
\end{remark}

\begin{figure}
\centering
\includegraphics[width=0.4\textwidth]{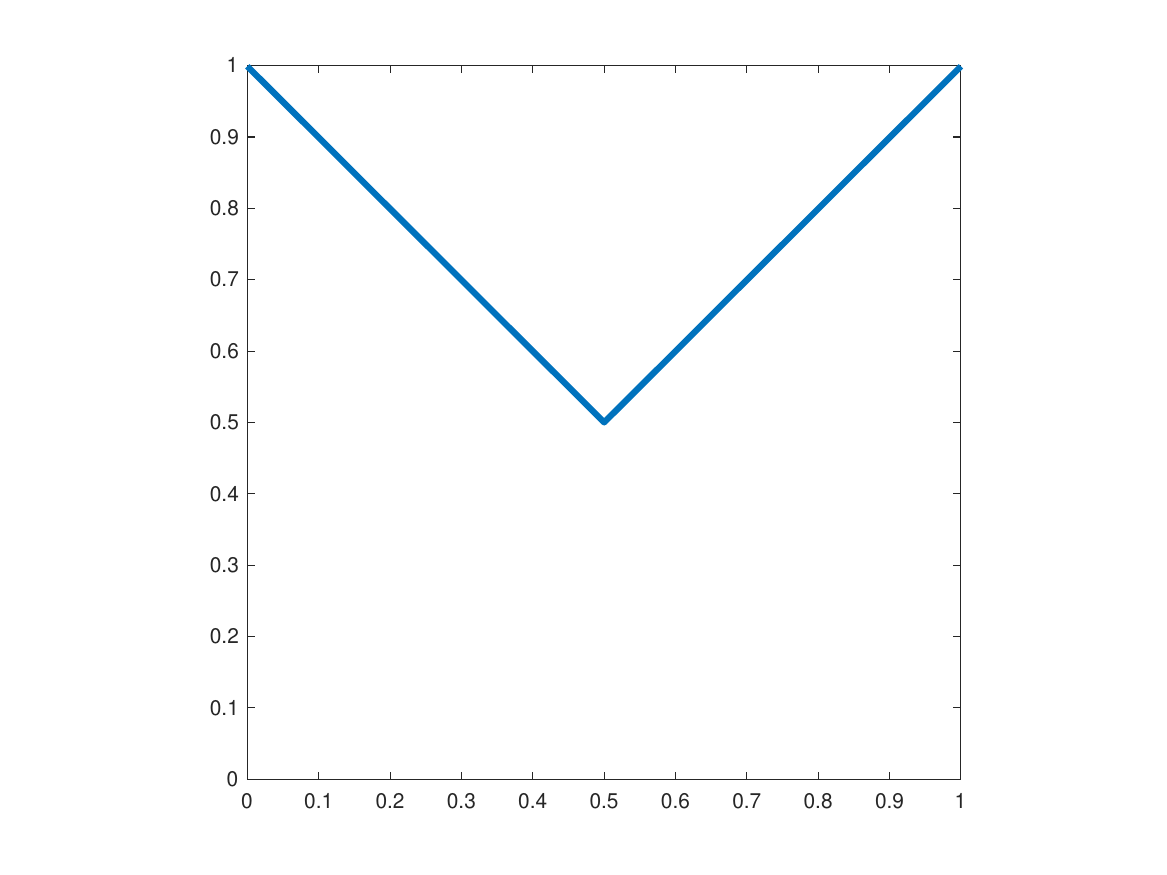}
\includegraphics[width=0.4\textwidth]{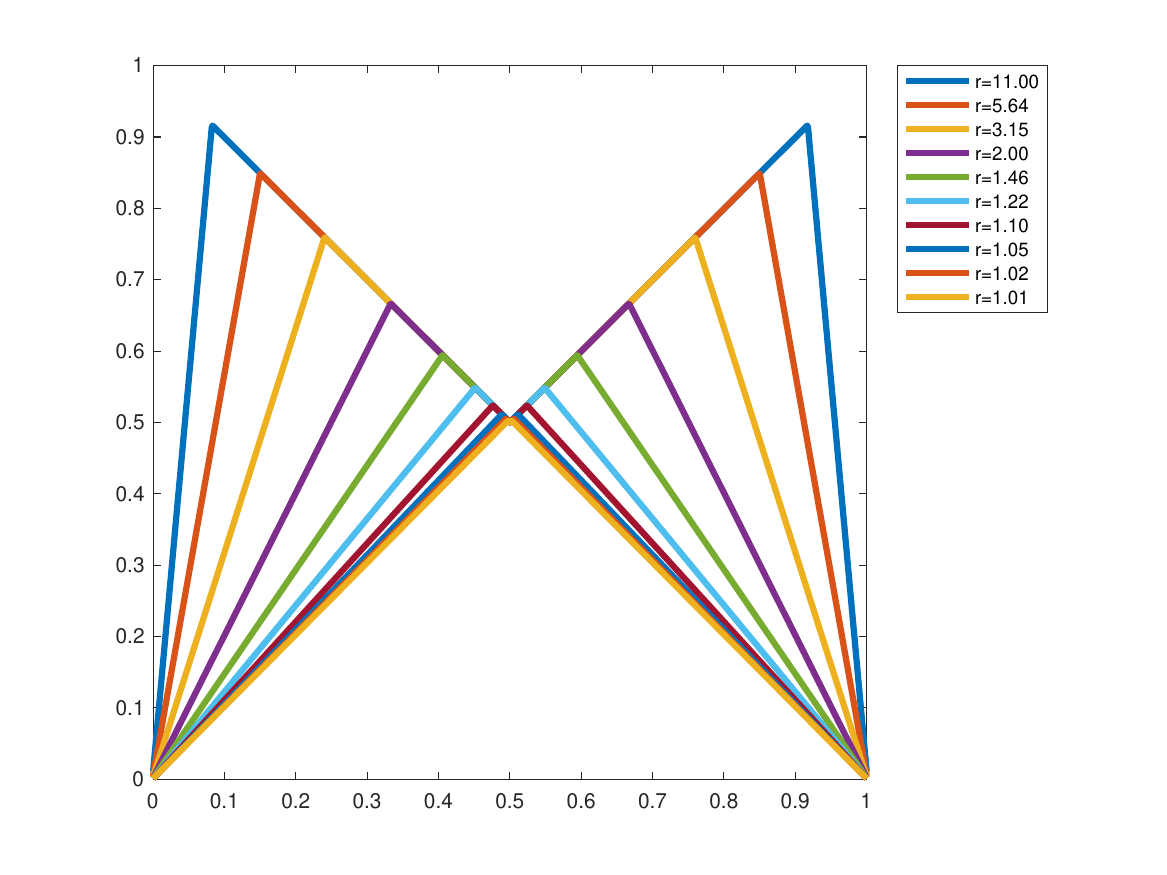}
\caption{To provide intuition for the Kuratowski-like embedding metric $\delta$ (cf. Definition \ref{def:kur-metric}), we plot $\d(e^{0},e^{i\theta})$ using $\w_{\us^1}$ (\textbf{left}) and $\rus$ (\textbf{right}) as $\theta$ varies along $\us^1$. We have rescaled $\us^1$ to be parametrized by $[0,1)$ for convenience. The figure on the left shows that for the strictly directed circle, $\d(e^0,\cdot)$ is minimized at the antipode $e^{i\pi}$, and that for all $\e$ values less than the diameter, the $\d$-ball of radius $\e$ centered at $e^0$ contains only $e^0$. The same holds for any other point $e^{i\theta}$ by rotational symmetry. Thus the metric topology induced by the Kuratowski-like embedding is actually discrete. Compare this with the figure on the right. The different colored lines correspond to different choices of $\rho$, taken on a logarithmic scale. We observe that $\d$ takes on an M-shape that converges to the shape of the geodesic distance as $\rho \downarrow 1$. The shape of the plot also tells us that for any finite reversibility parameter, the $\d$-ball of radius $\e$ centered at $e^0$ (and thus any other point by rotational symmetry) is centrally symmetric, as is the case for a ball of radius $\e$ using the standard geodesic metric.}
\label{fig:circle-delta}
\end{figure}

\subsection{The complexity of computing $\dn$}
\label{sec:np-hard}

Computing $\dn$ is at least as hard as computing the Gromov-Hausdorff distance between metric spaces, which has been shown by Schmiedl to not be computable in polynomial time unless P=NP \cite{schmiedl2017computational,schmiedl}. Moreover, Schmiedl showed that even obtaining a 3-approximation of the Gromov-Hausdorff distance is not possible in polynomial time, unless P=NP. This result was obtained by starting with an instance of the 3-partition problem and constructing an instance of the Gromov-Hausdorff distance such that a 3-approximation of its optimal solution would solve the starting 3-partition problem.

\begin{remark}\label{rem:dn-NP} For quite some time prior to Schmiedl's proof, it was remarked \cite{dgh-props} that computing the Gromov-Hausdorff distance is reminiscent of the \emph{quadratic bottleneck assignment problem} (QBAP) \cite{qap-book1}, which is NP-hard \cite{burkard2009assignment}. We reproduce this argument below in the case of $\dn$.

By Remark \ref{remark:bijection} and Proposition \ref{prop:bijection} we know that it is possible to obtain an upper bound on $\dn$, in the case $|X| = |Y|$, by using $\dnh$. This problem turns out to be a case of the \emph{quadratic bottleneck assignment problem} (QBAP) \cite{qap-book1}. Let $X=\set{x_1,\ldots, x_n}$ and let $Y=\set{y_1,\ldots,y_n}$. Let $\Pi$ denote the set of all $n\times n$ permutation matrices. Note that any $\mathbf{\pi} \in \Pi$ can be written as $\mathbf{\pi}=\matele{\pi_{ij}}_{i,j=1}^n$, where each $\pi_{ij}\in \set{0,1}$. Then $\sum_j\pi_{ij}=1$ for any $i$, and $\sum_i\pi_{ij}=1$ for any $j$. Computing $\dnh$ now becomes:
\[\dnh(X,Y) = \frac{1}{2}\min_{\mathbf{\pi}\in \Pi}\max_{1\leq i,k,j,l, \leq n}\Gamma_{ijkl}\pi_{ij}\pi_{kl}, \text{ where }\Gamma_{ikjl} = |\w_X(x_i,x_k)-\w_Y(y_j,y_l)|. \]
This is just the QBAP, which is known to be NP-hard  \cite{burkard2009assignment}. 
\end{remark}

\subsection{Comparison of $\dn$ with the cut metric on graphs}\label{sec:cutmetric}

In our work throughout this paper, we have developed the theoretical framework of a certain notion of network distance that has proven to be useful for applying methods from the topological data analysis literature to network data.  
In each of these applications, networks were modeled as generalizations of metric spaces, and so the appropriate notion of network distance turned out to be a generalization of the Gromov-Hausdorff distance between compact metric spaces. However, an alternative viewpoint would be to model networks as weighted, directed graphs. From this perspective, a well-known metric on the space of all graphs is the \emph{cut metric} \cite{lovasz2012large, borgs2008convergent}. In particular, it is known that the completion of the space of all graphs with respect to the cut metric is compact \cite[p. 149]{lovasz2012large}. An analogous result is known for the network distance, and we will establish it in a forthcoming publication. It turns out that there are other structural similarities between the cut metric and the network distance. In this section, we will develop an interpretation of an $\ell^\infty$ version of the cut metric in the setting of compact metric spaces, and show that it agrees with the Gromov-Hausdorff distance in this setting.

\subsubsection{The cut distance between finite graphs}
Let $G = (V,E)$ denote a vertex-weighted and edge-weighted graph on a vertex set $V = \set{1,2,\ldots, n}$. Let $\a_i$ denote the weight of node $i$, with the assumption that each $\a_i \geq 0$, and $\sum_{i}\a_i = 1$. Let $\b_{ij}\in \R$ denote the weight of edge $ij$. 
For any $S, T \subseteq V$, define: 
\[e_G(S,T) := \sum_{s\in S, t \in T}\a_s\a_t\b_{st}.\] 
Note for future reference that one may regard $e_G$ as a function from $\pow(V)\times \pow(V)$ into $\R$.

Let $A$ be an $n \times n$ matrix of real numbers. Some classical norms include the $\ell^1$ norm $\|A\|_1 = n^{-2}\sum_{i,j = 1}^n|A_{ij}|$, the $\ell_2$ norm $\|A\|_2 = (n^{-2}\sum_{i,j = 1}^n|A_{ij}|^2)^{1/2}$, and the $\ell^\infty$ norm $\|A\|_\infty = \max_{i,j}|A_{ij}|$. Note that the $n^{-2}$ term is included for normalization.

The \emph{cut norm of $A$} is defined as 
\[\|A\|_{\Box} := \frac{1}{n^2}\ds\max_{S,T \subseteq \set{1,\ldots ,n}}\Big\vert\sum_{i\in S, j \in T} A_{ij} \Big\vert.\]

The \emph{cut metric} or \emph{cut distance} \gls{deltabox} between weighted graphs on the same node set $V = \set{1,2, \ldots, n}$ as
\[\d_{\Box}(G,G') := \max_{S, T \subseteq V}|e_G(S,T) - e_{G'}(S,T)|.\]

Next we consider weighted graphs with different numbers of nodes. Let $G, G'$ be graphs on $n$ and $m$ nodes respectively, with node weights $(\a_i)_{i=1}^n, (\a'_k)_{k=1}^m $ and edge weights $(\b_{ij})_{i,j=1}^n, (\b'_{kl})_{k,l=1}^m $, respectively. A \emph{fractional overlay} is a non-negative $n\times m$ matrix $W$ such that

\[\sum_{i=1}^n W_{ik} = \a'_k \text{ for } 1\leq k \leq m, \text{ and } \sum_{k=1}^m W_{ik} = \a_i \text{ for } 1\leq i \leq n.\]

Define $\mathcal{W}(G,G')$ to be the set of all fractional overlays between $G$ and $G'$. Let $W \in \mathcal{W}(G,G')$. Consider the graphs $G(W), G'(W)$ on the node set $\set{(i,k) : i \leq n, k \leq m, i,k\in \N}$ defined in the following way: node $(i,k)$ carries weight $W_{ik}$ in both $G(W), G'(W)$, edge $((i,k),(j,l))$ carries weight $\b_{ij}$ in $G(W)$ and $\b'_{kl}$ in $G'(W)$. Then the cut distance \gls{dbox} between graphs of different sizes becomes

\begin{equation}\label{eq:cut}
\dbox(G,G') := \min_{W \in \mathcal{W}(G,G')}\delta_{\Box}(G(X),G'(X)).\tag{$\star$}
\end{equation}

\subsubsection{The cut distance and the Gromov-Hausdorff distance}

In our interpretation, a fractional overlay is analogous to a \emph{correspondence}, as defined in \S\ref{sec:networks}. We define correspondences between networks, but a similar definition can be made for metric spaces, and in the case of finite metric spaces, a correspondence can be regarded as a binary matrix. Since correspondences are used to define the Gromov-Hausdorff distance between compact metric spaces, and our definition of network distance is motivated by GH distance, we would like to reinterpret the cut distance in the setting of compact metric spaces. Our goal is to show that in this setting, a certain analogue of the cut distance agrees with the GH distance.

For any compact metric space $(X,d_X)$, let $\pow(X)$ denote the nonempty elements of the power set of $X$, and let $e_X$ be any $\R_+$-valued function defined on $\pow(X) \times \pow(X)$. In analogy with the definition of $e_G$ for graphs, %
one would like $e_X$ to absorb information about the metric $d_X$ on $X$.

Given a subset $R \subseteq X\times Y$, let $\pi_1$ and $\pi_2$ denote the canonical projections to the $X$ and $Y$ coordinates, respectively. Let $\Xi$ denote the map that takes a compact metric space $(X,d_X)$ and returns a function $e_X:\pow(X)\times \pow(X)\rightarrow\R$.  We impose the following two conditions on the assignment $d_X\mapsto e_X$ induced by the map $\Xi$: 
\begin{enumerate}
\item For all $x,x' \in X$,  
\begin{equation}\label{cond-point}
\Xi(d_X)(\set{x},\set{x'}) = d_X(x,x'). \tag{${\#}_1$}\end{equation}
Thus $e_X=\Xi(d_X)$ and $d_X$ agree on singleton sets.
\item For all $T,S \subseteq X \times Y$,
\begin{equation}\label{cond-sup}
\vert \Xi(d_X)(\pi_1(T),\pi_1(S)) - \Xi(d_Y)(\pi_2(T),\pi_2(S)) \vert \leq 
\max_{t\in T, s \in S}\vert d_X(\pi_1(t),\pi_1(s)) - d_Y(\pi_2(t),\pi_2(s))\vert.  \tag{${\#}_2$}
\end{equation} 
\end{enumerate}
The latter of the two conditions above can be viewed as a continuity condition.

\begin{example}
Some natural candidates for the assignment $d_X\mapsto e_X$ which satisfy the two conditions above are:
\begin{itemize}
\item $\Xi^{\textrm{H}}$ such that $e_X(A,B) = \dhdf^X(A,B),$ \text{the Hausdorff distance},\\
\item $\Xi^{\max}$  such that $e_X(A,B) = \sup_{a\in A, b\in B}d_X(a,b)$,\\
\item $\Xi^{\min}$ such that $e_X(A,B) = \inf_{a\in A, b\in B}d_X(a,b)$.\\

\end{itemize}

It is clear that all these satisfy $({\#}_1)$. In Proposition \ref{prop:dh-cond-2} we prove that they satisfy condition $({\#}_2)$.
\end{example}

An analogue of the cut distance (\ref{eq:cut}) in the setting of compact metric spaces is the following.

\begin{definition}[An analogue of the cut distance for compact metric spaces]
Let $\Xi$ be any map satisfying (${\#}_1$) and (${\#}_2$). Then for compact metric spaces $X$ and $Y$, define
\[\dbox^\Xi(X,Y) := \frac{1}{2}\inf_{R}\sup_{S,T \subseteq R} \vert \Xi(d_X)(\pi_1(T),\pi_1(S)) - \Xi(d_Y)(\pi_2(T),\pi_2(S))\vert,\]
where $R$ ranges over correspondences between $X$ and $Y$, and $\pi_1, \pi_2$ are the canonical projections from $X \times Y$ onto $X$ and $Y$ respectively. We include the coefficient $2^{-1}$ to make comparison with $\dgh$ simpler. We claim that this interpretation of $\dbox^\Xi$ is identical to $\dgh$. 
\end{definition}

Note that one definition of the Gromov-Hausdorff distance for compact metric spaces \cite[Theorem 7.3.25]{burago} is the following:

\[\dgh(X,Y) = \frac{1}{2}\inf_{R \in \Rsc}\sup_{(x,y),(x',y') \in R} \vert d_X(x,x') - d_Y(y,y') \vert.\]

We also make the following definitions for \emph{distortion}:
\begin{align*}
\dis_{\textrm{GH}}(R) &= \sup_{(x,y),(x',y') \in R} \vert d_X(x,x') - d_Y(y,y') \vert, \\
\dis_{\Box}^\Xi(R) &= \sup_{S,T \subseteq R} \vert \Xi(d_X)(\pi_1(T),\pi_1(S)) - \Xi(d_Y)(\pi_2(T),\pi_2(S))\vert.
\end{align*}

\begin{proposition}\label{prop:dgh-compatible} For all compact metric spaces $X,Y$ and any assignment $\Xi$ satisfying (${\#}_1$) and (${\# }_2$) above, one has  $\dgh(X,Y) = \dbox^\Xi(X,Y).$
\end{proposition}

\begin{proof}[Proof of Proposition \ref{prop:dgh-compatible}]
Write $e_X=\Xi(d_X)$ and $e_Y=\Xi(d_Y)$. Let $R \in \Rsc(X,Y)$. In computing $\dis^{\Xi}_{\gh}(R)$, we take the supremum over all subsets of $R$, including singletons. Since $e_X$ (resp. $e_Y$) agrees with $d_X$ (resp. $d_Y$) on singletons, it follows that $\dis_{\textrm{GH}}(R) \leq \dis_{\Box}^\Xi(R)$. Thus $\dgh \leq \dbox^\Xi$. 

We now need to show $\dbox^\Xi \leq \dgh$.

Let $\eta > 0$ such that $\dgh(X,Y) < \eta$. Then, by one of the characterizations of $\dgh$ \cite[Chapter 7]{burago}, there exists a joint-metric $\d$ defined on $X \sqcup Y$ and a correspondence $R$ such that $\d(x,y) < \eta $ for all $(x,y) \in R$. In particular, $\d$ agrees with $d_X$ and $d_Y$ when restricted to the appropriate spaces. Now we have
\begin{align*}
\dis_{\Box}^\Xi(R) &= \sup_{T, S \subseteq R} \vert e_X(\pi_1(T),\pi_1(S)) - e_Y(\pi_2(T),\pi_2(S))\vert\\
&\leq \sup_{T,S \subseteq R} \sup_{t\in T, s\in S} \vert d_X(\pi_1(t),\pi_1(s)) - d_Y(\pi_2(t),\pi_2(s))\vert\\
&= \sup_{T,S \subseteq R} \sup_{t\in T, s\in S} \vert \d(\pi_1(t),\pi_1(s)) - \d(\pi_2(t),\pi_2(s))\vert
\intertext{By the triangle inequality, we have the following for any $t\in T$ and $s\in S$:}
\vert \d(\pi_1(t),\pi_1(s)) - \d(\pi_2(t),\pi_2(s))\vert &\leq 
\vert \d(\pi_1(t),\pi_1(s)) - \d(\pi_1(s),\pi_2(t))\vert \\ 
& \qquad\qquad\qquad + \vert \d(\pi_1(s),\pi_2(t)) - \d(\pi_2(t),\pi_2(s))\vert \\
&\leq \vert \d(\pi_1(t),\pi_2(t))\vert + \vert \d(\pi_1(s),\pi_2(s))\vert\\
&< \eta + \eta = 2\eta. \text{ Thus we conclude}\\
\dis_{\Box}^\Xi(R) &\leq 2\eta.
\end{align*}
This shows $\dbox^\Xi(X,Y) \leq \eta$, and so $\dbox^\Xi(X,Y) \leq \dgh(X,Y)$. So we conclude $\dgh(X,Y) = \dbox^\Xi(X,Y)$, under the assumptions made in this discussion. 
\end{proof}

\begin{proposition}\label{prop:dh-cond-2}
Each of the maps $\Xi^{\min}$, $\Xi^{\max}$, and $\Xi^{\textrm{H}}$ satisfies condition (\ref{cond-sup}).
\end{proposition}

\begin{proof}
We only give details for $\Xi^{\textrm{H}}$. The argument for the other cases is similar. Let $X,Y$ be compact metric spaces, and let $T,S \subseteq X \times Y$. First we recall the Hausdorff distance between two closed subsets $E,F \subseteq X$:
\[\dhdf^X(E,F) = \max\set{\sup_{e\in E}\inf_{f \in F}d_X(e,f), \sup_{f \in F}\inf_{e \in E}d_X(e,f)}.\]

So $\dhdf^X$ between two sets in $X$ is written as a max of two numbers $a$ and $b$, and we have the general result 
\[|\max(a,b) - \max(a',b')| \leq \max(|a-a'|,|b-b'|).\]
Another general result about ``calculating" with $\sup$ is that $|\sup f-\sup g| \leq \sup |f-g|$ for real valued functions $f$ and $g$. Both these properties are consequences of the triangle inequality, and we use them here:
\begin{align*}
\vert \dhdf^X(\pi_1(T),\pi_1(S)) - \dhdf^Y(\pi_2(T),\pi_2(S))\vert &= |\max(a,b) - \max(a',b')|\\
&\leq \max(|a-a'|,|b-b'|), \text{ where }\\
a &= \sup_{t \in T}\inf_{s\in S} d_X(\pi_1(t),\pi_1(s))\\
a' &= \sup_{t\in T}\inf_{s \in S}d_Y(\pi_2(t),\pi_2(s))\\
b &= \sup_{s\in S}\inf_{t\in T}d_X(\pi_1(t),\pi_1(s))\\
b' &= \sup_{s\in S}\inf_{t\in T}d_Y(\pi_2(t),\pi_2(s))
\end{align*}
We consider only one of the terms $|a-a'|$; the other term can be treated similarly.
\begin{align*}
|a-a'| &= \vert \sup_{t \in T}\inf_{s\in S} d_X(\pi_1(t),\pi_1(s)) - \sup_{t\in T}\inf_{s \in S}d_Y(\pi_2(t),\pi_2(s))\vert\\
&\leq \sup_{t\in T} \vert \inf_{s\in S}d_X(\pi_1(t),\pi_1(s)) - \inf_{s\in S}d_Y(\pi_2(t),\pi_2(s))\vert\\
&\leq \sup_{t\in T, s\in S}|d_X(\pi_1(t),\pi_1(s)) - d_Y(\pi_2(t),\pi_2(s))|.
\end{align*} 
The same bound holds for $|b-b'|$. Thus $\dhdf$ satisfies condition (\ref{cond-sup}), as claimed. 
\end{proof}

\subsection{Proofs from Section \ref{sec:networks}}
\label{sec:pf-networks}

\subsubsection{Proofs related to Theorem \ref{thm:dN1}}
\label{sec:pf-dN1}

\weakisomequiv*

\begin{proof}[Proof of Proposition \ref{prop:weak-isom-equiv}]
The case for Type I weak isomorphism is similar to that of Type II, so we omit it. 
For Type II weak isomorphism, the reflexive and symmetric properties are easy to see, so we only provide details for verifying transitivity. Let $A,B,C\in \Ngen$ be such that $A\congwt B$ and $B\congwt C$. Let $\e > 0$, and let $P,S$ be sets with surjective maps $\ph_A:P \r A$, $\ph_B:P \r B$, $\psi_B: S \r B$, $\psi_C: S \r C$ such that:
\begin{align*}
|\w_A(\ph_A(p),\ph_A(p')) - \w_B(\ph_B(p),\ph_B(p'))| &< \e/2 \quad \text{for each } p,p'\in P, \text{ and }\\
|\w_B(\psi_B(s),\psi_B(s')) - \w_C(\psi_C(s),\psi_C(s'))| &< \e/2 \quad \text{for each } s,s'\in S.
\end{align*}
Next define $T:=\{(p,s)\in P\times S: \ph_B(p)=\psi_B(s)\}.$
\begin{claim} The projection maps $\pi_P: T \r P$ and $\pi_S:T \r S$ are surjective. 
\end{claim}
\begin{subproof} Let $p\in P$. Then $\ph_B(p)\in B$, and since $\psi_B:S \r B$ is surjective, there exists $s\in S$ such that $\psi_B(s)=\ph_B(p)$. Thus $(p,s)\in T$, and $\pi_P(p,s)=p$. This suffices to show that $\pi_P:T \r P$ is a surjection. The case for $\pi_S:T \r S$ is similar.
\end{subproof}

It follows from the preceding claim that $\ph_A\circ \pi_P:T \r A$ and $\psi_C\circ \pi_S: T \r C$ are surjective. Next let $(p,s),(p',s')\in T$. Then,
\begin{align*}
&|\w_A(\ph_A(\pi_P(p,s)),\ph_A(\pi_P(p',s')))-\w_C(\psi_C(\pi_S(p,s)),\psi_C(\pi_S(p',s')))| \\ 
&=|\w_A(\ph_A(p),\ph_A(p'))-\w_C(\psi_C(s),\psi_C(s'))| \\
&=|\w_A(\ph_A(p),\ph_A(p'))-\w_B(\ph_B(p),\ph_B(p')) + \w_B(\ph_B(p),\ph_B(p')) - \w_C(\psi_C(s),\psi_C(s'))| \\
&=|\w_A(\ph_A(p),\ph_A(p'))-\w_B(\ph_B(p),\ph_B(p')) + \w_B(\psi_B(s),\psi_B(s')) - \w_C(\psi_C(s),\psi_C(s'))| \\
& < \e/2 + \e/2 = \e.
\end{align*}
Since $\e > 0$ was arbitrary, it follows that $A \congwt C$. \end{proof}

\dN*

\begin{proof}[Proof of Theorem \ref{thm:dN1}] 

It is clear that $\dn(X,Y) \geq 0$. To show $\dn(X,X) = 0$, consider the correspondence $R = \set{(x,x): x \in X}$. Then for any $(x,x),(x',x') \in R$, we have $|\w_X(x,x') - \w_X(x,x')| = 0$. Thus $\dis(R) = 0$ and $\dn(X,X) = 0$. 

Next we show symmetry, i.e. $\dn(X,Y) \leq \dn(Y,X)$ and $\dn(Y,X) \leq \dn(X,Y)$. The two cases are similar, so we just show the second inequality. 
Let $\eta > \dn(X,Y)$. 
Let $R \in \Rsc(X,Y)$ be such that $\dis(R) < 2\eta$. Then define $\tilde{R} = \set{(y,x) : (x,y) \in R}$. Note that $\tilde{R} \in \Rsc(Y,X)$. We have:
\begin{align*}
\dis(\tilde{R}) &= \sup_{(y,x),(y',x') \in \tilde{R}}|\w_Y(y,y') - \w_X(x,x')|\\
 &= \sup_{(x,y),(x',y') \in R}|\w_Y(y,y') - \w_X(x,x')|\\
 &= \sup_{(x,y),(x',y') \in R}|\w_X(x,x') - \w_Y(y,y')| = \dis(R).
\end{align*} 
So $\dis(R) = \dis(\tilde{R})$. Then $
\dn(Y,X) = \tfrac{1}{2}\inf_{S \in \Rsc(Y,X)}\dis(S) \leq \tfrac{1}{2}\dis(\tilde{R}) < \eta.$
This shows $\dn(Y,X) \leq \dn(X,Y)$. The reverse inequality follows by a similar argument. 

Next we prove the triangle inequality. Let $R \in \Rsc(X,Y), S \in \Rsc(Y,Z)$, and let \[R\circ S = \set{(x,z) \in X \times Z \mid \exists y, (x,y) \in R, (y,z) \in S}\]
First we claim that $R \circ S \in \Rsc(X,Z)$. This is equivalent to checking that for each $x \in X$, there exists $z$ such that $(x,z) \in R\circ S$, and for each $z \in Z$, there exists $x$ such that $(x,z) \in R \circ S$. The proofs of these two conditions are similar, so we just prove the former. Let $x \in X$. Let $y \in Y$ be such that $(x,y) \in R$. Then there exists $z\in Z$ such that $(y,z) \in S$. Then $(x,z) \in R \circ S$.

Next we claim that $\dis(R\circ S) \leq \dis(R) + \dis(S)$. Let $(x,z), (x',z') \in R \circ S$. Let $y\in Y$ be such that $(x,y) \in R$ and $(y,z) \in S$. Let $y' \in Y$ be such that $(x',y') \in R, (y',z')\in S$. Then we have:
\begin{align*}
|\w_X(x,x') - \w_Z(z,z')| &= |\w_X(x,x') - \w_Y(y,y') + \w_Y(y,y') - \w_Z(z,z')|\\
&\leq |\w_X(x,x') - \w_Y(y,y')| + |\w_Y(y,y') - \w_Z(z,z')|\\
&\leq \dis(R) + \dis(S).
\end{align*} This holds for any $(x,z), (x',z') \in R\circ S$, and proves the claim.

Now let $\eta_1 > \dn(X,Y)$, let $\eta_2 > \dn(Y,Z$, and let $R \in \Rsc(X,Y)$, $S \in \Rsc(Y,Z)$ be such that $\dis(R) < 2\eta_1$ and $\dis(S) < 2\eta_2$. Then we have:
\[\dn(X,Z) \leq \tfrac{1}{2}\dis(R\circ S)
\leq \tfrac{1}{2}\dis(R) + \tfrac{1}{2}\dis(S) < 2\eta_1 + 2\eta_2.\]
This shows that $\dn(X,Z) \leq \dn(X,Y) + \dn(Y,Z)$, and proves the triangle inequality.

Finally, we claim that $X \congwt Y$ if and only if $\dn(X,Y) = 0$. Suppose $\dn(X,Y) = 0$. Let $\e > 0$, and let $R(\e) \in \Rsc(X,Y)$ be such that $\dis(R(\e)) < \e$. Then for any $z = (x,y), z' =(x',y') \in R(\e)$, we have $|\w_X(x,x') - \w_Y(y,y')| < \e$. But this is equivalent to writing $|\w_X(\pi_X(z),\pi_X(z')) - \w_Y(\pi_Y(z),\pi_Y(z'))| < \e$, where $\pi_X:R(\e) \r X$ and $\pi_Y:R(\e) \r Y$ are the canonical projection maps. This holds for each $\e > 0$. Thus $X \congwt Y$. 

Conversely, suppose $X \congwt Y$, and for each $\e > 0$ let $Z(\e)$ be a set with surjective maps $\phi^\e_X:Z(\e) \r X$, $\phi^\e_Y : Z \r Y$ such that $|\w_X(\phi_X(z),\phi_X(z')) - \w_Y(\phi_Y(z),\phi_Y(z'))| < \e$ for all $z,z'\in Z(\e)$. For each $\e > 0$, let $R(\e) = \set{(\phi^\e_X(z),\phi^\e_Y(z)): z \in Z(\e)}$. Then $R(\e) \in \Rsc(X,Y)$ for each $\e > 0$, and $\dis(R(\e))= \sup_{z,z' \in Z}|\w_X(\phi_X(z),\phi_X(z')) - \w_Y(\phi_Y(z),\phi_Y(z'))| < \e$.

We conclude that $\dn(X,Y) = 0$. Thus $\dn$ is a metric modulo Type II weak isomorphism.
\end{proof}

\subsubsection{Proofs related to Theorem \ref{thm:skel-terminal}}
\label{sec:pf-networks-skel-terminal}

The following lemma summarizes useful facts about weight preserving maps and the relation $\sim$. 

\begin{lemma}\label{lem:wp-maps} Let $(X,\w_X), (Y,\w_Y) \in \Ngen$, and let $f: X \r Y$ be a weight preserving surjection. Then,
\begin{enumerate}
\item $f$ preserves equivalence classes of $\sim$, i.e. $x\sim x'$ for $x,x'\in X$ iff $f(x) \sim f(x')$. 
\item $f$ preserves weights between equivalence classes, i.e. $\w_{X/\sim}([x],[x']) = \w_{Y/\sim}([f(x)],[f(x')])$ for any $[x],[x'] \in X/\sim$. 
\end{enumerate}
\end{lemma}

\begin{proof}[Proof of Lemma \ref{lem:wp-maps}]

For the first assertion, let $x\sim x'$ for some $x,x'\in X$. We wish to show $f(x) \sim f(x')$. Let $y\in Y$, and write $y=f(z)$ for some $z\in X$. Then,
\[\w_Y(f(x),y) = \w_Y(f(x),f(z)) = \w_X(x,z) = \w_X(x',z) = \w_Y(f(x'),f(z)) = \w_Y(f(x'),y).\]
Similarly we have $\w_Y(y,f(x)) = \w_Y(y,f(x'))$ for any $y\in Y$. Thus $f(x) \sim f(x')$.

Conversely suppose $f(x)\sim f(x')$. Let $z\in X$. Then,
\[\w_X(x,z) = \w_Y(f(x),f(z)) = \w_Y(f(x'),f(z)) = \w_X(x',z),\]
and similarly we get $\w_X(z,x) = \w_X(z,x')$. Thus $x\sim x'$. This proves the first assertion.

The second assertion holds by definition:
\[\w_{Y/\sim}([f(x)],[f(x')]) = \w_Y(f(x),f(x')) = \w_X(x,x') = \w_{X/\sim}([x],[x']).\qedhere\]
\end{proof}

The following proposition shows that skeletons inherit the property of weak coherence (cf. Definition \ref{defn:coherence-weak}).

\begin{proposition}\label{prop:skel-cohrt} Let $(X,\w_X)$ be a compact network with a weakly coherent topology. The quotient topology on $(\sk(X),\w_{\sk(X)})$ is also weakly coherent. 
\end{proposition}

\begin{proof}[Proof of Proposition \ref{prop:skel-cohrt}]

Let $([x_n])_n$ be a sequence in $\sk(X)$ converging to some $[x] \in \sk(X)$, and let $U$ be an open set in $X$ containing $x$. By Proposition \ref{prop:quotient-open} we know $V:= \sigma(U)$ is open. Thus all but finitely many of the terms of $([x_n])_n$ are inside $V$, and so all but finitely many of the terms of $(x_n)_n$ are inside $U$. Thus $x_n \to x$. By weak coherence, this means $\d_X(x_n,x) \to 0$. Because $\sigma$ preserves weights by construction, we then have $\d_{\sk(X)}([x_n],[x]) \to 0$. 

This completes one direction of the proof. For the other direction, suppose we have $\d_{\sk(X)}([x_n],[x]) \to 0$. We need to show $[x_n] \to [x]$ in $\sk(X)$. 

Let $V$ be an open set in $\sk(X)$ containing $[x]$, and denote $U:= \sigma\inv(V)$. Because $\sigma$ preserves weights, we know  $\d_X(x_n,x) \to 0$.
Then by coherence of $X$, we have $x_n \to x$. Thus $(x_n)_n$ is eventually inside $U$, and so $([x_n])_n$ is eventually inside $V$. Thus $[x_n] \to [x]$. This concludes the proof. \end{proof}

In addition to weak coherence, the skeleton has the following useful property.

\begin{proposition}\label{prop:skel-hdrf} Let $(X,\w_X)$ be a compact network with a weakly coherent topology. Then $(\sk(X),\w_{\sk(X)})$ is Hausdorff.
\end{proposition}

\begin{proof}[Proof of Proposition \ref{prop:skel-hdrf}] 
Let $[x]\neq [x'] \in \sk(X)$. By first countability, we take a countable open neighborhood base $\set{U_n : n \in \N}$ of $[x]$ such that $U_1 \supseteq U_2 \supseteq U_3 \ldots $ (if necessary, we replace $U_n$ by $\cap_{i=1}^n U_i$). Similarly, we take a countable open neighborhood base $\set{V_n : n \in \N}$ of $[x']$ such that $V_1 \supseteq V_2 \supseteq V_3 \ldots $. To show that $\sk(X)$ is Hausdorff, it suffices to show that there exists $n \in \N$ such that $U_n \cap V_n = \emptyset$. 

Towards a contradiction, suppose $U_n \cap V_n \neq \emptyset$ for each $n \in \N$. For each $n\in \N$, let $[y_n] \in U_n \cap V_n$. Any open set containing $[x]$ contains $U_N$ for some $N \in \N$, and thus contains $[y_n]$ for all $n \geq N$. Thus $[y_n] \r [x]$. Similarly, $[y_n] \r [x']$. Because $\sk(X)$ has a weakly coherent topology (Proposition \ref{prop:skel-cohrt}), we then have:
\begin{align*}
\| \w_{\sk(X)}([x],\cdot) -  \w_{\sk(X)}([x'],\cdot) \|_\infty \leq 
& \|  \w_{\sk(X)}([x],\cdot) -  \w_{\sk(X)}([y_n],\cdot) \|_\infty \\
& \quad + \|  \w_{\sk(X)}([y_n],\cdot) -  \w_{\sk(X)}([x'],\cdot) \|_\infty \to 0.
\end{align*}
Thus $ \w_{\sk(X)}([x],\cdot) =  \w_{\sk(X)}([x'],\cdot)$ and similarly $ \w_{\sk(X)}(\cdot,[x]) =  \w_{\sk(X)}(\cdot,[x'])$. 
But then $x \sim x'$ and so $[x]=[x']$, a contradiction. \end{proof}

We are now ready to prove that skeletons are terminal, in the sense of Definition \ref{defn:terminal} (also recall Definitions \ref{defn:automorph} and \ref{defn:poset}).

\skelterminal*

\begin{proof}[Proof of Theorem \ref{thm:skel-terminal}]
Let $Y\in \mf{p}(X)$. Let $f:X \r Y$ be a weight preserving surjection. We first prove that there exists a weight preserving surjection $g:Y \r \skel(X)$. 

Since $f$ is surjective, for each $y\in Y$ we can write $y=f(x_y)$ for some $x_y\in X$. Then define $g:Y \r \skel(X)$ by $g(y):=[x_y]$. 

To see that $g$ is surjective, let $[x]\in \skel(X)$. Write $y = f(x)$. Then there exists $x_y \in X$ such that $f(x_y) = y$ and $g(y) = [x_y]$. Since $f$ preserves equivalence classes (Lemma \ref{lem:wp-maps}) and $f(x_y) = f(x)$, we have $x\sim x_y$. Thus $[x_y] = [x]$, and so $g(y) = [x]$. 

To see that $g$ preserves weights, let $y,y'\in Y$. Then,
\[\w_Y(y,y') = \w_Y(f(x_y),f(x_{y'})) = \w_X(x_y,x_{y'})  = \w_{\sk(X)}([x_y],[x_{y'}]) = \w_{\sk(X)}(g(y),g(y')).\]
This proves that the skeleton satisfies the first condition for being terminal. 

Next suppose $g: Y \r \skel(X)$ and $h:Y \r \skel(X)$ are two weight preserving surjections. We wish to show $h=\psi\circ g$ for some $\psi \in \Aut(\skel(X))$. 

For each $[x]\in \skel(X)$, we use the surjectivity of $g$ to pick $y_x\in Y$ such that $g(y_x) = [x]$. Then we define $\psi: \skel(X) \r \skel(X)$ by $\psi([x]) = \psi(g(y_x)):=h(y_x)$.

To see that $\psi$ is surjective, let $[x] \in \skel(X)$. Since $h$ is surjective, there exists $y'_x \in Y$ such that $h(y'_x) = [x]$. Write $[u]=g(y_x')$. We have already chosen $y_u$ such that $g(y_u) = [u]$. Since $g$ preserves equivalence classes (Lemma \ref{lem:wp-maps}), it follows that $y'_x \sim y_u$. Then,
\[\psi([u]) = \psi(g(y_u)) = h(y_u) = h(y'_x) = [x],\]
where the second-to-last equality holds because $h$ preserves equivalence classes (Lemma \ref{lem:wp-maps}).

To see that $\psi$ is injective, let $[x],[x'] \in \skel(X)$ be such that $\psi([x]) =h(y_x) = h(y_{x'})= \psi([x'])$. Since $h$ preserves equivalence classes (Lemma \ref{lem:wp-maps}), we have $y_x \sim y_{x'}$. Next, $g(y_x) = [x]$ and $g(y_{x'}) = [x']$ by the choices we made earlier. Since $y_x \sim y_{x'}$ and $g$ preserves clusters, we have $g(y_x) \sim g(y_x')$. Thus $[x] = [x']$. 

Next we wish to show that $\psi$ preserves weights. Let $[x],[x'] \in \skel(X)$. Then,
\begin{align*}
\w_{\sk(X)}(\psi([x]),\psi([x'])) = \w_{\sk(X)}(h(y_x),h(y_{x'})) = \w_Y(y_x,y_{x'}) &= \w_{\sk(X)}(g(y_x),g(y_{x'}))\\
& = \w_{\sk(X)}([x],[x']).
\end{align*}

Thus $\psi$ is a bijective, weight preserving automorphism of $\skel(X)$. Finally we wish to show that $h=\psi\circ g$. Let $y\in Y$, and write $g(y) = [x]$ for some $x \in X$. Since $g$ preserves equivalence classes (Lemma \ref{lem:wp-maps}), we have $y\sim y_x$, where $g(y_x) = [x]$. Then,
\[\psi(g(y)) = \psi([x]) = \psi(g(y_x)) = h(y_x) = h(y),\]
where the last equality holds because $h$ preserves equivalence classes (Lemma \ref{lem:wp-maps}). Thus for each $y\in Y$, we have $h(y) = \psi(g(y))$. This shows that the skeleton satisfies the second condition for being terminal. We conclude the proof. \end{proof}

\subsubsection{Proof of Example \ref{prop:perm}}
\label{sec:pf-perm}
\perm*

\begin{proof}[Proof of Example \ref{prop:perm}]
We start with some notation: for $x,x' \in X$, $y,y' \in Y$, let \[\Ga(x,x',y,y') = |\w_X(x,x')-\w_Y(y,y')|.\]  

Let $\ph: X \r Y$ be a bijection. Note that $R_\ph := \set{(x,\ph(x)) : x \in X}$ is a correspondence, and this holds for any bijection (actually any surjection) $\ph$. Since we minimize over all correspondences for $\dn$, we conclude $\dn(X,Y) \leq \dnh(X,Y)$. 

For the reverse inequality, we represent all the elements of $\Rsc(X,Y)$ as 2-by-2 binary matrices $R$, where a 1 in position $ij$ means $(x_i,y_j) \in R$. Denote the matrix representation of each $R\in \Rsc(X,Y)$ by $\text{mat}(R)$, and the collection of such matrices as $\text{mat}(\Rsc)$. Then we have:
\[\text{mat}(\Rsc) = \set{\mattwo{1}{a}{b}{1} : a, b \in \set{0,1}} \cup 
\set{\mattwo{a}{1}{1}{b} : a,b \in \set{0,1}} \]

Let $A = \set{(x_1,y_1), (x_2,y_2)}$ (in matrix notation, this is $\mattwo{1}{0}{0}{1}$) and let $B = \set{(x_1,y_2),(x_2,y_1)}$ (in matrix notation, this is $\mattwo{0}{1}{1}{0}$). Let $R \in \Rsc(X,Y)$. Note that either $A \subseteq R$ or $B \subseteq R$. Suppose that $A \subseteq R$. 
Then we have: 
\[\max_{(x,y),(x',y')\in A}\Ga(x,x',y,y') \leq \max_{(x,y),(x',y')\in R}\Ga(x,x',y,y')\]
Let $\Om(A)$ denote the quantity on the left hand side. 
A similar result holds in the case $B \subseteq R$:
\[\max_{(x,y),(x',y')\in B}\Ga(x,x',y,y') \leq \max_{(x,y),(x',y')\in R}\Ga(x,x',y,y')\]
Let $\Om(B)$ denote the quantity on the left hand side. Since either $A \subseteq R$ or $B \subseteq R$, we have 
\[\min\set{\Om(A),\Om(B)} \leq \min_{R \in \Rsc}\max_{(x,y),(x',y')\in R}\Ga(x,x',y,y')\]
We may identify $A$ with the bijection given by $x_1 \mapsto y_1$ and $x_2 \mapsto y_2$. Similarly we may identify $B$ with the bijection sending $x_1 \mapsto y_2$, $x_2 \mapsto y_1$. Thus we have
\[
\min_{\ph}\max_{x,x'\in X}\Ga(x,x',\ph(x),\ph(x')) 
 \leq \min_{R \in \Rsc}\max_{(x,y),(x',y')\in R}\Ga(x,x',y,y').
 \]
So we have $\dnh(X,Y) \leq \dn(X,Y)$. Thus $\dnh = \dn$. 

Next, let $\set{p,q}$ and $\set{p',q'}$ denote the vertex sets of $X$ and $Y$. Consider the bijection $\ph$ given by $p \mapsto p'$, $q \mapsto q'$ and the bijection $\psi$ given by $p \mapsto q'$, $q \mapsto p'$. Note that the weight matrix is determined by setting $\w_X(p,p) = \a,\ \w_X(p,q) = \d,\ \w_X(q,p) = \b$, and $\w_X(q,q) = \g$, and similarly for $Y$. Then we get $\dis(\ph) = \max\left(|\a-\a'|,|\b-\b'|,|\g-\g'|,|\d-\d'|\right)$ and $\dis(\psi) = \max(\left(|\a-\g'|,|\g-\a'|,|\d-\b'|,|\b-\d'|\right)$. The formula follows immediately. \end{proof}

\subsubsection{Proof of Proposition \ref{prop:bijection}}
\label{sec:pf-networks-bijection}
\bijection*

\begin{proof}[Proof of Proposition \ref{prop:bijection}]
We begin with an observation. Given $X,Y \in \Ngen$, let $X', Y' \in \Ngen$ be such that $X \congwt X'$, $Y\congwt Y'$, and $|X'| = |Y'|$. Because $\dn$ is a metric on $\Ngen$ modulo Type II weak isomorphism (Theorem \ref{thm:dN1}), we have:
\[\dn(X,Y) \leq \dn(X,X') + \dn(X',Y') + \dn(Y',Y) = \dn(X',Y') \leq \dnh(X',Y'),\]
where the last inequality follows from Remark \ref{remark:bijection}.

Next let $\eta > \dn(X,Y)$, and let $R \in \Rsc(X,Y)$ be such that $\dis(R) < 2\eta$. We wish to find networks $X'$ and $Y'$ such that $\dnh(X',Y') < \eta$. Write $Z = X \times  Y$, and write $f: Z \r X$ and $g : Z \r Y$ to denote the (surjective) projection maps $(x,y) \mapsto x$ and $(x,y) \mapsto y$. Notice that we may write $R = \set{(f(z),g(z)) : z \in R \subseteq Z}.$ In particular, by the definition of a correspondence, the restrictions of $f, g$ to $R$ are still surjective. Also notice that $Z$ is a finite product of first countable spaces and is hence first countable. We equip $R$ with the subspace topology, which is also first countable.

Define two weight functions $\fst\w, \gst\w : Z\times Z \r \R$ by $\fst\w(z,z') = \w_X(f(z),f(z'))$ and $\gst\w(z,z') = \w_Y(g(z),g(z'))$. Note that $\fst\w$ is the composition of $\w_X$ with the continuous projection $(x,y,x',y')\mapsto(x,x')$, and is thus continuous. The restriction of a continuous function to a subspace is continuous, and so $\fst\w$ restricted to $R\times R$ is continuous. We note also that the subspace topology on $R\times R$ induced by $Z\times Z$ is just the product topology on $R\times R$. Similarly, $\gst\w$ restricted to $R\times R$ is continuous. We now abuse notation slightly to write $\fst\w,\gst\w$ to mean their restrictions to $R\times R$. Let $(U,\w_U) = (R, \fst\w)$ and let $(V,\w_V) = (R,\gst\w)$. Then $U,V \in \Ngen$.
Note that $\dn(X,U) = 0$ by Remark \ref{rem:surj}, because $|U| \geq |X|$ and for all $z,z' \in U$, we have $\w_U(z,z') = \fst\w(z,z') = \w_X(f(z),f(z'))$ for the surjective map $f$. Similarly $\dn(Y,V) = 0$. 

Next let $\ph:U \r V$ be the bijection $z \mapsto z$. Then we have:
\begin{align*}
\sup_{z,z' \in U}|\w_U(z,z') - \w_V(\ph(z),\ph(z'))| &= 
\sup_{z,z' \in U}|\w_U(z,z') - \w_V(z,z')|\\
&= \sup_{z,z' \in R}|\w_X(f(z),f(z')) - \w_Y(g(z),g(z'))|\\
&= \sup_{(x,y), (x',y') \in R}|\w_X(x,x') - \w_Y(y,y')|\\
&= \dis(R)\text{. In particular,}\\
\inf_{\ph: U \r V \;\text{bijection}}\dis(\ph) &\leq \dis(R).\\
\end{align*}

So there exist networks $U, V$ with the same node set (and thus the same cardinality) such that $\dnh(U,V) \leq \frac{1}{2}\dis(R) < \eta$. We have already shown that $\dn(X,Y) \leq \dnh(U,V)$. Since $\eta > \dn(X,Y)$ was arbitrary, it follows that we have:
\[\dn(X,Y) = \inf\set{\dnh(X',Y') : X' \congwt  X, Y' \congwt  Y, \text{ and }|X'| = |Y'|}.\qedhere\]
\end{proof}

\subsubsection{Proof of Theorem \ref{thm:dis-cpt-sisom}}
\label{sec:pf-dis-cpt-sisom}
\discptsisom*

\begin{proof}[Proof of Theorem \ref{thm:dis-cpt-sisom}] By Theorem \ref{thm:weak-isom-cpt}, we know that $X$ and $Y$ are Type I weakly isomorphic. So there exists a set $V$ with surjections $\ph_X: V \r X$, $\ph_Y:V \r Y$ such that $A_X(\ph_X(v),\ph_X(v')) = A_Y(\ph_Y(v),\ph_Y(v'))$ for all $v,v' \in V$. Thus we obtain (not necessarily unique) maps $f: X \r Y$ and $g: Y \r X$ that are weight-preserving. Hence the composition $g\circ f : X \r X$ is a weight-preserving map. Without loss of generality, assume that $X$ has a $\Psi$-relaxed triangle inequality. Recall that this means that there exists a continuous function $\Psi: \R_+ \times \R_+ \r \R_+$ such that $\Psi(0,0)=0$ and $A_X(x,x') \leq \Psi(A_X(x,x''),A_X(x',x''))$ for all $x,x',x'' \in X$.

It is known that an isometric embedding from a compact metric space into itself must be bijective \cite[Theorem 1.6.14]{burago}. We now prove a similar result using the assumptions of our theorem. Let $h: X \r X$ be a weight-preserving map. By the assumption of a dissimilarity network, we know that $f,g$ and $h$ are injective. 

We check that $h$ is continuous, using the assumptions about the topology on $X$. Let $V\subseteq h(X)$ be open. Define $U:=h\inv[V]$. We claim that $U$ is open. Let $x \in U$, and consider $h(x) \in V$. Since $V$ is open and the forward balls form a base for the topology, we pick $\e > 0$ such that $B^+(h(x),\e) \subseteq V$. Now let $x' \in B^+(x,\e)$. Then $A_X(h(x),h(x')) = A_X(x,x') < \e$, so $h(x') \in B^+(h(x),\e) \subseteq V$. Hence $x' \in U$. It follows that $B^+(x,\e) \subseteq U$. Hence $U$ is open, and $h$ is continuous.

Next we check that $X$ is Hausdorff, using the $\Psi_X$-relaxed triangle inequality assumption. Let $x,x' \in X$, where $x\neq x'$. Using continuity of $\Psi_X$, let $\e>0$ be such that $\Psi([0,\e),[0,\e)) \subseteq [0,A_X(x,x'))$. We wish to show that $B^+(x,\e) \cap B^+(x',\e) = \emptyset.$ Towards a contradiction, suppose this is not the case and let $z \in B^+(x,\e) \cap B^+(x',\e)$. But then $A_X(x,x') \leq \Psi\big( A_X(x,z), A_X(x',z) \big) < A_X(x,x'),$ a contradiction. It follows that $X$ is Hausdorff.

Now $h(X)$ is compact, being the continuous image of a compact space, and it is closed in $X$ because it is a compact subset of a Hausdorff space. 

Finally we show that $h$ is surjective. Towards a contradiction, suppose that the open set $X\setminus h(X)$ is nonempty, and let $x \in X\setminus h(X)$. Using the topology assumption on $X$, pick $\e > 0 $ such that $B^+(x,\e) \subseteq X\setminus h(X)$. Define $x_0:= x$, and $x_n := h(x_{n-1})$ for each $n \in \N$. Then for each $n\in \N$, we have $A_X(x_0,x_n) \geq \e$. Since $h$ is weight-preserving, we also have $A_X(x_k,x_{k+n}) \geq \e$ for all $k,n  \in \N$. Since $X$ is sequentially compact, the sequence $(x_k)_{k \geq 0}$ has a convergent subsequence $(x_{k_j})_{j\in \N}$ that limits to some $z\in X$. 
Thus $B^+(z,r)$ contains all but finitely many terms of this sequence, for any $r > 0$. Now for any $m,n \in \N$ we observe:
\begin{align*}
A_X(x_m,x_{n}) \leq \Psi\big( A_X(x_m,z),A_X(x_{n},z) \big), \text{ where } 
 A_X(x_m,z) &\leq \Psi\big(A_X(x_m,x_m),A_X(z,x_m) \big)\\
 &= \Psi\big(0,A_X(z,x_m)\big),\\
\text{ and similarly } 
 A_X(x_{n},z) &\leq \Psi\big(0,A_X(z,x_{n}) \big).
\end{align*}
Since $\Psi$ is continuous and vanishes at $(0,0)$, we choose $\d>0$ such that $\Psi([0,\d),[0,\d)) \subseteq [0,\e)$. We also choose $\eta>0$ such that $\Psi(0,[0,\eta)) \subseteq [0,\d)$. Since $B^+(z,\eta)$ contains all but finitely many terms of the sequence $(x_{k_j})_{j \geq \N}$, we pick $N \in \N$ so that $x_{k_m} \in B^+(z,\eta)$, for all $m \geq N$. Let $m,n \geq N$. Then $A_X(z,x_{k_n}) < \eta$ and $A_X(z,x_{k_m}) < \eta$. Thus $A_X(x_{k_n},z) \leq \Psi(0,A_X(z,x_{k_n})) < \d$ and $A_X(x_{k_m},z) \leq \Psi(0,A_X(z,x_{k_m})) < \d$. It follows that $A_X(x_{k_m},x_{k_n}) < \e$. 

But this is a contradiction to what we have shown before. Thus $h$ is surjective, hence bijective. Since $h$ was an arbitrary weight-preserving map from $X$ into itself, the same result holds for $g\circ f: X \r X$. This shows that $g$ is surjective. It follows that $X \cong^s Y$.\end{proof}

\section{Metric structure of networks}
\label{sec:metrics}

In this section we first characterize the weak isomorphism class of $\Ncom$ via a notion of tripods. Viewing $(\Ncom/\sim,\dn)$ as a bona fide metric space, we then recover results on completeness, geodesics, and compact families as can be described for the collection of compact metric spaces equipped with $\dgh$.

With a slight abuse of notation, $\dn:\Ngen/\congwt \times \Ngen/\congwt \r \R_+$ is defined as follows:
\[\dn([X],[Y]) := \dn(X,Y),\quad\text{for each }[X],[Y]\in \Ngen/\congwt.\]

To check that $\dn$ is well-defined on $[X],[Y] \in \Ngen/\congwt$, let $X' \in [X], Y' \in [Y]$. Then:
\[\dn([X'],[Y'])=\dn(X',Y')= \dn(X,Y)=\dn([X],[Y]),\]
where the second-to-last equality follows from the triangle inequality and the observation that $\dn(X,X') = \dn(Y,Y') = 0$.

Our main result is that the two types of weak isomorphism coincide in the setting of compact networks. As a stepping stone towards proving this result, we explore the notion of ``sampling" finite networks from compact networks.

\subsection{Compact networks and finite sampling}
\label{sec:cpt-finite-samp}
In this section, we prove that any compact network admits an approximation by a finite network up to arbitrary precision, in the sense of $\dn$.

\begin{example}[Some compact and noncompact networks] The nonreversible and finitely reversible directed circles in Section \ref{sec:dir-s1} are examples of noncompact and compact asymmetric networks, respectively. The finitely reversible circle $(\us^1,\rus)$ can be ``approximated" up to arbitrary precision by picking $n$ equidistant points on $\us^1$ and equipping this collection with the restriction of $\rus$. We view this process as ``sampling" finite networks from a compact network. In the next result, we present this sampling process as a theorem that applies to \emph{any} compact network.
\end{example}

\begin{definition}[$\e$-systems]\label{def:systems} 
Let $\e > 0$. For any network $(X,\w_X)$, an \emph{$\e$-system on $X$} is a finite open cover $\mc{U} = \set{U_1,\ldots, U_n}, n \in \N$, of $X$ such that for any $1\leq i,j \leq n$, we have $\w_X(U_i,U_j) \subseteq B(r_{ij},\e)$ for some $r_{ij} \in \R$.

In some cases, we will be interested in the situation where $X$ is a finite union of connected components $\set{X_1,\ldots, X_n}, n\in \N$. By a \emph{refined $\e$-system}, we will mean an $\e$-system such that each element of the $\e$-system is contained in precisely one connected component of $X$.
\end{definition}

The next result shows that by sampling points from all the elements of an $\e$-system, one obtains a finite, quantitatively good approximation to the underlying network.

\begin{restatable}[Sampling from a compact network]{theorem}{sampling}
\label{thm:sampling}
Let $(X,\w_X)$ be a compact network. Then for any $\e>0$,
\begin{itemize}
\item there exists a refined $\e/4$-system $\mathscr{G}$ on $X$
\item for $X'$ any finite subset of $X$ having nonempty intersection with each element in a refined $\e/4$ system $\mathscr{G}$, one has 
\[\dn((X,\w_X),(X',\w_X|_{X'\times X'})) < \e.\]
\end{itemize}
In particular, one can always choose a finite subset $X' \subseteq X$ such that $\dn(X,X') < \e$.
\end{restatable}

The proof of this result is in Section \ref{sec:pf-metrics-sampling}.

\begin{remark}[Compact metric spaces admit $\e$-systems]\label{rem:e-net} When considering a compact metric space $(X,d_X)$, the preceding theorem relates to the well-known notion of taking finite $\e$-nets in a metric space. Recall that for $\e >0$, a subset $S \subseteq X$ is an $\e$-net if for any point $x\in X$, we have $B(x,\e)\cap S \neq \emptyset$. Such an $\e$-net satisfies the nice property that $\dgh(X,S) < \e$ \cite[7.3.11]{burago}. In particular, one can find a finite $\e$-net of $(X,d_X)$ for any $\e > 0$ by compactness.

Moreover, one may always find an $\e$-system for a compact metric space, for any $\e>0$. To see this, let $(X,d_X)$ be a compact metric space and let $\e>0$. Consider the open cover $\{B(x,\e/4) : x \in X\}$. By compactness, we can take a finite subcover $\mc{U}= \{B(x_i,\e/4) : 1\leq i \leq n\}$ for some $n \in \N$. Now fix $1\leq i,j \leq n$. Let $x \in B(x_i,\e/4), x' \in B(x_j,\e/4)$. Then by triangle inequality, one has:
\begin{align*}
    |d_X(x,x') - d_X(x_i,x_j)| \leq d_X(x,x_i) + d_X(x_j,x') < \e/2.
\end{align*}
Consequently $d_X(x,x') \in B(d_X(x_i,x_j),\e)$. Since $x,x'$ were arbitrary, we have $d_X(B(x_i,\e/4), B(x_j,\e/4)) \subseteq B(d_X(x_i,x_j),\e)$. Setting $r_{ij} = d_X(x_i,x_j)$ now recovers the property of being an $\e$-system. Theorem \ref{thm:sampling} can be viewed as showing that this property of compact metric spaces can be obtained (albeit with more difficulty) as a consequence of compactness and continuity, rather than metric axioms.

Observe that we do not make quantitative estimates on the cardinality of the $\e$-approximation produced in Theorem \ref{thm:sampling}. In the setting of compact metric spaces, the size of an $\e$-net relates to the rich theory of metric entropy developed by Kolmogorov and Tihomirov \cite[Chapter 17]{edgar1993classics}.
\end{remark}

By virtue of Theorem \ref{thm:sampling}, one can always approximate a compact network up to any given precision. The next result implies that a sampled network limits to the underlying compact network as the sample gets more and more dense.

\begin{corollary}[Limit of dense sampling]\label{cor:dense-sampling} Let $(X,\w_X)$ be a compact network, and let $S = \set{s_1,s_2,\ldots}$ be a countable dense subset of $X$ with a fixed enumeration. For each $n\in \N$, let $X_n$ be the finite network with node set $\set{s_1,\ldots, s_n}$ and weight function $\w_X|_{X_n \times X_n}$. Then we have:
\[\dn(X,X_n) \downarrow 0 \text{ as } n \r \infty.\]
\end{corollary}

\begin{proof}[Proof of Corollary \ref{cor:dense-sampling}] Using the first statement of Theorem \ref{thm:sampling}, let $\mathscr{G}=\{G_1,\ldots, G_q\}$ be a refined $\e/4$-system on $X$. By density of $S$, choose $p(i)\in \N$ for $1\leq i \leq q$ such that $s_{p(i)} \in G_i$ for each $i$. Then define 
\[n:=\max\set{p(1),p(2),\ldots, p(q)}.\]
Now define $X_n$ to be the network with node set $\set{s_1,s_2,\ldots, s_n}$ and weight function given by the appropriate restriction of $\w_X$. Then by the second assertion of Theorem \ref{thm:sampling}, one has $\dn(X,X_n) \leq \e$. It follows from the construction that $\dn(X,X_n) \downarrow 0$ as $n\to \infty$. \end{proof}

\subsection{Weak isomorphism in compact networks and characterization via motifs}
\label{sec:cpt-weak-isom}

By Theorem \ref{thm:dN1}, $\dn$ is a proper metric on $\Ngen$ modulo Type II weak isomorphism, which is equivalent to Type I weak isomorphism when restricted to $\Ncal$. The comparison between $\Q \cap [0,1]$ and $[0,1]$ in Example \ref{ex:non-opt-corr} shows that in general, these two notions of weak isomorphism are not equivalent. This leads to the following natural question: when restricted to $\Ncom$, are we still able to recover equivalence between Type I and Type II weak isomorphism?

In the following theorem, we provide a positive answer to this question.

\begin{restatable}[Weak isomorphism in $\Ncom$]{theorem}{weakisomcpt}
\label{thm:weak-isom-cpt} 
Let $X,Y \in \Ncom$. Then $X$ and $Y$ are Type II weakly isomorphic if and only if $X$ and $Y$ are Type I weakly isomorphic, i.e. there exists a set $V$ and surjections $\ph_X:V \r X,\; \ph_Y:V \r Y$ such that:
\[\w_X(\ph_X(v),\ph_X(v')) = \w_Y(\ph_Y(v),\ph_Y(v')) \quad
\text{for all } v,v'\in V.\]
\end{restatable}

The proof of this result is in Section \ref{sec:pf-metrics-weak-isom-cpt}.
As a consequence of Theorem \ref{thm:weak-isom-cpt}, we see that weak isomorphisms of Types I and II coincide in the setting of $\Ncom$. Thus we recover a desirable notion of equivalence in the setting of compact networks.
By imposing slightly more control over topology using the notion of weak coherence, we obtain the next result, which roughly states that weakly isomorphic compact networks have strongly isomorphic skeleta. 

\begin{restatable}{theorem}{cptsisom}
\label{thm:cpt-sisom} Suppose $(X,\w_X), (Y,\w_Y)$  are separable, compact networks with weakly coherent topologies. Then the following are equivalent:
\begin{enumerate}
\item $X \cong^w Y$. \label{item:cpt-wisom}
\item $\M_n(X) = \M_n(Y)$ for all $n \in \N$.\label{item:cpt-motifs}
\item $\sk(X) \cong^s \sk(Y)$.\label{item:cpt-skel}
\end{enumerate}
\end{restatable}

The proof of this result is in Section \ref{sec:pf-motifs}.
Assertion (\ref{item:cpt-skel}) implies (\ref{item:cpt-wisom}) immediately, but motif sets form the crucial tool required to pass from (\ref{item:cpt-wisom}) to (\ref{item:cpt-skel}). When proving this result, we will use a property of motif sets called \emph{stability} to verify that (\ref{item:cpt-wisom}) implies (\ref{item:cpt-motifs}). This property of stability is later contextualized and proved in Theorem \ref{thm:motif-stab}.

To contrast this result with the metric setting, note that compact metric spaces are automatically separable and weakly coherent, and that weak and strong notions of isomorphism coincide for compact metric spaces. The directed circles with finite reversibility (Section \ref{sec:dir-s1}) form a family of non-metric spaces that satisfy the assumptions of Theorem \ref{thm:cpt-sisom}.

\subsection{The completeness of $\Ncomsim$}
\label{sec:completeness}
A very natural question regarding $\Ncomsim$ is if it is complete. This indeed turns out to be the case. The proof of the following result is in Section \ref{sec:pf-metrics-complete}.

\begin{restatable}[]{theorem}{complete}
\label{thm:complete} 
$(\Ncom/{\cong^w},\dn)$ is a complete metric space containing $(\Ncal/{\cong^w},\dn)$.
\end{restatable}

The result of Theorem \ref{thm:complete} can be summarized as follows:
\begin{center}
\textit{A convergent sequence of finite networks limits to a compact topological space with a continuous weight function.}
\end{center}

\begin{remark} The technique of composed correspondences used in the proof of Theorem \ref{thm:complete} can also be used to show that the collection of isometry classes of compact metric spaces endowed with the Gromov-Hausdorff distance is a complete metric space. Standard proofs of this fact \cite[\S 10]{petersen2006riemannian} do not use correspondences, relying instead on a method of endowing metrics on disjoint unions of spaces and then computing Hausdorff distances.
\end{remark}

Completeness of $\Ncomsim$ gives us a first useful criterion for convergence of networks. Ideally, we would also want a criterion for convergence along the lines of sequential compactness. In the setting of compact metric spaces, Gromov's Precompactness Theorem implies that the topology induced by the Gromov-Hausdorff distance admits many \emph{precompact families} of compact metric spaces (i.e. collections whose closure is compact) \cite{gromov1981structures, burago, petersen2006riemannian}. Any sequence in such a precompact family has a subsequence converging to some limit point of the family. In the next section, we extend these results to the setting of networks. Namely, we show that that there are many families of compact networks that are precompact under the metric topology induced by $\dn$.

\subsection{Precompact families in $\Ncomsim$}
We begin this section with some definitions.

\begin{definition}[Diameter for networks, \cite{dn-part1}] For any network $(X,\w_X)$, define $\diam(X):= \sup_{x,x'\in X}|\w_X(x,x')|$. For compact networks, the $\sup$ is replaced by $\max$. 
\end{definition}

\begin{definition} A family $\mathcal{F}$ of weak isomorphism classes of compact networks is \emph{uniformly approximable} if: (1) there exists $D \geq 0$ such that for every $[X] \in \mathcal{F}$, we have $\diam(X) \leq D$, and (2) for every $\e > 0$, there exists $N(\e) \in \N$ such that for each $[X] \in \mathcal{F}$, there exists a finite network $Y$ satisfying $|Y| \leq N(\e)$ and $\dn(X,Y) < \e$. 
\end{definition}

\begin{remark} The preceding definition is an analogue of the definition of \emph{uniformly totally bounded} families of compact metric spaces \cite[Definition 7.4.13]{burago}, which is used in formulating the precompactness result in the metric space setting. A family of compact metric spaces is said to be uniformly totally bounded if there exists $D \in \R_+$ such that each space has diameter bounded above by $D$, and for any $\e >0$ there exists $N_\e \in \N$ such that each space in the family has an $\e$-net with cardinality bounded above by $N_\e$. Recall that given a metric space $(X,d_X)$ and $\e >0$, a subset $S \subseteq X$ is an $\e$-net if for any point $x\in X$, we have $B(x,\e)\cap S \neq \emptyset$. Such an $\e$-net satisfies the nice property that $\dgh(X,S) < \e$ \cite[7.3.11]{burago}. Thus an $\e$-net is an $\e$-approximation of the underlying metric space in the Gromov-Hausdorff distance.

\end{remark}

\begin{theorem}\label{thm:precompactness}
Let $\mathfrak{F}$ be a uniformly approximable family in $\Ncomsim$. Then $\mathfrak{F}$ is precompact, i.e. any sequence in $\mathfrak{F}$ contains a subsequence that converges in $\Ncomsim$.
\end{theorem}

Our proof is modeled on the proof of an analogous result for compact metric spaces proposed by Gromov \cite{gromov1981structures}. We use one external fact \cite[stability of $\diam$]{dn-part1}: for compact networks $X,Y$ such that $\dn(X,Y) < \e$, we have $\diam(X) \leq \diam(Y) + 2\e$. 

\begin{proof}[Proof of Theorem \ref{thm:precompactness}]
Let $D \geq 0$ be such that $\diam(X) \leq D$ for each $[X]\in \mathfrak{F}$. It suffices to prove that $\mathfrak{F}$ is totally bounded, because Theorem \ref{thm:complete} gives completeness, and these two properties together imply precompactness. Let $\e > 0$. We need to find a finite family $\mathcal{G}\subseteq \Ncomsim$ such that for every $[F] \in \mathfrak{F}$, there exists $[G] \in \mathcal{G}$ with $\dn(F,G) < \e$. Define:
\[\mc{A}:= \set{A \in \Ncal : |A| \leq N(\e/2),\; 
\dn(A,F) < \e/2 \text{ for some } [F]\in \mathfrak{F}}.\]

Each element of $\mathcal{A}$ is an $n\times n $ matrix, where $1\leq n \leq N(\e/2)$. For each $A\in \mathcal{A}$, there exists $[F]\in \mathfrak{F}$ with $\dn(A,F) < \e/2$, and by the fact stated above, we have $\diam(A) \leq \diam(F) + 2(\e/2) \leq D+\e$. Thus the matrices in $\mathcal{A}$ have entries in $[-D-\e,D+\e]$. Let $N \gg 1$ be such that: 
\[\frac{2D+2\e}{N} < \frac{\e}{4},\]
and write the refinement of $[-D-\e,D+\e]$ into $N$ pieces as:  
\[W := \set{-D-\e + k\left(\tfrac{2D+2\e}{N}\right) : 0\leq k \leq N}.\]

Write $\mathcal{A} = \bigsqcup_{i=1}^{N(\e/2)}\mathcal{A}_i$, where each $\mathcal{A}_i$ consists of the $i\times i$ matrices of $\mathcal{A}$. For each $i$ define:
\[\mathcal{G}_i := \set{(G_{pq})_{1\leq p,q \leq i} : G_{pq} \in W}, \text{ the $i\times i$ matrices with entries in $W$.}\]

Let $\mathcal{G} = \bigsqcup_{i=1}^{N(\e/2)}\mathcal{G}_i$ and note that this is a finite collection. Furthermore, for each $A_i \in \mathcal{A}_i$, there exists $G_i\in \mathcal{G}_i$ such that 
\[\norm{A_i - G_i}_\infty < \frac{\e}{4}.\]

Taking the diagonal correspondence between $A_i$ and $G_i$, it follows that $\dn(A_i,G_i) < \e/2$. Hence for any $[F] \in \mathfrak{F}$, there exists $A \in \mathcal{A}$ and $G \in \mathcal{G}$ such that 
\[\dn(F,G) \leq \dn(F,A) + \dn(A,G) < \e/2 + \e/2 = \e.\]
This shows that $\mathfrak{F}$ is totally bounded, and concludes the proof. \end{proof}

\subsection{Geodesic structure on $\Ncomsim$}\label{sec:geodesic}

Thus far, we have motivated our discussion of compact networks by viewing them as limiting objects of finite networks. 
By the results of the preceding section, we know that $(\Ncomsim,\dn)$ is complete and obeys a well-behaved compactness criterion. In this section, we prove that this metric space is also \emph{geodesic}, i.e. any two compact networks can be joined by a rectifiable curve with length equal to the distance between the two networks. 

Geodesic spaces can have a variety of practical implications. For example, geodesic spaces that are also complete and locally compact are \emph{proper} (i.e. any closed, bounded subset is compact), by virtue of the Hopf-Rinow theorem \cite[\S 2.5.3]{burago}. Any probability measure with finite second moment supported on such a space has a \emph{barycenter} \cite[Lemma 3.2]{ohta2012barycenters}, i.e. a ``center of mass". Conceivably, such a result can be applied to a compact, geodesically convex region of $(\Ncomsim,\dn)$ to compute an ``average" network from a collection of networks. Such a result is of interest in statistical inference, e.g. when one wishes to represent a noisy collection of networks by a single network. Similar results on barycenters of geodesic spaces can be found in \cite{gouic2017existence, leustean2016barycenters}. We leave a treatment of this topic from a probabilistic framework as future work, and only use this vignette to motivate the results in this section. 

We begin with some definitions.

\begin{definition}[Curves and geodesics]
A \emph{curve} on $\Ngen$ joining $(X,\omega_X)$ to $(Y,\omega_Y)$ is any continuous map $\gamma:[0,1]\rightarrow \Ngen$ such that $\gamma(0)=(X,\w_X)$ and $\gamma(1)=(Y,\w_Y)$. We will write \emph{a curve on $\Ncal$} (resp. \emph{a curve on $\Ncom$}) to mean that the image of $\g$ is contained in $\Ncal$ (resp. $\Ncom$). 
Such a curve is called a \emph{geodesic} \cite[\S I.1]{bridson2011metric} between $X$ and $Y$ if for all $s,t\in[0,1]$ one has: 
\[\dn(\gamma(t),\gamma(s))=|t-s|\cdot\dn(X,Y).\]
A metric space is called a \emph{geodesic space} if any two points can be connected by a geodesic.
\end{definition}

The following theorem is a useful result about geodesics:

\begin{theorem}[\cite{burago}, Theorem 2.4.16]\label{thm:mid} Let $(X,d_X)$ be a complete metric space. If for any $x,x'\in X$ there exists a midpoint $z$ such that $d_X(x,z)=d_X(z,y)=\tfrac{1}{2}d_X(x,y)$, then $X$ is geodesic. 
\end{theorem}

The proofs of the next results are in Section \ref{sec:pf-metrics-geod-inf}. As a first step towards showing that $\Ncomsim$ is geodesic, we show that the collection of finite networks forms a geodesic space. 

\begin{restatable}[]{theorem}{geodesic}
\label{thm:geodesic} 
The metric space $(\Ncalsim, \dn)$ is a geodesic space. More specifically, let $[X],[Y] \in (\Ncalsim, \dn)$. Then, for any $R\in \Rsc^{\opt}(X,Y)$, we can construct a geodesic $\g_R:[0,1] \r \Ncalsim$ between $[X]$ and $[Y]$ as follows: 
\begin{align*}
&\g_R(0):=[(X,\w_X)],\, \g_R(1):=[(Y,d_Y)], \text{ and }\g_R(t):=[(R,\w_{\g_R(t)})] \text{ for } t\in (0,1),
\intertext{ where for each $(x,y),(x',y') \in R$ and $t\in(0,1),$} 
&\hspace{0.1\textwidth} \w_{\g_R(t)}\big((x,y),(x',y')\big):=(1-t)\cdot \w_X(x,x')+t\cdot \w_Y(y,y').
\end{align*}
\end{restatable}

A key step in the proof of the preceding theorem is to choose an optimal correspondence between two finite networks. This may not be possible, in general, for compact networks. However, using the additional results on precompactness and completeness of $\Ncomsim$, we are able to obtain the desired geodesic structure in Theorem \ref{thm:geod-inf}. The proof is similar to the one used by the authors of \cite{ivanov2016gromov} to prove that the metric space of isometry classes of compact metric spaces endowed with the Gromov-Hausdorff distance is geodesic.

\begin{restatable}[]{theorem}{geodinf}
\label{thm:geod-inf} 
The complete metric space $(\Ncomsim,\dn)$ is geodesic.
\end{restatable}

\begin{remark}
Consider the collection of compact metric spaces endowed with the Gromov-Hausdorff distance. This collection can be viewed as a subspace of $(\Ncomsim,\dn)$. It is known (via a proof relying on Theorem \ref{thm:mid}) that this restricted metric space is geodesic \cite{ivanov2016gromov}. 
Furthermore, it was proved in \cite{dgh-era} that an optimal correspondence always exists in this setting, and that such a correspondence can be used to construct explicit geodesics instead of resorting to Theorem \ref{thm:mid}. 
The key technique used in \cite{dgh-era} was to take a convergent sequence of increasingly-optimal correspondences, use a result about compact metric spaces called Blaschke's theorem \cite[Theorem 7.3.8]{burago} to show that the limiting object is closed, and then use metric properties such as the Hausdorff distance to guarantee that this limiting object is indeed a correspondence. A priori, such techniques cannot be readily adapted to the network setting, and while one can obtain a convergent sequence of increasingly-optimal correspondences, the obstruction lies in showing that the limiting object is indeed a correspondence. This is why we use the indirect proof via Theorem \ref{thm:mid}. 
\end{remark}

\begin{remark}[Branching and deviant geodesics] 
It is important to note that there exist geodesics in $\Ncomsim$ that deviate from the straight-line form given by Theorem \ref{thm:geodesic}. Even in the setting of compact metric spaces, there exist infinite families of branching and deviant geodesics \cite{dgh-era}.
\end{remark}

\subsection{Proofs from Section \ref{sec:metrics}}
\label{sec:pf-metrics}

\subsubsection{Proof of Theorem \ref{thm:sampling}}
\label{sec:pf-metrics-sampling}

\sampling*
\begin{proof}[Proof of Theorem \ref{thm:sampling}] The idea is to find a cover of $X$ by open sets $G_1,\ldots, G_q$ and representatives $x_i\in G_i$ for each $1\leq i\leq q$ such that whenever we have $(x,x')\in G_i \times G_j$, we know by continuity of $\w_X$ that $|\w_X(x,x') - \w_X(x_i,x_j)|< \e/4$. Then we define a correspondence that associates each $x\in G_i$ to $x_i$, for $1\leq i \leq q$. Such a correspondence has distortion bounded above by $\e/2$.

Let $\e > 0$. Let $\mathscr{B}$ be a base for the topology on $X$. 

Let $\set{B(r,\e/4) : r\in \R}$ be an open cover for $\R$. Then by continuity of $\w_X$, we get that 
\[\set{\w_X^{-1}[B(r,\e/4)] : r\in \R}\] 
is an open cover for $X\times X$. Each open set in this cover can be written as a union of open rectangles $U\times V$, for $U,V\in \mathscr{B}$. Thus the following set is an open cover of $X\times X$:
\[\mathscr{U}:=\set{U\times V : U, V \in \mathscr{B},\; U\times V \subseteq \w_X^{-1}[B(r,\e/4)],\; r\in \R}.\]

\begin{claim}\label{cl:sampling-1} There exists a finite open cover $\mathscr{G}=\set{G_1,\ldots, G_q}$ of $X$ such that for any $1\leq i,j\leq q$, we have $G_i \times G_j \subseteq U\times V$ for some $U\times V \in \mathscr{U}$.
\end{claim}

\begin{subproof}[Proof of Claim \ref{cl:sampling-1}] The proof of the claim includes elements that are often used to provide a proof of the Tube Lemma \cite[Lemma 26.8]{munkres-top}. Since $X\times X$ is compact, we take a finite subcover:
\[\mathscr{U}^f:=\{U_1\times V_1,\ldots, U_n\times V_n\}, \text{ for some } n\in \N.\]

Let $x\in X$. Then we define:
\[\mathscr{U}^f_x:=\{U\times V \in \mathscr{U}^f : x\in U\},\] 
and write 
\begin{align}
\mathscr{U}^f_x=\set{U^x_{i_1}\times V^x_{i_1},\ldots, U^x_{i_{m(x)}}\times V^x_{i_{m(x)}}}.\label{eq:sampling-U-decomp}
\end{align}
Here $m(x)$ is an integer depending on $x$, and $\set{i_1,\ldots, i_{m(x)}}$ is a subset of $\set{1,\ldots, n}$.

Since $\mathscr{U}^f$ is an open cover of $X\times X$, we know that $\mathscr{U}^f_x$ is an open cover of $\set{x}\times X$. 
Next define: 
\begin{align}A_x:=\ds\bigcap_{k=1}^{m(x)}U^x_{i_k}.\label{eq:sampling-A-in-U}
\end{align}
Then $A_x$ is open and contains $x$. In the literature \cite[p. 167]{munkres-top}, the set $A_x \times X$ is called a \emph{tube} around $\set{x}\times X$. Notice that $A_x \times X \subseteq \cup\mathscr{U}^f_x$. Since $x$ was arbitrary in the preceding construction, we define $\mathscr{U}^f_x$ and $A_x$ for each $x\in X$. Then note that $\set{A_x : x\in X}$ is an open cover of $X$. Using compactness of $X$, we choose $\set{s_1,\ldots, s_p} \subseteq X,\ p\in \N$, such that $\set{A_{s_1},\ldots, A_{s_p}}$ is a finite subcover of $X$. 

Once again let $x\in X$, and let $\mathscr{U}^f_x$ and $A_x$ be defined as above. Define the following:
\begin{align}
B_x:=\set{A_x \times V^x_{i_k} : 1\leq k \leq m(x)}. \label{eq:sampling-B-cover}
\end{align} 
Since $x\in A_x$ and $X \subseteq \cup_{k=1}^{m(x)}V^x_{i_k}$, it follows that $B_x$ is a cover of $\set{x}\times X$. Furthermore, since $\set{A_{s_1},\ldots, A_{s_p}}$ is a cover of $X$, it follows that the finite collection $\set{B_{s_1},\ldots, B_{s_p}}$ is a cover of $X\times X$. 

Let $z\in X$. Since $X \subseteq \ds\cup_{k=1}^{m(x)}V^x_{i_k}$, we pick $V^x_{i_k}$ for $1\leq k\leq m(x)$ such that $z\in V^x_{i_k}$. Since $x$ was arbitrary, such a choice exists for each $x\in X$. Therefore, we define:
\[C_{z}:=\set{V\in \mathscr{B} : z\in V, \; A_{s_i}\times V \in B_{s_i} \text{ for some } 1\leq i\leq p}.\]

Since each $B_{s_i}$ is finite and there are finitely many $B_{s_i}$, we know that $C_{z}$ is a finite collection. Next define:
\[D_{z}:=\bigcap_{V\in C_{z}}V.\]
Then $D_{z}$ is open and contains $z$. Notice that $X\times D_z$ is a tube around $X \times \set{z}$. Next, using the fact that $\set{A_{s_i} : 1\leq i \leq p}$ is an open cover of $X$, pick $A_{s_{i(z)}}$ such that $z\in A_{s_{i(z)}}$. Here $1\leq i(z) \leq p$ is some integer depending on $z$. Then define 
\begin{align}
G_{z}:=D_{z}\cap A_{s_{i(z)}}. \label{eq:sampling-G-is-D-and-A}
\end{align}
Then $G_{z}$ is open and contains $z$. Since $z$ was arbitrary, we define $G_{z}$ for each $z\in X$. Then $\set{G_{z} : z\in X}$ is an open cover of $X$, and we take a finite subcover:
\[\mathscr{G}:=\{G_1,\ldots, G_q\}, \; q\in \N.\]

Finally, we need to show that for any choice of $1\leq i,j\leq q$, we have $G_i\times G_j \subseteq U\times V $ for some $U\times V\in \mathscr{U}$. Let $1\leq i,j\leq q$. Note that we can write $G_i=G_w$ and $G_j = G_y$ for some $w,y\in X$. By the definition of $G_w$ (Equation \ref{eq:sampling-G-is-D-and-A}), we then have $G_w \subseteq A_{s_{i(w)}}$ for some index $i(w)$ depending on $w$. Referring to Equations (\ref{eq:sampling-U-decomp}) and (\ref{eq:sampling-A-in-U}) and using the property that $\mathscr{U}^f_{s_{i(w)}}$ is a cover of $\set{s_{i(w)}}\times X$, we choose $U^{s_{i(w)}}\times V^{s_{i(w)}} \in \mathscr{U}^f_{s_{i(w)}}$ such that $V^{s_{i(w)}}$ contains $y$. Thus we obtain:
\[G_w \subseteq A_{s_{i(w)}} \subseteq U^{s_{i(w)}} \text{ for }
U^{s_{i(w)}}\times V^{s_{i(w)}} \in \mathscr{U}^f_{s_{i(w)}}, \, 1\leq i(w) \leq p,\] 
From Equation (\ref{eq:sampling-B-cover}) we have that $A_{s_{i(w)}}\times V^{s_{i(w)}} \in B_{s_{i(w)}}$. Then $V^{s_{i(w)}} \in C_y$, and so we have:
\[G_y \subseteq D_y \subseteq V^{s_{i(w)}}.\]
It follows that $G_i\times G_j = G_w \times G_y \subseteq {\color{blue} U^{s_{i(w)}} \times V^{s_{i(w)}} }\in \mathscr{U}$.\end{subproof}

From the definition of $\mathscr{U}$, it follows that $\mathscr{G}$ is an $\e/4$-system on $X$. By further partitioning $\mathscr{G}$ against the connected components of $X$, we obtain a refined $\e/4$-system. This concludes the first part of the proof.

For the next part, let $\mathscr{G}$ be a refined $\e/4$-system on $X$, and suppose $X'$ is a finite set that intersects each element of $\mathscr{G}$. We first perform a disjointification step. Define:
\[\widetilde{G}_1:=G_1, \;
\widetilde{G}_2:= G_2 \setminus \widetilde{G}_1,\;
\widetilde{G}_3:= G_3 \setminus (\widetilde{G}_1 \cup \widetilde{G}_2),\; \ldots,\; 
\widetilde{G}_q:=G_q \setminus \left(\cup_{k=1}^{q-1}\widetilde{G}_k \right).
\] 

Next define a correspondence between $X$ and $X'$ as follows:
\[R:=\set{(x,x_i) : x\in \widetilde{G}_i, 1\leq i\leq q}.\]
Let $(x,x_i),(x',x_j) \in R$. Then we have $(x,x'),(x_i,x_j)\in \widetilde{G}_i\times \widetilde{G}_j \subseteq G_i \times G_j$. By assumption we know that $G_i\times G_j \subseteq U \times V$, for some $U\times V \in \mathscr{U}$. Therefore $\w_X(x,x'),\w_X(x_i,x_j) \in B(r,\e/4)$ for some $r\in \R$. It follows that:
\[|\w_X(x,x')-\w_X(x_i,x_j)| < \e/2.\] 
Since $(x,x_i),(x',x_j) \in R$ were arbitrary, we have $\dis(R)\leq  \e/2$. Hence $\dn(X,X') \leq \e$. \end{proof}

\subsubsection{Proof of Theorem \ref{thm:weak-isom-cpt}}
\label{sec:pf-metrics-weak-isom-cpt}
\weakisomcpt*
\begin{proof}[Proof of Theorem \ref{thm:weak-isom-cpt}]

By the definition of $\congwo$, it is clear that if $X\congwo Y$, then $\dn(X,Y) = 0$, i.e. $X \congwt Y$ (cf. Theorem \ref{thm:dN1}).

Conversely, suppose $\dn(X,Y) =0$. Our strategy is to obtain a set $Z\subseteq X\times Y$ with canonical projection maps $\pi_X: Z \r X, \pi_Y:Z \r Y$ and surjections $\psi_X:X \r \pi_X(Z), \psi_Y: Y \r \pi_Y(Z)$ as in the following diagram:

\begin{center}
\begin{tikzpicture}
\begin{scope}
\node (1) at (4,1.5){$X$};
\node (2) at (12,1.5){$Y$};
\node (l) at (2,0){$X$};
\node (r) at (14,0){$Y$};
\node (4) at (6,0){$\pi_X(Z)$};
\node (5) at (8,1.5){$Z$};
\node (6) at (10,0){$\pi_Y(Z)$};
\node at (4,0){$\congwo$};
\node at (8,0){$\congwo$};
\node at (12,0){$\congwo$};

\path[->>] (1) edge node[above right]{$\psi_X$} (4);
\path[->>] (1) edge node[above left]{$\id_X$} (l);
\path[->>] (2) edge node[above right]{$\id_Y$} (r);
\path[->>] (2) edge node[above left]{$\psi_Y$}(6);
\path[->>] (5) edge node[above left]{$\pi_X$}(4);
\path[->>] (5) edge node[above right]{$\pi_Y$}(6);
\end{scope}
\end{tikzpicture}
\end{center}

Furthermore, we will require:
\begin{align}
\w_X(\pi_X(z),\pi_X(z'))&=\w_Y(\pi_Y(z),\pi_Y(z')) \quad \text{for all } z,z'\in Z,\label{eq:weak-isom-1} \\
\w_X(x,x')&=\w_X(\psi_X(x),\psi_X(x')) \quad \text{for all } x,x'\in X,\label{eq:weak-isom-2}\\
\w_Y(y,y')&=\w_Y(\psi_Y(y),\psi_Y(y')) \quad \text{for all } y,y'\in Y.\label{eq:weak-isom-3}
\end{align}

As a consequence, we will obtain a chain of Type I weak isomorphisms 
\[X \congwo \pi_X(Z) \congwo \pi_Y(Z) \congwo Y.\] 
Since Type I weak isomorphism is an equivalence relation (Proposition \ref{prop:weak-isom-equiv}), it will follow that $X$ and $Y$ are Type I weakly isomorphic. 

By applying Theorem \ref{thm:sampling}, we choose sequences of finite subnetworks 
$\set{X_n \subseteq X :n\in \N}$ and $\set{Y_n\subseteq Y: n\in \N}$ such that $\dn(X_n, X) < 1/n$ and $\dn(Y_n,Y) < 1/n$ for each $n\in \N$. By the triangle inequality, $\dn(X_n,Y_n) < 2/n$ for each $n$. 

For each $n\in \N$, let $T_n \in \Rsc(X_n,X), P_n\in \Rsc(Y,Y_n)$ be such that $\dis(T_n) < 2/n$ and $\dis(P_n) < 2/n$. Define $\a_n:=4/n - \dis(T_n) - \dis(P_n)$, and notice that $\a_n \r 0$ as $n \r \infty$. Since $\dn(X,Y)=0$ by assumption, for each $n \in \N$ we let $S_n \in \Rsc(X,Y)$ be such that $\dis(S_n) < \a_n$. Then,
\[\dis(T_n \circ S_n \circ P_n) \leq \dis(T_n) + \dis(S_n) + \dis(P_n) < 4/n. \qquad\text{ (cf. Remark \ref{rem:chained-corr})}\]

Then for each $n\in \N$, we define $R_n:=T_n \circ S_n \circ P_n \in \Rsc(X_n,Y_n)$. By Remark \ref{rem:chained-corr}, we know that $R_n$ has the following expression:
\begin{align*}
R_n = \{(x_n,y_n)\in X_n\times Y_n : \text{ there exist }   \tilde{x}\in X,\; \tilde{y}\in Y &\text{ such that } 
(x_n,\tilde{x})\in T_n,\\ 
&(\tilde{x},\tilde{y})\in S_n,\; 
(\tilde{y},y_n)\in P_n\}.
\end{align*}
Next define:
\[\mc{S}:=\set{(\tilde{x}_n,\tilde{y}_n)_{n\in \N} \in (X\times Y)^\N : (\tilde{x}_n,\tilde{y}_n)\in S_n \text{ for each } n\in \N}.\]

Since $X,Y$ are first countable and compact, the product $X\times Y$ is also first countable and compact, hence sequentially compact. Any sequence in a sequentially compact space has a convergent subsequence, so for convenience, we replace each sequence in $\mc{S}$ by a convergent subsequence. Next define:
\[Z:= \set{(x,y)\in X\times Y : (x,y)\text{ a limit point of some } (\tilde{x}_n,\tilde{y}_n)_{n\in \N} \in \mc{S}}.\]

\begin{claim}\label{cl:wisom-cpt} $Z$ is a closed subspace of $X\times Y$. Hence it is compact and sequentially compact.
\end{claim}

The second statement in the claim follows from the first: assuming that $Z$ is a closed subspace of the compact space $X\times Y$, we obtain that $Z$ is compact. Any subspace of a first countable space is first countable, so $Z$ is also first countable. Next, observe that $\pi_X(Z)$ equipped with the subspace topology is compact, because it is a continuous image of a compact space. It is also first countable because it is a subspace of the first countable space $X$. Furthermore, the restriction of $\w_X$ to $\pi_X(Z)$ is continuous. Thus $\pi_X(Z)$ equipped with the restriction of $\w_X$ is a compact network, and by similar reasoning, we get that $\pi_Y(Z)$ equipped with the restriction of $\w_Y$ is also a compact network.

\begin{subproof}[Proof of Claim \ref{cl:wisom-cpt}]
We will show that $Z \subseteq X\times Y$ contains all its limit points. Let $(x,y)\in X\times Y$ be a limit point of $Z$. Let $\set{U_n \subseteq X: n\in \N, (x,y)\in U_n}$ be a countable neighborhood base of $(x,y)$. For each $n \in \N$, the finite intersection $V_n:=\cap_{i=1}^nU_i$ is an open neighborhood of $(x,y)$, and thus contains a point $(x_n,y_n) \in Z$ that is distinct from $(x,y)$ (by the definition of a limit point). Pick such an $(x_n,y_n)$ for each $n \in \N$. Then $(x_n,y_n)_{n\in \N}$ is a sequence in $Z$ converging to $(x,y)$ such that $(x_n,y_n)\in V_n$ for each $n\in \N$. 

For each $n\in \N$, note that because $(x_n,y_n) \in Z$ and $V_n$ is an open neighborhood of $(x_n,y_n)$, there exists a sequence in $\mc{S}$ converging to $(x_n,y_n)$ for which all but finitely many terms are contained in $V_n$. So for each $n\in \N$, let $(\tilde{x}_{n},\tilde{y}_{n})\in S_n$ be such that $(\tilde{x}_{n},\tilde{y}_{n}) \in V_n$. Then the sequence $(\tilde{x}_{n},\tilde{y}_{n})_{n\in \N} \in \mc{S}$ converges to $(x,y)$. 
Thus $(x,y)\in Z$. Since $(x,y)$ was an arbitrary limit point of $Z$, it follows that $Z$ is closed. \qedhere
\end{subproof}

\medskip
\noindent
\textbf{Proof of Equation \ref{eq:weak-isom-1}.}
We now prove Equation \ref{eq:weak-isom-1}. Let $z=(x,y), \ z'=(x',y') \in Z$, and let $(\tilde{x}_n,\tilde{y}_n)_{n\in \N}, (\tilde{x}'_n,\tilde{y}'_n)_{n\in \N}$ be elements of $\mc{S}$ that converge to $(x,y),(x',y')$ respectively. We wish to show $|\w_X(x,x')-\w_Y(y,y')| = 0$. Let $\e > 0$, and observe that:
\begin{align*}
&\hspace{0.3\textwidth} |\w_X(x,x') - \w_Y(y,y')|\\
&=|\w_X(x,x') - \w_X(\tilde{x}_n,\tilde{x}_n') + \w_X(\tilde{x}_n,\tilde{x}_n') -\w_Y(\tilde{y}_n,\tilde{y}_n')+ \w_Y(\tilde{y}_n,\tilde{y}_n') - \w_Y(y,y')| \\
&\leq |\w_X(x,x') - \w_X(\tilde{x}_n,\tilde{x}_n')| + |\w_X(\tilde{x}_n,\tilde{x}_n') -\w_Y(\tilde{y}_n,\tilde{y}_n')|+ |\w_Y(\tilde{y}_n,\tilde{y}_n') - \w_Y(y,y')|.
\end{align*}

\begin{claim}\label{cl:wisom-eq2} Suppose we are given sequences $(\tilde{x}_n,\tilde{y}_n)_{n\in \N}, (\tilde{x}'_n,\tilde{y}'_n)_{n\in \N}$ in $Z$ converging to $(x,y)$ and $(x',y')$ in $Z$, respectively. Then there exists $N \in \N$ such that for all $n \geq N$, we have:
\[|\w_X(x,x') - \w_X(\tilde{x}_n,\tilde{x}_n')| < \e/4, \qquad
|\w_Y(\tilde{y}_n,\tilde{y}_n') - \w_Y(y,y')| < \e/4.\]
\end{claim}

\begin{subproof}[Proof of Claim \ref{cl:wisom-eq2}]

Write $a:=\w_X(x,x'), b:=\w_Y(y,y')$. Since $\w_X,\ \w_Y$ are continuous, we know that $\w_X\inv[B(a,\e/4)]$ and $\w_Y\inv[B(b,\e/4)]$ are open neighborhoods of $(x,x')$ and $(y,y')$. Since each open set in the product space $X\times X$ is a union of open rectangles of the form $A\times A'$ for $A,A'$ open subsets of $X$, we choose an open set $A\times A' \subseteq \w_X\inv[B(a,\e/4)]$ such that $(x,x')\in A\times A'$. Similarly, we choose an open set $B\times B' \subseteq \w_Y\inv[B(b,\e/4)]$ such that $(y,y')\in B\times B'$. Then $A\times B,\ A'\times B'$ are open neighborhoods of $(x,y),\ (x',y')$ respectively. Since $(\tilde{x}_n,\tilde{y}_n)_{n\in \N}$ and $(\tilde{x}'_n,\tilde{y}'_n)_{n\in \N}$ converge to $(x,y)$ and $(x',y')$, respectively, we choose $N\in \N$ such that for all $n\geq N$, we have $(\tilde{x}_n,\tilde{y}_n)\in A\times B$ and $(\tilde{x}'_n,\tilde{y}'_n)\in A'\times B'$. The claim now follows.\qedhere
\end{subproof}

Now choose $N\in \N$ such that the property in Claim \ref{cl:wisom-eq2} is satisfied, as well as the additional property that $8/N < \e/4$. Then for any $n \geq N$, we have:
\[|\w_X(x,x') - \w_Y(y,y')| \leq \e/4 + |\w_X(\tilde{x}_n,\tilde{x}_n') -\w_Y(\tilde{y}_n,\tilde{y}_n')| + \e/4.\]

Separately note that for each $n\in \N$, having $(\tilde{x}_n,\tilde{y}_n), (\tilde{x}'_n,\tilde{y}'_n) \in S_n$ implies that there exist $(x_n,y_n)$ and $ (x'_n,y'_n) \in R_n$ such that $(x_n,\tilde{x}_n), (x'_n,\tilde{x}'_n) \in T_n$ and $(\tilde{y}_n, y_n), (\tilde{y}'_n, y'_n) \in P_n$. Thus we can bound the middle term above as follows:
\begin{align*}
&\hspace{0.3\textwidth} |\w_X(\tilde{x}_n,\tilde{x}_n') -\w_Y(\tilde{y}_n,\tilde{y}_n')|\\
&=|\w_X(\tilde{x}_n,\tilde{x}_n')- \w_X(x_n,x_n') + \w_X(x_n,x_n') - \w_Y(y_n,y_n') + \w_Y(y_n,y_n')-\w_Y(\tilde{y}_n,\tilde{y}_n')|\\
&\leq |\w_X(\tilde{x}_n,\tilde{x}_n')- \w_X(x_n,x_n')| +| \w_X(x_n,x_n') - \w_Y(y_n,y_n')| + |\w_Y(y_n,y_n')-\w_Y(\tilde{y}_n,\tilde{y}_n')|\\
&\leq \dis(T_n) + \dis(R_n) + \dis(P_n) < 8/n \leq 8/N < \e/4.
\end{align*}
The preceding calculations show that:
\[|\w_X(x,x') - \w_Y(y,y')| < \e.\]
Since $\e> 0$ was arbitrary, it follows that $\w_X(x,x')=\w_Y(y,y')$. This proves Equation \ref{eq:weak-isom-1}.

It remains to define surjective maps $\psi_X:X \r \pi_X(Z), \psi_Y:Y \r \pi_Y(Z)$ and to verify Equations \ref{eq:weak-isom-2} and \ref{eq:weak-isom-3}. Both cases are similar, so we only show the details of constructing $\psi_X$ and verifying Equation \ref{eq:weak-isom-2}. 

\medskip
\noindent
\textbf{Construction of $\psi_X$.} Let $x \in X$. Suppose first that $x \in \pi_X(Z)$. Then we simply define $\psi_X(x) = x$. We also make the following observation, to be used later: for each $n \in \N$, letting $y \in Y$ be such that $(x,y) \in S_n$, there exists $x_n \in X_n$ and $y_n \in Y_n$ such that $(x_n,x) \in T_n$ and $(y,y_n) \in P_n$. 

Next suppose $x \in X \setminus \pi_X(Z)$. For each $n\in \N$, let $x_n\in X_n$ be such that $(x_n,x)\in T_n$, and let $\tilde{x}_n \in X$ be such that $(x_n,\tilde{x}_n) \in T_n$. Also for each $n \in \N$, let $\tilde{y}_n\in Y$ be such that $(\tilde{x}_n,\tilde{y}_n) \in S_n$. Then for each $n\in \N$, let $y_n\in Y_n$ be such that $(\tilde{y}_n, y_n)\in P_n$. Then by sequential compactness of $X\times Y$, the sequence $(\tilde{x}_n,\tilde{y}_n)_{n\in \N}$ has a convergent subsequence which belongs to $\mc{S}$ and converges to a point $(\tilde{x},\tilde{y})\in Z$. In particular, we obtain a sequence $(\tilde{x}_n)_{n\in \N}$ converging to a point $\tilde{x}$, such that $(x_n,x)$ and $(x_n,\tilde{x}_n) \in T_n$ for each $n\in \N$. 
Define $\psi_X(x)=\tilde{x}$. 

Since $x\in X$ was arbitrary, this construction defines $\psi_X:X \r \pi_X(Z)$. Note that $\psi_X$ is simply the identity on $\pi_X(Z)$, hence is surjective. 

\medskip
\noindent
\textbf{Proof of Equation \ref{eq:weak-isom-2}.} Now we verify Equation \ref{eq:weak-isom-2}. Let $\e > 0$. There are three cases to check:
\begin{description}
\item[Case 1: $x,x' \in \pi_X(Z)$] In this case, we have:
\[|\w_X(x,x') - \w_X(\psi_X(x),\psi_X(x'))| = \w_X(x,x') - \w_X(x,x') = 0.\]
\item[Case 2: $x,x' \in X \setminus \pi_X(Z)$] By continuity of $\w_X$, we obtain an open neighborhood $U:=\w_X\inv[B(\w_X(\psi_X(x),\psi_X(x')),\e/2)]$ of $(x,x')$. By the definition of $\psi_X$ on $X\setminus \pi_X(Z)$, we obtain sequences $(\tilde{x}_n,\tilde{y}_n)_{n\in \N}$ and $(\tilde{x}'_n,\tilde{y}'_n)_{n\in \N}$ in $\mc{S}$ converging to $(\psi_X(x),\tilde{y})$ and $(\psi_X(x'),\tilde{y}')$ for some $\tilde{y}, \tilde{y}' \in Y$. By applying Claim \ref{cl:wisom-eq2}, we obtain $N \in \N$ such that for all $n \geq N$, we have $(\tilde{x}_n,\tilde{x}'_n) \in U$. Note that we also obtain sequences $(x_n)_{n\in \N}$ and $(x'_n)_{n\in \N}$ such that $(x_n,x), (x_n,\tilde{x}_n) \in T_n$ and $(x'_n,x'), (x'_n,\tilde{x}'_n) \in T_n$. Choose $N$ large enough so that it satisfies the property above and also that $4/N < \e/2$. Then for any $n \geq N$, 
\begin{align*}
&\hspace{0.3\textwidth} |\w_X(x,x') - \w_X(\psi_X(x),\psi_X(x'))|\\
&=|\w_X(x,x') -\w_X(x_n,x_n')+\w_X(x_n,x_n') - \w_X(\tilde{x}_n,\tilde{x}_n') +\w_X(\tilde{x}_n,\tilde{x}_n') - \w_X(\psi_X(x),\psi_X(x'))|\\
&\leq \dis(T_n) + \dis(T_n) + \e/2 < 4/n + \e/2 \leq 4/N + \e/2 < \e.
\end{align*}

\item[Case 3: $x\in \pi_X(Z), x' \in X\setminus \pi_X(Z)$] By the definition of $\psi_X$ on $X\setminus \pi_X(Z)$, we obtain: (1) a sequence $(\tilde{x}'_n)_{n\in \N}$ converging to $\psi_X(x')$, and (2) another sequence $(x'_n)_{n\in \N}$ such that $(x'_n,x')$ and $(x'_n,\tilde{x}'_n)$ both belong to $T_n$, for each $n\in \N$. By the definition of $\psi_X$ on $\pi_X(Z)$, we obtain a sequence $(x_n)_{n\in \N}$ such that $(x_n,x) \in T_n$ for each $n \in \N$. 

Let $U:=\w_X\inv[B(\w_X(x,\psi_X(x')),\e/2)]$. Since $(\tilde{x}'_n)_{n\in \N}$ converges to $\psi_X(x')$, we know that all but finitely many terms of the sequence $(x,\tilde{x}'_n)_{n\in \N}$ belong to $U$. So we choose $N$ large enough so that for each $n\geq N$, we have:
\begin{align*}
&\hspace{0.3\textwidth} |\w_X(x,x') - \w_X(x,\psi_X(x'))|\\
&=|\w_X(x,x') -\w_X(x_n,x_n')+\w_X(x_n,x_n') - \w_X(x,\tilde{x}_n') +\w_X(x,\tilde{x}_n') - \w_X(x,\psi_X(x'))|\\
&\leq \dis(T_n) + \dis(T_n) + \e/2 < 4/n + \e/2 \leq 4/N + \e/2 < \e.
\end{align*}
\end{description}

Since $\e >0$ was arbitrary, Equation \ref{eq:weak-isom-2} follows. The construction of $\psi_Y$ and proof for Equation \ref{eq:weak-isom-3} are similar. This concludes the proof of the theorem. \end{proof}

\subsubsection{Proofs related to Theorem \ref{thm:cpt-sisom}}
\label{sec:pf-motifs}

Our goal in this section is to prove that weak isomorphism, equality of motif sets, and strong isomorphism between skeleta are equivalent in the setting of compact networks with coherent topologies. 
Because the main result of this section will assume equality of motif sets, we will specifically only use the weak form of coherence (cf. Definition \ref{defn:coherence-weak}).

\begin{proposition}\label{prop:wp-motifs} Let $(X,\w_X), (Y,\w_Y)$ be compact networks such that $\M_n(X) = \M_n(Y)$ for all $n \in \N$. Suppose $X$ contains a countable subset $S_X$. Then there exists a weight-preserving map $f: S_X \r Y$.
\end{proposition}

\begin{proof}[Proof of Proposition \ref{prop:wp-motifs}] We proceed via a  diagonal argument. 
Write $S_X= \set{x_1,x_2,\ldots, x_n,\ldots}$. For each $n \in \N$, let $f_n : S_X \r Y$ be a map that preserves weights on $\set{x_1,\ldots, x_n}$. Such a map exists by the assumption that $\M_n(X)=\M_n(Y)$. 

Since $Y$ is first countable and compact, hence sequentially compact, the sequence $(f_n(x_1))_n$ has a convergent subsequence; we write this as $(f_{1,n}(x_1))_n$. Since $f_k$ is weight-preserving on $\set{x_1,x_2}$ for $k \geq 2$, we know that $f_{1,n}$ is weight-preserving on $\set{x_1,x_2}$ for $n \geq 2$. Using sequential compactness again, we have that $(f_{1,n}(x_2))_n$ has a convergent subsequence $(f_{2,n}(x_2))_n$. This sequence converges at both $x_1$ and $x_2$, and $f_{2,n}$ is weight-preserving on $\set{x_1,x_2}$ for $n \geq 2$. Proceeding in this way, we obtain the diagonal sequence $(f_{n,n})_n$ which converges pointwise on $S_X$. Furthermore, for any $n \in \N$, $f_{k,k}$ is weight-preserving on $\set{x_1,\ldots, x_n}$ for $k \geq n$.

Next define $f:S_X \r Y$ by setting $f(x):=\lim_n f_{n,n}(x)$ for each $x \in S_X$. It remains to show that $f$ is weight-preserving. Let $x_n,x_m \in S_X$, and let $k \geq \max(m,n)$. Then $\w_X(x_n,x_m) = \w_Y(f_{k,k}(x_n),f_{k,k}(x_m))$. Using (sequential) continuity of $\w_Y$, we then have:
\[\w_Y(f(x_n),f(x_m)) = \w_Y(\lim_kf_{k,k}(x_n),\lim_kf_{k,k}(x_m)) = \lim_k \w_Y(f_{k,k}(x_n),f_{k,k}(x_m)) =  \w_X(x_n,x_m).\]
In the second equality above, we used the fact that a sequence converges in the product topology iff the components converge.
Since $x_n,x_m \in S_X$ were arbitrary, this concludes the proof. \end{proof}

\begin{proposition}\label{prop:wp-extn} Let $(X,\w_X),(Y,\w_Y)$ be compact networks. Suppose $f: S_X \r Y$ is a weight-preserving function defined on a countable dense subset $S_X \subseteq X$. Then $f$ extends to a weight-preserving map on $X$. 
\end{proposition}

\begin{proof}[Proof of Proposition \ref{prop:wp-extn}]
Let $x\in X\setminus S_X$. By first countability, we take a countable neighborhood base $\set{U_n : n \in \N}$ of $x$ such that $U_1 \supseteq U_2 \supseteq U_3 \ldots $ (if necessary, we replace $U_n$ by $\cap_{i=1}^nU_i$). For each $n \in \N$, let $x_n \in U_n \cap S_X$. Then $x_n \r x$. To see this, let $U$ be any open set containing $x$. Then $U_n \subseteq U$ for some $n \in \N$, and so $x_k \in U_n \subseteq U$ for all $k \geq n$. 

Because $Y$ is compact and first countable, hence sequentially compact, the sequence $(f(x_n))_n$ has a convergent subsequence; let $y$ be its limit. Define $f(x) = y$. Extend $f$ to all of $X$ this way.

We need to verify that $f$ is weight-preserving. Let $x,x' \in X$. Invoking the definition of $f$, let $(x_n)_n, (x'_n)_n$ be sequences in $S_X$ converging to $x,x'$ such that $f(x_n) \r f(x)$ and $f(x'_n) \r f(x')$. 
By sequential continuity and the standard result that a sequence converges in the product topology iff the components converge, we have 
\[\lim_n \w_Y(f(x_n),f(x'_n)) = \w_Y(f(x),f(x')); \qquad
\lim_n \w_X(x_n,x'_n) = \w_X(x,x').\]

Let $\e > 0$. By the previous observation, fix $N \in \N$ such that for all $n\geq N$, we have $|\w_Y(f(x_n),f(x'_n)) - \w_Y(f(x),f(x'))| < \e$ and $|\w_X(x_n,x'_n) - \w_X(x,x')| < \e$. Then, 
\begin{align*}
|\w_X(x,x') - \w_Y(f(x),f(x'))| &= |\w_X(x,x') - \w_X(x_n,x'_n) + \w_X(x_n,x'_n) - \w_Y(f(x),f(x'))|\\
& \hspace{-4em} \leq |\w_X(x,x') - \w_X(x_n,x'_n)| + |\w_X(f(x_n),f(x'_n)) - \w_Y(f(x),f(x'))| < 2\e.
\end{align*}
Thus $\w_X(x,x') = \w_Y(f(x),f(x'))$. Since $x,x' \in X$ were arbitrary, this concludes the proof.\end{proof}

\begin{definition} Let $(X,\w_X), (Y,\w_Y)$ be networks. We say that a map $f:X \to Y$ is \emph{$l^\infty$-intrinsic} if we have $\delta_Y(f(x),f(x'))= \delta_X(x,x')$ for all $x,x' \in X$, i.e. 
\begin{align*}
&\max \left( \| \w_X(x,\cdot) - \w_X(x',\cdot)\|_\infty,
\| \w_X(\cdot,x) - \w_X(\cdot,x')\|_\infty \right) \\
=
&\max \left( \| \w_Y(f(x),\cdot) - \w_Y(f(x'),\cdot)\|_\infty,
\| \w_Y(\cdot,f(x)) - \w_Y(\cdot,f(x'))\|_\infty \right). \end{align*}

\end{definition}

Recall that a topological space is \emph{separable} if it contains a countable dense subset. 

\begin{lemma} 
\label{lem:intrinsic}
Let $(X,\w_X), (Y,\w_Y)$ be separable, compact, weakly coherent networks, and suppose $\M_n(X) = \M_n(Y)$ for all $n \in \N$. Let $f:X \to Y$ be a weight-preserving map. Then $f$ is $l^\infty$-intrinsic.
\end{lemma}

The interpretation of this lemma is that when two networks have the same motif sets, they contain the same information, and thus a mapping between them does not add any extrinsic data.
\begin{proof}[Proof of Lemma \ref{lem:intrinsic}]

Let $x,x' \in X$. Because $f$ is weight-preserving, we know 
\[ \| \w_Y(f(x),\cdot) - \w_Y(f(x'),\cdot)\|_\infty \geq \|\w_X(x,\cdot) - \w_X(x',\cdot) \|_\infty.\]
First we wish to show that this is actually an equality. Towards a contradiction, suppose the inequality above is strict. Let $y \in Y$ be such that 
\[ |\w_Y(f(x),y) - \w_Y(f(x'),y)| > \|\w_X(x,\cdot) - \w_X(x',\cdot)\|_\infty.\]
Using separability, fix a countable dense subset $\{x_1,x_2,\ldots \} \subseteq X$. Next, for each $n \in \N$, define $T_n$ to be the weight matrix obtained from $(y, f(x),f(x'),f(x_1),f(x_2),\ldots, f(x_n))$, i.e. the corresponding motif matrix. Because all the motif sets are equal, for each $T_n \in \M_{n+3}(Y)$ we have an equivalent copy in $\M_{n+3}(X)$. 

Because we do not maintain node labels when creating motif sets, it may be the case that even though $T_n$ is a submatrix of $T_{n+1}$, the points chosen to realize $T_n$ in $X$ are different from the points in $X$ used to realize $T_{n+1}$. Specifically, let $(z^1,a^1,b^1,p^1_1)$ denote the points in $X$ realizing $T_1$, let $(z^2,a^2,b^2,p^2_1,p^2_2)$ denote the points in $X$ realizing $T_2$, and more generally, let $(z^n,a^n,b^n,p^n_1,p^n_2,\ldots, p^n_n)$ denote the points in $X$ realizing $T_n$ for each $n \in \N$.

By sequential compactness and a diagonal argument as used in Proposition \ref{prop:wp-motifs}, we pass to a subsequence such that $z^n, a^n,b^n$, and $p_i^n$ for $i \leq n$ all converge as $n \to \infty$. Denote the limits by $z,a,b$, and $\{p_i\}_{i\in \N}$, respectively. Crucially we have $\w_X(a^n,z^n) = \w_Y(f(x),y)$, $\w_X(b^n,z^n) = \w_Y(f(x'),y)$, and $\w_X(a^n,b^n) = \w_Y(f(x),f(x')) = \w_X(x,x')$ for all $n \in \N$. We also have $\w_X(a^n,p^n_i) = \w_Y(f(x),f(x_i)) = \w_X(x,x_i)$ for all $n \geq i$, and likewise for other combinations of $z^n,a^n,b^n$, and $\{p^n_i\}_{i \in \N, n \geq i}$. Consequently we have the same equalities for the limits of these sequences of points.

Next define a map $g: \{a,b,p_1,p_2,\ldots\} \to X$ as $a \mapsto x, b \mapsto x'$, and $p_i \mapsto x_i$ for each $i \in \N$. Then $g$ is a weight-preserving map defined on a countable dense subset of $X$, so by Proposition \ref{prop:wp-extn} it extends to a weight-preserving map $g:X \to X$. In particular we have $\w_X(x,g(z)) = \w_X(a,z) = \w_Y(f(x),y)$ and $\w_X(x',g(z)) = \w_X(b,z) = \w_Y(f(x'),y)$. Thus we have 
\[ \| \w_X(x,\cdot) - \w_X(x',\cdot)\|_\infty \geq |\w_X(x,g(z)) - \w_X(x',g(z))| = | \w_Y(f(x),y) - \w_Y(f(x'), y) |,\]
which is a contradiction. Similarly one proves 
\[ \| \w_Y(\cdot,f(x)) - \w_Y(\cdot,f(x'))\|_\infty = \|\w_X(\cdot,x) - \w_X(\cdot,x') \|_\infty.\]
Because $x,x'\in X$ were arbitrary, this concludes the proof. \end{proof}

The following proposition mirrors the familiar notion that isometric maps between metric spaces are continuous. As an immediate application, we will obtain \emph{uniqueness} of the coherent topology, thus justifying its construction.

\begin{proposition}\label{prop:cohrt-cts} Let $(X,\w_X), (Y,\w_Y)$ be separable, compact, weakly coherent networks, and suppose $\M_n(X) = \M_n(Y)$ for all $n\in \N$. Let $f:X \r Y$ be a weight-preserving map. Then $f$ is continuous. 
\end{proposition}

\begin{proof} Let $V\subseteq Y$ be open. We need to show that $U:= f\inv(V)$ is open. By first-countability, it suffices to show that any sequence converging to a point of $U$ is eventually inside $U$. Let $(x_n)_n$ be a sequence in $X$ converging to $x \in U$. Because $f$ is $l^\infty$-intrinsic (Lemma \ref{lem:intrinsic}), we have
\[
\norm{\w_Y(f(x_n),\cdot) - \w_Y(f(x),\cdot)}_\infty = 
\norm{\w_X(x_n,\cdot) - \w_X(x,\cdot)}_\infty \to 0.\]
Here the limit holds because $X$ is coherent. Thus we have $f(x_n) \to f(x)$. But then there must exist $N \in \N$ such that $f(x_n) \in V$ for all $n\geq N$. Then $x_n \in U$ for all $n \geq N$. Since $(x_n)_n$ was arbitrary, it follows that $U$ is open. This concludes the proof. \end{proof}

\begin{corollary}[Uniqueness] 
\label{cor:cohrt-unique}
Let $(X,\w_X)$ be a network with a separable, compact, weakly coherent topology $\t_X$. Let $\t'$ be another separable, compact, coherent topology on $X$. Then $\t' = \t_X$. 
\end{corollary}
\begin{proof}
The identity map $\id :(X,\w_X,\t_X) \to (X,\w_X,\t')$ is weight-preserving, and since both the domain and codomain have the same motif set, we know by Proposition \ref{prop:cohrt-cts} that $\id$ is continuous. Thus $\t_X$ is finer than $\t'$. By the same argument, we have that $\t'$ is finer than $\t_X$. Thus the two topologies are equal. \end{proof}

The next result generalizes the result that an isometric embedding of a compact metric space into itself is automatically surjective \cite[Theorem 1.6.14]{burago}. However, before presenting the theorem we first discuss an auxiliary construction that is used in its proof.

\begin{definition}[A variant of $\d_X$ (cf. Equation (\ref{eq:delta}))] 
\label{defn:delta-witness}
Let $(X,\w_X)$ be any network. For any subset $A\subseteq X$, define $\text{\gls{deltaxa}}:X\times X \r \R_+$ by 
\[\d_X^A(x,x'):=\max\left(\sup_{a\in A}|\w_X(x,a) - \w_X(x',a)|,\sup_{a\in A}|\w_X(a,x) - \w_X(a,x')|\right).\]
Then $\d_X$ is a special case of $\d_X^A$ where $A=X$. More generally, $A$ acts as a set of ``observers", and $\d_X^A$ compares two points $x,x'$ with respect to this observation set.

Next, for any $E\subseteq X$ and any $y\in X$, define $\d_X^A(y,E):=\inf_{y'\in E}\d_X^A(y,y')$. Then $\d^A_X(\cdot,E)$ behaves as a proxy for the ``distance to a set" function, where the set is fixed to be $E$ and the ``distance" is measured relative to the observation set $A$.
\end{definition}

\begin{theorem}\label{thm:chrt-surj}
Let $(X,\w_X)$ be a compact network with a weakly coherent, Hausdorff topology. Suppose $f:X \r X$ is a weight-preserving map. Then $f$ is surjective.
\end{theorem}

\begin{proof}[Proof of Theorem \ref{thm:chrt-surj}]
Towards a contradiction, suppose $f(X) \neq X$. By Proposition \ref{prop:cohrt-cts}, $f$ is continuous. Define $X_0:=X$, and $X_n:=f(X_{n-1})$ for each $n \in \N$. The continuous image of a compact space is compact, and compact subspaces of a Hausdorff space are closed. Thus we obtain a decreasing sequence of nonempty compact sets $X_0 \supseteq X_1 \supseteq X_2 \supseteq \ldots$. Then $Z:= \cap_{n\in \N} X_n$ is nonempty and compact, hence closed. 

We now break up the proof up into several claims.

\begin{claim}
$f(Z) = Z.$
\end{claim}
To see this, first note that $f(\cap_{n\in \N}X_n) \subseteq \cap_{n\in \N}f(X_n) \subseteq Z$. Next let $v \in Z$. For each $n\in \N$, let $u_n \in X_n$ be such that $f(u_n) = v$. Since singletons in a Hausdorff space are closed, we know that $\set{v}$ is closed. By continuity, it follows that $f\inv(\set{v})$ is closed. 

By sequential compactness, the sequence $(u_n)_n$ has a convergent subsequence that converges to some limit $u$. Since each $u_n \in f\inv(\set{v})$ and a closed set contains its limit points, we then have $u \in f\inv(\set{v})$. Thus $f(u) = v$, and $v \in f(Z)$. Hence $Z = f(Z)$. This proves the claim. 

Let $x \in X_0 \setminus X_1$. Define $x_0:=x$, and for each $n \in \N$, define $x_n:=f(x_{n-1})$. Then $(x_n)_n$ is a sequence in the sequentially compact space $X$, and so it has a convergent subsequence $(x_{n_k})_k$. Let $z$ be the limit of this subsequence. 

\begin{claim}
$z \in Z$.
\end{claim}
To see this, suppose towards a contradiction that $z \not\in Z$. Then there exists $N\in \N$ such that $z \not\in X_N$. Since $X_N$ is closed, we have that $X\setminus X_N$ is open. By the definition of convergence, $X\setminus X_N$ contains all but finitely many terms of the sequence $(x_{n_k})_k$. But each $x_{n_k}$ belongs to $X_{n_k}$, which is a subset of $X_N$ for sufficiently large $k$. Thus infinitely many terms of the sequence $(x_{n_k})_k$ belong to $X_N$, a contradiction. Hence $z\in Z$.

Now we invoke the $\d_X^\bullet$ construction as in Definition \ref{defn:delta-witness}.

\begin{claim}
For any $E\subseteq X$ and any $y \in E$, 
\[\d_X^E(y,Z) = \d_X^{f(E)}(f(y),f(Z)).\]
\end{claim}

To see this claim, fix $y\in E$. 
Let $v\in f(Z)$. Then $v =f(y') $ for some $y' \in Z$, and $\d_X^{f(E)}(f(y),v) = \d_X^E(y,y')$. To see the latter assertion, let $u\in f(E)$; then $u =f(y'') $ for some $y'' \in E$. Because $f$ is weight-preserving, we then have:
\begin{align*}
|\w_X(f(y),u)-\w_X(v,u)| = |\w_X(f(y),f(y''))-\w_X(f(y'),f(y''))| &= |\w_X(y,y'') - \w_X(y',y'')|,\\
|\w_X(u,f(y))-\w_X(u,v)| = |\w_X(f(y''),f(y))-\w_X(f(y''),f(y'))| &= |\w_X(y'',y) - \w_X(y'',y')|.
\end{align*}
The preceding equalities show that for each $v \in f(Z)$, there exists $y'\in Z$ such that $\d_X^{f(E)}(f(y),v) = \d_X^E(y,y')$. Conversely, for any $y'\in Z$, we have $\d_X^{f(E)}(f(y),f(y')) = \d_X^E(y,y')$. It follows that $\d_X^{f(E)}(f(y),f(Z)) = \d_X^E(y,Z)$.

\begin{claim}
$\d_X^{X}(x,Z) = 0$.
\end{claim}
To see this, assume towards a contradiction that $\d_X^X(x,Z) = \e > 0$ ($\d_X^X$ is nonnegative by definition). 
Since $f(Z) = Z$, we have by the preceding claim that $\d_X^X(x,Z) = \d_X^{f(X)}(f(x),Z) = \ldots = \d_X^{f^n(X)}(f^n(x),Z)$ for each $n \in \N$. In particular, for any $k \in \N$,
\[
\e = \d_X^{f^{n_k}(X)}(f^{n_k}(x),Z) \leq  \d_X^{f^{n_k}(X)}(f^{n_k}(x),z) \leq \d_X^X(f^{n_k}(x),z). 
\]
Here the first inequality follows because the left hand side includes an infimum over $z \in Z$, and the second inequality holds because the right hand side includes a supremum over a larger set.

Since $x_{n_k} \r z$, we have that $\d_X^X(f^{n_k}(x),z) = \d_X^X(x^{n_k},z) = \d_X(x_{n_k},z) \to 0$, where the limit holds by the definition of weak coherence (cf. Definition \ref{defn:coherence-weak}). This contradiction proves the claim.

Recall that by assumption, $x\not\in Z$. For each $n \in \N$, let $z_n \in Z$ be such that $\d_X^X(x,z_n) < 1/n$. Then the sequence $(z_n)_n$ converges to $x$ by weak coherence. Hence any open set containing $x$ also contains infinitely many points of $Z$ that are distinct from $x$. Thus $x$ is a limit point of the closed set $Z$, and so $x \in Z$. This is a contradiction.\end{proof}

\cptsisom*

\begin{proof}[Proof of Theorem \ref{thm:cpt-sisom}]

(\ref{item:cpt-motifs}) follows from (\ref{item:cpt-wisom}) by the stability of motif sets (Theorem \ref{thm:motif-stab}). (\ref{item:cpt-wisom}) follows from (\ref{item:cpt-skel}) by the triangle inequality of $\dn$. We need to show that (\ref{item:cpt-motifs}) implies (\ref{item:cpt-skel}). 

First observe that $\sk(X)$, being a continuous image of the separable space $X$, is separable, and likewise for $\sk(Y)$. Let $S_X, S_Y$ denote countable dense subsets of $\sk(X)$ and $\sk(Y)$. Next, because $\dn(X,\sk(X)) = 0$, an application of Theorem \ref{thm:motif-stab} shows that $\M_n(X) = \M_n(\sk(X))$ for each $n\in \N$. The analogous result holds for $\sk(Y)$. Thus $\M_n(\sk(X)) = \M_n(\sk(Y))$ for each $n \in \N$. Since $X$ and $Y$ have weakly coherent topologies, so do $\sk(X)$ and $\sk(Y)$, by Proposition \ref{prop:skel-cohrt}. By Propositions \ref{prop:wp-motifs} and \ref{prop:wp-extn}, there exist weight-preserving maps $\ph: \sk(X) \r \sk(Y)$ and $\psi:\sk(Y) \r \sk(X)$. Define $X^{(1)}:=\psi(\sk(Y))$ and $Y^{(1)}:= \ph(\sk(X))$. Also define $\ph_1$ and $\psi_1$ to be the restrictions of $\ph$ and $\psi$ to $X^{(1)}$ and $Y^{(1)}$, respectively. Finally define $X^{(2)}:= \psi_1(Y^{(1)})$ and $Y^{(2)}:= \ph_1(X^{(1)})$. Then we have the following diagram.

\begin{center}
\begin{tikzpicture}
\node (11) at (0,2) {$\sk(X)$};
\node (12) at (0,0) {$\sk(Y)$};
\node (21) at (3,2) {$X^{(1)}$};
\node (22) at (3,0) {$Y^{(1)}$};
\node (31) at (6,2) {$X^{(2)}$};
\node (32) at (6,0) {$Y^{(2)}$};

\node at (1.5,0) {$\supseteq$};
\node at (1.5,2) {$\supseteq$};
\node at (4.5,0) {$\supseteq$};
\node at (4.5,2) {$\supseteq$};

\path[->] (11) edge node[above, pos=0.25]{$\ph$}(22);
\path[->,dashed] (12) edge node[pos=0.15, above]{$\psi$}(21);
\path[->] (21) edge node[above, pos=0.25]{$\ph_1$}(32);
\path[->,dashed] (22) edge node[pos=0.15, above]{$\psi_1$}(31);

\end{tikzpicture}

\end{center} 

Now $\psi \circ \ph$ is a weight-preserving map from $\sk(X)$ into itself. Furthermore, it is continuous by Proposition \ref{prop:cohrt-cts}. Since $\sk(X)$ is Hausdorff (Proposition \ref{prop:skel-hdrf}), an application of Theorem \ref{thm:chrt-surj} now shows that $\psi\circ \ph : \sk(X) \r \sk(X)$ is surjective. It follows from Definition \ref{defn:terminal} that $\psi\circ \ph$ is an automorphism of $\sk(X)$, hence a bijection. It follows that $\ph$ is injective. The dual argument for $\ph\circ \psi$ shows that $\psi$ is also injective. 

Since $\psi\circ\ph(\sk(X)) = X^{(2)} = \sk(X)$ and $X^{(2)} \subseteq X^{(1)} \subseteq \sk(X)$, we must have $X^{(1)} = \sk(X)$. Similarly, $Y^{(1)} = \sk(Y)$. Thus $\ph:\sk(X) \r \sk(Y)$ and $\psi:\sk(Y) \r \sk(X)$ are a weight-preserving bijections. In particular, we have $\sk(X) \cong^s \sk(Y)$. This concludes the proof.\end{proof}

\subsubsection{Proofs related to Theorem \ref{thm:complete}}
\label{sec:pf-metrics-complete}

\begin{lemma}
\label{lem:chain-corr}
Let $X_1,\ldots, X_n \in \Ncal$, and for each $i = 1,\ldots, n-1$, let $R_i \in \Rsc(X_i,X_{i+1})$. Define 
\begin{align*}
R&:=R_1\circ R_2 \circ \cdots \circ R_n\\
&=\set{(x_1,x_n)\in X_1\times X_n \mid \exists (x_i)_{i=2}^{n-1}, (x_i,x_{i+1})\in R_i \text{ for all } i}.
\end{align*}
Then $\dis(R) \leq \sum_{i=1}^n\dis(R_i).$
\end{lemma}
\begin{proof}
We proceed by induction, beginning with the base case $n=2$. For convenience, write $X:= X_1$, $Y:=X_2$, and $Z:= X_3$. Let $(x,z), (x',z') \in R_1 \circ R_2$. Let $y\in Y$ be such that $(x,y) \in R_1$ and $(y,z) \in R_2$. Let $y' \in Y$ be such that $(x',y') \in R_1, (y',z')\in R_2$. Then we have:
\begin{align*}
|\w_X(x,x') - \w_Z(z,z')| &= |\w_X(x,x') - \w_Y(y,y') + \w_Y(y,y') - \w_Z(z,z')|\\
&\leq |\w_X(x,x') - \w_Y(y,y')| + |\w_Y(y,y') - \w_Z(z,z')|\\
&\leq \dis(R) + \dis(S).
\end{align*} This holds for any $(x,z), (x',z') \in R\circ S$, and proves the claim.

Suppose that the result holds for $n=N \in \N$. Write $R'=R_1\circ \cdots \circ R_N$ and $R=R' \circ R_{N+1}.$ Since $R'$ is itself a correspondence, applying the base case yields:
\begin{align*}
\dis(R) &\leq \dis(R') + \dis(R_{N+1})\\
&\leq \sum_{i=1}^N\dis(R_i) + \dis(R_{N+1}) \quad \text{by induction}\\
&=\sum_{i=1}^{N+1}\dis(R_i).
\end{align*}
This proves the lemma. \end{proof}

\complete*

\begin{proof}[Proof of Theorem \ref{thm:complete}]

Let $([X_i])_{i\in \N}$ be a Cauchy sequence in $\Ncal/{\cong^w}$. First we wish to show this sequence converges in $\Ncom/{\cong^w}$. Note that $(X_i)_{i\in \N}$ is a Cauchy sequence in $\Ncal$, since the distance between two equivalence classes is given by the distance between any representatives. To show $(X_i)_i$ converges, it suffices to show that a subsequence of $(X_i)_i$ converges, so without loss of generality, suppose $\dn(X_i,X_{i+1})<2^{-i}$ for each $i$. Then for each $i$, there exists $R_i\in \Rsc(X_i,X_{i+1})$ such that $\dis(R_i)\leq 2^{-i+1}$. Fix such a sequence $(R_i)_{i\in \N}$. For $j> i$, define $$R_{ij}:=R_i\circ R_{i+1}\circ R_{i+2}\circ \cdots \circ R_{j-1}.$$ 
By Lemma \ref{lem:chain-corr}, $\dis(R_{ij})\leq \dis(R_i) + \dis(R_{i+1}) + \ldots +\dis(R_{j-1})\leq 2^{-i+2}$. Next define:
\[overline{X}:=\set{(x_j) : (x_j,x_{j+1})\in R_j \text{ for all } j\in \N} \subseteq \ds\prod_{i \in \N}X_i.\]
To see $\overline{X} \neq \emptyset$, let $x_1\in X_1$, and use the (nonempty) correspondences to pick a sequence $(x_1,x_2,x_3,\ldots)$. By construction, $(x_i)\in \overline{X}$. 

Define $\w_{\overline{X}}((x_j),(x_j')) = \limsup_{j\r \infty}\w_{X_j}(x_j,x_j')$. We claim that $\w_{\overline{X}}$ is bounded, and thus is a real-valued weight function. To see this, let $(x_j),(x_j') \in {\overline{X}}$. Let $j\in \N$. Then we have:
\begin{align*}
|\w_{X_j}(x_j,x_j')|&=|\w_{X_j}(x_j,x_j')-\w_{X_{j-1}}(x_{j-1},x'_{j-1})+\w_{X_{j-1}}(x_{j-1},x'_{j-1})-\ldots\\
&\quad -\w_{X_{1}}(x_{1},x'_{1})+\w_{X_{1}}(x_{1},x'_{1})|\\
&\leq |\w_{X_1}(x_1,x'_1)|+ \dis(R_1)+\dis(R_2)+\ldots+\dis(R_{j-1})\\
&\leq |\w_{X_1}(x_1,x'_1)|+ 2
\intertext{But $j$ was arbitrary. Thus we obtain:}
|\w_{\overline{X}}((x_j),(x_j'))| &= \limsup_{j\r \infty}|\w_{X_j}(x_j,x_j')| \leq |\w_{X_1}(x_1,x'_1)|+ 2< \infty.
\end{align*}

{
\newcommand{\Xbar}{\overline{X}}
\newcommand{\wbar}{\w_{\Xbar}}

\begin{claim}
\label{cl:X-cpt}
$(\Xbar,\wbar) \in \Ncom$. More specifically, $\Xbar$ is a first countable compact topological space, and $\wbar$ is continuous with respect to the product topology on $\Xbar \times \Xbar$.
\end{claim}
\begin{proof}[Proof of Claim \ref{cl:X-cpt}] We equip $\prod_{i\in \N}X_i$ with the product topology. First note that the countable product $\prod_{i\in \N} X_i$ of first countable spaces is first countable. Any subspace of a first countable space is first countable, so $\Xbar \subseteq \prod_{i\in \N} X_i$ is first countable. By Tychonoff's theorem, $\prod_{i\in \N} X_i$ is compact. So to show that $\Xbar$ is compact, we only need to show that it is closed. 

If $\Xbar = \prod_{i\in \N}X_i$, we would automatically know that $\Xbar$ is compact. Suppose not, and let $(x_i)_{i\in\N} \in \left(\prod_{i\in \N}X_i \right)\setminus \Xbar$. Then there exists $N\in \N$ such that $(x_N,x_{N+1}) \not\in R_N$. Define:
\[U:= X_1 \times X_2 \times \ldots \times \set{x_N} \times \set{x_{N+1}} \times X_{N+2} \times \ldots .\]
Since $X_i$ has the discrete topology for each $i\in \N$, it follows that $\set{x_N}$ and $\set{x_{N+1}}$ are open. Hence $U$ is an open neighborhood of $(x_i)_{i\in \N}$ and is disjoint from $\prod_{i\in \N}X_i$. It follows that $\left(\prod_{i\in \N}X_i \right)\setminus \Xbar$ is open, hence $\Xbar$ is closed and thus compact. 

It remains to show that $\wbar$ is continuous. We will show that preimages of open sets in $\R$ under $\wbar$ are open. Let $(a,b) \subseteq \R$, and suppose $\wbar\inv[(a,b)]$ is nonempty (otherwise, there is nothing to show). Let $(x_i)_{i\in \N}, (x'_i)_{i\in \N} \in \Xbar \times \Xbar$ be such that 
\[\a:= \wbar((x_i)_i,(x'_i)_i) \in (a,b).\] 
Write $r':= \min(|\a-a|,|b-\a|)$, and define $r:=\tfrac{1}{2}r'$. 

Let $N\in \N$ be such that $2^{-N+3} < r$. Consider the following open sets:
\begin{align*}
U&:=\set{x_1}\times \set{x_2} \times \ldots \times \set{x_N} \times X_{N+1} \times X_{N+2} \times \ldots \subseteq \prod_{i\in \N}X_i,\\
V&:=\set{x'_1}\times \set{x'_2} \times \ldots \times \set{x'_N} \times X_{N+1} \times X_{N+2} \times \ldots \subseteq \prod_{i\in \N}X_i.
\end{align*}

Next write $A:=\Xbar \cap U$ and $B:=\Xbar \cap V$. Then $A$ and $B$ are open with respect to the subspace topology on $\Xbar$. Thus $A\times B$ is open in $\Xbar \times \Xbar$. Note that $(x_i)_{i\in \N} \in A$ and $(x'_i)_{i\in \N} \in B$. We wish to show that $A\times B \subseteq \wbar\inv[(a,b)]$, so it suffices to show that $\wbar(A,B) \subseteq (a,b)$. 

Let $(z_i)_{i\in \N} \in A$ and $(z'_i)_{i\in \N} \in B$. Notice that $z_i= x_i$ and $z'_i = x'_i$ for each $i \leq N$. So for $n\leq N$, we have $|\w_{X_n}(z_n,z'_n) - \w_{X_n}(x_n,x'_n)| = 0.$ 

Next let $n\in \N$, and note that:
\begin{align*}
&|\w_{X_{N+n}}(z_{N+n},z'_{N+n}) - \w_{X_{N+n}}(x_{N+n},x'_{N+n})|\\
&=|\w_{X_{N+n}}(z_{N+n},z'_{N+n}) - \w_{X_N}(z_N,z'_N) + \w_{X_N}(z_N,z'_N) - \w_{X_{N+n}}(x_{N+n},x'_{N+n})|\\
&=|\w_{X_{N+n}}(z_{N+n},z'_{N+n}) - \w_{X_N}(z_N,z'_N) + \w_{X_N}(x_N,x'_N) - \w_{X_{N+n}(x_{N+n}},x'_{N+n})|\\
&\leq \dis(R_{N,N+n}) + \dis(R_{N,N+n}) \leq 2^{-N+2} + 2^{-N+2} = 2^{-N+3} < r.
\end{align*}
Here the second to last inequality follows from Lemma \ref{lem:chain-corr}. The preceding calculation holds for arbitrary $n\in \N$. It follows that:
\[\limsup_{i \r \infty}\w_{X_i}(x_i,x'_i) - 
\limsup_{i \r \infty}\w_{X_i}(z_i,z'_i) \leq 
\limsup_{i \r \infty}(\w_{X_i}(x_i,x'_i)-\w_{X_i}(z_i,z'_i)) < r,\]
and similarly $\limsup_{i \r \infty}\w_{X_i}(z_i,z'_i) - 
\limsup_{i \r \infty}\w_{X_i}(x_i,x'_i) < r.$ 
Thus we have $\wbar((z_i)_i,(z'_i)_i) \in (a,b)$. This proves continuity of $\wbar$. \end{proof}

}

Next we claim that $X_i \xrightarrow{\dn} \overline{X}$ as $i\r \infty$. Fix $i\in \N$. We wish to construct a correspondence $S\in \Rsc(X_i,\overline{X})$. Let $y\in X_i$. We write $x_i=y$ and pick $x_1,x_2,\ldots, x_{i-1},x_{i+1},\ldots$ such that $(x_j,_{j+1}) \in R_j$ for each $j\in \N$. We denote this sequence by $(x_j)^{x_i=y}$, and note that by construction, it lies in $\overline{X}$. Conversely, for any $(x_j) \in \overline{X}$, we simply pick its $i$th coordinate $x_i$ as a corresponding element in $X_i$. We define:
\begin{align*}
S &:=A \cup B, \text{ where }\\
A &:= \set{(y,(x_j)^{x_i=y}) : y\in X_i}\\
B &:= \set{(x_i,(x_k)) : (x_k) \in \overline{X}}
\end{align*}
Then $S \in \Rsc(X_i,\overline{X})$. We claim that $\dis(S) \leq 2^{-i+2}$. Let $z=(y,(x_k)),z'=(y',(x'_k)) \in B$. Let $n\in \N, n \geq i$. Then we have:
\begin{align*}
|\w_{X_i}(y,y') - \w_{X_n}(x_n,x'_n)|
&=|\w_{X_i}(y,y') - \w_{X_{i+1}}(x_{i+1},x'_{i+1})+
\w_{X_{i+1}}(x_{i+1},x'_{i+1})-\ldots\\
&\quad + \w_{X_{n-1}}(x_{n-1},x'_{n-1})+
\w_{X_{n}}(x_{n},x'_{n})|\\
&\leq \dis(R_i) + \dis(R_{i+1})+\ldots+\dis(R_{n-1})\\
&\leq 2^{-i+1} + 2^{-i} + \ldots + 2^{-n+2}\\
&\leq 2^{-i+2}.
\intertext{This holds for arbitrary $n\geq i$. It follows that we have:}
|\w_{X_i}(y,y') - \w_{\overline{X}}((x_k),(x'_k))|&\leq 2^{-i+2}.
\end{align*}
Similar inequalities hold for $z,z' \in A$, and for $z\in A,z'\in B$. Thus $\dis(S) \leq 2^{-i+2}$. It follows that $\dn(X_i,\overline{X})\leq 2^{-i+1}$. Thus the sequence $([X_i])_i$ converges to $[\overline{X}]\in \Ncom/{\cong^w}$. 

Finally, we need to check that $(\Ncom/{\cong^w},\dn)$ is complete. Let $([Y_n])_n$ be a Cauchy sequence in $\Ncom/{\cong^w}$. By invoking Theorem \ref{thm:sampling}, for each $n$ let $[X_n] \in \Ncal/{\cong^w}$ be such that $\dn([X_n],[Y_n])<\tfrac{1}{n}$. Let $\e > 0$. Then for sufficiently large $m$ and $n$, we have: 
\[\dn([X_n],[X_m])\leq \dn([X_n],[Y_n])+\dn([Y_n],[Y_m])+\dn([Y_m],[X_m])<\e.\]
Thus $([X_n])_n$ is a Cauchy sequence in $\Ncal/{\cong^w}$. By applying what we have shown above, this sequence converges to some $[\overline{X}] \in \Ncom/{\cong^w}$. By applying the triangle inequality, we see that the sequence $([Y_n])_n$ also converges to $[\overline{X}]$. This shows completeness, and concludes the proof. \end{proof}

\begin{remark} In the proof of Theorem \ref{thm:complete}, note that the construction of the limit is dependent upon the initial choice of optimal correspondences. However, all such limits obtained from different choices of optimal correspondences belong to the same weak isomorphism class. 
\end{remark}

\subsubsection{Proofs related to Theorem \ref{thm:geod-inf}}
\label{sec:pf-metrics-geod-inf}

\geodesic*

\begin{proof}[Proof of Theorem \ref{thm:geodesic}]
Let $[X],[Y] \in \Ncalsim$. We will show the existence of a curve $\g: [0,1] \r \Ncal$ such that $\g(0)=(X,\w_X)$, $\g(1)=(Y,\w_Y)$, and for all $s, t \in [0,1]$, 
\[\dn(\g(s),\g(t)) = |t-s| \cdot \dn(X,Y).\]
Note that this yields $\dn([\g(s)],[\g(t)]) = |t-s| \cdot  \dn([X],[Y])$ for all $s,t \in [0,1]$, which is what we need to show. 

Let $R\in\Rsc^{\opt}(X,Y)$, i.e. let $R$ be a correspondence such that $\dis(R)=2\dn(X,Y)$. For each $t\in(0,1)$ define $\g(t):=\big(R,\w_{\gamma(t)}\big)$, where 
\[\w_{\gamma(t)}\big((x,y),(x',y')\big):=(1-t) \cdot \w_X(x,x')+t \cdot \w_Y(y,y') \quad\text{for all $(x,y),(x',y')\in R$}.\]
Also define $\g(0)=(X,\w_X)$ and $\g(1)=(Y,\w_Y)$. 

\begin{claim}\label{cl:ineq} For any $s,t \in [0,1]$, 
\[\dn(\g(s),\g(t)) \leq |t-s| \cdot \dn(X,Y).\]
\end{claim}
Suppose for now that Claim \ref{cl:ineq} holds. We further claim that this implies, for all $s, t \in [0,1]$, 
\[\dn(\g(s),\g(t)) = |t-s| \cdot \dn(X,Y).\] 
To see this, assume towards a contradiction that there exist $s_0 < t_0$ such that :
\begin{align*}
\dn(\g(s_0),\g(t_0)) &< |t_0 - s_0| \cdot \dn(X,Y).\\
\text{Then }\dn(X,Y) &\leq \dn(X,\g(s_0)) + \dn(\g(s_0),\g(t_0)) + \dn(\g(t_0),Y)\\
&< |s_0-0| \cdot \dn(X,Y) + |t_0-s_0|\cdot \dn(X,Y) + |1-t_0|\cdot \dn(X,Y)\\
&=\dn(X,Y), \text{ a contradiction}.
\end{align*}

Thus it suffices to show Claim \ref{cl:ineq}. There are three cases: (i) $s,t \in (0,1)$, (ii) $s=0, t \in (0,1)$, and (iii) $s\in (0,1), t=1$. The latter two cases are similar, so we just prove (i) and (ii).
For (i), fix $s,t \in (0,1)$. Notice that $\Delta := \textrm{diag}(R\times R) := \set{(r,r):r \in R}$ is a correspondence in $\Rsc(R,R).$ Then we obtain:
\begin{align*}
\dis(\Delta) &= \max_{(a,a),(b,b) \in \Delta}\vert \w_{\g(t)}(a,b) - \w_{\g(s)}(a,b)\vert\\
&= \max_{(x,y),(x',y') \in R}\vert \w_{\g(t)}((x,y),(x',y')) - \w_{\g(s)}((x,y),(x',y'))\vert\\
&= \max_{(x,y),(x',y') \in R}\vert (1-t)\w_X(x,x') + t\cdot  \w_Y(y,y') - (1-s)\w_X(x,x') -s\cdot \w_Y(y,y') \vert\\
&= \max_{(x,y),(x',y') \in R}\vert (s-t)\w_X(x,x') - (s-t)\w_Y(y,y')\vert\\
&= |t-s| \cdot \max_{(x,y),(x',y') \in R}\vert \w_X(x,x') - \w_Y(y,y')\vert\\
&\leq 2|t-s|\cdot \dn(X,Y).
\end{align*}
Finally $\dn(\g(t),\g(s))\leq \frac{1}{2}\dis(\Delta) \leq |t-s|\cdot \dn(X,Y)$.

For (ii), fix $s=0, t\in (0,1)$. Define $R_X = \set{(x,(x,y)) : (x,y) \in R}$. Then $R_X$ is a correspondence in $\Rsc(X, R)$.
\begin{align*}
\dis(R_X) &= \max_{(x,(x,y)),(x',(x',y')) \in R_X}|\w_X(x,x') - (1-t)\cdot \w_X(x,x') - t\cdot \w_Y(y,y')|\\
&=\max_{(x,(x,y)),(x',(x',y')) \in R_X} t\cdot  |\w_X(x,x') - \w_Y(y,y')|  \\
&=t  \dis(R) = 2t\cdot  \dn(X,Y).
\end{align*}
Thus $\dn(X,\g(t)) \leq t \cdot \dn(X,Y)$. The proof for case (iii), i.e. that $\dn(\g(s),Y) \leq |1-s| \cdot \dn(X,Y)$, is similar. This proves Claim \ref{cl:ineq}, and the result follows. \end{proof}

\geodinf*

\begin{proof}[Proof of Theorem \ref{thm:geod-inf}] Let $[X],[Y] \in \Ncomsim$. It suffices to find a geodesic between $X$ and $Y$, because the distance between any two equivalence classes is given by the distance between any two representatives, and hence we will obtain a geodesic between $[X]$ and $[Y]$. 

Let $(X_n)_n, (Y_n)_n$ be sequences in $\Ncal$ such that $\dn(X_n,X)<\tfrac{1}{n}$ and $\dn(Y_n,Y)<\tfrac{1}{n}$ for each $n$. For each $n$, let $R_n$ be an optimal correspondence between $X_n$ and $Y_n$, endowed with the weight function
\[\w_n((x,y),(a,b))=\tfrac{1}{2}\w_{X_n}(x,a)+\tfrac{1}{2}\w_{X_n}(y,b).\]
By the proof of Theorem \ref{thm:geodesic}, the network $(R_n,\w_n)$ is a midpoint of $X_n$ and $Y_n$. 
\begin{claim}\label{cl:mid} The collection $\set{R_n : n \in \N}$ is precompact.
\end{claim}

Assume for now that Claim \ref{cl:mid} is true. Then we can pick a sequence $(R_n)$ that converges to some $R \in \Ncom$. Then we obtain:
\begin{align*}
\dn(X,R) &\leq \dn(X,X_n)+\dn(X_n,R_n)+\dn(R_n,R)\\
&=\dn(X,X_n) + \tfrac{1}{2}\dn(X_n,Y_n)+\dn(R_n,R) \r \tfrac{1}{2}\dn(X,Y).
\end{align*}
Similarly $\dn(R,Y) \leq \tfrac{1}{2}\dn(X,Y)$. Furthermore, equality holds in both inequalities, because we would get a contradiction otherwise. Thus $R$ is a midpoint of $X$ and $Y$, and moreover, $[R]$ is a midpoint of $[X]$ and $[Y]$. The result now follows by an application of Theorem \ref{thm:mid}. 

It remains to prove Claim \ref{cl:mid}. By Theorem \ref{thm:precompactness}, it suffices to show that $\set{R_n}$ is uniformly approximable. 

Since $\dn(X_n,X) \r 0$ and $\dn(Y_n,Y) \r 0$, we can choose $D>0$ large enough so that $\diam(X_n) \leq \tfrac{D}{2}$ and $\diam(Y_n) \leq \tfrac{D}{2}$ for all $n$. Then $\diam(R_n) \leq D$ for all $n$. 

Let $\e > 0$. Fix $N$ large enough so that $\tfrac{1}{N}<\tfrac{\e}{2}$, and write $N(\e) = \max_{n \leq N}|R_n|$. We wish to show that every $R_n$ is $\e$-approximable by a finite network with cardinality up to $N(\e)$. For any $n\leq N$, we know $R_n$ approximates itself, and $|R_n| \leq N(\e)$. Next let $n > N$. It will suffice to show that $R_n$ is $\e$-approximable by $R_N$. 

Let $S, T$ be optimal correspondences between $X_n,X_N$ and $Y_n,Y_N$ respectively. Note that $\dn(X_N,X_n)\leq \dn(X_N,X) + \dn(X,X_n) \leq \tfrac{1}{N} + \tfrac{1}{N} = \tfrac{2}{N}$, and similarly $\dn(Y_N,Y_n) \leq \tfrac{2}{N}$. Thus $\dis(S) \leq \tfrac{4}{N}$ and $\dis(T)\leq \tfrac{4}{N}$. Next write
\[Q=\set{(x,y,x',y') \in R_N\times R_n : (x,x')\in S, (y,y')\in T}.\]
Observe that since $S$ and $T$ are correspondences, $Q$ is  a correspondence between $R_N$ and $R_n$. Next we calculate $\dis(Q)$:
\begin{align*}
\dis(Q)&= \max_{\substack{(x,y,x',y'),\\(a,b,a',b')\in Q}}|\w_N((x,y),(a,b)) - \w_n((x',y'),(a',b'))|\\
&= \max_{\substack{(x,y,x',y'),\\(a,b,a',b')\in Q}}|\tfrac{1}{2}\w_{X_N}(x,a)+\tfrac{1}{2}\w_{Y_N}(y,b) - \tfrac{1}{2}\w_{X_n}(x',a') - \tfrac{1}{2}\w_{Y_n}(y',b'))|\\
&\leq \tfrac{1}{2}\max_{(x,x'),(a,a')\in S}|\w_{X_N}(x,a)-\w_{X_n}(x',a')| + \tfrac{1}{2}\max_{(y,y'),(b,b')\in S}|\w_{Y_N}(y,b)-\w_{Y_n}(y',b')|\\
&=\tfrac{1}{2}\dis(S) + \tfrac{1}{2}\dis(T) \leq \tfrac{4}{N}.
\end{align*} 
Thus $\dn(R_N,R_n) \leq \tfrac{2}{N} < \e$. This shows that any $R_n$ can be $\e$-approximated by a network having up to $N(\e)$ points. Thus $\set{R_n}$ is uniformly approximable, hence precompact. Thus Claim \ref{cl:mid} and the result follow. \end{proof}

\section{Invariants, stability, and convergence}
\label{sec:inv-conv}

At this point, we have computed $\dn$ between several examples of networks, as in Example \ref{prop:perm} and Remark \ref{remark:bijection}. We also asserted in Remark \ref{rem:compute-hard} that $\dn$ is in general difficult to compute. 
A standard approach that is used in practice to circumvent the problem of computing $\dn$ is to compute {lower bounds} instead. We frame this approach in terms of defining \emph{quantitatively stable invariants} of networks, and then comparing the invariants instead of comparing the networks directly. 
In this section, we set up the framework of stable network invariants, explain how this incorporates existing stability results on persistent homology and hierarchical clustering, and go on to prove convergence of these methods in the infinite setting. In particular, the setting of compact networks gives a suitable family of infinite networks with well-defined persistence diagrams.

Throughout this section, we will reserve the notation $X$ for random variables, and switch to using letters $Y,Z$ to denote networks.

\subsection{Preliminary setup}

Throughout this section, all the underlying topological spaces will be assumed to be compact, unless specified otherwise.
We use the following convention for an open cover of a topological space: An \emph{open cover} of a topological space $Z$ is a collection of open sets $\{U_i \subseteq Z : i \in I\}$ indexed by some set $I$ such that each $U_i$ is \emph{nonempty} and $\bigcup_{i\in I}U_i = Z$.
Given a topological space $Z$, we will write $\borel(Z)$ to denote the Borel $\s$-field on $Z$. We often write $(\Om, \mc{F}, \P)$ to denote a probability space. The support \gls{supp} of a measure $\mu_Z$ on a topological space $Z$ is defined as:
\[\supp(\mu_Z):= \set{ z\in Z : \text{ for each open neighborhood } N_z \ni z,\text{ we have } \mu_Z(N_z) > 0}.\]
The complement of $\supp(\mu_Z)$ is the union of open sets of measure zero. It follows that $\supp(\mu_Z)$ is closed, hence compact.

Given a probability space $(\Om, \mc{F},\P)$ and a measurable space $(Z,\mc{G})$, a \emph{random variable} (defined on $\Om$ with values in $Z$) is a measurable function $\mathbf{z}: \Om \r Z$. The \emph{pushforward} or \emph{image measure} of $\mathbf{z}$ is defined to be the measure $\mathbf{z}_\#\P$ on $\mc{G}$ given by writing $\mathbf{z}_\#\P(A):= \P(\mathbf{z}\inv[A])$ for all $A\in \mc{G}$. The pushforward is often called the \emph{distribution} of $\mathbf{z}$. 

We recall an important corollary of the existence of infinite products of probability measures. For any probability space $(Z,\mc{E},\mu_Z)$, there exists a probability space $(\Om,\mc{F},\P)$ on which there are independent random variables $\mathbf{z}_1,\mathbf{z}_2,\ldots $ taking values in $Z$ with distribution $\mu_Z$ \cite[\S8.2]{dudley}. This is done by letting $\Om:= \prod_{n\in \N} Z$ and taking each $\mathbf{z}_i$ to be the canonical projection map $(z_i)_{i\in \N} \mapsto z_i$.

\begin{definition}[Measure network, \cite{gwnets}]
\label{defn:measure-net}
A \emph{measure network} is a triple $(Z,\w_Z,\mu_Z)$ where $Z$ is a compact Polish space, $\mu_Z$ is a Borel probability measure on $Z$, and $\w_Z:Z\times Z \to \R$ is a bounded measurable function.
\end{definition}

We now specify what it means to sample from a measure network.

\begin{definition}[Sampling from a measure network]
\label{defn:measure-net-sampling}
Let $(Z,\w_Z,\mu_Z)$ be a measure network, and let $(\Om, \mc{F},\P)$ be a probability space. A measurable function $\mathbf{z}:\Om \to Z$ is said to be a random variable taking values in $Z$. For each $i\in \N$, let $\mathbf{z}_i: \Om \r Z$ be an independent random variable with distribution $\mu_Z$. For each $n\in \N$, let $\mc{Z}_n=\set{\mathbf{z}_1,\mathbf{z}_2,\ldots, \mathbf{z}_n}$. To say that we sample $n$ points randomly from $(Z,\w_Z,\mu_Z)$ means that we take a selection $\mc{Z}_n(\omega)=\{ z_1:= \mathbf{z}_1(\omega),z_2:=\mathbf{z}_2(\omega),\ldots, z_n:=\mathbf{z}_n(\omega)\}$ and equip the collection $Z_n:=\{z_1,\ldots,z_n\}$ with the restriction of $\w_Z$.

A related concept is that of \emph{almost sure convergence w.r.t. $\dn$}. To say that $\mc{Z}_n$ converges almost surely to $Z$ with respect to $\dn$ means that for each $\e > 0$,
\[\P\left( \limsup_{n\to \infty} \{ \omega \in \Omega : \dn(\mc{Z}_n(\omega),Z) \geq \e \} \right) = 0.\]
\end{definition}

This notion of sampling is the form used when sampling from metric spaces, and it assumes that selecting the points $Z_n=\{z_1,\ldots, z_n\}$ is enough to obtain the restriction of $\w_Z$ to $Z_n\times Z_n$. This is often the case, such as when $\w_Z$ is given by a Gaussian kernel. Note that this is quite different from sampling from a combinatorial graph, where one has to use techniques such as Markov chain Monte Carlo so as to not sample very sparse graphs \cite{lyu2019sampling}.

The presence of $\mu_Z$ allows us to ask questions such as the following: Given a measure network $(Z,\w_Z,\mu_Z)$ and some $\e>0$, what should it mean to take an \emph{optimal} $\e$-system (cf. Definition \ref{def:systems}) on $Z$? The next definition sheds some light on this question.

\begin{definition}[Minimal mass function \gls{mfm}]
\label{defn:fraktur-m} Let $(Z,\w_Z,\mu_Z)$ be a measure network. Let $\mc{U}$ be any $\e$-system on $Z$ for $\e>0$. We define the \emph{minimal mass function}
$\mf{m}(\mc{U}):=\min\set{\mu_Z(U) : U \in \mc{U}, \;\mu_Z(U) > 0}.$ Note that $\mf{m}$ returns the minimal non-zero mass of an element in $\mc{U}$.

Next let $\e > 0$. Define an \emph{optimal minimal mass function} $\text{\gls{mfM}}_{\e}: \Ncom \r (0,1]$ as follows: 
\[\mf{M}_{\e}(Z):=\sup\set{\mf{m}(\mc{U}) : \mc{U} \text{ a refined $\e$-system on $Z$}}.\]
Recall that a \emph{refined} $\e$-system was defined in Definition \ref{def:systems}.
Since $\mc{U}$ covers $Z$, we know that the total mass of $\mc{U}$ is $1$. Thus the set of elements in $\mc{U}$ with positive mass is nonempty, and so $\mf{m}(\mc{U})$ is strictly positive. It follows that $\mf{M}_{\e}(Z)$ is strictly positive. More is true when $\mu_Z$ is fully supported on $Z$: given any $\e$-system $\mc{U}$ on $Z$ and any $U \in \mc{U}$, we automatically have $\mu_Z(U) > 0$ (recall our convention above that empty elements are excluded from an open cover).
\end{definition}

In the preceding definition, for a given $\e> 0$, the function $\mf{M}_\e(Z)$ considers the collection of all \emph{refined} $\e$-systems on $Z$, and then maximizes the minimal mass of any element in such an $\e$-system. 
The functions in Definition \ref{defn:fraktur-m} are crucial to the next result, which shows that as we sample points from a distribution on a network, the sampled subnetwork converges almost surely to the support of the distribution.

\begin{theorem}[Probabilistic network approximation]\label{thm:prob-net-approx}
Let $(Z,\w_Z,\mu_Z)$ be a measure network. For each $i\in \N$, let $\mathbf{z}_i: \Om \r Z$ be an independent random variable defined on some probability space $(\Om, \mc{F},\P)$ with distribution $\mu_Z$. For each $n\in \N$, let $\mc{Z}_n=\set{\mathbf{z}_1,\mathbf{z}_2,\ldots, \mathbf{z}_n}$. Let $\e > 0$.  Then we have: 
\[\P\big(\set{\omega \in \Om : \dn(\supp(\mu_Z),Z_n)\geq \e } \big) \leq \frac{\left(1-\mf{M}_{\e/2}(\supp(\mu_Z))\right)^n}{\mf{M}_{\e/2}(\supp(\mu_Z))},\]
where $Z_n=\mc{Z}_n(\omega)$ is an $n$-point sample as in Definition \ref{defn:measure-net-sampling}. 
In particular, the subnetwork $\mc{Z}_n$ converges almost surely to $Z$ w.r.t. $\dn$.
\end{theorem}

\begin{remark}[Relation to the metric setting] 
The invariants in Definition \ref{defn:fraktur-m} are abstractions of invariants that are standard in the metric setting. For a concrete example, note that
\textcolor{black}{the mass of an $\e$-ball in $d$-dimensional Euclidean space scales as $\e^d$. Thus in the setting of Euclidean space $\R^d$, the quantity on the right} of Theorem \ref{thm:prob-net-approx} \textcolor{black}{would scale as $\e^{-d}(1-\e^d)^n$.} For stronger results in the metric setting, see \cite[Theorem 30]{clust-um}.  
\end{remark}

Before proving the theorem, we prove the following useful lemma:

\begin{lemma}
\label{lem:sampling-e-system}
Assume the setup of $(Z,\w_Z,\mu_Z)$, $(\Om, \mc{F},\P)$, and $\mc{Z}_n$ for each $n\in \N$ as in Theorem \ref{thm:prob-net-approx}. 
Fix $\e > 0$, and let $\mc{U} = \set{U_1,\ldots, U_m}$ be a refined $\e$-system on $\supp(\mu_Z)$. For each $1\leq i \leq m$ and each $n\in \N$, define the following event:
\[A_i:=\bigcap_{k=1}^n\set{\omega\in \Om: \mathbf{z}_k(\omega) \not\in U_i} \subseteq \Om.\]
Then we have $\P\left(\bigcup_{k=1}^mA_k\right) \leq \frac{1}{\mf{m}(\mc{U})}(1-\mf{m}(\mc{U}))^n.$
\end{lemma}
\begin{proof}[Proof of Lemma \ref{lem:sampling-e-system}]
Here we are considering the probability that at least one of the $U_i$ has empty intersection with $\mc{Z}_n$. As the points are sampled independently, we can invoke independence and write $\P(A_i) = (1- \mu_Z(U_i))^n$. Then we have:
\begin{align*}
\P\left(\bigcup_{k=1}^mA_k\right) &\leq \sum_{k=1}^m\P(A_k) 
= \sum_{k=1}^m(1-\mu_Z(U_k))^n
\leq m\cdot \max_{1\leq k \leq m } (1-\mu_Z(U_k))^n
\leq \frac{(1-\mf{m}(\mc{U}))^n}{\mf{m}(\mc{U})}.
\end{align*}
Here the first inequality follows by subadditivity of measure, and the last inequality follows because the total mass $\mu_Z(\supp(\mu_Z))=1$ is an upper bound for $m\cdot \mf{m}(\mc{U})$. Note also that each $U \in \mc{U}$ has nonzero mass, by the observation in Definition \ref{defn:fraktur-m}. \end{proof}

\begin{proof}[Proof of Theorem \ref{thm:prob-net-approx}]
By endowing $\supp(\mu_Z)$ with the restriction of $\w_Z$, it may itself be viewed as a network with full support. So for notational convenience, we assume $Z=\supp(\mu_Z)$. 

First observe that $\mf{M}_{\e/2}(Z) \in (0,1]$. Let $ r \in (0,\mf{M}_{\e/2}(Z))$, and let $\mc{U}_r$ be an $\e/2$-system on $Z$ such that $\mf{m}(\mc{U}_r) \in (r, \mf{M}_{\e/2}(Z)]$. For convenience, write $m:= |\mc{U}_r|$, and also write $\mc{U}_r= \set{U_1,\ldots, U_m}$. For each $1\leq i \leq m$, define $A_i$ as in the statement of Lemma \ref{lem:sampling-e-system}. 
Then by Lemma \ref{lem:sampling-e-system}, the probability that at least one $U_i$ has empty intersection with $Z_n$ is bounded as $\P\left(\bigcup_{k=1}^mA_k\right) \leq \frac{1}{\mf{m}(\mc{U}_r)}(1-\mf{m}(\mc{U}_r))^n$.  On the other hand, if $U_i$ has nonempty intersection with $Z_n$ for each $1\leq i \leq m$, then by Theorem \ref{thm:sampling} we obtain $\dn(Z,Z_n) < \e$. For each $n\in \N$, define: $B_n:=\set{\omega\in \Om : \dn(Z,\mc{Z}_n(\omega)) \geq \e}.$ Then we have:
\begin{align*}
\P(B_n) \leq \P\left(\bigcup_{k=1}^mA_k\right) \leq  \frac{(1-\mf{m}(\mc{U}_r))^n}{\mf{m}(\mc{U}_r)}.
\end{align*}
Since $ r \in (0,\mf{M}_{\e/2}(Z))$ was arbitrary, letting $r$ approach $\mf{M}_{\e/2}(Z)$ shows that
$\P(B_n) \leq \frac{(1-\mf{M}_{\e/2}(Z))^n}{\mf{M}_{\e/2}(Z)}.$
We have by Definition \ref{defn:fraktur-m} that $\mf{M}_{\e/2}(Z)$ is strictly positive. Thus the term on the right side of the inequality is an element of a convergent geometric series, so
\[\sum_{n=1}^\infty\P(B_n) \leq \frac{1}{\mf{M}_{\e/2}(Z)}\sum_{n=1}^\infty (1-\mf{M}_{\e/2}(Z))^n < \infty.\]
By the Borel-Cantelli lemma, we have $\P(\limsup_{n\r \infty} B_n) = 0$. The result follows.\end{proof}

In the upcoming sections, we will see applications of these results to two different fields of exploratory data analysis: \emph{persistent homology} and \emph{hierarchical clustering}. These methods are classically defined for metric spaces, but generalizations have recently been proposed to allow for meaningful constructions on generalized datasets such as the networks in this work.

\subsection{Network invariants and basic examples}

An $\R$-invariant of a network is a map $\iota:\Ncom \rightarrow \R$ such that for any $Y,Z \in \Ncom$, if $Y\cong^w Z$ then $\iota(Y) = \iota(Z)$. Any $\R$-invariant is an example of a \emph{pseudometric} (and in particular, a \emph{metric}) \emph{space valued invariant}, which we define next. Recall that a pseudometric space $(V,d_V)$ is a metric space where we allow $d_V(v,v')=0$ even if $v\neq v'$.

\begin{definition} Let $(V,d_V)$ be any metric or pseudometric space. A \emph{$V$-valued invariant} is any map $\iota:\Ncom\rightarrow V$ such that $\iota(X,\w_X) = \iota(Y,\w_Y)$ whenever $X\cong^w Y$.
\end{definition}

\begin{figure}
\begin{center}
\begin{tikzpicture}
\node[](dgm) at (5.5,1.5){$(\R^2_{\geq 0},\db)$};
\node[](U) at (5,0.5){$(\Ncom^{\operatorname{ult}},\dn)$};
\node[](N) at (0,0){$\Ncom$};
\node[](Rn2) at (3,0){$(\R^{n\times n}, \dhdf^{l^\infty})$};
\node[](Rdh) at (4.3,-1){$(\R,\dhdf)$};
\node[](Rlp) at (6,-1.2){$(\R,|\cdot |)$};
\draw[->] (N) to [bend left=60] node[above] {$\dgm$} (dgm);
\draw[->] (N) to [bend left=30] node[above] {$\H$} (U);
\draw[->] (N) to [bend left=0] node[above] {$\M_n$} (Rn2);
\draw[->] (N) to [bend left=-30] node[above] {$\spec$} (Rdh);
\draw[->] (N) to [bend left=-60] node[below] {$\diam$} node[above]{$\vdots$} (Rlp);
\end{tikzpicture}
\end{center}
\caption{The framework of quantitatively stable invariants allows us to produce easily computable proxies for $\dn$. Apart from the hierarchical clustering map $\H$, all the other maps embed network features into simple spaces with metrics that can be computed in polynomial time. The $\H$ map embeds an input compact network into $\Ncomult$ (cf. Definition \ref{def:ultrametric}), which can then be displayed as a dendrogram-like object \cite{nets-asilo}.}
\label{fig:invariants}
\end{figure}

In what follows, we will construct several maps that give rise to pseudometric space valued invariants (Proposition \ref{prop:n0invariants}). We will eventually prove that our proposed invariants are \emph{quantitatively stable}. This notion is made precise in \S\ref{sec:quant-stab}.

\begin{example} Define the \emph{diameter} map to be the map 
\[\diam:\Ncom\r \R \text{ given by }(Z,\w_Z)\mapsto \max_{z,z'\in Z} |\w_Z(z,z')|.\] 
Then $\diam$ is an $\R$-invariant. Observe that the maximum is achieved for $(Z,\w_Z) \in \Ncom$ because $Z$ (hence $Z\times Z$) is compact and $\w_: Z\times Z \r \R$ is continuous. An application of $\diam$ to Example \ref{ex:comparison} gives an upper bound on $\dn(Z,Y)$ for $Z,Y \in \Ncom$ in the following way: 
\[\dn(Z,Y) \leq \dn(Z,N_1(0)) + \dn(N_1(0),Y)
= \tfrac{1}{2}(\diam(Z)+\diam(Y)) \text{ for any } Z,Y\in \Ncom.\]
\end{example}

Recall that $\pow(\R)$, the nonempty elements of the power set of $\R$, is a pseudometric space when endowed with the Hausdorff distance \cite[Proposition 7.3.3]{burago}.

\begin{example} 
Define the \emph{spectrum} map 
\[\spec: \Ncom \r \pow(\R) \text{ by }(Z,\w_Z)\mapsto \{\w_Z(z,z') : z,z'\in Z\}.\] 
The spectrum also has two local variants. Define the $\mathrm{out}$-local spectrum of $Z$ by $z \mapsto \spec^{\cout}_Z(z):=\{\w_Z(z,z'),\,z' \in Z\}.$ 
Notice that $\spec(Z) = \bigcup_{z\in Z} \spec^{\cout}_Z(z)$ for any network $Z$, thus justifying the claim that this construction localizes $\spec$. Similarly, we define the $\mathrm{in}$-spectrum of $Z$ as the map $z\mapsto \spec^{\cin}_Z(z) := \set{\w_Z(z',z) : z'\in Z}.$ Notice that one still has $\spec(Z) = \bigcup_{z\in Z} \spec^{\cin}_Z(z)$ for any network $Z$. Finally, we observe that the two local versions of $\spec$ do not necessarily coincide in an asymmetric network.

The spectrum is closely related to the \textit{multisets} used by Boutin and Kemper \cite{boutin2007lossless} to produce invariants of weighted undirected graphs. For an undirected graph $G$, they considered the collection of all subgraphs with three nodes, along with the edge weights for each subgraph (compare to our notion of spectrum). Then they proved that the distribution of edge weights of these subgraphs is an invariant when $G$ belongs to a certain class of graphs.
\end{example}

\begin{example} Define the \emph{trace} map $\tr:\Ncom \r \pow(\R)$ by $(Z,\w_Z) \mapsto \tr(Z):= {\set{\w_Z(z,z) : z \in Z}}.$ This also defines an associated map $z \mapsto \tr_Z(z):= \w_Z(z,z).$ An example is provided in Figure \ref{fig:tr}: in this case, we have $(Z,\tr_Z) = (\{p,q\},(\alpha,\beta))$. 

\end{example}

\begin{figure}
\begin{center}
\begin{tikzpicture}[every node/.style={font=\footnotesize}]
\node[circle,draw,above](2) at (0,0){$p$};
\node[circle,draw](3) at (1.5,0){$q$};

\path[->] (2) edge [loop left, min distance = 5mm] node[above]{$\a$}(2);
\path[->] (3) edge [loop right, min distance = 5mm] node[below]{$\beta$}(3);
\path[->] (2) edge [bend left] node[above]{$\d$} (3);
\path[->] (3) edge [bend left] node[below]{$\g$} (2);

\node[above] at (3,0){$\Longrightarrow$};

\node[circle,draw,above,fill = orange!20](4) at (4.5,0){$p$};
\node[circle,draw,fill=purple!20](5) at (6,0){$q$};
\path[->] (4) edge [loop left, min distance = 5mm,dashed] node[above]{$\a$}(4);
\path[->] (5) edge [loop right, min distance = 5mm,dashed] node[below]{$\beta$}(5);

\begin{scope}[on background layer]
\node[draw,dashed,rounded corners, fill=orange!20, fill opacity = 0.5, fit = (2)]{};
\node[draw, dashed,rounded corners, fill=purple!20, fill opacity = 0.5, fit = (3)]{};
\end{scope}

\node at (0.5,-1) {$(Z,\w_Z)$};
\node at (5,-1) {$(Z,\tr_Z)$};
\end{tikzpicture}
\end{center}
\caption{The trace map erases data between pairs of nodes.}
\label{fig:tr}
\end{figure}

\begin{example}[The $\cout$ and $\cin$ maps]\label{ex:cin-cout}
Let $(Z,\w_Z) \in \Ncom$, and let $z \in Z$. 
Now define $\cout: \Ncom \r \pow(\R)$ and $\cin:\Ncom \rightarrow \pow(\R)$ by 
\begin{align*}\cout(Z)&= \set{\max_{z'\in Z} |\w_Z(z,z')| : z \in Z}\qquad \text{for all $(Z,\w_Z)\in \Ncom$}\\
\cin(Z) &=\set{\max_{z'\in Z} |\w_Z(z',z)| : z \in Z} \qquad\text{for all $(Z,\w_Z) \in \Ncom$}.
\end{align*}
For each $z \in Z$, $\max_{z'\in Z}|\w_Z(z,z')|$ and $\max_{z'\in Z}|\w_Z(z',z)|$ are achieved because $\set{z}\times Z$ and $Z\times \set{z}$ are compact. We also define the associated maps $\cout_Z$ and $\cin_Z$ by writing, for any $(Z,\w_Z)\in \Ncom$ and any $z\in Z$,
\[\cout_Z(z) = \max_{z'\in Z}|\w_Z(z,z')| \qquad
\cin_Z(z) = \max_{z' \in Z}|\w_Z(z',z)|.\]
To see how these maps operate on a network, let $Z = \set{p,q,r}$ and consider the weight matrix $\Si = \left(
\begin{smallmatrix}
1&2&3\\
0&0&4\\
0&0&5
\end{smallmatrix}
\right)
.$
The network corresponding to this matrix is shown in Figure \ref{fig:outin}. We ascertain the following directly from the matrix:
\begin{align*}
\cout_Z(p) = 3 & \qquad \cin_Z(p) = 1\\
\cout_Z(q) = 4 & \qquad \cin_Z(q) = 2\\
\cout_Z(r) = 5 & \qquad \cin_Z(r) = 5.
\end{align*}
So the $\cout$ map returns the maximum (absolute) value in each row, and the $\cin$ map pulls out the maximum (absolute) value in each column of the weight matrix. As in the preceding example, we may use the Hausdorff distance to compare the images of networks under the $\cout$ and $\cin$ maps.

\end{example}

\begin{figure}
\begin{center}
\begin{tikzpicture}[every node/.style={font=\footnotesize}]
\node[circle,draw,fill=orange!20](1) at (0,1.5){$p$};
\node[circle,draw,fill=purple!20](2) at (-1.5,0){$q$};
\node[circle,draw,fill=violet!30](3) at (1.5,0){$r$};

\path[->] (1) edge [loop above, min distance = 5mm] node[right]{$1$}(1);
\path[->] (2) edge [dashed,loop, out=240, in = 210, min distance = 5mm] node[below]{$0$}(2);
\path[->] (3) edge [loop, out=300,in=330, min distance = 5mm] node[below]{$5$}(3);
\path[->] (1) edge [bend left] node[above,pos=0.5]{$2$} (2);
\path[->] (2) edge [dashed,bend left] node[above,pos=0.5]{$0$} (1);
\path[->] (1) edge [bend left] node[above,pos=0.5]{$3$} (3);
\path[->] (3) edge [dashed,bend left] node[above,pos=0.5]{$0$} (1);
\path[->] (2) edge [ bend left] node[below,pos=0.5]{$4$} (3);
\path[->] (3) edge [dashed,bend left] node[above,pos=0.5]{$0$} (2);

\node at (0,-1){$Z$};
\end{tikzpicture}
\caption{The $\cout$ map applied to each node yields the greatest weight of an arrow \emph{leaving} the node, and the $\cin$ map returns the greatest weight \emph{entering} the node. }
\label{fig:outin}
\end{center}
\end{figure}

\begin{example}[$\min$-$\cout$ and $\min$-$\cin$]
Define the maps $\mincout:\Ncom\rightarrow \R$ and $\mincin:\Ncom\rightarrow \R$ by 
\begin{align*}
\mincout((Z,\w_Z)) &= \min_{z\in Z}\cout_Z(z)\qquad\text{for all $(Z,\w_Z) \in \Ncom$}\\
\mincin((Z,\w_Z)) &= \min_{z\in Z}\cin_Z(z) \qquad\text{for all $(Z,\w_Z) \in \Ncom$}.
\end{align*}
Then both $\mincin$ and $\mincout$ are $\R$-invariants. We take the minimum when defining $\mincout, \mincin$ because for any network $(Z,\w_Z)$, we have $\max_{z\in Z}\cout_Z(z) = \max_{z\in Z}\cin_Z(z) = \diam(Z).$ Also observe that the minima are achieved above because $Z$ is compact.
\end{example}

\begin{proposition}\label{prop:n0invariants}
The maps $\cout, \cin, \tr, \spec,$ and $\spec^\bullet$ are $\pow(\R)$-invariants. Similarly, $\diam, \mincout,$ and $\mincin$ are $\R$-invariants.
\end{proposition}

This follows from the stronger statement of Proposition \ref{prop:lips}.

\begin{definition}[Motif sets are metric space valued invariants]\label{dfn:motif-inv} 
Recall the motif sets described in Definition \ref{defn:motif}.
Our use of motif sets is motivated by the following observation, which appeared in \cite[Section 5]{dgh-props}. 
For any $n\in \N$, let $\mc{C}(\R^{n\times n})$ denote the set of closed subsets of $\R^{n\times n}$. Under the Hausdorff distance induced by the $\ell^\infty$ metric on $\R^{n\times n}$, this set becomes a valid metric space \cite[Proposition 7.3.3]{burago}. The motif sets defined in Definition \ref{defn:motif} define a metric space valued invariant as follows: for each $n \in \N$, let $\M_n : \Ncom \r \mc{C}(\R^{n \times n})$ be the map $Z \mapsto \M_n(Z)$. We call this the \emph{motif set invariant}. So for $(Z,\w_Z), (Y,\w_Y) \in \Ncom$, for each $n \in \N$, we let $(U,d_U) = (\R^{n \times n},\ell^\infty)$ and consider the following distance between the $n$-motif sets of $Z$ and $Y$: 
\[d_n(\M_n(Z),\M_n(Y)) := \dhdf^U(\M_n(Z),\M_n(Z)).\] 
Since $\dhdf$ is a proper distance between closed subsets, $d_n(\M_n(Z),\M_n(Y))=0$ if and only if $\M_n(Z)=\M_n(Y).$
\end{definition}

\subsection{Quantitative stability of basic invariants}
\label{sec:quant-stab}

Let $(V,d_V)$ be a given pseudometric space. The $V$-valued invariant $\iota:\Ncom\rightarrow V$ is said to be \emph{quantitatively stable} if there exists a constant $L> 0$ such that 
\[d_{V}\big(\iota(Z),\iota(Y)\big)\leq L\cdot \dn(Z,Y)\] 
for all networks $Z$ and $Y$. The least constant $L$ such that the above holds for all $Z,Y\in\Ncom$ is the Lipschitz constant of $\iota$ and is denoted $\mathbf{L}(\iota).$

Note that by identifying a non-constant quantitatively stable $V$-valued invariant $\iota$, we immediately obtain a lower bound for the $\dn$ distance between any two compact networks $(Z,\omega_Z)$ and $(Y,\omega_Y)$. Furthermore, given a finite family $\iota_\alpha:\Ncom\rightarrow V$, $\alpha\in A$, of non-constant quantitatively stable invariants, we may obtain the following lower bound for the distance between compact networks $Z$ and $Y$:
\[\dn(Z,Y)\geq \max_{\alpha\in A}\mathbf{L}(\iota_\alpha)^{-1}d_V(\iota_\alpha(Z),\iota_\alpha(Y)).\]

It is often the case that computing $d_V(\iota(Z),\iota(Y))$ is substantially simpler than computing the $\dn$ distance between $Z$ and $Y$ (which leads to a possibly NP-hard problem). The invariants described in the previous section are quantitatively stable. The proofs of the next results are in Section \ref{sec:pf-stab}.

\begin{lemma}\label{lem:dn0-exp} Let $(Z,\w_Z), (Y,\w_Y) \in \Ncom$. Let $f$ represent any of the maps $\tr, \cout,$ and $\cin$, and let $f_Z$ (resp. $f_Y$) represent the corresponding map $\tr_Z, \cout_Z, \cin_X$ (resp. $\tr_Y, \cout_Y, \cin_Y$). Then we obtain:
\[\dhdf^\R(f(Z),f(Y)) = \inf_{R\in\Rsc(Z,Y)}\sup_{(z,y)\in R}\big|f_Z(z)-f_Y(y)\big|.\]
\end{lemma}

\begin{proposition}\label{prop:lips} The invariants $\diam, \tr, \cout,\cin, \mincout$, and $\mincin$ are quantitatively stable, with Lipschitz constant $\mathbf{L} = 2$.

\end{proposition}

\begin{example}
Proposition \ref{prop:lips} provides simple lower bounds for the $\dn$ distance between compact networks. One application is the following: for all networks $Z$ and $Y$, we have $\dn(Z,Y)\geq \frac{1}{2} \big|\diam(Z)-\diam(Y)\big|.$ For example, for the networks $Z=N_2(\mattwo{1}{5}{2}{4})$ and $Y=N_k(\mathbbm{1}_{k\times k})$ (the all-ones matrix---also see Example \ref{ex:simple-networks} to recall the $N_2$ and $N_k$ notation) we have $\dn(Z,Y)\geq \frac{1}{2}|5-1|= 2,$ for all $k\in\N.$ For another example, consider the weight matrices 
\[
\Sigma:=
\left(
\begin{smallmatrix}
0&5&2\\
3&1&4\\
1&4&3
\end{smallmatrix}
\right)
\,\,\mbox{and}
\,\,
\Sigma':=
\left(
\begin{smallmatrix}
3&4&2\\
3&1&5\\
3&3&4
\end{smallmatrix}
\right).
\]
Let $Z=N_3(\Sigma)$ and $Y=N_3(\Sigma')$. By comparing the diagonals, we can easily see that $Z\not\cong^s Y$, but let us see how the invariants we proposed can help. Note that $\diam(Z)=\diam(Y)=5$, so the lower bound provided by diameter ($\frac{1}{2}|5-5| = 0$) does not help in telling the networks apart. However, $\tr(Z)=\set{0,1,3}$ and $\tr(Y)=\set{3,1,4}$, and Proposition \ref{prop:lips} then yields 
\[\dn(Z,Y)\geq  \frac{1}{2} \dhdf^\R(\{0,1,3\},\{1,3,4\}) = \frac{1}{2}. \]

Consider now the $\cout$ and $\cin$ maps. Note that one has $\cout(Z) = \set{5,4}$, $\cout(Y)=\set{4,5}$,  $\cin(Z) = \set{3,5,4}$, and $\cin(Y) = \set{3,4,5}$. Then $\dhdf^\R(\cout(Z),\cout(Y)) = 0$, and $\dhdf^\R(\cin(Z),\cin(Y)) = 0$. Thus in both cases, we obtain $\dn(Z,Y) \geq 0$. So in this particular example, the $\cout$ and $\cin$ maps are not useful for obtaining a lower bound to $\dn(X,Y)$ via Proposition \ref{prop:lips}.
\end{example}

Now we state a proposition regarding the stability of global and local spectrum invariants. These will be of particular interest for computational purposes as we explain in \S\ref{sec:computations}.

\begin{proposition}\label{prop:spec}
Let $\spec^\bullet$ refer to either the $\cout$ or $\cin$ version of local spectrum. Then, 
for all $(Z,\w_Z),(Y,\w_Y) \in\Ncom$ we have 
\begin{align*}
\dn(Z,Y)&\geq \frac{1}{2}\inf_{R\in\Rsc}\sup_{(z,y)\in R}\dhdf^\R(\spec^\bullet_Z(z),\spec^\bullet_Y(y))\\
&\geq \frac{1}{2}\dhdf^\R(\spec(Z),\spec(Y)).
\end{align*}
As a corollary, we get $\mathbf{L}(\spec^\bullet)=\mathbf{L}(\spec)=2$.
\end{proposition}

\begin{figure}
\begin{center}
\begin{tikzpicture}[every node/.style={font=\footnotesize}]
\node[circle,draw, above, fill = orange!20](2) at (-1,0){$p$};
\node[circle,draw, fill = purple!20](3) at (1,0){$q$};
\path[->] (2) edge [loop left, min distance = 5mm] node[above]{$2$}(2);
\path[->] (3) edge [loop right, min distance = 5mm] node[below]{$1$}(3);
\path[->] (2) edge [bend left] node[above]{$1$} (3);
\path[->] (3) edge [bend left] node[below]{$2$} (2);

\node[circle,draw, above, fill = orange!60](2) at (3,0){$r$};
\node[circle,draw,fill = violet!30](3) at (5,0){$s$};
\path[->] (2) edge [loop left, min distance = 5mm] node[above]{$2$}(2);
\path[->] (3) edge [loop right, min distance = 5mm] node[below]{$3$}(3);
\path[->] (2) edge [bend left] node[above]{$3$} (3);
\path[->] (3) edge [bend left] node[below]{$1$} (2);

\node at (0,-1){$Z$};
\node at (4,-1){$Y$};
\end{tikzpicture}
\caption{Lower-bounding $\dn$ by using global spectra (cf. Example \ref{ex:prop-spec}).}
\label{fig:spec}
\end{center}
\end{figure}

\begin{example}[An application of Proposition \ref{prop:spec}]
\label{ex:prop-spec}
Consider the networks in Figure \ref{fig:spec}. By Proposition \ref{prop:spec}, we may calculate a lower bound for $\dn(Z,Y)$ by simply computing the Hausdorff distance between $\spec(Z)$ and $\spec(Y)$, and dividing by 2. In this example, $\spec(Z) = \set{1,2}$ and $\spec(Y) = \set{1,2,3}$. Thus $\dhdf^\R(\spec(Z),\spec(Y)) = 1$, and $\dn(Z,Y) \geq \frac{1}{2}$.  
\end{example}

Computing the lower bound involving local spectra requires solving a  bottleneck linear assignment problem over the set of all correspondences between $Z$ and $Y$. This can be solved in polynomial time; details are provided in \S\ref{sec:computations}. The second lower bound stipulates computing the Hausdorff distance on $\R$ between the (global) spectra of $Z$ and $Y$ -- a computation which can be carried out in  (smaller) polynomial time as well. 

To conclude this section, we state a theorem asserting that motif sets form a family of quantitatively stable invariants. Recall that we used this result in the proof of Theorem \ref{thm:cpt-sisom}.

\begin{theorem}\label{thm:motif-stab}
For each $n\in\N$, $\M_n$ is a stable invariant with $\mathbf{L}(\M_n)=2$.
\end{theorem}

\subsection{Persistent homology invariants and their convergence}
\label{sec:ph-methods}

\emph{Persistent homology} is a technique for studying multiscale features of data by creating geometric complexes, associating algebraic invariants to these complexes, and studying the structure of these invariants \cite{robins1999towards,carlsson2009topology, edelsbrunner2010computational, edelsbrunner2014persistent,frosini1992measuring}. In the conventional setting, persistent homology takes Euclidean or metric data as input and produces topological summaries called \emph{persistence diagrams} or \emph{barcodes} as output, one in each dimension $k \in \Z_+$. In an intermediate step, this technique first constructs a collection of vector spaces with linear maps, and these vector spaces are implicitly built on top of geometric complexes. The structure of these linear maps is then studied via matrix reductions and summarized in persistence barcodes. Persistent homology research has seen rapid progress in recent years, and in particular, the notion of persistent homology for general, possibly asymmetric settings such as that of networks has been studied in \cite{turner2019rips, dowker-jact, pph, phmlp,pinto2020motivic,huntsman2020generalizing,dey2020efficient}. In these works, persistent homology methods are defined for finite networks.

Here and in the next section, we use our machinery of $\e$-systems (cf. Definition \ref{def:systems} and subsequent constructions) to prove that persistence diagrams of infinite, compact networks are well-defined. This answers a question posed by Turner in \cite{turner2019rips}.

We begin with the definition of a \emph{persistent vector space}. Throughout this section, all our vector spaces are assumed to be over a fixed ground field $\mb{F}$.

\begin{definition} A \emph{persistent vector space} $\mc{V}$ (also denoted \gls{pvec}) is a family $\{V^\d \xr{\nu_{\d,\d'}} V^{\d'}\}_{\d \leq \d' \in \R}$ of vector spaces and linear maps such that: (1) $\nu_{\d,\d}$ is the identity map for any $\d \in \R$, and (2) $\nu_{\d,\d''} = \nu_{\d',\d''}\circ\nu_{\d,\d'}$ whenever $\d \leq \d'\leq \d''$. 
\end{definition}

Conventional hierarchical clustering methods take in metric data as input and produce ultrametrics as output that are in turn faithfully visualized as dendrograms \cite{clust-um}. A conventional persistent homology method (e.g. Vietoris-Rips, defined below) yields a higher dimensional analogue of this process: it takes a metric dataset as input, and outputs a persistent vector space $\mc{V}$ that is faithfully represented as a \emph{persistence diagram} $\dgm(\mc{V})$. A classification result in \cite[\S5.2]{carlsson2005persistence} shows that the persistence diagram is a full invariant of a persistent vector space. This completes the analogy with the setting of hierarchical clustering. 

Persistence diagrams can be compared using the \emph{bottleneck distance}, which we denote by \gls{db}. We point the reader to \cite{pers-mod-book} and references therein for details.

While the persistence diagram and bottleneck distance are the primary tools in practical applications, theoretical proofs are often made simpler through the language of \emph{interleavings} and \emph{interleaving distance}. We present this next.

\begin{definition}[$\e$-interleaving,  \cite{chazal2009proximity}]
\label{defn:alg-interleaving}
Let $\mc{U} = \{U^\d \xr{s_{\d,\d'}} U^{\d'}\}_{\d \leq \d' \in \R}$ and $\mc{V} = \{V^\d \xr{t_{\d,\d'}} V^{\d'}\}_{\d \leq \d' \in \R}$ be two persistent vector spaces. Given $\e \geq 0$, $\mc{U}$ and $\mc{V}$ are said to be \emph{$\e$-interleaved} if there exist two families of linear maps $\{\ph_{\d}:U^\d \r V^{\d + \e}\}_{\d \in \R}$ and $\{\psi_{\d}:V^\d \r U^{\d + \e}\}_{\d \in \R}$ such that: (1) $\ph_{\d'} \circ s_{\d,\d'} = t_{\d+\e,\d'+\e}\circ \ph_\d$, (2) $\psi_{\d'} \circ t_{\d,\d'} = s_{\d+\e,\d'+\e}\circ \psi_\d$, (3) $s_{\d,\d+2\e} = \psi_{\d + \e} \circ \ph_{\d}$, and (4) $t_{\d,\d+2\e} = \ph_{\d + \e} \circ \psi_{\d}$ for each $\d \leq \d' \in \R$. 
\end{definition}

The \emph{interleaving distance \gls{di}} between $\mc{U}$ and $\mc{V}$ is then defined as:
\[\di(\mc{U},\mc{V}):=\inf \{\e\geq 0 : \text{$\mc{U}$ and $\mc{V}$ are $\e$-interleaved}\}.\]
The interleaving and bottleneck distances are connected by the \emph{Isometry Theorem}, which states that the two distances are in fact equivalent. Various forms of this theorem have appeared in the literature; we will end this section with a statement of this result that appears in \cite{pers-mod-book}.

Our aim in this work is to describe the \emph{convergence} of persistent homology methods applied to network data. When dealing with finite networks, the vector spaces resulting from applying a persistent homology method will necessarily be finite dimensional. However, our setting is that of infinite (more specifically, \emph{compact}) networks, and so we need additional machinery to ensure that our methods output well-defined persistent vector spaces. The following definition and theorem are provided in full detail in \cite{pers-mod-book}.

\begin{definition}[\S2.1, \cite{pers-mod-book}] A persistent vector space $\mc{V}=\{V^\d \xr{\nu_{\d,\d'}} V^{\d'}\}_{\d \leq \d' \in \R}$ is \emph{q-tame} if $\nu_{\d,\d'}$ has finite rank whenever $\d < \d'$. 
\end{definition}

\begin{theorem}[\cite{pers-mod-book}, also \cite{chazal2014persistence} Theorem 2.3]
\label{thm:tame-well-dfn-diag} 
Any q-tame persistent vector space $\mc{V}$ has a well-defined persistence diagram $\dgm(\mc{V})$. If $\mc{U}, \mc{V}$ are $\e$-interleaved q-tame persistent vector spaces, then $\db(\dgm(\mc{U}),\dgm(\mc{V})) \leq \e.$
\end{theorem}

We conclude this section with a statement of the isometry theorem.

\begin{theorem}[Theorem 5.14, \cite{pers-mod-book}]
Let $\mc{U}, \mc{V}$ be $q$-tame persistent vector spaces. Then, \[\di(\mc{U},\mc{V}) =  \db(\mc{U},\mc{V}).\]
\end{theorem}

\subsection{Vietoris-Rips, Dowker, Path, and Ordered Tuple persistent homology methods on networks}
\label{sec:ph-method-details}

We now present a collection of methods for producing persistent vector spaces from network data. For finite networks, these methods have appeared in \cite{dowker-jact, turner2019rips}. Here our goal is to define these methods for compact networks and to establish their convergence properties. First we will focus on defining the different types of \emph{filtered complexes}---i.e. complexes nested by inclusion---from which persistent vector spaces can be derived. Once the complexes are defined, the persistent vector spaces are obtained by applying the \emph{homology functor} (cf. \cite{munkres-book} for details on homology). Recall that an \emph{abstract simplicial complex} on a set $Z$ is a collection $\Sigma$ of non-empty finite subsets of $Z$ such that whenever $\sigma \in \Sigma$ and $\tau \subseteq \sigma$, we also have $\tau \in \Sigma$. A \emph{filtered simplicial complex} is a nested collection of simplicial complexes $\{\Sigma_\d \subseteq \Sigma_{\d'}\}_{\d \leq \d'}$.

\begin{definition}[Vietoris-Rips complexes (\gls{vr})] Given a compact network $(Z,\w_Z)$ and $\d \in \R$, the \emph{Vietoris-Rips complex at resolution $\d$} is defined as:
\[{\vr}_\d(Z):=\{\s \in \pow(Z) : \s \text{ finite, } \max_{z,z'\in \s}\w_Z(z,z') \leq \d \}.\] 
\end{definition}

Practitioners of persistent homology will recognize that the Vietoris-Rips complex construction given above is a direct generalization of the Vietoris-Rips complex of a metric space. This definition yields a simplicial filtration $\{{\vr}_\d(Z) \hookrightarrow {\vr}_{\d'}(Z)\}_{\d \leq \d' \in \R}$. Applying the simplicial homology functor in dimension $k$ (for $k\in \Z_+$) to this filtration yields the \emph{Vietoris-Rips persistent vector space} $\pvec_k^{\vr}(Z)$.

Next we describe two constructions---the Dowker source and sink complexes---that practitioners will recognize as asymmetric generalizations of the \v{C}ech complex of a metric space \cite{chazal2014persistence,dowker-jact}. 

\begin{definition}[\v{C}ech complex (\gls{cech}) of a metric space] Given a metric space $(Z,d_Z)$ and $\d \in \R$, the \emph{\v{C}ech complex at resolution $\d$} is defined as:
\[ \cech_\d(Z) := \{ \s \in \pow(Z) : \s \text{ finite, } \min_{p \in Z} \max_{z \in \s} d_Z(z,p) \leq \d\}.\]
\end{definition}

\begin{definition}[Dowker complexes (\gls{dow})] Given a compact network $(Z,\w_Z)$ and $\d \in \R$, the \emph{Dowker sink-complex at resolution $\d$} is defined as:
\[\dow^{\si}_\d(Z):=\{\s \in \pow(Z) : \s \text{ finite, } \min_{p\in Z}\max_{z\in \s}\w_Z(z,p) \leq \d\}. \]
Similarly, the \emph{Dowker source-complex at resolution $\d$} is defined as:
\[\dow^{\so}_\d(Z):=\{\s \in \pow(Z) : \s \text{ finite, } \min_{p\in Z}\max_{z\in \s}\w_Z(p,z) \leq \d \}. \]
\end{definition}

The Dowker sink and source complexes are different in general when $X$ is asymmetric. Surprisingly, the persistent vector spaces obtained from the sink and source filtrations \emph{are equivalent}. This result was established in \cite{dowker-jact} in the setting of finite networks. For compact networks, the statement is as follows. Here we momentarily denote the Dowker source and sink persistent vector spaces as $\pvec^{\so}$ and $\pvec^{\si}$, respectively.

\begin{theorem}[Dowker duality] Let $(Z,\w_Z)$ be a compact network, and let $k \in \Z_+$. Then,
\[ \pvec_k^{\si}(Z) = \pvec_k^{\so}(Z).\] 
\end{theorem}

The proof is via a functorial generalization of Dowker's Theorem \cite{dowker1952homology}, which holds in the case of infinite sets. Alternatively, a functorial generalization of the Nerve Lemma can also be used to prove this result, as suggested in \cite{chazal2014persistence}. Hence we denote the resulting persistent vector space (in dimension $k \in \Z_+$) as $\pvec_k^{\dow}(X)$, without distinguishing between sink and source constructions.

Whereas the Vietoris-Rips and Dowker complexes generate conventional simplicial complexes, it is possible to build more general complexes that better capture asymmetry in the function $\w_Z$. One example is via \emph{ordered-tuple $(OT)$} complexes \cite{turner2019rips}. An $\ot$-complex $\Sigma$ over a set $Z$ is a collection of nonempty, finite tuples $(z_1,\ldots, z_p)$ such that whenever $(z_1,\ldots, z_p) \in \Sigma,$ we also have $(z_1,\ldots, \hat{z_i},\ldots, z_p) \in \Sigma$ for all $1\leq i \leq p$, where the hat denotes omission from the sequence. Notions related to homology can be defined for OT complexes without difficulty, and thus a persistent vector space can be immediately obtained after defining a filtered OT complex. The following construction was phrased in \cite{turner2019rips} in terms of set-function pairs, which are just finite networks in our setting.

\begin{definition}[Vietoris-Rips OT complexes (\gls{ot})] 
Given a compact network $(Z,\w_Z)$ and $\d \in \R$, the \emph{Vietoris-Rips OT complex} at resolution $\delta$ is defined as:
\[ \ot^{\vr}_\d(Z) := \{ (z_1,z_2,\ldots,z_p) : \w_Z(z_i,z_j) \leq \d \text{ for all } i \leq j\}.\]
We denote the corresponding persistent vector space by $\pvec^{\ot}_k(Z)$.
\end{definition}

Yet another approach is via the notion of \emph{path homology} \cite{grigor2014homotopy,pph}. A preliminary notion that we will need is that of a \emph{boundary operator}. Suppose we are given a set $Z$ and the free vector space of finite-length tuples in $Z$, i.e. the vector space $\mb{F}[\{(z_0,z_1,\ldots, z_k) : k \in \Z_+, \, z_i \in Z \text{ for all } i \}]$. Then one defines a boundary operator via the following alternating sum:
\[ \p_k((z_0,z_1,\ldots, z_k)) = \sum_{i=0}^{k} (-1)^k (z_0,z_1,\ldots,\hat{z_i},\ldots, z_k).\]

\begin{definition}[Persistent path homology (\gls{path})] Given a compact network $(Z,\w_Z)$ and $\d \in \R$, first construct the digraph $G^\d$ with $V(G^\d):=Z$ and $E(G^\d):= \{(v,v') : \w_Z(v,v') \leq \d\}$. Next, for each $k \in \Z_+$, define the free vector space of \emph{allowed $k$-paths} $\mc{A}_k(G^\d) := \mb{F}[ (v_0,v_1,\ldots, v_k) : (v_i, v_{i+1}) \in E(G^\d) \text{ for each } 0 \leq i \leq k-1\}]$. Finally define $\Om_k(G^\d):= \mb{F}[ \{u \in \mc{A}_k(G^\d) : \p_k(u) \in \mc{A}_{k-1}(G^\d)\}]$. Then the degree-$k$ path homology vector space of $G^\d$ is given by $H_k^{\Xi}(G^\d):= \ker(\p_k)/\im(\p_{k+1})$, and the inclusions $\{G^\d \subseteq G^{\d'}\}_{\d\leq \d'}$ are transformed into a persistent vector space by functoriality of homology. We denote the corresponding persistent vector space by $\pvec^{\path}_k(Z)$.

\end{definition}

The following lemma collects results from \cite{dowker-jact, pph, turner2019rips} on the \emph{stability of persistent homology invariants} of networks.

\begin{lemma}[Quantitative stability between $\dn$ and $\di$] 
\label{lem:interleaving}
Let $(Z,\w_Z)$ and $(Y,\w_Y)$ be two networks. Let $\e > 2\dn(X,Y)$. Then $\pvec_k^{*}(X)$ and $\pvec_k^{*}(Y)$ are $\e$-interleaved, where $*$ denotes any of the methods in $\{\vr,\dow,\path,\ot\}$. 
\end{lemma}

\begin{theorem}\label{thm:tame-cpt-net} Let $(Z,\w_Z) \in \Ncom,$ $k \in \Z_+$. Then $\pvec_k^{*}(Z)$ is q-tame, where $*$ denotes any of the methods in $\{\vr,\dow,\path,\ot\}$.
\end{theorem}

The metric space analogue of Theorem \ref{thm:tame-cpt-net} appeared in \cite[Proposition 5.1]{chazal2014persistence}; the same proof structure works in the setting of networks after applying our results on approximation via $\e$-systems.

\begin{proof}[Proof of Theorem \ref{thm:tame-cpt-net}]
For convenience, write $\pvec_k^{*}(Z) = \{V^{\d}\xr{\nu_{\d,\d'}} V^{\d'}\}_{\d \leq \d' \in \R}$. Let $\d < \d'$. We need to show $\nu_{\d,\d'}$ has finite rank. Write $\e:=(\d'-\d)/2$. Let $\mc{U}$ be an $\e/4$-system on $Z$ (this requires Theorem \ref{thm:sampling}). Then by Theorem \ref{thm:sampling} we pick a finite subset $Z'\subseteq Z$ such that $\dn(Z,Z') < \e/2$. Then $\pvec_k^{*}(Z')$ and $\pvec_k^{*}(Z)$ are $\e$-interleaved by Lemma \ref{lem:interleaving}. For convenience, write $\pvec_k^{*}(Z') = \{U^{\d}\xr{\mu_{\d,\d'}} U^{\d'}\}_{\d \leq \d' \in \R}$. Then the map $\nu_{\d,\d'}:V^{\d} \r V^{\d'}$ factorizes through $U^{\d+\e}$ via interleaving maps $V^{\d} \r U^{\d+\e} \r V^{\d+2\e} = V^{\d'}$. Since $U^{\d+\e}$ is finite dimensional, it follows that $\nu_{\d,\d'}$ has finite rank. This concludes the proof.\end{proof}

\begin{corollary}[Stability]
\label{cor:stab}
Let $(Z,\w_Z),(Y,\w_Y) \in \Ncom$, $k \in \Z_+$. Then for any $*$ in $\{\vr,\dow,\path,\ot\}$,
\[\db(\dgm_k^{*}(Z),\dgm_k^{*}(Y)) \leq 2\dn(X,Y).\]
\end{corollary}
\begin{proof}
By Theorem \ref{thm:tame-cpt-net}, $\pvec^*(Z), \pvec^*(Y)$ are both tame. Thus they have well-defined persistence diagrams (Theorem \ref{thm:tame-well-dfn-diag}). The result follows by Lemma \ref{lem:interleaving} and Theorem \ref{thm:tame-well-dfn-diag}.
\end{proof}

The following result presents a probabilistic version of the prior result. Probabilistic approaches have practical value due to the inevitable measurement error found in network data \cite{newman2018network, newman2018estimating}. 

\begin{theorem}[Convergence]\label{thm:consistency-rips-cech}
Let $(Z,\w_Z,\mu_Z)$ be a measure network. For each $i\in \N$, let $\mathbf{z}_i: \Om \r Z$ be an independent random variable defined on some probability space $(\Om, \mc{F},\P)$ with distribution $\mu_Z$. For each $n\in \N$, let $\mc{Z}_n=\set{\mathbf{z}_1,\mathbf{z}_2,\ldots, \mathbf{z}_n}$. Let $\e > 0$.  Then we have: 
\[\P\big(\set{\omega \in \Om : \db(\dgm^{*}(\supp(\mu_Z)),\dgm^{*}(\mc{Z}_n(\omega)))\geq \e } \big) \leq \frac{\left(1-\mf{M}_{\e/4}(\supp(\mu_Z))\right)^n}{\mf{M}_{\e/4}(\supp(\mu_Z))},\]
where $\mc{Z}_n(\omega)$ is the subnetwork induced by $\set{\mathbf{z}_1(\omega),\ldots, \mathbf{z}_n(\omega)}$ and $*$ belongs to $\{\vr,\dow,\path,\ot\}$. In particular, $\pvec^*(\mc{Z}_n)$ converges almost surely to $\pvec^*(\supp(\mu_Z))$ w.r.t. bottleneck distance. 
\end{theorem}

\begin{proof}[Proof of Theorem \ref{thm:consistency-rips-cech}] 
We can consider $\supp(\mu_Z)$ as a network with full support by endowing it with the restriction of $\w_Z$ to $\supp(\mu_Z)\times \supp(\mu_Z)$. So for convenience, assume $Z= \supp(\mu_Z)$. 
For any $\omega \in \Omega$ such that $\dn(Z,\mc{Z}_n(\omega)) < \e/2$, we have by Corollary \ref{cor:stab} that $\db(\dgm^{*}(Z),\dgm^{*}(\mc{Z}_n)) < \e$. Thus by applying Theorem \ref{thm:prob-net-approx}, we have:
\begin{align*}
\P\big(\set{\omega \in \Om : \db(\dgm^{*}(Z),\dgm^{*}(\mc{Z}_n(\omega)))\geq \e } \big) 
&\leq \P\big(\set{\omega \in \Om : \dn(Z,\mc{Z}_n(\omega))\geq \e/2 } \big)\\
&\leq \frac{\left(1-\mf{M}_{\e/4}(Z)\right)^n}{\mf{M}_{\e/4}(Z)}.
\end{align*}
We conclude the proof with an application of the Borel-Cantelli lemma, as in the proof of Theorem \ref{thm:prob-net-approx}.\end{proof}

\subsection{Hierarchical clustering methods on asymmetric networks}
We now provide similar results in the special case of hierarchical clustering for asymmetric networks. First we state some definitions relating to connectivity and induce networks from preexisting partitions of a network. Then we will be ready to define hierarchical clustering methods on asymmetric networks and to prove related convergence results. The reciprocal and nonreciprocal clustering methods described below were introduced in a slightly restricted setting (finite networks $(Z,\w_Z)$ with zero diagonal and nonnegative $\w_Z$) in \cite{carlsson2013axiomatic}. The stability of these methods, given in the form $\dn(\H(Y,\w_Y),\H(Z,\w_Z)) \leq \dn(Y,Z)$ where $\H$ denotes either one of these hierarchical clustering methods, was also shown in \cite{carlsson2013axiomatic}.

Applying either of the hierarchical clustering methods we describe next to a network $(Z,\w_Z)$ produces a new weight function $u_Z$ that satisfies the ultrametric inequality (cf. Definition \ref{def:ultrametric}). Thus we get maps of the form $\H:\Ncom \to \Ncomult$ as shown in Figure \ref{fig:invariants}.

Before proceeding, we make some historical remarks and explain the connection to the next few sections. The methods described next will be reminiscent of single linkage hierarchical clustering (SLHC). In the Euclidean case, it has been known for some time that SLHC \cite{hartigan1981consistency,hartigan1985statistical} is not sensitive to modes of the data distribution. This observation was refined in \cite{clust-um} where it was shown that in the case of metric measure spaces, SLHC depended on only the support of the measure. Despite this drawback, SLHC has long been of mathematical interest \cite{hartigan1985statistical}. In more recent years, it has been shown that under certain axioms that characterize the well-behavedness of an HC method for metric spaces, every well-behaved HC method factors through SLHC \cite{carlsson2013classifying}. This thread of axiomatic characterization has been further developed in the setting of finite, asymmetric networks in \cite{carlsson2013axiomatic, carlsson2018hierarchical} and most recently in \cite{carlsson2021robust}. The reciprocal and nonreciprocal clustering methods we describe next are key players in these axiomatic treatments of directed networks, as they constitute extremal examples that bound an infinite family of axiomatic HC methods on directed networks. Over the next few sections, we treat the reciprocal and nonreciprocal methods in the infinite network setting and connect back to results in \cite{clust-um} by showing that such methods remain sensitive only to the support of the data distribution.

\subsubsection{Chain cost and path-connectedness}
\begin{definition}[The modified $\w$-weights and $\w$-path-connectedness]
\label{defn:path-conn}
Given a network $(Z,\w_Z)$, one defines new weight functions $\widetilde{\w}_Z:Z\times Z \to \pow(\R)$ and $\overline{\w}_Z, \underline{\w}_Z: Z\times Z \r \R$ by writing the following for $z,z'\in Z$:
\begin{align*}
\widetilde{\w}_Z(z,z') &:= \{\w_Z(z,z),\w_Z(z,z'),\w_Z(z',z')\},\\
\overline{\w}_Z(z,z') &:= \max\widetilde{\w}_Z(z,z'), \text{ and }\\
\underline{\w}_Z(z,z') &:= \min\widetilde{\w}_Z(z,z').
\end{align*}

To say that $(Z,\w_Z)$ is \emph{$\w$-path-connected} means that given any $x,x' \in X$, there exists $r_{x,x'} \in X$ and a continuous function $\g:[0,1] \r X$ such that $\g(0)=x$, $\g(1) = x'$, and for any $\e > 0$, there exist $0=t_0 \leq t_1 \leq t_2 \leq \ldots \leq t_n = 1$ for $n\in \Z_+$ such that: 
\begin{equation}\label{eq:conn}\widetilde{\w}_X\big(\g(t_i),\g(t_{i+1})\big) \subseteq (r_{x,x'}-\e,r_{x,x'}+\e) \text{ for each } 0\leq i \leq n-1.\end{equation}
\end{definition}

\begin{remark}
When $(Z,\w_Z)$ is a metric space, all the self weights $\w_Z(z,z)$ are zero. So one would have $\overline{\w}_Z = \w_Z$ and $\underline{\w}_Z = 0$.
The notion of $\w$-path-connectedness agrees with the standard notion of path-connectedness in a metric space, as we elaborate in Lemma \ref{lem:path-conn} below.
In general, however, if $\gamma(t)$ connecting $x$ to $x'$ satisfies (\ref{eq:conn}), because of the asymmetry of ${\w}_X$, it does not follow that the reverse curve $\gamma(1-t)$ connecting $x'$ to $x$ will satisfy (\ref{eq:conn}).
\end{remark}

\begin{example}
Similarity kernels, e.g. the cosine similarity or the Gaussian kernel, constitute simple examples where $\w_Z(z,z)$ is nonzero. Such examples abound in machine learning applications. The unit circle equipped with a Gaussian kernel is an example of a $\w$-path-connected space with nonzero self weights.
\end{example}

\begin{lemma}\label{lem:path-conn} Let $(Z,\w_Z)$ be a $\w$-path-connected network. Then there exists a unique $r_Z \in \R$ such that $\w_Z(z,z) = r_Z$ for all $z \in Z$. In the case of metric spaces, one has $r_Z = 0$. 
\end{lemma}

\begin{proof}[Proof of Lemma \ref{lem:path-conn}] Let $z,z'\in Z$ and let $r_{z,z'}\in \R$ be as in Definition \ref{defn:path-conn}. Let $(\e_n)_{n\in \N}$ be a sequence decreasing to 0. Fix $n\in \N$, and let $\g:[0,1] \r X$ be a continuous function such that $\g(0)= z$, $\g(1) = z'$, and there exist $t^n_0 = 0 \leq t^n_1 \leq t^n_2 \leq \ldots, t^n_{k_n} = 1$ such that: 
\[\widetilde{\w}_Z(\g(t^n_i),\g(t^n_{i+1})) \subseteq B(r_{z,z'},\e_n) \text{ for each } 0\leq i \leq k_n-1.\]
In particular, because $\widetilde{\w}_Z(z,\g(t^n_1)) \subseteq B(r_{z,z'},\e_n)$ and $\widetilde{\w}_Z(\g(t^n_{k_n-1}),z') \subseteq B(r_{z,z'},\e_n)$, we have:
\[\set{{\w}_Z(z,z),{\w}_Z(z',z')}  \subseteq B(r_{z,z'},\e_n).\]
Thus $|\w_Z(z,z) - \w_Z(z',z')| \leq 2\e_n$. Letting $n \r \infty$, we obtain $\w_Z(z,z) = \w_Z(z',z')$. Since $z' \in Z$ was arbitrary, we get that $\w_Z(z',z') = \w_Z(z,z)$ for all $z' \in Z$. The result now follows.\end{proof}

\begin{definition}[Path-connectivity constant] Let $(Z,\w_Z)$ be a $\w$-path-connected network. Then we define its \emph{path-connectivity constant} $\mf{pc}_Z$ to be the real number $r_Z$ obtained via Lemma \ref{lem:path-conn}.
\end{definition}

\begin{definition}[Networks arising from disconnected networks]\label{defn:induced-net} 
Let $(Z,\w_Z)$ be a network such that $Z$ is a finite, disjoint union of $\w$-path-connected components $\set{U_a : a\in A}$, where $A$ is a (finite) indexing set and each $U_a$ is compact. Let $\nu_A:A \times A \r \R$ be the map given by writing, for each $a,a' \in A$,
\[\nu_A(a,a'):= \min\set{\overline{\w}_Z(z,z') : z \in U_a, z' \in U_a'}.\]
Then $(A,\nu_A)$ is a network. 
Analogously, one induces a \emph{symmetric} network by defining $\lambda_A$ as follows:
\begin{align*}
\lambda_A(a,a')&:=\min\set{\max(\overline{\w}_Z(z,z'),\overline{\w}_Z(z',z)) : z\in U_a,\; z'\in U_{a'}}.
\end{align*}
\end{definition}

The following definition will be useful in the next section.

\begin{definition}[Chains and directed cost]
A \emph{chain} $c$ from $z$ to $z'$ is defined to be a finite ordered set of points starting at $z$ and reaching $z'$:
\[c=\set{z_0,z_1,z_2,\ldots, z_n : z_0=z, z_n=z', z_i \in Z \text{ for all $i$}}.\]
The collection of all chains from $z$ to $z'$ will be denoted $C_Z(z,z')$. 
The \emph{(directed) cost} of a chain $c \in C_Z(z,z')$ is defined as follows:
$\cost_Z(c) := \max_{z_i,z_{i+1} \in c}\overline{\w}_z(z_i,z_{i+1}).$
\end{definition}

\begin{remark}[Equivalence of ultrametrics and dendrograms]
\label{rem:ultra-dendro}
Before proceeding to the next section, we remind the reader that any ultrametric has a lossless representation as a dendrogram, and conversely, any dendrogram has a lossless representation as an ultrametric \cite{jardine1971mathematical}. By virtue of this result, we write the outputs of hierarchical clustering methods as ultrametrics. As shown in \cite{nets-asilo}, a similar duality holds even in the setting of (asymmetric) networks, up to a small modification of definitions. In particular, the output of an HC method on a network is a network in itself, along with some special structure that allows it to be visualized as a (generalized) dendrogram.
\end{remark}

\subsubsection{The nonreciprocal clustering method: definition and convergence}

We now present the nonreciprocal hierarchical clustering method for directed networks.

\begin{definition}[Nonreciprocal clustering] The \emph{nonrecriprocal clustering method (\gls{hnr})} is a map $\hnr: \Ncom \r \Ncom$ given by $(Z,\w_Z) \mapsto (Z,\unr_Z)$, where $\unr_Z: Z \times Z \r \R$ is defined by writing, for each $z,z' \in Z$,
\[ \hspace{-1in} \unr_Z(z,z'):=\max\left(\inf_{c\in C_Z(z,z')}\cost_Z(c),\inf_{c\in C_Z(z',z)}\cost_Z(c) \right).\]

\end{definition}

The output $\unr_X$ is symmetric and satisfies the \emph{ultrametric inequality}, so it can be represented as a tree \cite[\S 7.2]{phylo}. Compare this to the \emph{cluster trees} discussed by \cite{hartigan1975clustering}. The idea behind this definition is easily summarized: two points $z$ and $z'$ belong to the same cluster at resolution $\d$ if there are directed paths $z\r z'$ and $z'\r z$, each with cost $ \leq \d$. 

The next result shows that nonreciprocal clustering essentially recovers only the structure of the support of the data distribution, and no information about the data density. Such a result is also known for metric spaces \cite{clust-um}; the proof for the current result uses the machinery of $\e$-systems.
Proofs related to the next result are provided in Section \ref{sec:pf-hnr-consistent}.

\begin{restatable}[Convergence of nonreciprocal clustering]{theorem}{hnrconsistent}
\label{thm:hnr-consistent}
Let $(Z,\w_Z,\mu_Z)$ be a measure network. Suppose $\supp(\mu_Z)$ is a finite, disjoint union of compact, $\w$-path-connected components $\set{Z_a : a\in A}$. Let $(A,\nu_A)$ be as in Definition \ref{defn:induced-net}, and let $(A,\unr_A)=\hnr(A,\nu_A)$. For each $i\in \N$, let $\mathbf{z}_i: \Om \r Z$ be an independent random variable defined on some probability space $(\Om, \mc{F},\P)$ with distribution $\mu_Z$. For each $n\in \N$, let $\mc{Z}_n=\set{\mathbf{z}_1,\mathbf{z}_2,\ldots, \mathbf{z}_n}$, and for each $\omega \in \Om$, let $\mc{Z}_n(\omega)$ denote the subnetwork induced by $\set{\mathbf{z}_1(\omega),\ldots, \mathbf{z}_n(\omega)}$. Let $\e > 0$. Then,
\[\P \left(\set{\omega \in \Om : \dn((A,\unr_A),\hnr(\mc{Z}_n(\omega)))\geq \e }\right) \leq 
\frac{\left(1-\mf{M}_{\e/2}(\supp(\mu_Z))\right)^n}{\mf{M}_{\e/2}(\supp(\mu_Z))}.\] 
In particular, the output of the nonreciprocal clustering method applied to the sampled network $\mc{Z}_n$ converges almost surely to $(A,\unr_A)$ in the sense of $\dn$ as the sample size increases.
\end{restatable}

We end this section with an application of nonreciprocal clustering to Finsler manifolds. 

\begin{proposition}[Nonreciprocal clustering on Finsler manifolds]
\label{prop:finsler-nr} 
Let $(M,F_M, \w_M)$ be a compact, connected Finsler manifold without boundary, where $\w_M$ is the asymmetric weight function induced by the Finsler function $F_M$. Then $\unr_M(z,z') = 0 $ for all $z,z' \in M$.
\end{proposition}

\begin{proof}[Proof of Proposition \ref{prop:finsler-nr}] Let $z,z' \in M$. Let $\e > 0$, and let $\g, \g':[0,1] \r M$ be curves from $z$ to $z'$ and from $z'$ to $z$, respectively. By choosing $n$ uniformly separated points on $\g([0,1])$ and $\g'([0,1])$ for sufficiently large $n$, we obtain finite chains $c$ and $c'$ on $\g([0,1])$ and $\g'([0,1])$ such that $\max( \cost_M(c), \cost_M(c')) < \e$. Since $\e > 0$ was arbitrary, we obtain $\unr_M(z,z') = 0$.\end{proof}

\subsubsection{The reciprocal clustering method: definition and convergence}

\begin{definition}[Reciprocal clustering] The \emph{recriprocal clustering method (\gls{hr})} is a map $\hr: \Ncom \r \Ncom$ given by $(Z,\w_Z) \mapsto (Z,\ur_Z)$, where $\ur_Z: Z \times Z \r \R$ is defined by writing, for each $z,z' \in Z$,
\begin{equation}\label{eq:hr}
\ur_Z(z,z'):=\inf_{c\in C_Z(z,z')}\max_{z_i,z_{i+1} \in c} \bigg(\max \big(\overline{\w}_Z(z_i,z_{i+1}), \overline{\w}_Z(z_{i+1},z_{i})\big)\bigg).
\end{equation}
The function $\ur_Z$ satisfies the ultrametric inequality, so it can be represented as a tree \cite[\S 7.2]{phylo}.
\end{definition}

Our convergence result for reciprocal clustering requires two additional assumption on the underlying network: (1) the weight function is a dissimilarity measure (i.e. self-weights are 0), and (2) the asymmetry is bounded, i.e. the network has finite reversibility (cf. Section \ref{sec:dissim-rev}). 

Under these assumptions, we are able to show that with increasing sample size, reciprocal clustering converges to an object that is insensitive to the shape of the data distribution. However, the limiting network is slightly different from that of the nonreciprocal clustering case. The proof of the next result is provided in Section \ref{sec:pf-hr-consistent}. 

\begin{restatable}[Convergence of reciprocal clustering]{theorem}{hrconsistent}
\label{thm:hr-consistent}
Let $(Z,\w_Z,\mu_Z)$ be a measure network with dissimilarity weights and finite reversibility. Suppose $\supp(\mu_Z)$ is a finite, disjoint union of compact, $\w$-path-connected components $\set{Z_a : a\in A}$. Let $(A,\lambda_A)$ be as in Definition \ref{defn:induced-net}, and let $(A,\ur_A)=\hr(A,\lambda_A)$. For each $i\in \N$, let $\mathbf{z}_i: \Om \r Z$ be an independent random variable defined on some probability space $(\Om, \mc{F},\P)$ with distribution $\mu_Z$. For each $n\in \N$, let $\mc{Z}_n=\set{\mathbf{z}_1,\mathbf{z}_2,\ldots, \mathbf{z}_n}$, and for each $\omega \in \Om$, let $\mc{Z}_n(\omega)$ denote the subnetwork induced by $\set{\mathbf{z}_1(\omega),\ldots, \mathbf{z}_n(\omega)}$. 
Let $\e > 0$. Then,
\[\P \left(\set{\omega \in \Om : \dn((A,\ur_A),\hr(\mc{Z}_n(\omega)))\geq \e }\right) \leq 
\frac{\left(1-\mf{M}_{\e/2}(\supp(\mu_Z),A)\right)^n}{\mf{M}_{\e/2}(\supp(\mu_Z),A)}.\]
In particular, the output of the reciprocal clustering method applied to the sampled network $\mc{Z}_n$ converges almost surely to $(A,\ur_A)$ in the sense of $\dn$ as the sample size increases.
\end{restatable}

In the case of Finsler manifolds with finite reversibility, we can also recover the result of Proposition \ref{prop:finsler-nr}.

\begin{proposition}[Reciprocal clustering on Finsler manifolds with finite reversibility]
\label{prop:finsler-r} 
Let $(M,F_M, \w_m)$ be a compact, connected finitely-reversible Finsler manifold without boundary. Here $\w_M$ is the asymmetric weight function induced by the Finsler function $F_M$. Then $\ur_M(z,z') = 0 $ for all $z,z' \in M$.
\end{proposition}

\begin{proof}[Proof of Proposition \ref{prop:finsler-r}] Let $z,z' \in M$. Let $\e > 0$, and let $\g:[0,1] \r M$ be a curve from $z$ to $z'$. By invoking the finite reversibility of $M$, choose $n$ uniformly separated points $\{z_1,\ldots, z_n\}$ on $\g([0,1])$ for sufficiently large $n$ such that $\max(\overline{\w}_M(z_i,z_{i+1}), \overline{\w}_M(z_{i+1},z_{i})) < \e$ for each $i = 1,\ldots, n-1$. Here $z_1= z$ and $z_{n} = z'$. Then $\ur_M(z,z') < \e$. Since $\e > 0 $ was arbitrary, the result follows.\end{proof}

\subsection{Clustering and persistence on the directed circle with finite reversibility}
\label{sec:dir-s1-pers}

Recall the directed circle with finite reversibility $(\us^1,\rus)$ presented in Section \ref{sec:dir-s1}. Also consider the \emph{directed circle with finite reversibility on $n$ nodes $(\us^1_n,{\rusn})$} obtained by writing
\[\us^1_n := \set{ e^{2\pi i \frac{k}{n}} \in \C : k\in \set{0,1,\ldots, n-1}}\]
and defining $\rusn$ to be the restriction of $\rus$ on this set.

Directed circles on $n$ nodes without the finite reversibility condition were introduced in \cite{dowker-jact}, and their Dowker persistent homology was fully characterized by drawing on a connection to the \v{C}ech complex of a standard unit circle. The persistent homology of this complex had in turn been fully characterized in \cite{adamaszek2016nerve,adamaszek2017vietoris}. In this section we further elaborate on this connection by showing that the Dowker complexes of directed circles on $n$ nodes with reversibility $\rho$ actually \emph{interpolate} between the \v{C}ech complex of the standard circle and the Dowker complex of the directed, irreversible circle. We will utilize the language of interleavings of simplicial complexes, and we describe this construction briefly. For more details we direct the reader to \cite{dowker-jact}.

\begin{definition}[Interleavings of simplicial complexes]
\label{defn:top-interleaving}
Simplicial maps between simplicial complexes $f,g: U \to V$ are said to be \emph{contiguous} if $f(\s)\cup g(\s)$ is a simplex in $V$ whenever $\s$ is a simplex in $U$ \cite{munkres-book}. 
Two filtered simplicial complexes $\{ U_\d \xrightarrow{\iota_{\d,\d'}} U_{\d'} \}_{\d \in \R}, \{V_\d \xrightarrow{\iota_{\d,\d'}} V_{\d'} \}_{\d \in \R}$ are said to be $\e$-interleaved for $\e \geq 0$ if there exist two families of simplicial maps $\{\phi_\d:  U_\d \to V_{\d+\e}\}_{\d \in \R}$ and $\{\psi_\d:  V_\d \to U_{\d+\e}\}_{\d \in \R}$ such that the following pairs of maps are contiguous for any $\d \leq \d' \in \R$:
\begin{itemize}
\item $\ph_{\d'} \circ \iota_{\d,\d'}$ and $\iota_{\d+\e,\d'+\e} \circ \ph_{\d}$
\item $\psi_{\d'} \circ \iota_{\d,\d'}$ and $\iota_{\d+\e,\d'+\e} \circ \psi_{\d}$
\item $\psi_{\d+\e}\circ \ph_{\d}$ and $\iota_{\d,\d+2\e}$
\item $\ph_{\d+\e}\circ \psi_{\d}$ and $\iota_{\d,\d+2\e}$.
\end{itemize}
Importantly, if two filtered simplicial complexes are $\e$-interleaved, then their corresponding persistent vector spaces are $\e$-interleaved in the sense of Definition \ref{defn:alg-interleaving} (cf. \cite{dey2016multiscale} and also \cite{dowker-jact} for details on these concepts using a similar setup).

\end{definition}

Now we set up some notation. To use the results of \cite{adamaszek2016nerve,adamaszek2017vietoris}, we need the circle parametrized by the unit interval $[0,1)$, which we denote as \gls{S}. Let $(S^1,d)$ denote the $[0,1)$-circle with geodesic distance. Then define:
\[S^1_n := \{0, \tfrac{1}{n}, \tfrac{2}{n},\ldots, \tfrac{n-1}{n}\}.\]
Also define:
\begin{center}
\begin{minipage}{0.4\textwidth}
\begin{align*}
\ph:S^1_n &\to \mb{S}^1_n \\
\tfrac{k}{n} &\mapsto \exp(2\pi i \tfrac{k}{n}) 
\end{align*}
\end{minipage}
\begin{minipage}{0.4\textwidth}
\begin{align*}
\psi:\mb{S}^1_n & \to S^1_n\\
\exp(2\pi i \tfrac{k}{n})  &\mapsto \tfrac{k}{n} 
\end{align*}
\end{minipage}
\end{center}

In what follows we replace $\cech_\d(-),\dow_\d(-)$ notation with $\cech(-;\d),\dow(-;\d)$ for readability. For each $\d \in \R$, define:
\[U_\d := \cech(S^1_n; \tfrac{1}{2}(\tfrac{\d}{2\pi} + \tfrac{\d}{2\pi\rho}) ), \qquad V_\d := \dow^{\so}(\mb{S}^1_n; \d).\]
To motivate this definition, first we consider different cases of $\d$. For $\d<0$, both $U_\d$ and $V_\d$ are empty. For $\d = 0$, $U_\d$ and $V_\d$ comprise $n$ 0-simplices. 
For $\d \geq 2\pi$, both $U_\d$ and $V_\d$ contain the full simplex on $n$ points. 
Let $\d \in (0,2\pi)$, and let $\tau \in \dow^{\so}(\mb{S}^1_n;\d)$. Then $\tau$ has the form 
\[ 
\tau = [ \exp(2\pi i (\tfrac{j}{n}), \exp(2\pi i (\tfrac{j}{n} + \tfrac{1}{n})),\ldots, \exp(2\pi i (\tfrac{j}{n} + \tfrac{k}{n})), \ldots, \exp(2\pi i (\tfrac{j}{n} + \tfrac{m}{n}))],\]
where $\exp(2\pi i (\tfrac{j}{n} + \tfrac{k}{n}))$ acts as the source. This means that we have $2\pi\tfrac{m-k}{n} < \d$ and $2\pi\tfrac{k-0}{n} < \tfrac{\d}{\rho}$. In particular, the total geodesic distance on the unit circle from $\exp(2\pi i (\tfrac{j}{n})$ to $\exp(2\pi i (\tfrac{j}{n} + \tfrac{m}{n}))$ is $2\pi \tfrac{m-0}{n} < \d + \tfrac{\d}{\rho}$. Rescaling to the $[0,1)$-circle, we find the total geodesic distance from $\tfrac{j}{n}$ to $\tfrac{j+m}{n}$ to be $\tfrac{\d}{2\pi} + \tfrac{\d}{2\pi\rho}$. Finally consider $\s := \psi(\tau) = [\tfrac{j}{n}, \tfrac{j}{n} + \tfrac{1}{n},\ldots, \tfrac{j}{n} + \tfrac{\floor{m/2}}{n},\ldots, \tfrac{j}{n}+\tfrac{m}{n}]$. If $m$ is even, then each $\tfrac{l}{n} \in \s$ is contained in the closed ball 
\[\overline{B}(\tfrac{j}{n}+ \tfrac{m/2}{n}, \tfrac{\d}{4\pi} + \tfrac{\d}{4\pi\rho}  ).  \]
In the odd $m$ case, we may need to extend the radius of the closed ball by $1/n$. This can be done via the update $\d \gets \d + 4\pi/n$, as we have $U_{\d+4\pi/n} = \cech(S^1_n; \tfrac{\d}{4\pi} + \tfrac{1}{n} + \tfrac{\d}{4\pi\rho} + \tfrac{1}{n\rho}) \supseteq \cech(S^1_n; \tfrac{\d}{4\pi} + \tfrac{\d}{4\pi\rho}  + \tfrac{1}{n})$. Thus for any simplex $\tau \in V_\d$, we have that $\psi(\tau)$ is a simplex in $U_{\d+4\pi/n}$.

Conversely, suppose $\s = [\tfrac{j}{n}, \tfrac{j}{n} + \tfrac{1}{n},\ldots, \tfrac{j}{n}+\tfrac{m}{n}]$ is a simplex in $U_\d$. Then the geodesic distance from $\tfrac{j}{n}$ to $\tfrac{j+m}{n}$ is bounded above by $\tfrac{\d}{2\pi} + \tfrac{\d}{2\pi\rho}$. Rescaling to the $[0,2\pi)$-circle, the geodesic distance from $\exp(2\pi i (\tfrac{j}{n})$ to $\exp(2\pi i (\tfrac{j}{n} + \tfrac{m}{n}))$ is bounded above by $\d + \d/\rho$. Thus there is a point $\xi \in \mb{S}^1$ such that the geodesic distance from $\exp(2\pi i (\tfrac{j}{n})$ to $\xi$ is $\d/\rho$ and the distance from $\xi$ to $\exp(2\pi i (\tfrac{j}{n} + \tfrac{m}{n}))$ is $\d$. We also know that $\xi$ is in a $\tfrac{1}{2}\cdot\tfrac{2\pi}{n}$-neighborhood of $\mb{S}^1_n$. Let $k$ be such that $\exp(2\pi i (\tfrac{j}{n} + \tfrac{k}{n}))$ is the closest element of $\mb{S}^1_n$ clockwise of $\xi$. Then we have:
\[\rusn \left( e^{2\pi i \tfrac{j}{n}} \cdot e^{2\pi i \tfrac{k}{n}}, e^{2\pi i \tfrac{j}{n}} \cdot e^{2\pi i \tfrac{m}{n}} \right)< \d + 2\pi/n, \qquad
\rusn \left( e^{2\pi i \tfrac{j}{n}} , e^{2\pi i \tfrac{j}{n}} \cdot e^{2\pi i \tfrac{k}{n}} \right)< \d/\rho.\]
In particular, $\tau= \ph(\s)$ is a simplex in $V_{\d+4\pi/n}$.

It follows from the preceding work that the rescaling maps $\ph,\psi$, which are essentially identity maps on simplices, induce interleavings. We record this below:

\begin{claim}
\label{cl:cech-dow}
Let $U_\d$ and $V_\d$ be the simplicial complexes on $n$ points as defined above. Then they are $4\pi/n$ interleaved via the simplicial maps induced by $\ph$ and $\psi$.
\end{claim}

Finally we invoke a characterization of the \v{C}ech complex of $S^1$. 

\begin{theorem}[\cite{adamaszek2016nerve, adamaszek2017vietoris}]
\label{thm:cech}
Let $n \in \N$, and let $0 \leq k \leq n-2$ be an integer. Then,
\[ \cech(S^1_n; \tfrac{k}{2n}) = 
\begin{cases}
\bigvee^{n-k-1} S^{2l} &: \frac{k}{n} = \frac{l}{l+1},\\
S^{2l+1} &: \frac{l}{l+1} < \frac{k}{n} < \frac{l+1}{l+2},
\end{cases}\]
where $l \in \Z_+$.
Moreover, if $\d < \d' < \tfrac{1}{2}$ and $\cech(S^1_n,\d), \cech(S^1_n,\d')$ have the same homotopy type, then the induced inclusion map is a homotopy equivalence. Consequently the persistent homology of $\cech(S^1_n;-)$ is fully characterized. 
\end{theorem}
The latter part of the preceding theorem is not contained in \cite{adamaszek2016nerve}, but follows from work in \cite{adamaszek2017vietoris}. An explicit verification of this portion, due to Henry Adams, can be found in \cite[Theorem 49]{dowker-jact}. 

\begin{theorem}
\label{thm:dowker-circle} 
Let $1\leq \rho < \infty$ be a finite reversibility parameter. The directed circle $(\mb{S}^1,\rus)$ with reversibility $\rho$ has the following Dowker persistence barcodes in each dimension $2l + 1$ for $l \in \Z_+$:
\[ \dgm_{2l+1}^{\dow}(\mb{S}^1,\rus) = \left\{\left( \frac{l (4 \pi \rho)}{2(l+1)(1+\rho)} , \frac{(l+1)(4\pi\rho)}{2(l+2)(1+\rho)} \right)\right\}.\]
In even dimensions, the persistence diagram is trivial. Moreover, in the fully reversible case $\rho=1$, rescaling from the $[0,2\pi)$-circle to the $[0,1)$-circle recovers the persistent homology of the \v{C}ech complex of the circle \cite[Section 9]{adamaszek2017vietoris}. An illustration of this result is provided in Figure \ref{fig:dowker-circle}.
\end{theorem}

\begin{proof}
By Theorems \ref{thm:sampling} and \ref{thm:complete}, we know that it suffices to characterize the Dowker persistence for $\mb{S}^1_n$ and then take the limit as $n \to \infty$. First we observe from Theorem \ref{thm:cech} that $\cech(S^1;\e) \simeq S^{2l+1}$ on $(\tfrac{l}{2(l+1)}, \tfrac{l+1}{2(l+2)})$. By the work preceding Claim \ref{cl:cech-dow}, we know that $\cech(S^1_n;\e)$ is related to $\dow(\mb{S}^1_n;\d)$ via the affine transformation $\e =  \tfrac{\d}{4\pi} + \tfrac{\d}{4\pi\rho}$. Thus letting $\e = \tfrac{\d}{4\pi} + \tfrac{\d}{4\pi\rho}$ and solving for $\d$ then yields:
\[ \d \in \left( \frac{l (4 \pi \rho)}{2(l+1)(1+\rho)} , \frac{(l+1)(4\pi\rho)}{2(l+2)(1+\rho)} \right).\]
From the interleaving result in Claim \ref{cl:cech-dow} and Definition \ref{defn:top-interleaving}, we know that as $n \to \infty$, $\dgm^{\dow}(\mb{S}^1_n) \to \dgm^{\cech}(S^1_n)$. Taking the affine transformation of scale parameters into account then yields the result.  \end{proof}

We now specialize to the case of hierarchical clustering on $(\mb{S}^1,\rus)$. 

\begin{theorem}[Nonreciprocal clustering on $(\us^1,\rus)$]
\label{thm:dir-s1-nr}
\[\unr_{\sr}(x,x') = 0 \text{ for all } x, x' \in \us^1.\]
\end{theorem}

\begin{proof}[Proof of Theorem \ref{thm:dir-s1-nr}]
We claim that $\us^1$ is $\w$-path-connected, with path-connectivity constant $\mf{pc}_{\us^1} = 0$ (invoking Lemma \ref{lem:path-conn}). The result then follows by Lemma \ref{lem:hnr-path-conn}. 
Let $x,y \in \us^1.$ Without loss of generality, suppose $\rus(x,y) = \w_{\us^1}(x,y)$ (otherwise, switch $x$ and $y$). Let $\e > 0$. Pick $x_0=x,x_1,x_2,\ldots, x_n=y$ such that $\rus(x_i,x_{i+1}) = \w_{\us^1}(x_i,x_{i+1}) < \e$ for each $0\leq i \leq n-1$. We automatically have $\rus(x_i,x_i) =0$ for all $0\leq i \leq n$, and hence $\widetilde{\w}_{\sr}(x_i,x_{i+1}) \subseteq B(0,\e)$ for all $i$. Since $x,y \in \us^1$ were arbitrary, it follows by Definition \ref{defn:path-conn} that $\us^1$ is $\w$-path-connected. The preceding work shows that $\mf{pc}_{\us^1} = 0$. The result follows.\end{proof}

The case of nonreciprocal clustering essentially followed from an application of Theorem \ref{thm:hnr-consistent} (i.e. the special case of Lemma \ref{lem:hnr-path-conn}. Next we consider the application of reciprocal clustering to $(\us^1,\rus)$ via Theorem \ref{thm:hr-consistent}.

\begin{theorem}[Reciprocal clustering on $(\us^1,\rus)$]
\label{thm:dir-s1-r}
\[\ur_{\sr}(x,x') = 0 \text{ for all } x, x' \in \us^1.\]
\end{theorem}

\begin{proof}[Proof of Theorem \ref{thm:dir-s1-r}]
The proof is exactly analogous to that of Theorem \ref{thm:dir-s1-nr}, except that we utilize the finite reversibility to obtain reciprocal connections.
\end{proof}

\subsection{Proofs from Section \ref{sec:inv-conv}}
\label{sec:pf-inv-conv}

\subsubsection{Proofs related to stability of basic invariants}
\label{sec:pf-stab}

\begin{proof}[Proof of Lemma \ref{lem:dn0-exp}] Observe that $f(Z) = \set{f_Z(z) : z\in Z} = f_Z(Z)$, so we need to show 
\[\dhdf^{\R}(f_Z(Z),f_Y(Y)) = \inf_{R\in \Rsc(Z,Y)}\sup_{(z,y)\in R} \big|f_Z(z) - f_Y(y)\big|.\] Recall that by the definition of Hausdorff distance on $\R$, we have 
\[\dhdf^{\R}(f_Z(Z),f_Y(Y))
= \max\big\{\sup_{z\in Z}\inf_{y\in Y}|f_Z(z) - f_Y(y)|,
\sup_{y\in Y}\inf_{z\in Z}|f_Z(z) - f_Y(y)|\big\}.\]
Let $a\in Z$ and let $R \in \Rsc(Z,Y)$. Then there exists $b \in Y$ such that $(a,b) \in R$. Then we have:
\begin{align*}
|f_Z(a) - f_Y(b)| &\leq \sup_{(z,y) \in R} |f_Z(z) - f_Y(y)|, \text{ and so }\\
\inf_{b\in Y} |f_Z(a) - f_Y(b)| &\leq \sup_{(z,y) \in R} |f_Z(z) - f_Y(y)|.
\intertext{This holds for all $a \in Z$. Then,} 
\sup_{a\in Z}\inf_{b \in Y} |f_Z(a) -f_Y(b)| &\leq  \sup_{(z,y) \in R} |f_Z(z) - f_Y(y)|.
\intertext{This holds for all $R \in \Rsc(Z,Y)$. So we have }
\sup_{a\in Z}\inf_{b \in Y} |f_Z(a) -f_Y(b)| &\leq  \inf_{R \in \Rsc}\sup_{(z,y) \in R} |f_Z(z) - f_Y(y)|.
\intertext{By a similar argument, we also have}
\sup_{b\in Y}\inf_{a \in Z} |f_Z(a) -f_Y(b)| &\leq  \inf_{R \in \Rsc}\sup_{(z,y) \in R} |f_Z(z) - f_Y(y)|.\\
\text{Thus } \dhdf^{\R}(f_Z(Z),f_Y(Y)) &\leq \inf_{R \in \Rsc}\sup_{(z,y) \in R} |f_Z(z) - f_Y(y)|.
\end{align*}

Now we show the reverse inequality. Let $z \in Z$, and let $\eta > \dhdf^{\R}(f_Z(Z),f_Y(Y))$. Then there exists $y \in Y$ such that $|f_Z(z) - f_Y(y)| < \eta $. Define $\ph(z) = y$, and extend $\ph$ to all of $Z$ in this way. Let $y \in Y$. Then there exists $z \in Z$ such that $|f_Z(z) - f_Y(y)| < \eta$. Define $\psi(y) = z$, and extend $\psi$ to all of $Y$ in this way. Let $R = \set{(z,\ph(z)) : z \in Z} \cup \set{(\psi(y),y) : y \in Y}$. 
Then for each $(a,b) \in R$, we have $|f_Z(a) - f_Y(b)| < \eta.$ 
Thus we have $\inf_{R \in \Rsc}\sup_{(z,y) \in R} |f_Z(z) - f_Y(y)| < \eta$. Since $\eta > \dhdf^{\R}(f_Z(Z),f_Y(Y))$ was arbitrary, it follows that
\[\inf_{R \in \Rsc(Z,Y)}\sup_{(z,y) \in R} |f_Z(z) - f_Y(y)| \leq \dhdf^{\R}(f_Z(Z),f_Y(Y)).\qedhere \]
\end{proof}

\begin{proof}[Proof of Proposition \ref{prop:lips}] Let $\eta > \dn(Z,Y)$. We break this proof into three parts.

\noindent
\textbf{The $\diam$ case.}
Recall that $\diam$ is an $\R$-valued invariant, so we wish to show $|\diam(Z) - \diam(Y)|\leq 2\dn(Z,Y)$. Let $R \in \Rsc(Z,Y)$ be such that for any $(a,b), (a',b') \in R$, we have $|\w_Z(a,a') - \w_Y(b,b')| < 2\eta. $

Let $z,z' \in Z$ such that $|\w_Z(z,z')| = \diam(Z)$, and let $y, y'$ be such that $(z,y), (z',y') \in R$. Then we have:
\begin{align*}
|\w_Z(z,z') - \w_Y(y,y')| &< 2\eta\\
|\w_Z(z,z') - \w_Y(y,y')| + |\w_Y(y,y')| &<
2\eta + |\w_Y(y,y')|\\
|\w_Z(z,z')| &< \diam(Y) + 2\eta.\\
\text{Thus } \diam(Z) &< \diam(Y) + 2\eta
\end{align*}
Similarly, we get $\diam(Y) < \diam(Z) + 2\eta$. It follows that $|\diam(Z) - \diam(Y)| < 2\eta.$ Since $\eta > \dn(Z,Y)$ was arbitrary, it follows that:
\[|\diam(Z) - \diam(Y)| \leq 2\dn(Z,Y).\]

For tightness, consider the networks $Z=N_1(1)$ and $Y=N_1(2)$. By Example \ref{ex:comparison}, we have that $\dn(Z,Y)=\frac{1}{2}$. On the other hand, $\diam(Z)=1$ and $\diam(Y)=2$ so that $|\diam(Z)-\diam(Y)|=1 = 2\dn(Z,Y).$\\
\noindent
\textbf{The cases $\tr$, $\cout$, and $\cin$.}
First we show $\mathbf{L}(\tr) = 2$. By Lemma \ref{lem:dn0-exp}, it suffices to show
\begin{align*}
\inf_{R\in \Rsc(Z,Y)}\sup_{(z,y)\in R}|\tr_Z(z) - \tr_Y(y)| < 2\eta. 
\end{align*}
Let $R \in \Rsc(Z,Y)$ be such that for any $(a,b), (a',b') \in R$, we have $|\w_Z(a,a') - \w_Y(b,b')| < 2\eta.$ Then we obtain $|\w_Z(a,a) - \w_Y(b,b)| < 2\eta$. Thus $|\tr_Z(a) - \tr_Y(b)| < 2\eta.$ Since $(a,b) \in R$ was arbitrary, it follows that $\sup_{(a,b) \in R}|\tr_Z(a) - \tr_Z(b)| < 2\eta.$ It follows that $\inf_{R\in \Rsc} \sup_{(a,b) \in R}|\tr_Z(a) - \tr_Z(b)| < 2\eta.$
The result now follows because $\eta > \dn(Z,Y)$ was arbitrary. The proofs for $\cout$ and $\cin$ are similar, so we just show the former. By Lemma \ref{lem:dn0-exp}, it suffices to show
\begin{align*}
\inf_{R\in \Rsc(Z,Y)}\sup_{(z,y)\in R}|\cout_Z(z) - \cout_Y(y)| < 2\eta. 
\end{align*}

Recall that $\cout_Z(z) = \max_{z'\in Z}|\w_Z(z,z')|$. Let $R \in \Rsc(Z,Y)$ be such that $|\w_Z(z,z') - \w_Y(y,y')| < 2\eta$ for any $(z,y),(z',y') \in R$. By triangle inequality, it follows that $|\w_Z(z,z')| < |\w_Y(y,y')| + 2\eta.$ In particular, for $(z',y') \in R$ such that $|\w_Z(z,z')| = \cout_Z(z)$, we have 
$\cout_Z(z) < |\w_Y(y,y')| + 2\eta.$ Hence 
$\cout_Z(z) < \cout_Y(y) + 2\eta.$ Similarly, $\cout_Y(y) < \cout_Z(z) + 2\eta.$ Thus we have $|\cout_Z(z) - \cout_Y(y)| < 2\eta$. This holds for all $(z,y) \in R$, so we have: 
\[\sup_{(z,y) \in R}|\cout_Z(z) - \cout_Y(y)| < 2\eta.\]
Minimizing over all correspondences, we get:
\[\inf_{R \in \Rsc}\sup_{(a,b)\in R}|\cout_Z(a) - \cout_Y(b)| < 2\eta.\] 
The result follows because $\eta > \dn(Z,Y)$ was arbitrary.

Finally, we need to show that our bounds for the Lipschitz constant are tight. Let $Z = N_1(1)$ and let $Y = N_1(2)$. Then $\dn(Z,Y) = \frac{1}{2}$. We also have $\dhdf^{\R}(\tr(Z),\tr(Y)) = |1-2| = 1$, and similarly $\dhdf^{\R}(\cout(Z),\cout(Y)) = \dhdf^{\R}(\cin(Z),\cin(Y)) = 1$.\\

\noindent
\textbf{The cases $\mincout$ and $\mincin$.}
The two cases are similar, so let's just prove $\mathbf{L}(\mincout) = 2$. Since $\mincout$ is an $\R$-invariant, we wish to show $|\mincout(Z) - \mincout(Y)| < 2\eta$. It suffices to show:
\[|\mincout(Z) - \mincout(Y)| \leq \dhdf^{\R}(\cout(Z),\cout(Y)),\]
because we have already shown 
\[\dhdf^{\R}(\cout(Z),\cout(Y)) = \inf_{R\in \Rsc(Z,Y)}\sup_{(z,y)\in R}|\cout_Z(z) - \cout_Y(y)| < 2\eta.\]
Here we have used Lemma \ref{lem:dn0-exp} for the first equality above.

Let $\e > \dhdf^{\R}(\cout(Z),\cout(Y))$. Then for any $z \in Z$, there exists $y\in Y$ such that:
\[|\cout_Z(z) - \cout_Y(y)| < \e.\]

Let $a \in Z$ be such that $\mincout(Z) = \cout_Z(a)$. Then we have:
\[|\cout_Z(a) - \cout_Y(y)| < \e,\]
for some $y\in Y$. In particular, we have:
\[\mincout(Y) \leq \cout_Y(y) < \e + \cout_Z(a) = \e + \mincout(Z).\]
Similarly, we obtain:
\[\mincout(Z) < \e + \mincout(Y).\]

Thus we have $|\mincout(Z) - \mincout(Y)| < \e.$ Since $\e > \dhdf^{\R}(\cout(Z),\cout(Y))$ was arbitrary, we have:
\[|\mincout(Z) - \mincout(Y)| \leq \dhdf^{\R}(\cout(Z),\cout(Y)).\]

The inequality now follows by Lemma \ref{lem:dn0-exp} and our proof in the case of the $\cout$ map.

For tightness, note that $|\mincout(N_1(1)) - \mincout(N_1(2))| = |1-2| = 1 = 2\cdot \frac{1}{2} = 2\dn(N_1(1),N_1(2)).$ The same example works for the $\mincin$ case. \end{proof}

\begin{proof}[Proof of Proposition \ref{prop:spec}]
(First inequality.) Let $Z,Y \in \Ncom$ and let $\eta > \dn(Z,Y)$. Let $R \in \Rsc(Z,Y)$ be such that $\sup_{(z,y),(z',y') \in R}|\w_Z(z,z') - \w_Y(y,y')| < 2\eta.$
Let $(z,y) \in R$, and let $\a \in \spec_Z(z)$. Then there exists $z' \in Z$ such that $\w_Z(z,z') = \a$. Let $y' \in Y$ be such that $(z',y') \in R$. Let $\b = \w_Y(y,y')$. Note $\b \in \spec_Y(y)$. Also note that $|\a - \b| < 2\eta$. By a symmetric argument, for each $\b \in \spec_Y(y)$, there exists $\a \in \spec_Z(z)$ such that $|\a-\b| < 2\eta$. So $\dhdf^\R(\spec_Z(z),\spec_Y(y)) < 2\eta$. This is true for any $(z,y) \in R$, and so we have $\sup_{(z,y)\in R}\dhdf^\R(\spec_Z(z),\spec_Y(y)) \leq 2\eta$. Then we have: 
\[\inf_{R \in \Rsc}\sup_{(z,y)\in R}\dhdf^\R(\spec_Z(z),\spec_Y(y)) \leq 2\eta.\] 
Since $\eta > \dn(Z,Y)$ was arbitrary, the first inequality follows.

(Second inequality.) Let $R \in \Rsc(Z,Y)$. Let $\eta(R) = \sup_{(z,y) \in R}\dhdf^\R(\spec_Z(z),\spec_Y(y))$. Let $\a \in \spec(Z)$. Then $\a \in \spec_Z(z)$ for some $z \in Z$. Let $y \in Y$ such that $(z,y) \in R$. Then there exists $\b \in \spec_Y(y)$ such that $|\a - \b|\leq \dhdf^\R(\spec_Z(z),\spec_Y(y))$, and in particular, $|\a - \b| \leq \eta(R)$. In other words, for each $\a \in \spec(Z)$, there exists $\b \in \spec(Y)$ such that $|\a-\b| \leq \eta(R)$. By a symmetric argument, for each $\b \in \spec(Y)$, there exists $\a \in \spec(Z)$ such that $|\a-\b| \leq \eta(R)$. Thus $\dhdf^\R(\spec(Z),\spec(Y)) \leq \eta(R)$. This holds for any $R \in \Rsc$. Thus we have 
\begin{align*}
\dhdf^\R(\spec(Z),\spec(Y)) \leq \inf_{R \in \Rsc}\sup_{(z,y)\in R}\dhdf^\R(\spec_Z(z),\spec_Y(y)).
\end{align*} 
This proves the second inequality. \end{proof}

\begin{proof}[Proof of Theorem \ref{thm:motif-stab}] Let $n \in \N$. We wish to show $d_n(\M_n(Z),\M_n(Y)) \leq 2\dn(Z,Y)$. Let $R \in \Rsc(Z,Y)$. Let $(z_i) \in Z^n$, and let $(y_i) \in Y^n$ be such that for each $i$, we have $(z_i,y_i) \in R$. Then for all $j, k \in \set{1, \ldots, n}$, $|\w_Z(z_i,z_j) - \w_Y(y_i,y_j)| \leq \dis(R)$. 

Thus $\inf_{(y_i) \in Y^n}|\w_Z(z_i,z_j) - \w_Y(y_i,y_j)| \leq \dis(R)$. This is true for any $(z_i) \in X^n$. Thus we get:
\[\sup_{(z_i) \in Z^n}\inf_{(y_i) \in Y^n}|\w_Z(z_i,z_j) - \w_Y(y_i,y_j)| \leq \dis(R).\]
By a symmetric argument, we get $\sup_{(y_i) \in Y^n}\inf_{(z_i) \in Z^n}|\w_Z(z_i,z_j) - \w_Y(y_i,y_j)| \leq \dis(R).$ Thus $d_n(\M_n(Z),\M_n(Y)) \leq \dis(R)$. This holds for any $R \in \Rsc(Z,Y)$. Thus $d_n(\M_n(Z),\M_n(Y)) \leq \inf_{R \in \Rsc(Z,Y)}\dis(R) = 2\dn(Z,Y)$.

For tightness, let $Z = N_1(1)$ and let $Y = N_1(2)$. Then $\dn(Z,Y) = \frac{1}{2}$, so we wish to show $d_n(\M_n(Z),\M_n(Y)) = 1$ for each $n \in \N$. Let $n \in \N$. Let $\mathbbm{1}_{n \times n}$ denote the $n\times n$ matrix with $1$ in each entry. Then $\M_n(Z) = \set{\mathbbm{1}_{n\times n}}$ and $\M_n(Y) = \set{2\cdot\mathbbm{1}_{n \times n}}$. Thus $d_n(\M_n(Z),\M_n(Y)) = 1$. Since $n$ was arbitrary, we conclude that equality holds for each $n \in \N$.  
\end{proof}

\subsubsection{Proofs related to Theorem \ref{thm:hnr-consistent}}
\label{sec:pf-hnr-consistent}

\begin{lemma}[Nonreciprocal clustering on an $\w$-path-connected network] 
\label{lem:hnr-path-conn}
Let $(Z,\w_Z)$ be a $\w$-path-connected network with path-connectivity constant $\mf{pc}_Z$ (cf. Definition \ref{defn:path-conn}). Then $(Z,\unr_Z) = \hnr(Z,\w_Z)$ is given by writing $\unr_Z(z,z') = \mf{pc}_Z$ for all $z,z'\in Z$.
\end{lemma}

\begin{proof}[Proof of Lemma \ref{lem:hnr-path-conn}] Let $z,z'\in Z$, and let $\e > 0$. By Definition \ref{defn:path-conn}, there exist chains $c\in C_Z(z,z')$ and $c'\in C_Z(z',z)$ such that $\max(\cost_Z(c),\cost_Z(c')) < \mf{pc}_Z + \e$. Thus $\unr_Z(z,z') < \mf{pc}_Z+ \e$. This holds for each $\e > 0$, and for any $z,z'\in Z$. Thus $\unr_Z \leq \mf{pc}_Z$. We also have $\unr_Z \geq \mf{pc}_Z$ by the definition of chain cost and $\overline{\w}_Z$. This concludes the proof.\end{proof}

\begin{remark}[Distortion and $\overline{\w}_Z$] 
\label{rem:wbar-dist}
Let $(Z,\w_Z), (Y,\w_Y) \in \Ngen$ and let $R\in \Rsc(Z,Y)$. Then,
\[\sup_{(z,y),(z',y') \in R}|\overline{\w}_Z(z,z') - \overline{\w}_Y(y,y')| \leq \dis(R).\]

To see this, fix $(z,y),(z',y') \in R$. Suppose $u,u' \in \set{z,z'}$ are such that $\overline{\w}_Z(z,z') = \w_Z(u,u')$. Let $v,v' \in \set{y,y'}$ be such that $(u,v),(u',v') \in R$. Then we have 
\[\overline{\w}_Z(z,z') = \w_Z(u,u') \leq \dis(R) + \w_Y(v,v') \leq \dis(R) + \overline{\w}_Y(y,y').\]
Now let $v,v' \in \set{y,y'}$ be such that $\overline{\w}_Y(y,y') = \w_Y(v,v')$. Let $u,u' \in \set{z,z'}$ be such that $(u,v),(u',v') \in R$. Then,
\[\overline{\w}_Y(y,y') = \w_Y(v,v') \leq \dis(R) + \w_Z(u,u') \leq \dis(R) + \overline{\w}_Z(z,z').\]
It follows that $|\overline{\w}_Z(z,z') - \overline{\w}_Y(y,y')| \leq \dis(R)$.
\end{remark}

\begin{lemma}[Nonreciprocal clustering collapses $\w$-path-connected subsets] 
\label{lem:hnr-main}
Let $(Z,\w_Z)$ be a network such that $Z$ can be written as a finite, disjoint union of compact, $\w$-path-connected components $\set{Z_a : a\in A}$. Let $(A,\nu_A)$ be as in Definition \ref{defn:induced-net}, and let $(A,\unr_A)=\hnr(A,\nu_A)$. Also let $\mc{U}=\set{U_1,\ldots, U_m}$ be a refined $\e/2$-system on $Z$. 

Suppose that $S \subseteq Z$ is a finite subset equipped with the restriction $\w_S:= \w_Z|_{S \times S}$ such that $S$ has nonempty intersection with $Z_a$ for each $a\in A$, and with $U_i$ for each $1\leq i \leq m$. Then,
\[\dn\left((S,\unr_S),(A,\unr_A)\right) < \e.\]
\end{lemma}

\begin{proof}[Proof of Lemma \ref{lem:hnr-main}]
For each $z \in Z$, let $a(z) \in A$ denote the index such that $z \in Z_{a(z)}$. Here the index is unique because $Z$ is a disjoint union. Then define:
\[R:=\set{(s,a(s)) : s \in S}.\]
Then $R\in \Rsc(S,A)$. 
We wish to show $\dis(R) < 2\e$, where the distortion is calculated with respect to $\unr_S$ and $\unr_A$. Let $(s,a(s)), (s',a(s')) \in R$. 

\begin{claim} We have $\unr_A(a(s),a(s')) \leq \unr_S(s,s')$.
\end{claim}
\begin{proof} Pick chains 
\begin{align*}
c_1 &:= \set{r_0=s,r_1,r_2,\ldots, r_k = s'}  \in C_S(s,s') \text{ and}\\
c_2 &:= \set{t_0=s',t_1,t_2,\ldots, t_j=s} \in C_S(s',s)
\end{align*} 
such that $\unr_S(s,s') = \max(\cost(c_1),\cost(c_2))$. 
Then for each $0 \leq i \leq k-1$, we have $\overline{\w}_S(r_i,r_{i+1}) \leq \unr_S(s,s')$. 
Similarly for each $0 \leq i \leq j-1$ we have $\overline{\w}_S(t_i,t_{i+1}) \leq \unr_S(s,s')$.
Now observe that for each $z,z' \in Z$, we have:
\begin{equation}
\nu_A(a(z),a(z')) \leq \overline{\w}_Z(z,z').
\end{equation}
Then for each $0\leq i \leq k-1$, we have: 
\[\nu_A(a(r_i),a(r_{i+1})) \leq \overline{\w}_S(r_i,r_{i+1})\leq \unr_S(s,s').\]
Similarly for each $0\leq i \leq j-1$, we have:
\[\nu_A(a(t_i),a(t_{i+1})) \leq \overline{\w}_S(t_i,t_{i+1})\leq \unr_S(s,s').\]
It follows that $\unr_A(a(s),a(s')) \leq \unr_S(s,s')$. 
\end{proof}

By Theorem \ref{thm:sampling}, we obtain a correspondence $R' \in \Rsc(S,Z)$ with $\dis(R') < 2\e$ such that for each $(s,z) \in R'$, we have $\set{s,z} \subseteq U$ for some $U\in \mc{U}$. Here the distortion is measured with respect to $\w_Z$ and $\w_S$. We will use this correspondence $R'$ as follows: for each $z\in Z$, there exists $s\in S$ such that $(s,z) \in R'$. In other words, there exists $s \in S$ such that $\set{s,z} \subseteq U$ for some $U \in \mc{U}$. Since $\mc{U}$ is a refined $\e/2$-system, we know also that $s,z$ belong to the same connected component $Z_a$, for some $a \in A$. 

For each $z\in Z$, we will write $s(z)$ to denote the element of $S$ obtained by the preceding construction.

\begin{claim}\label{cl:hnr-same-component}
Let $s,s'\in S$ be such that $a(s)=a(s')$, i.e. $s,s'$ belong to the same $\w$-path-connected component of $Z$. Then $\unr_S(s,s') < \nu_A(a(s),a(s)) + 2\e\leq \unr_A(a(s),a(s')) + 2\e$. 
\end{claim} 
\begin{proof}[Proof of Claim \ref{cl:hnr-same-component}] Since $Z_{a(s)}$ is $\w$-path-connected, there exists a unique $r \in \R$ such that $\w_Z(z,z) = r$ for all $z\in Z_{a(s)}$ by Lemma \ref{lem:path-conn}. By Definition \ref{defn:induced-net}, we have $r \leq \nu_A(a(s),a(s))$. Let $\eta > 0$, and let $z,z' \in Z$ be such that $(s,z),(s',z') \in R'$. Then by the definition of $\w$-path connectivity, we can take a chain $c= \set{z_0 = z,z_1,z_2,\ldots, z_n = z'}$ joining $z$ to $z'$ such that $\cost_Z(c)\leq r + \eta$. We can now convert this to a chain in $S$ by using the correspondence $R'$. Define:
\[c_S:=\set{s,s(z_1),s(z_2),\ldots, s(z_{n-1}),s'}.\]
By construction, $(s(z_i),z_i)\in R'$ for each $1\leq i \leq n-1$. Furthermore we have $(s,z),(s',z')\in R'$ by our choice of $z,z'$. Now by using Remark \ref{rem:wbar-dist} and the fact that $\dis(R') < 2\e$, we have $\cost_S(c_S) <  r+ \eta + 2\e \leq \nu_A(a(s),a(s)) + \eta + 2\e$. 

By a similar process, we can obtain a chain $c'_S \in C_S(s',s)$ such that $\cost_S(c'_S) < \nu_A(a(s),a(s)) + \eta + 2\e$. Thus $\unr_S(s,s') < \nu_A(a(s),a(s)) + \eta + 2\e$. Since $\eta > 0$ was arbitrary, the result follows. \end{proof}

\begin{claim} We have $\unr_S(s,s') < \unr_A(a(s),a(s')) + 2\e$.
\end{claim}
\begin{proof} Let $\overrightarrow{c}:=\set{r_0,\ldots, r_k}$ be a chain in $A$ such that $r_0 = a(s)$, $r_k=a(s')$, and for each $0\leq i \leq k-1$, we have $\overline{\w}_A(r_i,r_{i+1}) \leq \unr_A(a(s),a(s'))$. Similarly let $\overleftarrow{c}:=\set{t_0,\ldots, t_j}$ be a chain in $A$ such that $t_0 = a(s')$, $t_j=a(s)$, and for each $0\leq i \leq j-1$, we have $\overline{\w}_A(t_i,t_{i+1}) \leq \unr_A(a(s),a(s'))$.

By construction, we have $\nu_A(r_i,r_{i+1}) \leq \unr_A(a(s),a(s'))$ for all $0\leq i \leq k-1$. Similarly we have $\nu_A(t_i,t_{i+1}) \leq \unr(a(s),a(s'))$ for all $0\leq i \leq j-1$. 

Next observe that by compactness of $Z$, for each $a,a'\in A$ we can obtain $z(a) \in Z_a$, $z(a') \in Z_{a'}$ such that $\overline{\w}_Z(z(a),z(a')) = \nu_A(a,a')$. Applying this construction to consecutive elements in the chains $\overrightarrow{c}$ and $\overleftarrow{c}$, we obtain the following chains in $Z$: 
\begin{align*}
\overrightarrow{c_Z}&=\set{z(r_0),\ldots, z(r_k)},\text{ joining } z(a(s))=z(r_0) \text{ to } z(a(s')) = z(r_k),\\
\overleftarrow{c_Z}&=\set{z(t_0),\ldots, x(t_j)},\text{ joining } z(a(s'))=z(t_0) \text{ to } z(a(s)) = z(t_j).
\end{align*}

In particular, for each $0\leq i \leq k-1$, we have $\overline{\w}_Z(z(r_i),z(r_{i+1})) = \nu_A(r_i,r_{i+1}).$ Similarly for each $0\leq i \leq j-1$, we have $\overline{\w}_Z(z(t_i),z(t_{i+1})) = \nu_A(t_i,t_{i+1}).$ Furthermore, we have $z(r_i) \in Z_{r_i}$ for each $0\leq i \leq k$, and $z(t_i) \in Z_{t_i}$ for each $0 \leq i \leq j$. 

Now we can use the correspondence $R' \in \Rsc(S,Z)$ that we had fixed earlier. Recall the use of the notation $s(z) \in S$ for $z\in Z$ from the discussion preceding Claim \ref{cl:hnr-same-component}. Now we obtain the following chains in $S$:

\begin{align*}
\overrightarrow{c_S}&=\set{s(z(r_0)),\ldots, s(z(r_k))},\text{ joining } s(z(a(s))) \in Z_{a(s)} \text{ to } s(z(a(s'))) \in Z_{a(s')}, \text{ and }\\
\overleftarrow{c_S}&=\set{s(z(t_0)),\ldots, s(z(t_j))},\text{ joining } s(z(a(s'))) \in Z_{a(s')} \text{ to } s(z(a(s))) \in Z_{a(s)},\text{ such that }
\end{align*}
\begin{align*}
\overline{\w}_S(s(z(r_i)),s(z(r_{i+1}))) &< \overline{\w}_Z(z(r_i),z(r_{i+1})) + 2\e \text{ for all } 0\leq i \leq k-1,\text{ and }\\
\overline{\w}_S(s(z(t_i)),s(z(t_{i+1}))) &< \overline{\w}_Z(z(t_i),z(t_{i+1})) + 2\e \text{ for all } 0 \leq i \leq j-1.
\end{align*}

Here we have applied Remark \ref{rem:wbar-dist} on consecutive points in the chains to obtain the inequalities.

We know that $s$ and $s(z(r_0))=s(z(t_j))$ belong to the same $\w$-path-connected component $Z_{a(s)}$, and similarly $s'$ and $s(z(r_k))=s(z(t_0))$ belong to the same $\w$-path-connected component $Z_{a(s')}$. By Claim \ref{cl:hnr-same-component}, we have:
\begin{align*}
\unr_S(s,s(z(\rho_0))) &= \unr_S(s,s(z(\t_v))) < \unr_A(a(s),a(s')) + 2\e,\\
\unr_S(s',s(z(\rho_u))) &= \unr_S(s',s(z(\t_0))) < \unr_A(a(s),a(s')) + 2\e.
\end{align*}
Finally it follows that: 
\[\unr_S(s,s') < \unr_A(a(s),a(s')) + 2\e.\] 
\end{proof}

Thus we have $|\unr_S(s,s') - \unr_A(a(s),a(s'))| < 2\e$. Since $(s,a(s)), (s',a(s')) \in R$ were arbitrary, it now follows that $\dn\left((S,\unr_S),(A,\unr_A)\right) < \e$. \end{proof}

\hnrconsistent*

\begin{proof}[Proof of Theorem \ref{thm:hnr-consistent}]
The proof is analogous to that of Theorem \ref{thm:prob-net-approx}, but we repeat the argument here to facilitate the assessment of details. 
First observe that $\mf{M}_{\e/2}(\supp(\mu_Z)) \in (0,1]$. Let $ r \in (0,\mf{M}_{\e/2}(\supp(\mu_Z)))$, and let $\mc{U}_r$ be a refined $\e/2$-system on $\supp(\mu_Z)$ such that $\mf{m}(\mc{U}_r) \in (r, \mf{M}_{\e/2}(\supp(\mu_Z))]$. 
For convenience, write $m:= |\mc{U}_r|$, and also write $\mc{U}_r= \set{U_1,\ldots, U_m}$. 

For each $1\leq i \leq m$, define $A_i$ as in the statement of Lemma \ref{lem:sampling-e-system}. 
Then by Lemma \ref{lem:sampling-e-system}, the probability that at least one $U_i$ has empty intersection with $\mc{Z}_n(\omega)$ is bounded as $\P\left(\bigcup_{k=1}^mA_k\right) \leq \frac{1}{\mf{m}(\mc{U}_r)}(1-\mf{m}(\mc{U}_r))^n$. 

On the other hand, if $U_i$ has nonempty intersection with $\mc{Z}_n(\omega)$ for each $1\leq i \leq m$, then by Lemma \ref{lem:hnr-main}, we obtain $\dn((A,\unr_A),\hnr(\mc{Z}_n(\omega))) < \e$. Now define:
\[B_n:=\set{\omega \in \Om : \dn((A,\unr_A),\hnr(\mc{Z}_n(\omega)))\geq \e }.\]
Then we have:
\begin{align*}
\P \left(B_n\right) \leq \P\left(\bigcup_{k=1}^mA_k\right) \leq  \frac{1}{\mf{m}(\mc{U}_r)} \left(1-{\mf{m}(\mc{U}_r)}\right)^n.
\end{align*}
Since $ r \in (0,\mf{M}_{\e/2}(Z))$ was arbitrary, it follows that:
\[\P \left(B_n\right) \leq 
\frac{1}{\mf{M}_{\e/2}(\supp(\mu_Z))} \left(1-{\mf{M}_{\e/2}(\supp(\mu_Z))}\right)^n.\]
By an application of the Borel-Cantelli lemma as in Theorem \ref{thm:prob-net-approx}, we have $\P(\limsup_{n\r \infty} B_n) = 0$. The result now follows. \end{proof}

\subsubsection{Proofs related to Theorem \ref{thm:hr-consistent}}
\label{sec:pf-hr-consistent}

\begin{lemma}[Reciprocal clustering on a $\w$-path-connected network] 
\label{lem:hr-path-conn}
Let $(Z,\w_Z)$ be an $\w$-path-connected network with dissimilarity weights and finite reversibility $\rho_Z$. Then $(Z,\ur_Z) = \hr(Z,\w_Z)$ is given by writing $\ur_Z(z,z') = 0$ for all $z,z'\in Z$.
\end{lemma}
\begin{proof} Let $z,z'\in Z$, and let $\e > 0$. By Definition \ref{defn:path-conn} and the assumption that $\w_Z(z,z) = 0$ for all $z\in Z$, there exists a continuous function $\g:[0,1] \r Z$ such that $\g(0)=z$, $\g(1) = z'$, and there exist $t_0=0 \leq t_1 \leq t_2 \leq \ldots \leq t_n = 1$ such that: 
\[{\w}_Z(\g(t_i),\g(t_{i+1})) \in [0,\tfrac{\e}{\rho_Z}) \text{ for each } 0\leq i \leq n-1.\] 
Since $\w_Z$ is a dissimilarity, we have $\w_Z(\g(t_{i+1}),\g(t_i)) \geq 0$. By finite reversibility, we also have 
\[\w_Z(\g(t_{i+1}),\g(t_i)) \leq \rho_Z\cdot \w_Z(\g(t_i),\g(t_{i+1})) < \e.\]
Thus by using the chain $\set{\g(t_0),\g(t_1),\ldots, \g(t_n)}$, we have $\ur_Z(z,z') < \e$. Since $z,z'\in Z$ and $\e > 0$ were arbitrary, the result now follows. \end{proof}

\begin{lemma}[Reciprocal clustering collapses $\w$-path-connected subsets] 
\label{lem:hr-main}
Let $(Z,\w_Z)$ be a network with dissimilarity weights and finite reversibility such that $Z$ is a disjoint collection $\set{Z_a : a\in A}$, where $A$ is a finite indexing set and each $Z_a$ is compact and $\w$-path-connected. Let $(A,\lambda_A)$ be as in Definition \ref{defn:induced-net}, and let $(A,\ur_A)=\hr(A,\lambda_A)$. Also let $\mc{U}=\set{U_1,\ldots, U_m}$ be a refined $\e/2$-system on $Z$. 

Suppose that $S \subseteq Z$ is a finite subset equipped with the restriction $\w_S:= \w_Z|_{S \times S}$ such that $S$ has nonempty intersection with $U_i$ for each $1\leq i \leq m$. Then we have:
\[\dn\left((S,\ur_S),(A,\ur_A)\right) < \e.\]
\end{lemma}

\begin{proof}[Proof of Lemma \ref{lem:hr-main}]
For each $z \in Z$, let $a(z) \in A$ denote the index such that $z \in Z_{a(z)}$. Then define:
\[R:=\set{(s,a(s)) : s \in S}.\]
Then $R\in \Rsc(S,A)$. 
We wish to show $\dis(R) < 2\e$, where the distortion is calculated with respect to $\ur_S$ and $\ur_A$. Let $(s,a(s)), (s',a(s')) \in R$. 

\begin{claim} We have $\ur_A(a(s),a(s')) \leq \ur_S(s,s')$.
\end{claim}
\begin{proof} Pick a chain $c := \set{r_0=s,r_1,r_2,\ldots, r_k = s'}  \in C_S(s,s')$ such that 
\[\ur_S(s,s') = \max_{0\leq i \leq k-1}(\max(\overline{\w}_S(r_i,r_{i+1}),\overline{\w}_S(r_{i+1},r_i))).\]
Next consider the chain $c_A:=\set{a(r_0),\ldots, a(r_k)}$. By Definition \ref{defn:induced-net}, we have:
\[\max_{0\leq i \leq k-1}\lambda_A(a(r_i),a(r_{i+1})) \leq \ur_S(s,s').\]
It follows that $\ur_A(a(s),a(s')) \leq \ur_S(s,s')$. 
\end{proof}

By Theorem \ref{thm:sampling}, we obtain a correspondence $R' \in \Rsc(S,Z)$ with $\dis(R') < 2\e$ such that for each $(s,z) \in R'$, we have $\set{s,z} \subseteq U$ for some $U\in \mc{U}$. Here the distortion is measured with respect to $\w_Z$ and $\w_S$. We will use this correspondence $R'$ as follows: for each $z\in Z$, there exists $s\in S$ such that $(s,z) \in R'$. In other words, there exists $s \in S$ such that $\set{s,z} \subseteq U$ for some $U \in \mc{U}$. Since $\mc{U}$ is a refined $\e/2$-system, we know also that $s,z$ belong to the same connected component $Z_a$, for some $a \in A$. 

For each $z\in Z$, we will write $s(z)$ to denote the element of $S$ obtained by the preceding construction.

\begin{claim}\label{cl:hr-same-component}
Let $s,s'\in S$ be such that $a(s)=a(s')$, i.e. $s,s'$ belong to the same $\w$-path-connected component of $Z$. Then $\ur_S(s,s') <  2\e$. 
\end{claim} 
\begin{proof}[Proof of Claim \ref{cl:hr-same-component}] Let $\eta > 0$, and let $z,z' \in Z$ be such that $(s,z),(s',z') \in R'$. By Lemma \ref{lem:hr-path-conn}, we can take a chain $c= \set{z_0 = z,z_1,z_2,\ldots, z_n = z'}$ from $z$ to $z'$ such that 
\[\max_{0\leq i \leq n-1}\max(\overline{\w}_Z(z_i,z_{i+1}),\overline{\w}_Z(z_{i+1},z_i)) <  \eta.\] 
We can now convert this to a chain in $S$ by using the correspondence $R'$. Define:
\[c_S:=\set{s,s(z_1),s(z_2),\ldots, s(z_{n-1}),s'}.\]
By construction, $(s(z_i),z_i)\in R'$ for each $1\leq i \leq n-1$. Furthermore we have $(s,z),(s',z')\in R'$ by our choice of $z,z'$. Now by using Remark \ref{rem:wbar-dist} and the fact that $\dis(R') < 2\e$, we have $\ur_S(s,s') < \eta + 2\e$. Since $\eta > 0$ was arbitrary, the result follows. \end{proof}

\begin{claim} We have $\ur_S(s,s') < \ur_A(a(s),a(s')) + 2\e$.
\end{claim}
\begin{proof} Let $c:=\set{r_0,\ldots, r_k}$ be a chain in $A$ such that $r_0 = a(s)$, $r_k=a(s')$, and 
\[\max_{0\leq i \leq k-1}\lambda_A(r_i,r_{i+1}) \leq \ur_A(a(s),a(s')).\] 

Next observe that by compactness of $Z$, for each $a,a'\in A$ we can obtain $z(a) \in Z_a$, $z(a') \in Z_{a'}$ such that: 
\[\max(\overline{\w}_Z(z(a),z(a')), \overline{\w}_Z(z(a'),z(a)))  = \lambda_A(a,a').\] 
Applying this construction to consecutive elements in the chain $c$, we obtain a chain in $Z$: 
\begin{align*}
{c_Z}&=\set{z(r_0),\ldots, z(r_k)},\text{ joining } z(a(s)) \text{ to } z(a(s')).
\end{align*}
In particular, for each $0\leq i \leq k-1$, we have 
\[\max(\overline{\w}_Z(z(r_i),z(r_{i+1})), \overline{\w}_Z(z(r_{i+1}),z(r_{i}))) = \lambda_A(r_i,r_{i+1}).\]

Now we can use the correspondence $R' \in \Rsc(S,Z)$ that we had fixed earlier. Recall the use of the notation $s(z) \in S$ for $z\in Z$ from the discussion preceding Claim \ref{cl:hr-same-component}. Now we obtain the following chain in $S$:

\begin{align*}
{c_S}&=\set{s(z(r_0)),\ldots, s(z(r_k))},\text{ joining } s(z(a(s))) \in Z_{a(s)} \text{ to } s(z(a(s'))) \in Z_{a(s')},\text{ such that }
\end{align*}
\begin{align*}
\overline{\w}_S(s(z(r_i)),s(z(r_{i+1}))) &< \overline{\w}_Z(z(r_i),z(r_{i+1})) + 2\e \text{ for all } 0\leq i \leq k-1,\text{ and }\\
\overline{\w}_S(s(z(r_{i+1})),s(z(r_{i}))) &< \overline{\w}_Z(z(r_{i+1}),z(r_{i})) + 2\e \text{ for all } 0 \leq i \leq k-1.
\end{align*}

Here we have applied Remark \ref{rem:wbar-dist} on consecutive points in the chains to obtain the inequalities.

We know that $s$ and $s(z(r_0))$ belong to the same $\w$-path-connected component $Z_{a(s)}$, and similarly $s'$ and $s(z(r_k))$ belong to the same $\w$-path-connected component $Z_{a(s')}$. By Claim \ref{cl:hr-same-component}, we have:
\begin{align*}
\ur_S(s,s(z(r_0))) & < 2\e,\\
\ur_S(s',s(z(r_k))) & <  2\e.
\end{align*}
Finally it follows that: 
\[\ur_S(s,s') < \ur_A(a(s),a(s')) + 2\e.\] 
\end{proof}

Thus we have $|\ur_S(s,s') - \ur_A(a(s),a(s'))| < 2\e$. Since $(s,a(s)), (s',a(s')) \in R$ were arbitrary, it now follows that $\dn\left((S,\ur_S),(A,\ur_A)\right) < \e$. \end{proof}

\hrconsistent*

\begin{proof}[Proof of Theorem \ref{thm:hr-consistent}]
The proof is analogous to that of Theorem \ref{thm:hnr-consistent}, but we repeat the argument here to facilitate the assessment of details. 
First observe that $\mf{M}_{\e/2}(\supp(\mu_Z)) \in (0,1]$. Let $ r \in (0,\mf{M}_{\e/2}(\supp(\mu_Z)))$, and let $\mc{U}_r$ be a refined $\e/2$-system on $\supp(\mu_Z)$ such that $\mf{m}(\mc{U}_r) \in (r, \mf{M}_{\e/2}(\supp(\mu_Z))]$. 
For convenience, write $m:= |\mc{U}_r|$, and also write $\mc{U}_r= \set{U_1,\ldots, U_m}$. 

For each $1\leq i \leq m$, define $A_i$ as in the statement of Lemma \ref{lem:sampling-e-system}. Then by Lemma \ref{lem:sampling-e-system}, the probability that at least one $U_i$ has empty intersection with $\mc{Z}_n$ is bounded as $\P\left(\bigcup_{k=1}^mA_k\right) \leq \frac{1}{\mf{m}(\mc{U}_r)}(1-\mf{m}(\mc{U}_r))^n$. 

On the other hand, if $U_i$ has nonempty intersection with $\mc{Z}_n(\omega)$ for each $1\leq i \leq m$, then by Lemma \ref{lem:hr-main}, we obtain $\dn((A,\ur_A),\hr(\mc{Z}_n(\omega))) < \e$. Now define:
\[B_n:=\set{\omega \in \Om : \dn((A,\ur_A),\hr(\mc{Z}_n(\omega)))\geq \e }.\]
Then we have:
\begin{align*}
\P \left(B_n\right) \leq \P\left(\bigcup_{k=1}^mA_k\right) \leq  \frac{1}{\mf{m}(\mc{U}_r)} \left(1-{\mf{m}(\mc{U}_r)}\right)^n.
\end{align*}
Since $ r \in (0,\mf{M}_{\e/2}(X))$ was arbitrary, it follows that:
\[\P \left(B_n\right) \leq 
\frac{1}{\mf{M}_{\e/2}(\supp(\mu_Z))} \left(1-{\mf{M}_{\e/2}(\supp(\mu_Z))}\right)^n.\]
By an application of the Borel-Cantelli lemma as in Theorem \ref{thm:prob-net-approx}, we have $\P(\limsup_{n\r \infty} B_n) = 0$. The result now follows.\end{proof}

\section{Discussion}
\label{sec:discussion}

The original motivation for this work came from a scientific problem: that of having a correspondence framework for graph-theoretic structures arising in systems neuroscience \cite{sporns2012discovering} in which undirected (\emph{functional connectivity}) and directed (\emph{effective connectivity}) graphs would be compatible with one another. The $\dn$ metric that had recently been introduced in \cite{carlsson2013axiomatic,clust-net} for extending hierarchical clustering to asymmetric networks was a natural proposal for addressing this limitation. After developing the $\dn$ theory in the finite case \cite{nets-allerton,nets-icassp} and observing concurrent uses for developing stable persistent homology methods for finite networks \cite{dowker-jact,turner2019rips}, it was apparent that one needed to understand the convergence of sequences of finite networks to continuous objects, as had been done for the cut metric in \cite{lovasz2012large} and related earlier works \cite{borgs2008convergent,borgs2012convergent}. In particular, this question was made explicit in \cite{turner2019rips}. Studying this convergence led us to realize that even though $\dn$ metrizes the collection $\{(X,\w_X)\}$ where $X$ is a set and $\w_X:X\times X \to \R$ is any function, such a collection was not amenable to meaningful convergence results or even a satisfactory characterization of isomorphism. This first set of obstacles led to the subsequent introduction of topologies and the intermediate collection $\Ncom$. This turned out to be a suitable ambient space which was sufficiently general while also having the required topological regularity. 

Although our study of the consequences of relaxing metric assumptions was driven by empirical and applied considerations which stem from the data-driven scientific outlook of the recent decades, ultimately we were led to the literature developed by Fr\'{e}chet and his contemporaries at the beginning of the 20th century \cite{frechet1906quelques,pitcher1918foundations,niemytzki1927third}. These authors had been motivated to understand different constructions so they could settle on an appropriate notion of a metric space. Later works by Busemann \cite{busemann1950geometry} and others on quasimetric spaces also turned out to be useful for understanding the subtleties that arise with asymmetry. Along these lines, a possible latent connection that we have not yet explored is to the field of directed topology \cite{fajstrup2016directed}.

In this work, we provided a comprehensive theoretical study of $\dn$, its isomorphism structure, and the metric structure that it induces on $\Ngen$. In particular, we introduced $\Ncom$ as a subfamily that is nested between $\Ncal$ and $\Ngen$, proved the sampling and convergence results on $\Ncom$ that finally enabled one to obtain well-defined persistence diagrams on \emph{infinite} asymmetric networks, and characterized the Dowker persistence diagrams of the directed circles with finite reversibility. These results now suggest numerous natural lines of continuation, some of which are as follows:

\begin{itemize}
    \item \textbf{Developing new embeddings:} As described in the introduction, graph structured data may be embedded into $\Ncom$ in a multitude of ways. By virtue of this work, different embeddings $G\mapsto \Ncom$ are now seen to share a common ambient metric space. Development of new and informative embeddings, perhaps by aggregating existing embeddings, is a broad-scale challenge. Some previous works \cite{dowker-jact,cosyne} have shown that building asymmetric Markov transition matrices can give better accuracy on certain classification tasks than their symmetric counterparts, so the systematic development of ``asymmetric embeddings" seems to hold promise in extracting additional signal from data. 
    With increasing availability of benchmark data sets for machine learning purposes, using data-driven supervision to generate ``learned" embeddings is a natural next step \cite{litman2013learning}.

    \item \textbf{New inputs to persistent homology methods:} Along these lines, it would be interesting to see which embeddings could serve as good preprocessing steps for the hierarchical clustering and persistent homology methods described earlier in this work. These would automatically benefit from $\dn$-stability with respect to the bottleneck distance. 
    From a different perspective, methods such as the Dowker complex have natural connections to \emph{hypergraphs} that are used to model multi-way relations \cite{zhou2006learning}. Elaborating on this latent connection by extending our constructions to accept hypergraph data as input would open new directions in hypergraph data analysis.

    \item \textbf{More directed models with geometric regularity:} 
    Another interesting thread for future development would be to construct directed models of shapes beyond the circles with finite reversibility. This would likely provide an interesting connection to the theoretical development for Randers manifolds in \cite{bao2004zermelo} and help towards characterizing how such spaces are embedded in the $(\Ngen,\dn)$ hierarchy. Implicit in this construction should be some notion of geometric regularity such as the finite reversibility assumption.

\end{itemize}

Finally we remark that embeddings of the form $G\mapsto \Ngen$ can be further extended via maps $\Ngen \to \Ngen^m$ where the latter is the collection of measure networks equipped with a generalized Gromov-Wasserstein distance \cite{gwnets}. This approach follows the formulation of \cite{dgh-sm} in taking a relaxation of the $\dn$ problem, and benefits from the deep study of Riemannian structure developed by Sturm \cite{sturm2012space}. 
This particular context has seen extensive growth in the machine learning community in recent years \cite{pcs16,s16,gwa}, and as has already been shown in \cite{xu2019scalable, chowdhury2020generalized}, in various cases the results demonstrated in these works arise from utilizing different embeddings $G\mapsto \Ngen$ and studying the particular structure that each embedding imposes upon the problem.   

\subsubsection*{Acknowledgements.}
This work is partially supported by the National Science Foundation under grants DMS-1723003, CCF-1740761,  CCF-1526513, IIS-1422400, and DMS-1547357. Facundo M\'emoli acknowledges support from the Mathematical Biosciences Institute at The Ohio State University.

\subsubsection*{Conflict of interest statement.}
The authors declare no conflict of interest.

\newpage
\bibliographystyle{alpha}
\bibliography{bibliography/biblio}

\newcommand{\etalchar}[1]{$^{#1}$}
\begin{thebibliography}{CCSG{\etalchar{+}}09b}

\bibitem[AA17]{adamaszek2017vietoris}
Micha{\l} Adamaszek and Henry Adams.
\newblock The {V}ietoris--{R}ips complexes of a circle.
\newblock {\em Pacific Journal of Mathematics}, 290(1):1--40, 2017.

\bibitem[AAF{\etalchar{+}}16]{adamaszek2016nerve}
Micha{\l} Adamaszek, Henry Adams, Florian Frick, Chris Peterson, and Corrine
  Previte-Johnson.
\newblock Nerve complexes of circular arcs.
\newblock {\em Discrete \& Computational Geometry}, 56(2):251--273, 2016.

\bibitem[BBI01]{burago}
Dmitri Burago, Yuri Burago, and Sergei Ivanov.
\newblock {\em A Course in Metric Geometry}, volume~33 of {\em AMS Graduate
  Studies in Math.}
\newblock American Mathematical Society, 2001.

\bibitem[BCL{\etalchar{+}}08]{borgs2008convergent}
Christian Borgs, Jennifer~T Chayes, L{\'a}szl{\'o} Lov{\'a}sz, Vera~T S{\'o}s,
  and Katalin Vesztergombi.
\newblock Convergent sequences of dense graphs {I}: Subgraph frequencies,
  metric properties and testing.
\newblock {\em Advances in Mathematics}, 219(6):1801--1851, 2008.

\bibitem[BCL{\etalchar{+}}12]{borgs2012convergent}
Christian Borgs, Jennifer~T Chayes, L{\'a}szl{\'o} Lov{\'a}sz, Vera~T S{\'o}s,
  and Katalin Vesztergombi.
\newblock Convergent sequences of dense graphs ii. multiway cuts and
  statistical physics.
\newblock {\em Annals of Mathematics}, pages 151--219, 2012.

\bibitem[BCS12]{bao2012introduction}
David Bao, S-S Chern, and Zhongmin Shen.
\newblock {\em An introduction to Riemann-Finsler geometry}, volume 200.
\newblock Springer Science \& Business Media, 2012.

\bibitem[BDM09]{burkard2009assignment}
Rainer~E Burkard, Mauro Dell'Amico, and Silvano Martello.
\newblock {\em Assignment Problems}.
\newblock SIAM, 2009.

\bibitem[BFT{\etalchar{+}}98]{brown1998statistical}
Emery~N Brown, Loren~M Frank, Dengda Tang, Michael~C Quirk, and Matthew~A
  Wilson.
\newblock A statistical paradigm for neural spike train decoding applied to
  position prediction from ensemble firing patterns of rat hippocampal place
  cells.
\newblock {\em The Journal of Neuroscience}, 18(18):7411--7425, 1998.

\bibitem[BH11]{bridson2011metric}
Martin~R Bridson and Andr{\'e} Haefliger.
\newblock {\em Metric spaces of non-positive curvature}, volume 319.
\newblock Springer Science \& Business Media, 2011.

\bibitem[BK07]{boutin2007lossless}
Mireille Boutin and Gregor Kemper.
\newblock Lossless representation of graphs using distributions.
\newblock {\em arXiv preprint arXiv:0710.1870}, 2007.

\bibitem[BL17]{blumberg2017universality}
Andrew~J Blumberg and Michael Lesnick.
\newblock Universality of the homotopy interleaving distance.
\newblock {\em arXiv preprint arXiv:1705.01690}, 2017.

\bibitem[BLM20]{bauer-tripods-reeb}
Ulrich Bauer, Claudia Landi, and Facundo M{\'{e}}moli.
\newblock The {R}eeb graph edit distance is universal.
\newblock In Sergio Cabello and Danny~Z. Chen, editors, {\em 36th International
  Symposium on Computational Geometry, SoCG 2020, June 23-26, 2020,
  Z{\"{u}}rich, Switzerland}, volume 164 of {\em LIPIcs}, pages 15:1--15:16.
  Schloss Dagstuhl - Leibniz-Zentrum f{\"{u}}r Informatik, 2020.

\bibitem[BRS04]{bao2004zermelo}
David Bao, Colleen Robles, and Zhongmin Shen.
\newblock Zermelo navigation on riemannian manifolds.
\newblock {\em Journal of Differential Geometry}, 66(3):377--435, 2004.

\bibitem[Bus50]{busemann1950geometry}
Herbert Busemann.
\newblock The geometry of {F}insler spaces.
\newblock {\em Bulletin of the American Mathematical Society}, 56(1):5--16,
  1950.

\bibitem[BWM01]{best2001spatial}
Phillip~J Best, Aaron~M White, and Ali Minai.
\newblock Spatial processing in the brain: the activity of hippocampal place
  cells.
\newblock {\em Annual review of neuroscience}, 24(1):459--486, 2001.

\bibitem[Car09]{carlsson2009topology}
Gunnar Carlsson.
\newblock Topology and data.
\newblock {\em Bulletin of the American Mathematical Society}, 46(2):255--308,
  2009.

\bibitem[CCM{\etalchar{+}}20]{filt-func}
Samir Chowdhury, Nathaniel Clause, Facundo M{\'e}moli, Jose~{\'A}ngel
  S{\'a}nchez, and Zoe Wellner.
\newblock New families of stable simplicial filtration functors.
\newblock {\em Topology and its Applications}, 2020.

\bibitem[CCSG{\etalchar{+}}09a]{chazal2009proximity}
Fr{\'e}d{\'e}ric Chazal, David Cohen-Steiner, Marc Glisse, Leonidas~J Guibas,
  and Steve~Y Oudot.
\newblock Proximity of persistence modules and their diagrams.
\newblock In {\em Proceedings of the twenty-fifth annual symposium on
  Computational geometry}, pages 237--246. ACM, 2009.

\bibitem[CCSG{\etalchar{+}}09b]{dgh-pers}
Fr{\'e}d{\'e}ric Chazal, David Cohen-Steiner, Leonidas~J Guibas, Facundo
  M{\'e}moli, and Steve~Y Oudot.
\newblock {G}romov-{H}ausdorff stable signatures for shapes using persistence.
\newblock In {\em Computer Graphics Forum}, volume~28, pages 1393--1403. Wiley
  Online Library, 2009.

\bibitem[CDM17]{cosyne}
Samir Chowdhury, Bowen Dai, and Facundo M{\'e}moli.
\newblock Topology of stimulus space via directed network persistent homology.
\newblock {\em Cosyne Abstracts}, 2017.

\bibitem[CDSGO16]{pers-mod-book}
Fr{\'e}d{\'e}ric Chazal, Vin De~Silva, Marc Glisse, and Steve Oudot.
\newblock {\em The structure and stability of persistence modules}.
\newblock Springer, 2016.

\bibitem[CDSO14]{chazal2014persistence}
Fr{\'e}d{\'e}ric Chazal, Vin De~Silva, and Steve Oudot.
\newblock Persistence stability for geometric complexes.
\newblock {\em Geometriae Dedicata}, 173(1):193--214, 2014.

\bibitem[CFV20]{calissano2020populations}
Anna Calissano, Aasa Feragen, and Simone Vantini.
\newblock Populations of unlabeled networks: Graph space geometry and geodesic
  principal components.
\newblock {\em MOX Report}, 2020.

\bibitem[CGHY19]{phmlp}
Samir Chowdhury, Thomas Gebhart, Steve Huntsman, and Matvey Yutin.
\newblock Path homologies of deep feedforward networks.
\newblock In {\em 2019 18th IEEE International Conference On Machine Learning
  And Applications (ICMLA)}, pages 1077--1082. IEEE, 2019.

\bibitem[Chu05]{chung2005laplacians}
Fan Chung.
\newblock Laplacians and the {C}heeger inequality for directed graphs.
\newblock {\em Annals of Combinatorics}, 9(1):1--19, 2005.

\bibitem[CI08]{curto2008cell}
Carina Curto and Vladimir Itskov.
\newblock Cell groups reveal structure of stimulus space.
\newblock {\em PLoS Computational Biology}, 4(10), 2008.

\bibitem[CM10]{clust-um}
Gunnar Carlsson and Facundo M\'emoli.
\newblock Characterization, stability and convergence of hierarchical
  clustering methods.
\newblock {\em Journal of Machine Learning Research}, 11:1425--1470, 2010.

\bibitem[CM13]{carlsson2013classifying}
Gunnar Carlsson and Facundo M{\'e}moli.
\newblock Classifying clustering schemes.
\newblock {\em Foundations of Computational Mathematics}, 13(2):221--252, 2013.

\bibitem[CM15]{nets-allerton}
Samir Chowdhury and Facundo M\'{e}moli.
\newblock Metric structures on networks and applications.
\newblock In {\em 2015 53rd Annual Allerton Conference on Communication,
  Control, and Computing (Allerton)}, pages 1470--1472, Sept 2015.

\bibitem[CM16a]{nets-icassp}
Samir Chowdhury and Facundo M{\'e}moli.
\newblock Distances between directed networks and applications.
\newblock In {\em 2016 IEEE International Conference on Acoustics, Speech and
  Signal Processing (ICASSP)}, pages 6420--6424. IEEE, 2016.

\bibitem[CM16b]{dowker-asilo}
Samir Chowdhury and Facundo M{\'e}moli.
\newblock Persistent homology of directed networks.
\newblock In {\em 2016 50th Asilomar Conference on Signals, Systems and
  Computers}, pages 77--81. IEEE, 2016.

\bibitem[CM17]{dn-part1}
Samir Chowdhury and Facundo M{\'e}moli.
\newblock Distances and isomorphism between networks and the stability of
  network invariants.
\newblock {\em arXiv preprint arXiv:1708.04727}, 2017.

\bibitem[CM18a]{dgh-era}
Samir Chowdhury and Facundo M{\'e}moli.
\newblock Explicit geodesics in {G}romov-{H}ausdorff space.
\newblock {\em Electronic Research Announcements}, 25:48--59, 2018.

\bibitem[CM18b]{dowker-jact}
Samir Chowdhury and Facundo M{\'e}moli.
\newblock A functorial {D}owker theorem and persistent homology of asymmetric
  networks.
\newblock {\em Journal of Applied and Computational Topology}, 2(1-2):115--175,
  2018.

\bibitem[CM18c]{pph}
Samir Chowdhury and Facundo M{\'e}moli.
\newblock Persistent path homology of directed networks.
\newblock In {\em Proceedings of the Twenty-Ninth Annual ACM-SIAM Symposium on
  Discrete Algorithms}, pages 1152--1169. SIAM, 2018.

\bibitem[CM19]{gwnets}
Samir Chowdhury and Facundo M{\'e}moli.
\newblock The {G}romov--{W}asserstein distance between networks and stable
  network invariants.
\newblock {\em Information and Inference: A Journal of the IMA}, 8(4):757--787,
  2019.

\bibitem[CMRS13]{carlsson2013axiomatic}
Gunnar Carlsson, Facundo M{\'e}moli, Alejandro Ribeiro, and Santiago Segarra.
\newblock Axiomatic construction of hierarchical clustering in asymmetric
  networks.
\newblock In {\em 2013 IEEE International Conference on Acoustics, Speech and
  Signal Processing (ICASSP)}, pages 5219--5223. IEEE, 2013.

\bibitem[CMRS14]{clust-net}
Gunnar Carlsson, Facundo M{\'{e}}moli, Alejandro Ribeiro, and Santiago Segarra.
\newblock Hierarchical quasi-clustering methods for asymmetric networks.
\newblock In {\em Proceedings of the 31th International Conference on Machine
  Learning, {ICML} 2014}, 2014.

\bibitem[CMRS18]{carlsson2018hierarchical}
Gunnar Carlsson, Facundo M{\'e}moli, Alejandro Ribeiro, and Santiago Segarra.
\newblock Hierarchical clustering of asymmetric networks.
\newblock {\em Advances in Data Analysis and Classification}, 12(1):65--105,
  2018.

\bibitem[CMS21]{carlsson2021robust}
Gunnar Carlsson, Facundo M{\'e}moli, and Santiago Segarra.
\newblock Robust hierarchical clustering for directed networks: An axiomatic
  approach.
\newblock {\em SIAM Journal on Applied Algebra and Geometry}, 5(4):675--700,
  2021.

\bibitem[CN20a]{chowdhury2020generalized}
Samir Chowdhury and Tom Needham.
\newblock Generalized spectral clustering via {G}romov-{W}asserstein learning.
\newblock {\em arXiv preprint arXiv:2006.04163}, 2020.

\bibitem[CN20b]{gwa}
Samir Chowdhury and Tom Needham.
\newblock Gromov-{W}asserstein averaging in a {R}iemannian framework.
\newblock In {\em Proceedings of the IEEE/CVF Conference on Computer Vision and
  Pattern Recognition Workshops}, pages 842--843, 2020.

\bibitem[Cob12]{cobzas2012functional}
Stefan Cobzas.
\newblock {\em Functional analysis in asymmetric normed spaces}.
\newblock Springer Science \& Business Media, 2012.

\bibitem[CZCG05]{carlsson2005persistence}
Gunnar Carlsson, Afra Zomorodian, Anne Collins, and Leonidas~J Guibas.
\newblock Persistence barcodes for shapes.
\newblock {\em International Journal of Shape Modeling}, 11(02):149--187, 2005.

\bibitem[DLW20]{dey2020efficient}
Tamal~K Dey, Tianqi Li, and Yusu Wang.
\newblock An efficient algorithm for 1-dimensional (persistent) path homology.
\newblock In {\em 36th International Symposium on Computational Geometry (SoCG
  2020)}. Schloss Dagstuhl-Leibniz-Zentrum f{\"u}r Informatik, 2020.

\bibitem[DMFC12]{dabaghian2012topological}
Yuri Dabaghian, Facundo M{\'e}moli, Loren Frank, and Gunnar Carlsson.
\newblock A topological paradigm for hippocampal spatial map formation using
  persistent homology.
\newblock {\em PLoS Computational Biology}, 8(8), 2012.

\bibitem[DMW16]{dey2016multiscale}
Tamal~K Dey, Facundo M{\'e}moli, and Yusu Wang.
\newblock Multiscale mapper: topological summarization via codomain covers.
\newblock In {\em Proceedings of the twenty-seventh annual ACM-SIAM Symposium
  on Discrete Algorithms}, pages 997--1013. Society for Industrial and Applied
  Mathematics, 2016.

\bibitem[Dow52]{dowker1952homology}
Clifford~H Dowker.
\newblock Homology groups of relations.
\newblock {\em Annals of Mathematics}, pages 84--95, 1952.

\bibitem[Dud02]{dudley}
Richard~M Dudley.
\newblock {\em Real analysis and probability}, volume~74.
\newblock Cambridge University Press, 2002.

\bibitem[Edg93]{edgar1993classics}
Gerald~A Edgar.
\newblock Classics on fractals.
\newblock 1993.

\bibitem[EH10]{edelsbrunner2010computational}
Herbert Edelsbrunner and John Harer.
\newblock {\em Computational topology: an introduction}.
\newblock American Mathematical Society, 2010.

\bibitem[EM14]{edelsbrunner2014persistent}
Herbert Edelsbrunner and Dmitriy Morozov.
\newblock Persistent homology: theory and practice.
\newblock 2014.

\bibitem[FGH{\etalchar{+}}16]{fajstrup2016directed}
Lisbeth Fajstrup, Eric Goubault, Emmanuel Haucourt, Samuel Mimram, and Martin
  Raussen.
\newblock {\em Directed algebraic topology and concurrency}.
\newblock Springer, 2016.

\bibitem[For17]{santo}
Santo Fortunato.
\newblock Benchmark graphs to test community detection algorithms.
\newblock {\em
  \url{https://www.santofortunato.net/resources#h.p_u6MEEWAKyhN0}}, 2017.

\bibitem[Fra65]{franklin1965spaces}
Stanley~P Franklin.
\newblock Spaces in which sequences suffice.
\newblock {\em Fundamenta Mathematicae}, 57(1):107--115, 1965.

\bibitem[Fr{\'e}06]{frechet1906quelques}
Maurice Fr{\'e}chet.
\newblock Sur quelques points du calcul fonctionnel.
\newblock {\em Rendiconti del Circolo Matematico di Palermo (1884-1940)},
  22(1):1--72, 1906.

\bibitem[Fri94]{friston1994functional}
Karl~J Friston.
\newblock Functional and effective connectivity in neuroimaging: a synthesis.
\newblock {\em Human brain mapping}, 2(1-2):56--78, 1994.

\bibitem[Fro92]{frosini1992measuring}
Patrizio Frosini.
\newblock Measuring shapes by size functions.
\newblock In {\em Intelligent Robots and Computer Vision X: Algorithms and
  Techniques}, pages 122--133. International Society for Optics and Photonics,
  1992.

\bibitem[FS98]{fagin1998relaxing}
Ronald Fagin and Larry Stockmeyer.
\newblock Relaxing the triangle inequality in pattern matching.
\newblock {\em International Journal of Computer Vision}, 30(3):219--231, 1998.

\bibitem[GLMY14]{grigor2014homotopy}
Alexander Grigor’yan, Yong Lin, Yuri Muranov, and Shing-Tung Yau.
\newblock Homotopy theory for digraphs.
\newblock {\em Pure and Applied Mathematics Quarterly}, 10(4):619--674, 2014.

\bibitem[Gro81]{gromov1981structures}
Mikhail Gromov.
\newblock Structures m{\'e}triques pour les vari{\'e}t{\'e}s {R}iemanniennes.
\newblock {\em Textes Math{\'e}matiques [Mathematical Texts]}, 1, 1981.

\bibitem[Gro99]{gromov-book}
Mikhail Gromov.
\newblock {\em Metric structures for {R}iemannian and non-{R}iemannian spaces},
  volume 152 of {\em Progress in Mathematics}.
\newblock Birkh\"auser Boston Inc., Boston, MA, 1999.

\bibitem[Gru84]{gruenhage1984generalized}
Gary Gruenhage.
\newblock Generalized metric spaces.
\newblock {\em Handbook of set-theoretic topology}, pages 423--501, 1984.

\bibitem[GS84]{galvin1984completeness}
Fred Galvin and Samuel Shore.
\newblock Completeness in semimetric spaces.
\newblock {\em Pacific Journal of Mathematics}, 113(1):67--75, 1984.

\bibitem[GS91]{galvin1991distance}
Fred Galvin and Samuel Shore.
\newblock Distance functions and topologies.
\newblock {\em The American Mathematical Monthly}, 98(7):620--623, 1991.

\bibitem[Har75]{hartigan1975clustering}
John~A Hartigan.
\newblock {\em Clustering algorithms}, volume 209.
\newblock Wiley New York, 1975.

\bibitem[Har81]{hartigan1981consistency}
John~A Hartigan.
\newblock Consistency of single linkage for high-density clusters.
\newblock {\em Journal of the American Statistical Association},
  76(374):388--394, 1981.

\bibitem[Har85]{hartigan1985statistical}
John~A Hartigan.
\newblock Statistical theory in clustering.
\newblock {\em Journal of classification}, 2(1):63--76, 1985.

\bibitem[Hei12]{heinonen2012lectures}
Juha Heinonen.
\newblock {\em Lectures on analysis on metric spaces}.
\newblock Springer Science \& Business Media, 2012.

\bibitem[Hen16]{hendrikson}
Reigo Hendrikson.
\newblock Using {G}romov-{W}asserstein distance to explore sets of networks.
\newblock Master's thesis, University of Tartu, 2016.

\bibitem[Hil05]{hill2005free}
Terrell~L Hill.
\newblock {\em Free energy transduction and biochemical cycle kinetics}.
\newblock Courier Corporation, 2005.

\bibitem[HLS07]{horster2007image}
Eva H{\"o}rster, Rainer Lienhart, and Malcolm Slaney.
\newblock Image retrieval on large-scale image databases.
\newblock In {\em Proceedings of the 6th ACM international conference on Image
  and video retrieval}, pages 17--24, 2007.

\bibitem[Hun20]{huntsman2020generalizing}
Steve Huntsman.
\newblock Generalizing cyclomatic complexity via path homology.
\newblock {\em arXiv preprint arXiv:2003.00944}, 2020.

\bibitem[INT16]{ivanov2016gromov}
AO~Ivanov, NK~Nikolaeva, and AA~Tuzhilin.
\newblock The {G}romov--{H}ausdorff metric on the space of compact metric
  spaces is strictly intrinsic.
\newblock {\em Mathematical Notes}, 5(100):883--885, 2016.

\bibitem[JO09]{jain2009structure}
Brijnesh~J Jain and Klaus Obermayer.
\newblock Structure spaces.
\newblock {\em Journal of Machine Learning Research}, 10(11), 2009.

\bibitem[JS71]{jardine1971mathematical}
N.~Jardine and R.~Sibson.
\newblock {\em Mathematical Taxonomy}.
\newblock Wiley series in probability and mathematical statistics. Wiley, 1971.

\bibitem[JSHV08]{jegou2008accurate}
Herve Jegou, Cordelia Schmid, Hedi Harzallah, and Jakob Verbeek.
\newblock Accurate image search using the contextual dissimilarity measure.
\newblock {\em IEEE Transactions on Pattern Analysis and Machine Intelligence},
  32(1):2--11, 2008.

\bibitem[Kel63]{kelly1963bitopological}
JC~Kelly.
\newblock Bitopological spaces.
\newblock {\em Proceedings of the London Mathematical Society}, 3(1):71--89,
  1963.

\bibitem[KO99]{kalton1999distances}
Nigel~J Kalton and Mikhail~I Ostrovskii.
\newblock Distances between {B}anach spaces.
\newblock In {\em Forum Mathematicum}, volume~11, pages 17--48. Walter de
  Gruyter, 1999.

\bibitem[LB13]{litman2013learning}
Roee Litman and Alexander~M Bronstein.
\newblock Learning spectral descriptors for deformable shape correspondence.
\newblock {\em IEEE transactions on pattern analysis and machine intelligence},
  36(1):171--180, 2013.

\bibitem[LGL17]{gouic2017existence}
Thibaut Le~Gouic and Jean-Michel Loubes.
\newblock Existence and consistency of {W}asserstein barycenters.
\newblock {\em Probability Theory and Related Fields}, 168(3-4):901--917, 2017.

\bibitem[LMS19]{lyu2019sampling}
Hanbaek Lyu, Facundo Memoli, and David Sivakoff.
\newblock Sampling random graph homomorphisms and applications to network data
  analysis.
\newblock {\em arXiv preprint arXiv:1910.09483}, 2019.

\bibitem[LNZ16]{leustean2016barycenters}
Laurentiu Leustean, Adriana Nicolae, and Alexandru Zaharescu.
\newblock Barycenters in uniformly convex geodesic spaces.
\newblock {\em arXiv preprint arXiv:1609.02589}, 2016.

\bibitem[Lov12]{lovasz2012large}
L{\'a}szl{\'o} Lov{\'a}sz.
\newblock {\em Large networks and graph limits}, volume~60.
\newblock American Mathematical Society, 2012.

\bibitem[M{\etalchar{+}}89]{matsumoto1989slope}
Makoto Matsumoto et~al.
\newblock A slope of a mountain is a {F}insler surface with respect to a time
  measure.
\newblock {\em Journal of Mathematics of Kyoto University}, 29(1):17--25, 1989.

\bibitem[Mel07]{melleray2007geometry}
Julien Melleray.
\newblock On the geometry of {U}rysohn's universal metric space.
\newblock {\em Topology and its Applications}, 154(2):384--403, 2007.

\bibitem[M{\'e}m07]{dgh-sm}
Facundo M{\'e}moli.
\newblock On the use of {G}romov-{H}ausdorff distances for shape comparison.
\newblock {\em The Eurographics Association}, 2007.

\bibitem[M{\'e}m11]{dghlp-focm}
Facundo M{\'e}moli.
\newblock Gromov--{W}asserstein distances and the metric approach to object
  matching.
\newblock {\em Foundations of Computational Mathematics}, 11(4):417--487, 2011.

\bibitem[M{\'e}m12]{dgh-props}
Facundo M{\'e}moli.
\newblock Some properties of {G}romov–{H}ausdorff distances.
\newblock {\em Discrete \& Computational Geometry}, pages 1--25, 2012.
\newblock 10.1007/s00454-012-9406-8.

\bibitem[M{\'{e}}m17]{memoli2017distance}
Facundo M{\'{e}}moli.
\newblock A distance between filtered spaces via tripods.
\newblock {\em arXiv preprint arXiv:1704.03965}, 2017.

\bibitem[Men13]{mennucci2013asymmetric}
Andrea~CG Mennucci.
\newblock On asymmetric distances.
\newblock {\em Analysis and Geometry in Metric Spaces}, 1:200--231, 2013.

\bibitem[MS04]{dgh-sgp}
Facundo M\'{e}moli and Guillermo Sapiro.
\newblock Comparing point clouds.
\newblock In {\em SGP '04: Proceedings of the 2004 Eurographics/ACM SIGGRAPH
  symposium on Geometry processing}, pages 32--40, New York, NY, USA, 2004.
  ACM.

\bibitem[MS05]{dgh-focm}
Facundo M{\'e}moli and Guillermo Sapiro.
\newblock A theoretical and computational framework for isometry invariant
  recognition of point cloud data.
\newblock {\em Foundations of Computational Mathematics}, 5(3):313--347, 2005.

\bibitem[Mun84]{munkres-book}
James~R Munkres.
\newblock {\em Elements of algebraic topology}, volume~7.
\newblock Addison-Wesley Reading, 1984.

\bibitem[Mun00]{munkres-top}
James~R Munkres.
\newblock {\em Topology}.
\newblock Prentice Hall, 2000.

\bibitem[New10]{newman2010networks}
Mark Newman.
\newblock {\em Networks: an introduction}.
\newblock Oxford University Press, 2010.

\bibitem[New18a]{newman2018estimating}
Mark~EJ Newman.
\newblock Estimating network structure from unreliable measurements.
\newblock {\em Physical Review E}, 98(6):062321, 2018.

\bibitem[New18b]{newman2018network}
MEJ Newman.
\newblock Network structure from rich but noisy data.
\newblock {\em Nature Physics}, 14(6):542--545, 2018.

\bibitem[Nie27]{niemytzki1927third}
VW~Niemytzki.
\newblock On the ``third axiom of metric space".
\newblock {\em Transactions of the American Mathematical Society},
  29(3):507--513, 1927.

\bibitem[OD71]{o1971hippocampus}
John O'{K}eefe and Jonathan Dostrovsky.
\newblock The hippocampus as a spatial map. preliminary evidence from unit
  activity in the freely-moving rat.
\newblock {\em Brain research}, 34(1):171--175, 1971.

\bibitem[Oht12]{ohta2012barycenters}
Shin-ichi Ohta.
\newblock Barycenters in alexandrov spaces of curvature bounded below.
\newblock {\em Advances in geometry}, 12(4):571--587, 2012.

\bibitem[OLP19]{oles2019efficient}
Vladyslav Oles, Nathan Lemons, and Alexander Panchenko.
\newblock Efficient estimation of a {G}romov--{H}ausdorff distance between
  unweighted graphs.
\newblock {\em arXiv preprint arXiv:1909.09772}, 2019.

\bibitem[ON78]{o1978hippocampus}
John O'{K}eefe and Lynn Nadel.
\newblock {\em The hippocampus as a cognitive map}, volume~3.
\newblock Clarendon Press Oxford, 1978.

\bibitem[PC18]{pitcher1918foundations}
Arthur~Dunn Pitcher and Edward~Wilson Chittenden.
\newblock On the foundations of the calcul fonctionnel of {F}r{\'e}chet.
\newblock {\em Transactions of the American Mathematical Society},
  19(1):66--78, 1918.

\bibitem[PCS16]{pcs16}
Gabriel Peyr{\'e}, Marco Cuturi, and Justin Solomon.
\newblock Gromov-{W}asserstein averaging of kernel and distance matrices.
\newblock In {\em International Conference on Machine Learning}, pages
  2664--2672, 2016.

\bibitem[Pet06]{petersen2006riemannian}
Peter Petersen.
\newblock {\em Riemannian geometry}, volume 171.
\newblock Springer Science \& Business Media, 2006.

\bibitem[PFA06]{pennec2006riemannian}
Xavier Pennec, Pierre Fillard, and Nicholas Ayache.
\newblock A {R}iemannian framework for tensor computing.
\newblock {\em International Journal of computer vision}, 66(1):41--66, 2006.

\bibitem[PHC{\etalchar{+}}20]{perez2020quantifying}
Leron Perez, Kabir Husain, Samir Chowdhury, Benjamin Schweinhart, Vahe
  Galstyan, Pankaj Mehta, Nikta Fakhri, and Arvind Murugan.
\newblock Quantifying scale-dependent irreversibility using persistent
  homology.
\newblock {\em Bulletin of the American Physical Society}, 65, 2020.

\bibitem[Pin20]{pinto2020motivic}
Guilherme Vituri~Fernandes Pinto.
\newblock Motivic constructions on graphs and networks with stability results.
\newblock 2020.

\bibitem[PJM11]{perrault2011directed}
Dominique~C Perrault-Joncas and Marina Meila.
\newblock Directed graph embedding: an algorithm based on continuous limits of
  laplacian-type operators.
\newblock In {\em Advances in Neural Information Processing Systems}, pages
  990--998, 2011.

\bibitem[PW94]{qap-book1}
Panos~M Pardalos and Henry Wolkowicz, editors.
\newblock {\em Quadratic assignment and related problems}.
\newblock DIMACS Series in Discrete Mathematics and Theoretical Computer
  Science, 16. American Mathematical Society, Providence, RI, 1994.

\bibitem[Rob99]{robins1999towards}
Vanessa Robins.
\newblock Towards computing homology from finite approximations.
\newblock In {\em Topology proceedings}, volume~24, pages 503--532, 1999.

\bibitem[Sch15]{schmiedl}
Felix Schmiedl.
\newblock {\em Shape matching and mesh segmentation: mathematical analysis,
  algorithms and an application in automated manufacturing}.
\newblock PhD thesis, M{\"u}nchen, Technische Universit{\"a}t M{\"u}nchen,
  Diss., 2015, 2015.

\bibitem[Sch17]{schmiedl2017computational}
Felix Schmiedl.
\newblock Computational aspects of the {G}romov--{H}ausdorff distance and its
  application in non-rigid shape matching.
\newblock {\em Discrete \& Computational Geometry}, 57(4):854--880, 2017.

\bibitem[SCM16]{nets-asilo}
Zane Smith, Samir Chowdhury, and Facundo M{\'e}moli.
\newblock Hierarchical representations of network data with optimal distortion
  bounds.
\newblock In {\em 2016 50th Asilomar Conference on Signals, Systems and
  Computers}, pages 1834--1838. IEEE, 2016.

\bibitem[SMC07]{mapper}
Gurjeet Singh, Facundo M{\'e}moli, and Gunnar Carlsson.
\newblock Topological methods for the analysis of high dimensional data sets
  and 3d object recognition.
\newblock In {\em Symposium on Point-Based Graphics}, pages 91--100, 2007.

\bibitem[SPKS16]{s16}
Justin Solomon, Gabriel Peyr{\'e}, Vladimir~G Kim, and Suvrit Sra.
\newblock Entropic metric alignment for correspondence problems.
\newblock {\em ACM Transactions on Graphics (TOG)}, 35(4):72, 2016.

\bibitem[Spo11]{sporns2011networks}
Olaf Sporns.
\newblock {\em Networks of the Brain}.
\newblock MIT press, 2011.

\bibitem[Spo12]{sporns2012discovering}
Olaf Sporns.
\newblock {\em Discovering the human connectome}.
\newblock MIT press, 2012.

\bibitem[SS78]{counterexamples}
Lynn~Arthur Steen and J~Arthur Seebach.
\newblock {\em Counterexamples in topology}, volume~18.
\newblock Springer, 1978.

\bibitem[SS03]{phylo}
Charles Semple and Mike~A Steel.
\newblock {\em Phylogenetics}, volume~24.
\newblock Oxford University Press on Demand, 2003.

\bibitem[SSS13]{sabau2013metric}
Sorin~V Sabau, Kazuhiro Shibuya, and Hideo Shimada.
\newblock Metric structures associated to {F}insler metrics.
\newblock {\em arXiv preprint arXiv:1305.5880}, 2013.

\bibitem[Stu12]{sturm2012space}
Karl-Theodor Sturm.
\newblock The space of spaces: curvature bounds and gradient flows on the space
  of metric measure spaces.
\newblock {\em arXiv preprint arXiv:1208.0434}, 2012.

\bibitem[SY09]{stojmirovic2009geometric}
Aleksandar Stojmirovi{\'c} and Yi-Kuo Yu.
\newblock Geometric aspects of biological sequence comparison.
\newblock {\em Journal of Computational Biology}, 16(4):579--610, 2009.

\bibitem[SZ10]{shen2010gromov}
Yi-Bing Shen and Wei Zhao.
\newblock Gromov pre-compactness theorems for nonreversible {F}insler
  manifolds.
\newblock {\em Differential Geometry and its Applications}, 28(5):565--581,
  2010.

\bibitem[TG82]{tversky1982similarity}
Amos Tversky and Itamar Gati.
\newblock Similarity, separability, and the triangle inequality.
\newblock {\em Psychological review}, 89(2):123, 1982.

\bibitem[Tur19]{turner2019rips}
Katharine Turner.
\newblock Rips filtrations for quasimetric spaces and asymmetric functions with
  stability results.
\newblock {\em Algebraic \& Geometric Topology}, 19(3):1135--1170, 2019.

\bibitem[VJP20]{venkatesh2020comparing}
Manasij Venkatesh, Joseph Jaja, and Luiz Pessoa.
\newblock Comparing functional connectivity matrices: A geometry-aware approach
  applied to participant identification.
\newblock {\em NeuroImage}, 207:116398, 2020.

\bibitem[Was13]{waszkiewicz2013local}
Pawel Waszkiewicz.
\newblock The local triangle axiom in topology and domain theory.
\newblock {\em Applied General Topology}, 4(1):47--70, 2013.

\bibitem[XLC19]{xu2019scalable}
Hongteng Xu, Dixin Luo, and Lawrence Carin.
\newblock Scalable {G}romov-{W}asserstein learning for graph partitioning and
  matching.
\newblock In {\em Advances in neural information processing systems}, pages
  3052--3062, 2019.

\bibitem[XLZD19]{xu2019gromov}
Hongteng Xu, Dixin Luo, Hongyuan Zha, and Lawrence~Carin Duke.
\newblock Gromov-wasserstein learning for graph matching and node embedding.
\newblock In {\em International conference on machine learning}, pages
  6932--6941. PMLR, 2019.

\bibitem[Xu20]{xu2020gromov}
Hongtengl Xu.
\newblock Gromov-wasserstein factorization models for graph clustering.
\newblock In {\em Proceedings of the AAAI Conference on Artificial
  Intelligence}, volume~34, pages 6478--6485, 2020.

\bibitem[Zau59]{zaustinsky1959spaces}
Eugene~M Zaustinsky.
\newblock {\em Spaces with non-symmetric distance}, volume~34.
\newblock American Mathematical Society, 1959.

\bibitem[ZHS06]{zhou2006learning}
Dengyong Zhou, Jiayuan Huang, and Bernhard Sch{\"o}lkopf.
\newblock Learning with hypergraphs: Clustering, classification, and embedding.
\newblock {\em Advances in neural information processing systems},
  19:1601--1608, 2006.

\end{thebibliography}

\appendix
\section{Experiments}
\label{sec:experiments}

\subsection{Computational aspects and an algorithm for computing minimum matchings}
\label{sec:computations}
In this section we first discuss some algorithmic details on how to compute the lower bounds for $\dn$ involving local spectra, and then present computational examples. All networks in this section are assumed to be finite. Our software and datasets are available on \url{https://github.com/fmemoli/PersNet} as part of the \texttt{PersNet} software package.

Lower bounds for $\dn$ involving the comparison of local spectra of two networks such as those in Proposition \ref{prop:spec} require  computing the minimum of a functional $J(R):=\max_{(z,y)\in R} C(z,y)$ where $C:Z\times Y\rightarrow \R_+$ is a given \emph{cost} function and $R$ ranges in $\Rsc(Z,Y)$. This is an instance of a bottleneck linear assignment problem (or LBAP) \cite{burkard2009assignment}. We remark that the current instance differs from the standard formulation in that one is now optimizing over correspondences and not over permutations. Hence the standard algorithms need to be modified.

Assume $n=|Z|$ and $m=|Y|$. In this section we adopt matrix notation and regard $R$ as a matrix $\matele{r_{i,j}}\in\{0,1\}^{n\times m}$. The condition $R\in\Rsc(Z,Y)$ then requires that $\sum_i r_{i,j}\geq 1$ for all $j$ and $\sum_{j}r_{i,j}\geq 1$ for all $i$. 
We denote by $C=\matele{c_{i,j}}\in\R_+^{n\times m}$ the matrix representation of the cost function $C$ described above. With the goal of identifying a suitable algorithm, the key observation is that the optimal value $\min_{R\in \Rsc}J(R)$ must coincide with a value realized in the matrix $C$.

An algorithm with complexity $O(n^2\times m^2)$ is the one in Algorithm \ref{algo:match} (we give it in Matlab pseudo-code). The algorithm belongs to the family of \emph{thresholding algorithms} for solving matching problems over permutations, see \cite{burkard2009assignment}.  Notice that $R$ is a binary matrix and that procedure \textbf{TestCorrespondence} has complexity $O(n\times m)$. In the worst case, the matrix $C$ has $n\times m$ distinct entries, and the \textbf{while} loop will need to exhaustively test them all, hence the claimed complexity of $O(n^2\times m^2).$ Even though a more efficient version (with complexity $O((n\times m) \log(n\times m))$ can be obtained by using a bisection strategy on the range of possible values contained in the matrix $C$ (in a manner similar to what is described for the case of permutations in \cite{burkard2009assignment}), here for clarity we limit our presentation to the version detailed above. 

\begin{algorithm}
\caption{MinMax matching}\label{matching}
\begin{algorithmic}[1]
\Procedure{MinMaxMatch}{C}
\State $v = \textbf{sort}(\textbf{unique}(C(:)));$
\State $k=1;$
\While {$\sim\!\text{done}$} 
\State $c=v(k);$
\State $R = (C<=c)$;
\State $\text{done} = \Call{TestCorrespondence}{R};$
\State $k=k+1;$
\EndWhile
\State\Return $c$
\EndProcedure \label{nada}
\Procedure{TestCorrespondence}{R}
\State done = \textbf{prod}(\textbf{sum}(R))*\textbf{prod}(\textbf{sum}(R')) $>$ 0;
\State \Return done
\EndProcedure
\end{algorithmic}
\label{algo:match}
\end{algorithm}

\subsection{Computational example: randomly generated networks}

\begin{figure}
    \centering
    \includegraphics[width=0.45\textwidth]{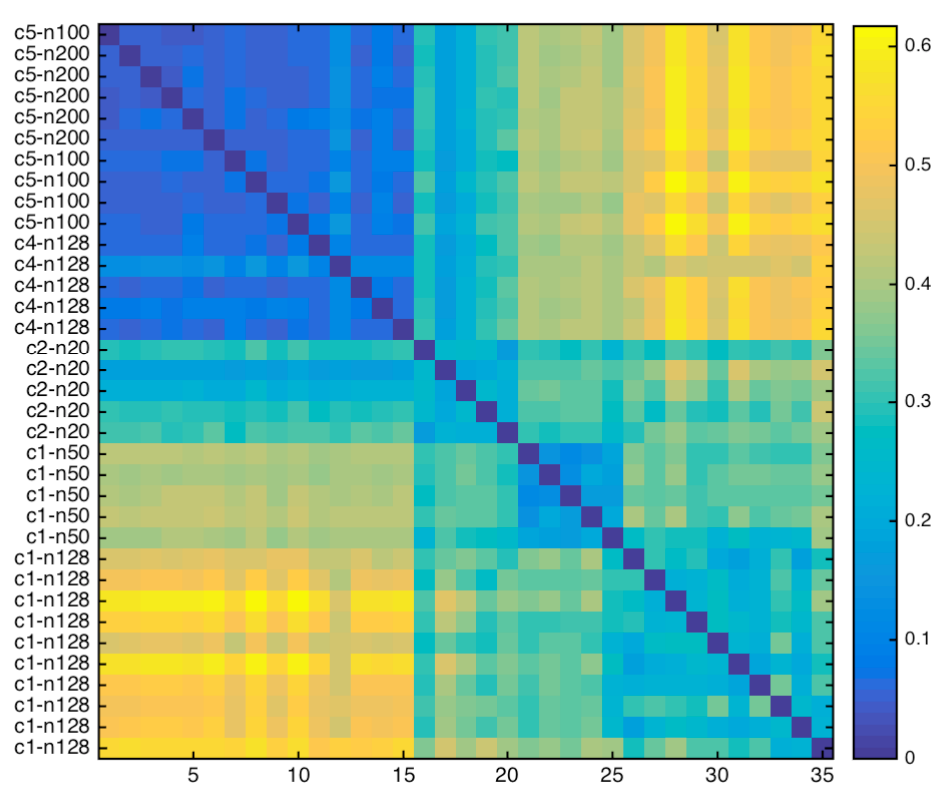}
    \includegraphics[width=0.45\textwidth]{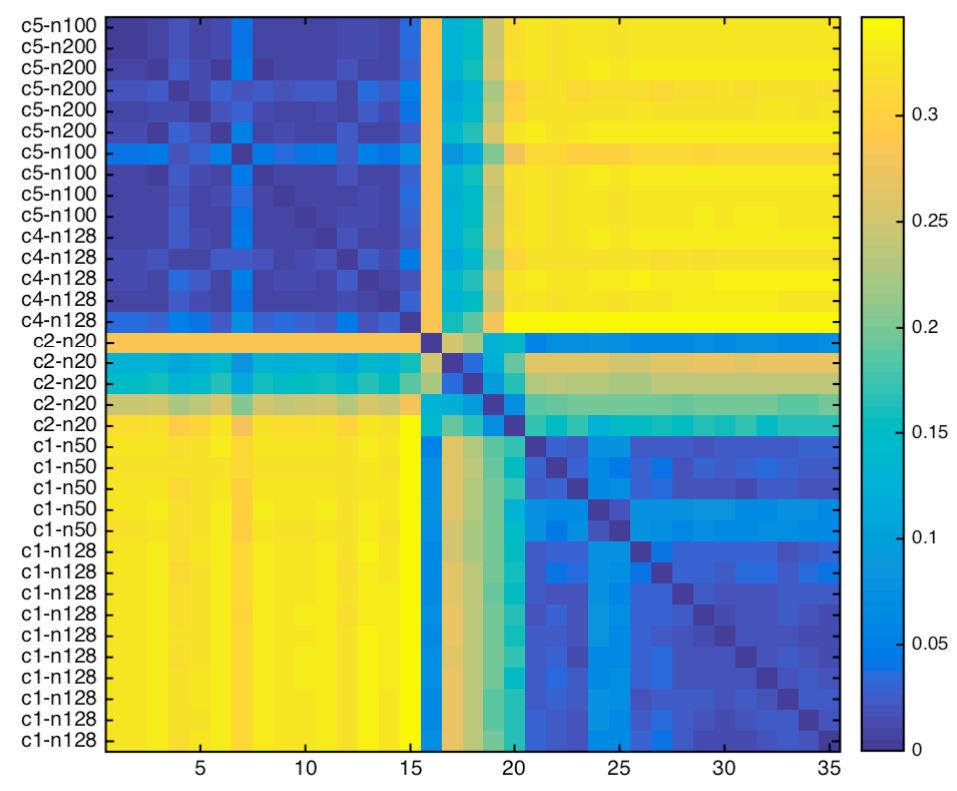}
    \includegraphics[width=0.45\textwidth]{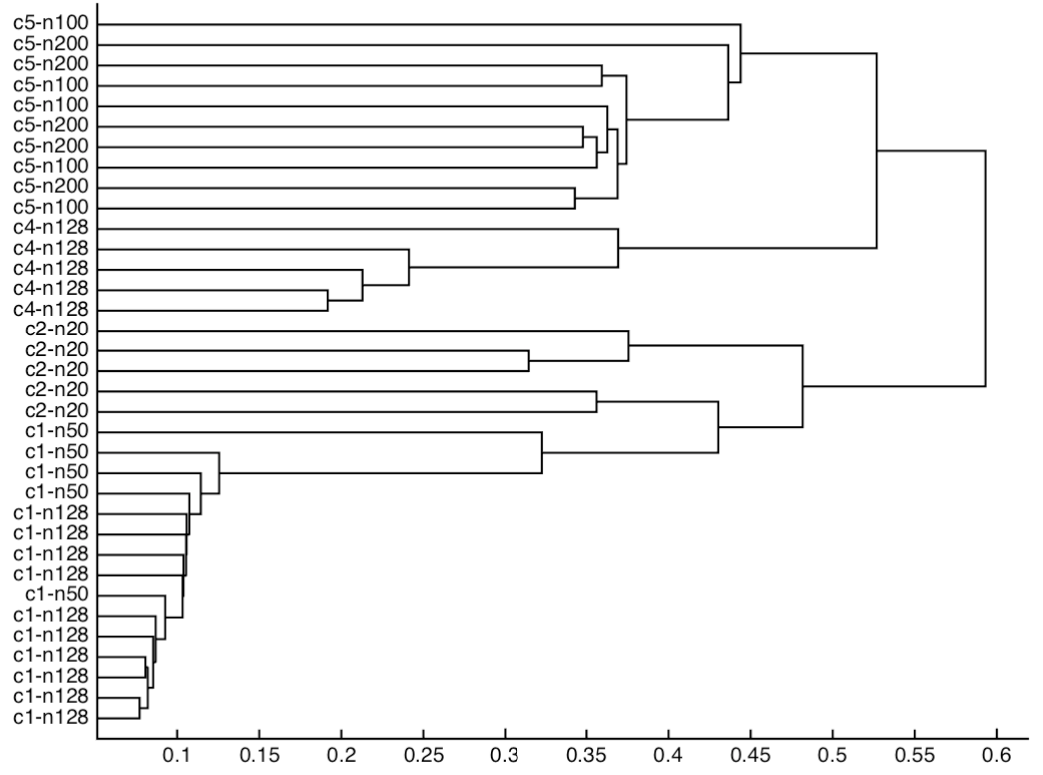}
    \includegraphics[width=0.45\textwidth]{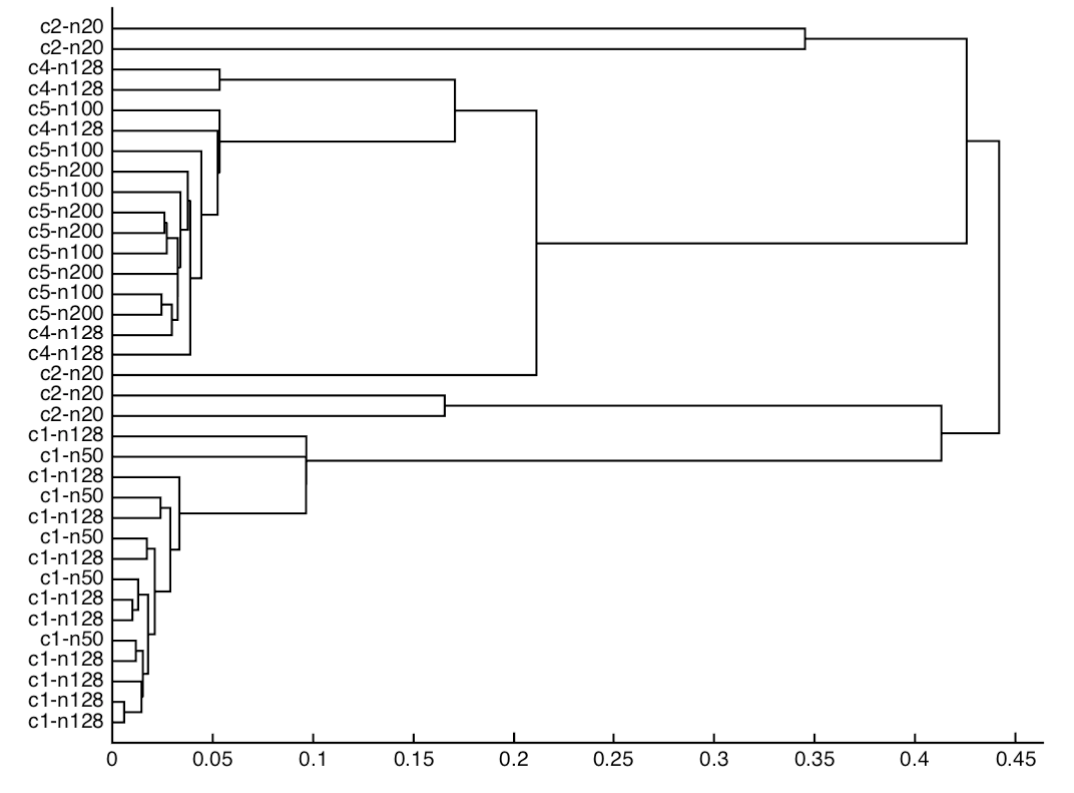}
    \caption{Performance of local (\textbf{left}) and global (\textbf{right}) spectrum lower bounds on the random community graph classification task. The global spectrum can reasonably tell apart classes of graphs with 4-5 communities from those with a single community, but is confounded by graphs with two communities. In comparison, the local spectrum is able to better distinguish the two-community graphs. }
    \label{fig:exp-comm}
\end{figure}

As a first application of our ideas we generated a database of weighted directed networks with different numbers of ``communities'' and different total cardinalities using the software provided by \cite{santo}. Using this software, we generated 35 random networks as follows: 5 networks with 5 communities and 200 nodes each (class \texttt{c5-n200}), 5 networks with 5 communities and 100 nodes each (class \texttt{c5-n100}), 5 networks with 4 communities and 128 nodes each (class \texttt{c4-n128}), 5 networks with 2 communities and 20 nodes each (class \texttt{c2-n20}), 5 networks with 1 community and 50 nodes each (class \texttt{c1-n50}), and 10 networks with  1 community and 128 nodes each (class \texttt{c1-n128}). 
In order to make the comparison more realistic, as a preprocessing step we divided all the weights in each network by the diameter of the network. In this manner, discriminating between networks requires differentiating their structure and not just the scale of the weights. Note that the (random) weights produced by the software \cite{santo} are all non-negative.

Using a Matlab implementation of Algorithm \ref{algo:match} we computed a 35 $\times $ 35 matrix of values corresponding to a lower bound based simultaneously on both the $\cin$ and $\cout$ local spectra. This strengthening of Proposition \ref{prop:spec} is stated below.

\begin{proposition}\label{prop:loc-spec-strong} For all $X,Y\in\Ncal$,
\begin{align*}
d_{\Ngen}(X,Y) &\geq \tfrac{1}{2}\min_{R\in\Rsc}\max_{(x,y)\in R}C(x,y), \text{ where }\\
C(x,y) &=\max\big(\dhdf^\R(\spec_X^{\cin}(x),\spec_Y^{\cin}(y)), \dhdf^\R(\spec_X^{\cout}(x),\spec_Y^{\cout}(y)) \big).
\end{align*}

\end{proposition}
This bound follows from Proposition \ref{prop:spec} by the discussion at the beginning of \S\ref{sec:quant-stab}.

The results are shown in the form of the lower bound matrix and its single linkage dendrogram in Figure \ref{fig:exp-comm}. Notice that the labels in the dendrogram permit ascertaining the quality of the classification provided by the local spectra bound.  With only very few exceptions, networks with similar structure (same number of communities) were clustered together regardless of their cardinality. Notice furthermore how networks with 4 and 5 communities merge together before merging with networks with 1 and 2 communities, and vice versa. For comparison, we also provide details about the performance of the global spectra lower bound on the same database in Figure \ref{fig:exp-comm}. The results are clearly inferior to those produced by the local version, as predicted by the inequality in Proposition \ref{prop:spec}.

\subsection{Computational example: simulated hippocampal networks}

A natural observation about humans is that as they navigate an environment, they produce ``mental maps" which enable them to recall the features of the environment at a later time. This is also true for other animals with higher cognitive function. In the neuroscience literature, it is accepted that the hippocampus in an animal's brain is responsible for producing a mental map of its physical environment \cite{best2001spatial, brown1998statistical}. More specifically, it has been shown that as a rat explores a given environment, specific physical regions (``place fields") become linked to specific ``place cells" in the hippocampus \cite{o1971hippocampus, o1978hippocampus}. Each place cell shows a spike in neuronal activity when the rat enters its place field, accompanied by a drop in activity as the rat goes elsewhere. In order to understand how the brain processes this data, a natural question to ask is the following: Is the time series data of the place cell activity, often referred to as ``spike trains", enough to recover the structure of the environment? 

Approaches based on homology \cite{curto2008cell} and persistent homology \cite{dabaghian2012topological} have shown that the preceding question admits a positive answer. We were interested in determining if, instead of computing homology groups, we could represent the time series data as networks, and then apply our invariants to distinguish between different environments. Our preliminary results on simulated hippocampal data indicate that such may be the case.  

In our experiment, there were two environments: (1) a square of side length $L$, and (2) a square of side length $L$, with a disk of radius $0.33L$ removed from the center. In what follows, we refer to the environments of the second type as \emph{1-hole environments}, and those of the first type as \emph{0-hole environments}. For each environment, a random-walk trajectory of 5000 steps was generated, where the agent could move above, below, left, or right with equal probability. If one or more of these moves took the agent outside the environment (a disallowed move), then the probabilities were redistributed uniformly among the allowed moves. The length of each step in the trajectory was $0.1L$. 

In the first set of 20 trials for each environment, 200 place fields of radius $0.1L$ were scattered uniformly at random. In the next two sets, the place field radii were changed to $0.2L$ and $0.05L$. This produced a total of 60 trials for each environment. For each trial, the corresponding network $(X,\w_X)$ was constructed as follows: $X$ consisted of 200 place cells, and for each $1\leq i,j\leq 200$, the weight $\w_X(x_i,x_j)$ was given by:
\[\w_X(x_i,x_j) =1-\frac{\text{\# times cell $x_j$ spiked in a window of five time units after cell $x_i$ spiked}}{\text{\# times cell $x_j$ spiked}}.\]

The results of applying the local spectra lower bound are shown in the inset of Figure \ref{fig:dendro-all}. The labels \texttt{env-0, env-1} correspond to $0$ and $1$-hole environments, respectively. Note that with some exceptions, networks corresponding to the same environment are clustered together, regardless of place field radius. In Figure \ref{fig:dendro-all} we also present the single linkage dendrogram obtained from comparing all 120 networks together. In light of these results, we are interested in seeing how these methods can applied to other time series data arising from biology.

As a final remark, we note that it is possible to obtain better clustering on the hippocampal network dataset by using $\dn$-invariants that arise from \emph{persistent homology}. We refer the reader to \cite{dowker-jact} for details.

\begin{figure}
    \centering
    \includegraphics[width=\textwidth]{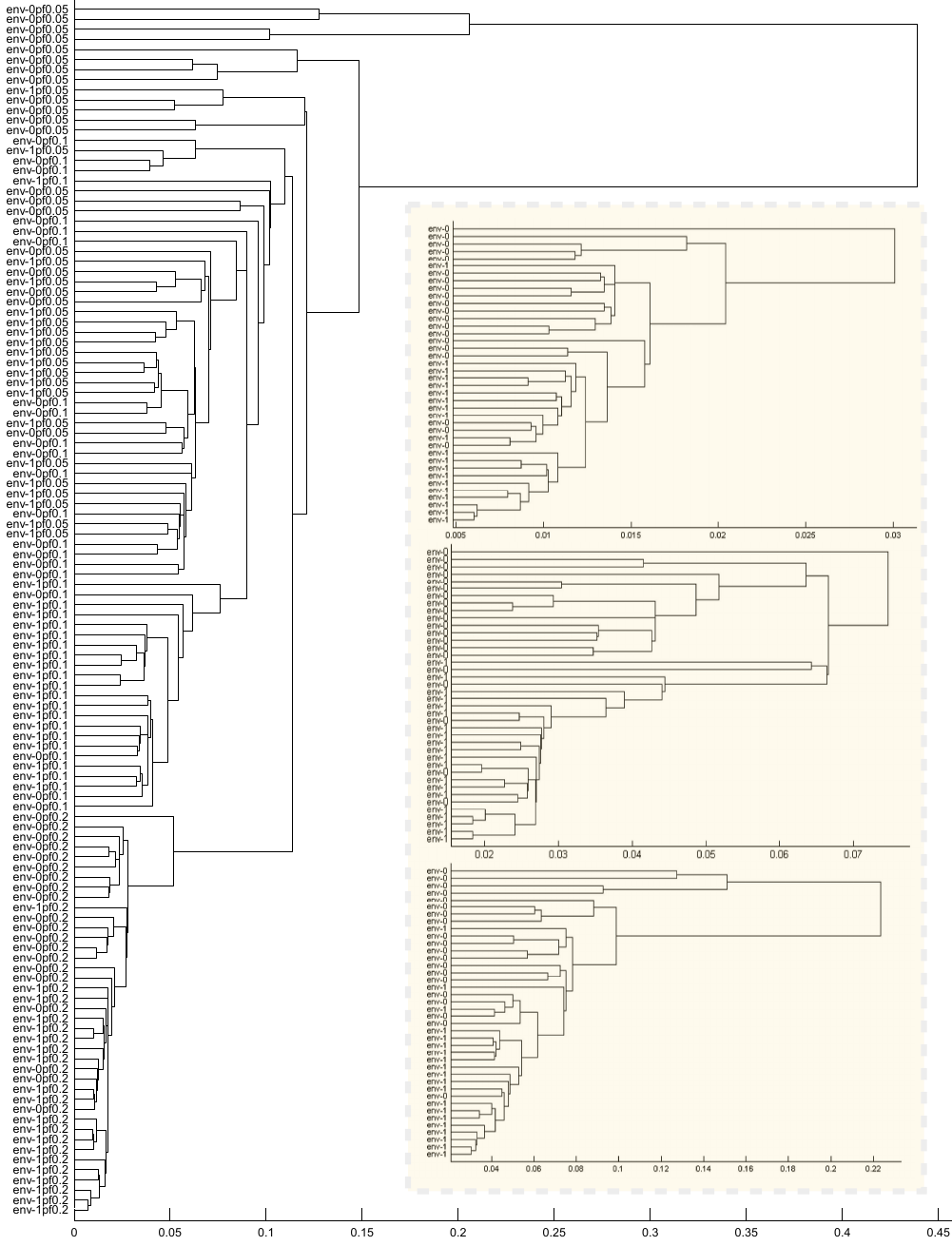}
    \caption{Single linkage dendrogram corresponding to 120 hippocampal networks of place field radii $0.05L, 0.1L,$ and $0.2L$. Results are based on the local spectrum lower bound of Proposition \ref{prop:loc-spec-strong}. \textbf{Inset, top to bottom:} Individual dendrograms for place field radii of $0.2L, 0.1L, 0.05L$, respectively.}
    \label{fig:dendro-all}
\end{figure}

\glsaddallunused

\end{document}